%% file: main_revise_arxiv.tex
\documentclass[prd,aps,amsfonts,showpacs,longbibliography,notitlepage,nofootinbib,superscriptaddress]{revtex4-1}
\usepackage[dvipsnames]{xcolor}
\usepackage{bm}
\usepackage{bbm}

\include{preamble}

\newcommand{\where}{\quad \mathrm{where}\quad}
\newcommand{\kpsi}{\ket{\psi_0}}
\newcommand{\tkpsi}{\ket{\widetilde{\psi}_0}}
\newcommand{\diam}{\mathrm{diam}}
\newcommand{\kzero}{\ket{\bm{0}}}
\newcommand{\Pzero}[1]{\ket{\bm{0}}_{#1} \bra{\bm{0}}}
\newcommand{\ball}{{\mathcal{B}}}

\newcommand{\bs}{\mathbf{s}}

\newcommand{\revise}[1]{{\color{black}#1}}
\newcommand{\comment}[1]{}

\begin{document}
\title{Theory of metastable states in many-body quantum systems}

\author{Chao Yin}
\email{chao.yin@colorado.edu}
\affiliation{Department of Physics and Center for Theory of Quantum Matter, University of Colorado, Boulder, CO 80309, USA}

\author{Federica M. Surace}
\email{fsurace@caltech.edu}
\affiliation{Department of Physics and Institute for Quantum Information and Matter,
California Institute of Technology, Pasadena, CA 91125, USA}

\author{Andrew Lucas}
\email{andrew.j.lucas@colorado.edu}
\affiliation{Department of Physics and Center for Theory of Quantum Matter, University of Colorado, Boulder, CO 80309, USA}

\date{January 21, 2025}

\begin{abstract}
   We present a mathematical theory of metastable pure states in closed many-body quantum systems with finite-dimensional Hilbert space.  Given a Hamiltonian, a pure state is defined to be metastable when all sufficiently local operators either stabilize the state, or raise its average energy. We prove that short-range entangled metastable states are necessarily eigenstates (scars) of a perturbatively close Hamiltonian.  Given any metastable eigenstate of a Hamiltonian, in the presence of perturbations, we prove the presence of prethermal behavior: local correlation functions decay at a rate bounded by a time scale nonperturbatively long in the inverse metastability radius, rather than Fermi's Golden Rule. Inspired by this general theory, we prove that the lifetime of the false vacuum in certain $d$-dimensional quantum models grows at least as fast as  $\exp(\epsilon^{-d})$, where $\epsilon\rightarrow 0$ is the relative energy density of the false vacuum; this lower bound matches, for the first time, explicit calculations using quantum field theory. We identify metastable states at finite energy density in the PXP model, along with exponentially many metastable states in ``helical" spin chains and the two-dimensional Ising model.  Our inherently quantum formalism reveals precise connections between many problems, including prethermalization, robust quantum scars, and quantum nucleation theory, and applies to systems without known semiclassical and/or field theoretic limits.
\end{abstract}

\maketitle
\tableofcontents

\section{Introduction}
The notion of metastability has a long and rich history in physics.  Perhaps the most famous illustration of metastability is the abrupt freezing of supercooled water at temperature $T<T_{\mathrm{c}}$ (the freezing temperature), which occurs only after the water is shaken to induce ice crystal nucleation. This phenomenon is not only easy to demonstrate experimentally, but is also well-understood through the framework of nucleation theory in classical statistical physics \cite{langer}.  The classical theory of metastability can be formulated more precisely using the framework of Markov chains \cite{cassandro}, where one shows that the ``mixing time needed to escape" the metastable part of state space is non-perturbatively large in the strength of a small parameter, such as a small free-energy density difference between the two phases.

Metastability is also an extremely important phenomenon in \emph{quantum systems}.  Radioactive decay is a property of metastable atomic isotopes decaying to a lower energy state by emitting energetic particles \cite{tunnel_wavef28}.  Such metastability amounts to a \emph{few-body} problem, which has been carefully analyzed for many decades and which (in its simplest realization) is textbook quantum mechanics.  There are also intrinsically many-body metastable quantum states.
For example, evidence of metastable unconventional superconducting states or other types of ``hidden'' quantum states have been observed in various materials after the excitation with strong laser pulses \cite{Nieva1992,Stojchevska,budden2021evidence}.
Metastability is also crucial in fields beyond condensed matter physics. In cosmology, a famous field-theoretic calculation proposed that the universe may be in a metastable state -- a \emph{false vacuum} \cite{coleman,FVD_rev22}.  Extrapolations based on current estimates of the parameters of the Standard Model suggests that the electroweak vacuum might be metastable \cite{SHER1989273,wilczek,ELIASMIRO2012222}.
 
 The simplest manifestation of this false vacuum arises in a quantum ferromagnet, with $\mathbb{Z}_2$ symmetry breaking, placed in a small longitudinal field of strength $\epsilon$ (which picks out one of the two degenerate ground states as lower energy).  One expects that ``the lifetime" of the false vacuum -- the original ground state which is now anti-aligned with the longitudinal field -- is non-perturbatively large: $t\sim \exp[C\epsilon^{-d}]$, for some constant $C$ \cite{falseVac_Ising,falseVac_spinchain}.  

Unlike in the classical models, we \emph{do not} know of a mathematically precise theory of metastability in many-body quantum systems.  This does not mean that useful calculations have not been done -- for example, very famous calculations have estimated the lifetime of false vacua in quantum field theory \cite{coleman}.  But these calculations are \emph{not well-posed mathematically}  -- as physicists, we know ``what" is being calculated, but perhaps only have a handwavy argument for ``why" the calculation works in the first place! Careful attempts to even give a ``physicist's derivation" -- within a nonrigorous  path integral formalism -- of quantum metastability are quite recent \cite{Andreassen:2016cff}. 

In this paper, we aim at formulating a rigorous theory of metastability for quantum many-body systems. We will essentially define a metastable pure state $|\psi\rangle$ in a many-body quantum system as one in which $\mathcal{O}|\psi\rangle - \langle \psi|\mathcal{O}|\psi\rangle |\psi\rangle $ has higher energy than $|\psi\rangle$, for arbitrary operators $\mathcal{O}$ that act on $R$ or fewer degrees of freedom, for some ``large" number $R$.  This definition is (hopefully) intuitive, and appears closely related to the definition used in recent work \cite{Chen:2023idc}.  Given this simple definition, we are able to prove that given a spatially local Hamiltonian in a finite number of spatial dimensions, metastable states have non-perturbatively long lifetimes scaling at least as fast as $\mathrm{e}^{R^\alpha}$ for any $\alpha \in (0,\frac{1}{2})$, as measured by local correlation functions.  As in a recent paper by two of us \cite{our_preth}, we will be able to give a precise definition for the lifetime of a metastable state, and explain why our definition is relevant to experimental and measurable physics.

The remainder of this introduction serves primarily to give a slightly more quantitative historical background to previous theories of metastability, as well as to outline the remainder of the paper.

\subsection{Few-body metastability}
To contrast the hardness of defining a theory of many-body quantum metastability, let us first review the much simpler few-body problem.

We first discuss classical physics (with stochastic fluctuations).  Consider -- just for pedagogical purposes -- a one-dimensional particle moving along the line and undergoing overdamped dynamics.  Suppose that the system is in thermal equilibrium, such that the (unnormalized) probability density function in equilibrium is \begin{equation}
    P_{\mathrm{stationary}}(x) = \exp[-U(x)/T],
\end{equation}
where $U(x)$ is the potential energy of the particle.   Under Markovian dynamics with Gaussian noise, the dynamics of the probability density $P(x,t)$ can be described by the Fokker-Planck equation: \cite{Fokker,Planck} \begin{equation}
    \partial_t P = \partial_x \left(\gamma(x) \left(T\partial_x + U(x)\right) P\right).
\end{equation}
One can show explicitly that if $U(x)$ has a deep local minimum, as depicted in Figure~\ref{fig:summary}(c),  then so long as $\gamma(x)$ varies sufficiently slowly, the particle is trapped in this metastable region for a time \cite{gardiner} \begin{equation}
    t_{\mathrm{c}} \sim \exp[c_{\mathrm{c}}\Delta U/T],
\end{equation}
where $\Delta U$ is the energy barrier and $c_{\mathrm{c}}$ is an O(1) constant.  This is the textbook Arrhenius law.

In quantum mechanics, \emph{tunneling} plays the role of stochastic fluctuations in a closed system (which does not decohere with its environment).  Given a single quantum particle of mass $m$ obeying the nonrelativistic Schr\"odinger equation, we can estimate the rate at which it tunnels out of the barrier using the WKB approximation:  \begin{equation}
    t_{\mathrm{q}} \sim \exp[ c_{\mathrm{q}}\sqrt{m \Delta U}R_{\mathrm{q}}/\hbar  ], \label{eq:WKB}
\end{equation}
where $c_{\mathrm{q}}$ is an O(1) constant and  $R_{\mathrm{q}}$ is the length scale of the barrier.  The physical interpretation given for $t_{\mathrm{q}}$ is that it is the time needed for the particle to escape the metastable region $R$: \begin{equation} \label{eq:tqinterp}
    \int\limits_R \mathrm{d}x |\psi(x,t)|^2 \sim \mathrm{e}^{-t/t_{\mathrm{q}}}.
\end{equation}
There are two related ways\footnote{The reader may instead be familiar with the \emph{even more heuristic} textbook argument whereby one calculates the WKB wave function, and estimates the ``bouncing time" between the walls of the classically bounding region.  This method is difficult to generalize to higher dimensions \cite{WKB_highd73}, nor does it give an intrinsically quantum mechanical justification for the prescription.} that such a calculation is usually done in the literature.  The first is to solve the Schr\"odinger equation subject to \emph{outgoing} (radiation) boundary conditions on the wave function \cite{tunnel_wavef28,tunnel_wavef39}.   This leads to a complex eigenvalue \begin{equation} \label{eq:complexE}
    E = E_0 - \mathrm{i}\frac{\hbar}{2t_{\mathrm{q}}},
\end{equation}
 which is aligned with the interpretation (\ref{eq:tqinterp}) that $t_{\mathrm{q}}$ is the decay time of the metastable state.  Such a complex eigenvalue is problematic: it renders the quantum dynamics of the system non-unitary, in direct contradiction to the starting assumption that the quantum system was isolated!  Related to the non-unitarity, the resulting eigenfunctions are not normalizable -- they exponentially diverge at large distances.   Despite the lack of mathematical rigor, the predictions of this approach are physically sound and have been compared to experiments, e.g. on the decay of radioactive nuclei \cite{tunnel_wavef28}.  The second approach, which is more common in modern calculations, is to perform a path integral calculation in imaginary time \cite{coleman2,Andreassen:2016cff}; the resummation of instanton contributions corresponding to particles escaping the metastable state also leads to (\ref{eq:complexE}) and the same physical interpretation for $t_{\mathrm{q}}$.

We remark that the prediction (\ref{eq:WKB}) makes sense only in a \emph{semiclassical} limit.  If the metastable region $R$ does not admit a large number of bound states (assuming that the potential $U(x)$ continued to increase outside $R$), then the approximation that $E\ll \Delta U$ -- required to justify the WKB prediction (\ref{eq:WKB}) -- fails.


\subsection{False vacuum: a many-body metastable state}
Let us now turn to the many-body theory of metastability.  Again, we begin with the classical setting, and focus on an illustrative example: the $d$-dimensional Ising model with $d>1$ in a longitudinal field, with Hamiltonian \begin{equation}\label{eq:H=ZZ-Z}
    H = -\Delta \sum_{i\sim j} Z_i Z_j - \epsilon \sum_i Z_i
\end{equation}
where $i\sim j$ represents nearest-neighbor pairs in $d$ spatial dimensions, and each $Z_i = \pm 1$. Suppose we start in the state $Z_i(t=0) = -1$, and consider temporal dynamics generated by a Gibbs sampler \cite{Metropolis_1953}, which tends towards a very low temperature stationary state $P(Z) \propto \exp[-\beta H(Z)]$.  At low enough temperature, the system is overwhelmingly likely to be found in a high magnetization $\langle Z_i\rangle \approx +1$ state, yet nevertheless one can prove that the time it takes to reach a state in which the average magnetization is positive will grow as \begin{equation}
    t_{\mathrm{c}}^\prime \gtrsim \exp[\beta \Delta \cdot (\Delta/\epsilon)^{d-1}] \label{eq:thit}
\end{equation}
where $\beta$ is the inverse temperature. This notion of hitting time can be bounded using the classical theory of Markov chains \cite{cassandro}, and so we have a precise theory of classical metastability, following \cite{langer}.

The problem above can be made quantum by simply considering $Z$ to be a Pauli operator acting on a qubit (or spin-1/2) on each site of a lattice, and considering Hamiltonian \begin{equation}\label{eq:ZZ-X-Z}
    H = -\Delta \sum_{i\sim j} Z_i Z_j - g \sum_i X_i - \epsilon \sum_i Z_i.
\end{equation}
This time, nearly all calculations are done by the heuristic path integral calculations \cite{coleman,anders} which neatly generalize the few-body case mentioned previously.   In a path integral approach, one can identify the false vacuum as a local minimum of the classical effective potential. In analogy, one could for example consider the (classical) product state  with $Z_i=-1$ as the false vacuum state for the Hamiltonian (\ref{eq:ZZ-X-Z}). 

However, 
the use of the path integral does not have a rigorous foundation,
and perhaps more importantly, when Feynman diagrammatic perturbation theory fails at strong coupling, very basic questions remain mysterious. (\emph{1}) Is there a meaningful notion of a minimum of a semiclassical effective potential, when our framework for strongly coupled (e.g. conformal) quantum field theory is not organized around a Lagrangian, but instead around operator content? In contrast, we expect that metastability is a property of quantum systems that does not rely on any semiclassical approximation: for example, we expect that metastability should in general arise close to any first order quantum phase transition, even when semiclassical approximations fail.  (\emph{2}) What \emph{is} the false vacuum state which has a long lifetime?  In the Ising example above, the product state with $Z_i=-1$ at time $t=0$ will rapidly decohere because of the non-zero transverse field $g$, which does not need to be small.  (\emph{3}) What do we mean by the lifetime of the false vacuum? Taken literally, in perturbative quantum field theory, the probability of remaining trapped in the false vacuum after time $t$ is estimated by path integral to be
\begin{equation}
\label{eq:Pt}
    P_{\mathrm{trapped}} = \exp[-C \Lambda t],
\end{equation}
where $\Lambda$ denotes the spatial volume of the entire system. We might be tempted to define the characteristic time scale of the decay as $(C\Lambda)^{-1}$, but this quantity vanishes in the thermodynamic limit.  This is the effect of the orthogonality catastrophe \cite{anderson}: the false vacuum state $|\psi\rangle$ is not a ``long-lived state" as a wave function: it is a rapidly decohering superposition of many-body eigenstates, so $\langle \psi(t)|\psi\rangle \rightarrow 0$ for any finite $t$, in the thermodynamic limit.  So, we instead interpret the smallness of prefactor $C$ as controlling the lifetime $t_{\mathrm{q}}$ of the false vacuum: \begin{equation}
    C^{-1}\sim t_{\mathrm{q}} \gtrsim \exp[(\Delta/\epsilon)^d]. \label{eq:thitq}
\end{equation} Intuitively, this is reasonable: experiments in physics measure local observables, so we should define the lifetime of the state to be the decay time of local correlation functions in the false vacuum.  But, the decay time of these correlation functions is not directly what the path integral computed!  
Recent works studied the decay of the false vacuum using methods for real-time evolution \cite{falseVac_confine,falseVac_spinchain, Bubble2021,Pomponio22,Milsted22,lagnese2023detecting, FVDTakacs, FVDphi4, Maki23, Batini} (for example, as the evolution after a quantum quench) rather than Euclidean path integrals. The connection of these results with traditional calculations of the decay rates is often unclear.
Clear answers to these questions are now pertinent, as ultracold atom experiments and other quantum simulation platforms are able to probe false vacuum decay directly \cite{ferrari,darbha2024false,vodeb2024stirring}. 




Two of us have recently suggested some precise answers to these questions \cite{our_preth}, in quantum many-body lattice models.  We found an explicit (but not elegant) expression for a particular many-body state, which we conjectured represented the correct nonperturbative definition of the false vacuum.  In this state, we proved that local correlation functions decay on a nonperturbatively slow time scale $t_{\mathrm{q}}^\prime$.  
The rigorous bound found in \cite{our_preth} showed that the decay time obeys \begin{equation}\label{eq:tgap=2d-1}
    t^\prime_{\mathrm{q}} \ge \exp[(\Delta/\epsilon)^{1/(2d-1)}],
\end{equation}
which is nonperturbatively long in $\epsilon$.  Still, the path integral calculation \cite{coleman} suggests a much longer time scale \eqref{eq:thitq} in $d>1$, suggesting 
that the actual physics that protects the false vacuum (in local correlation functions) is not properly captured by the framework of \cite{our_preth} in $d>1$.  This observation led us to conjecture  that a better understanding of what it means for a quantum system to be metastable is necessary, before we can attempt to prove (\ref{eq:thitq}) quantum mechanically. 

\subsection{Slow thermalization, robustness of scars, 
Hilbert space shattering, and beyond}
The fate of the quantum false vacuum has been studied for half a century.  The recent literature has also seen extensive work on slow (or absent) thermalization, which turns out to also be deeply related to the  rigorous theory of quantum metastability.

A metastable state like the one of the quantum Ising model (\ref{eq:ZZ-X-Z}) is a high-energy state, in the sense that it has an extensive energy proportional to $\epsilon \Lambda$ (a finite energy density) above the ground state. The evolution of a state with finite energy density is generally understood within the framework of the eigenstate thermalization hypothesis (ETH) \cite{ETH_91,ETH_94,ETH_rev16,ETH_rev18}, which is the quantum manifestation of the assumptions of statistical mechanics that all correlation functions will eventually appear thermal. In this context, such metastable state is expected to thermalize. Metastability can thus be understood as an example of slow thermalization \cite{ETH_rev18,deroeckprl,abanin2017rigorous,Floq_KS_16,Floq_KS_PRL,Floq_PRB,preth_PRX19,roeck2019weak_driven,preth_Kehmani20,quasiperiodic,our_preth,SuraceMotrunich} (sometimes described as prethermalization, i.e., the thermalization to a different Hamiltonian for transient, but long times). 

A prominent example of slow (or absent) thermalization at high energies is ``scarring" in many-body quantum systems, observed in a Rydberg atom experiment \cite{PXP_exp17} and investigated in a plethora or recent theoretical works \cite{scar_Mori17,PXP_NP18,scar_exact18,scar_exact18_1,PXP_exact19,scar_XY19,scar_decofree24,scar_HSF_rev22,PXP_rev_Papic,PXP_rev_Kehmani}. Scars are eigenstates of the many-body Hamiltonian that violate the ETH. They represent a very small number of states, compared to the exponential number of thermal states that satisfy the ETH. Nevertheless, they can have significant practical consequences: a simple initial state $\ket{\psi}$ can have a large overlap with scar eigenstates, and hence the time-evolved state $\mathrm{e}^{-\mathrm{i}Ht}\ket{\psi}$ will exhibit nonthermal local correlation functions for very long times.


However, the existence of scar states does not truly settle the question of why such slow dynamics ought to be visible experimentally -- scar states are typically fine-tuned, and are believed to decay with Fermi's Golden Rule (or with a parametrically smaller, but still perturbative rate) \cite{LinChandranMotrunich,Surace2021} in the presence of generic perturbations \cite{ETH_rev18}.  
A metastable state is in general not an exact scar (it is not an energy eigenstate), but for sufficiently small system sizes it will have large overlap with a single eigenstate, which may appear as an approximate scar in the many-body spectrum. (This can be understood from the long lifetime, which implies a small energy variance.) In contrast to fine-tuned scar states, a false vacuum is expected to be robust, having a non-perturbative decay rate rather than the perturbative Fermi's Golden Rule rate.

If we wish to prepare a scarlike state in the lab, either as a check on the quantum coherence of a simulator, or for metrological purposes \cite{scar_metro21}, it is desirable to find models where slow thermalization is robust at higher orders in perturbation theory.  One setting where such models can be found is in systems with Hilbert space shattering (also called fragmentation), whereby kinematic (or other) constraints forbid matrix elements between certain ``Krylov sectors" in Hilbert space \cite{shatter1,shatter2,frag_Rahul_Bern,iadecola,scar_HSF_rev22}, which cannot be associated with any local symmetry.  Again, in most models, this Hilbert space shattering is expected to be fragile to generic perturbations, in the sense that if $H_0$ is shattered, $H_0+V$ will not be for a generic small $V$; however, see \cite{loop_David24,loop_Charlie24} for some recent progress towards more stable shattered models.   In the general case, the time scale over which one predicts ``shattering physics" can be detected experimentally is still often quite large: if $V\sim \epsilon$ and $H_0$ has a spectrum with gap $\Delta$, the time  scale over which the shattered dynamics would be detectable in experiment scales as 
\begin{equation} \label{eq:prethermalscaling}
    t_{\mathrm{p}} \gtrsim \exp[\Delta/\epsilon].
\end{equation}

The bound in \eqref{eq:prethermalscaling} can be derived rigorously in more general settings \cite{abanin2017rigorous} (that do not depend in shattering).  Consider a Hamiltonian $H=H_0+V$, where $H_0$ has an integer spectrum and is a sum of commuting terms, and $V$ represents a small perturbation.  An example of physical interest for such a model is the Fermi-Hubbard model at large repulsive interaction strength, $H_0 = \sum_i \Delta c_{i\uparrow}^\dagger c_{i\uparrow} c_{i\downarrow}^\dagger c_{i\downarrow}  $ corresponds to the number of doubly occupied sites, while $V=\epsilon \sum_{i\sim j, \sigma} c_{i\sigma}^\dagger c_{j\sigma}$ corresponds to the hopping of single fermions from one site to the next.  In this model, \eqref{eq:prethermalscaling} manifests in the approximate conservation of the number of doubly occupied sites for $t\lesssim t_{\mathrm{p}}$:  although a doubly occupied site is highly energetic, one must go to very high order $\Delta/\epsilon$ in perturbation theory to find an energetically resonant many-body state with only singly occupied sites.  Mathematically, the methods used to obtain the rigorous bound \eqref{eq:tgap=2d-1} are similar to those used to prove \eqref{eq:prethermalscaling}.  

The intuition that prethermalization arises due to the lack of many-body resonances is quite similar to our previous argument for metastability in classical/quantum Ising models \cite{our_preth}.  To obtain a genuine quantum theory of metastability, however, we must relax (at least, a priori) many of the assumptions listed above.  It is not clear that a metastable Hamiltonian $H$ is close to a solvable Hamiltonian $H_0$ in the sense of $H=H_0+V$, nor that a Hamiltonian $H_0$ even exists for which the metastable $|\psi\rangle$ is an exact eigenstate.  The purpose of this paper is to develop a theory of quantum metastability, starting from rather minimal assumptions.  We will present a precise, and intrinsically quantum, theory of metastable states, and show that they indeed exhibit the expected phenomenology of slow (pre)thermal dynamics.  

\subsection{Summary of our results}
We now briefly summarize the key ideas in this paper, which are illustrated in Figure~\ref{fig:summary}.

\begin{figure}
    \centering
    \includegraphics[width=\linewidth]{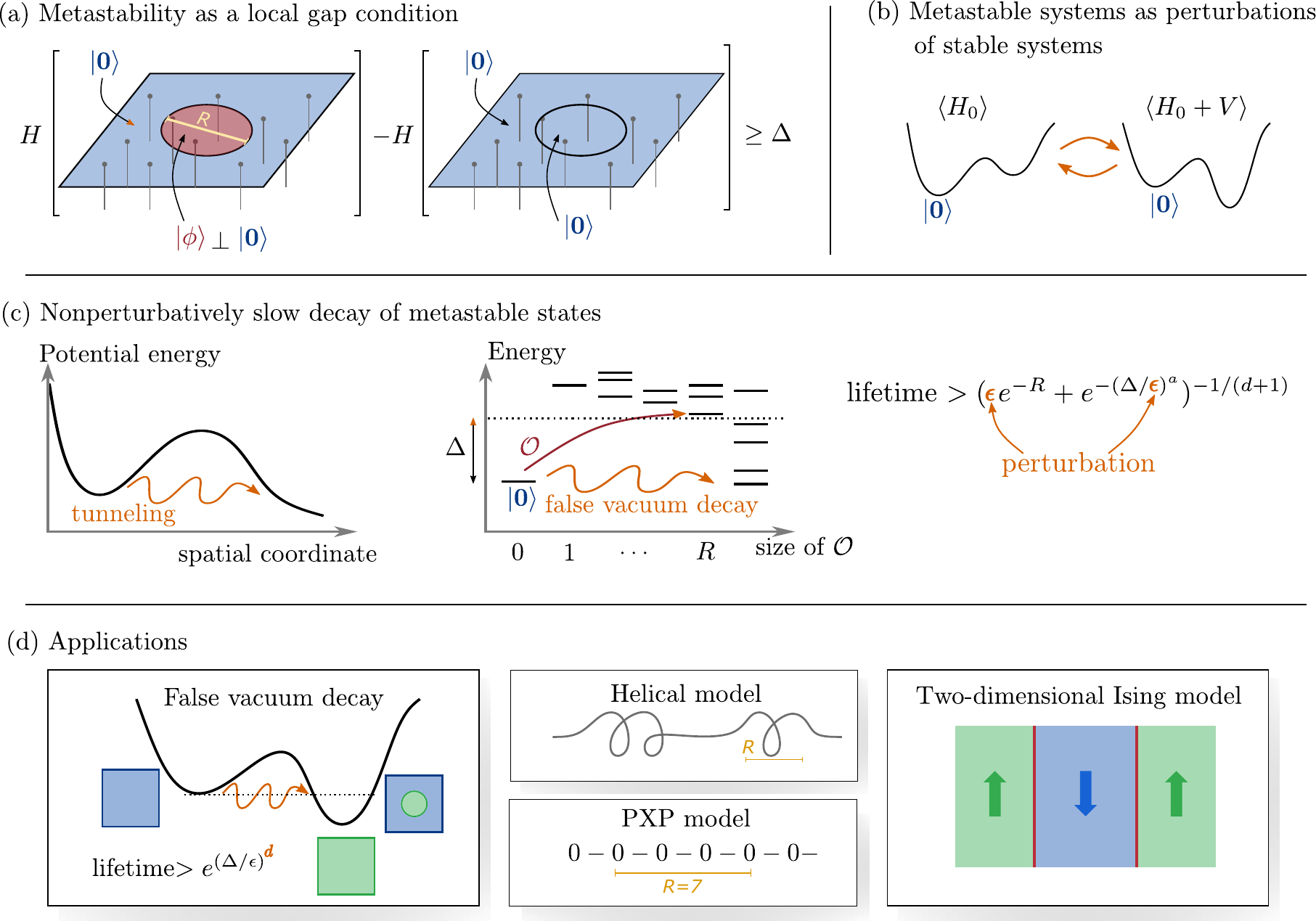}
    \caption{A summary of our results.  (a) Metastable states have a ``local gap" -- all sufficiently small operators acting on a set of diameter $\le R$ raise the average energy by $\ge \Delta$. While our exposition mostly focuses on the all-zero metastable state (or, more generally, short-range entangled states) as depicted here, the metastability concept generalizes to any state.  (b) Theorem \ref{thm:H=H0+V} proves that any short-range entangled metastable state of Hamiltonian $H$ is an \emph{eigenstate} of a ``nearby" Hamiltonian $H_0 = H-V$, where (heuristically) $V$ is perturbatively small in $R^{-1}$.  (c) Theorem \ref{thm:main} proves that given the decomposition $H=H_0+V$, all metastable states exhibit prethermal slow dynamics: local correlation functions decay on the nonperturbatively slow time scale depicted.  (d) 
    Our metastability theory leads us to new insight into a broad range of many-body models, giving us rigorous guarantees on the lifetime of false vacua, as well as uncovering slow thermalization in helical spin chains.   We also provide a new perspective on slow thermalization in the PXP model, and the two-dimensional Ising model.}
    \label{fig:summary}
\end{figure}
Our first step is to give a precise definition of a metastable state.
In Section \ref{sec:define}, we will define a pure state $|\psi\rangle$ and a Hamiltonian $H$ to be ($\Delta,R$) metastable if (loosely) all operators acting on $\le R$ bits increase $\langle H\rangle$ by a finite amount $\ge \Delta$ (Fig.~\ref{fig:summary}-a).  This definition follows closely our classical intuition about metastable states as being \emph{local} minima in the energy landscape, and a similar intuition was also used to define a ``quantum local minimum" in \cite{Chen:2023idc}. By analogy with the tunneling picture in quantum mechanics, the local energy gap $\Delta$ and the size $R$ represent the height and the width of the barrier respectively. A difference with respect to other possible definitions (such as the one used in \cite{Chen:2023idc}), is that metastability is here a robust property with respect to local perturbations of a state: rather than single isolated states, our definition defines ``metastable basins" of states. A metastable state can be pictured as a state that is {\it sufficiently close} to a local minimum in the energy landscape (with this clarification in mind, in the following, we will simply refer to metastable states as a {\it local minima}).

We then show that, with this definition, a metastable state is stable with respect to a close Hamiltonian.
Theorem \ref{thm:H=H0+V} proves that if a short-range-entangled state $|\psi\rangle$ is metastable, then the Hamiltonian $H$ can always be written as $H=H_0+V$, where $|\psi\rangle$ is an exact eigenstate of $H_0$, and $V$ is a small perturbation (possibly on each lattice site) whose magnitude decays with $R$ (Fig.~\ref{fig:summary}-b).  Rather surprisingly, then, this demonstrates that a broad range of metastable states do indeed imply the existence of ``scars" in a perturbatively close Hamiltonian. 

The implications of Theorem \ref{thm:H=H0+V} are profound.  In Section \ref{sec:consequence}, we prove in Theorem \ref{thm:main} that all such metastable states exhibit prethermal behavior, whereby local correlation functions decay over a nonperturbatively long time scale $\exp[R^\alpha]$ for some $\alpha>0$ (Fig.~\ref{fig:summary}-c).  This demonstrates that our quantum definition of metastability is a good one: it generally implies that experimental observables are well-approximated (over very long times) by their value in $|\psi\rangle$, even if $|\psi\rangle$ is not close to a single eigenstate of the Hamiltonian $H$.

The remainder of the paper explores a number of applications and extensions of our formal theory of metastability (Fig.~\ref{fig:summary}-d).  Theorem \ref{thm:Ising_t_d} in Section \ref{sec:falsevacuum} proves that the lifetime of the false vacuum in Ising-like models is lower bounded by (\ref{eq:thitq}), which provides a mathematically well-posed confirmation of a 50-year-old  prediction of quantum field theory \cite{coleman}.   Section \ref{sec:models} then describes a broad range of explicit microscopic models with metastable states.  We find a metastable pair of product states in the one-dimensional PXP model \cite{PXP_NP18}, a popular model of arrested thermalization and  quantum scarring, which provably exhibit extraordinarily slow thermalization, despite being at finite energy \emph{density} above the ground state.  We numerically study a family of ``helical" spin chains for which there are an exponential number of metastable states, and observe numerically that these finite  energy density states are extremely long lived and generate many-body entanglement extremely slowly.  Lastly, we also demonstrate that two-dimensional Ising models exhibit metastable states, which provides a connection to recent discussions of prethermal Hilbert space shattering \cite{2dIsing_PRL22,2dIsing_Hart22} in this setting.  

Theorems \ref{thm:H=H0+V}, \ref{thm:main}, and \ref{thm:Ising_t_d} all have rather technically involved proofs, which are contained in appendices.

\section{Defining metastable states}\label{sec:define}
In this section, we will carefully introduce both the class of models which we study, as well as our definition of quantum metastable pure states.
\subsection{Formal preliminaries: notation, graphs, and local Hamiltonians}
We consider a system of $N$ qudits.   The local Hilbert space dimensions of the qudits need not be equal, but so long as they are uniformly upper bounded by $q$, we might as well take each to be $q$.  A local basis for a single qudit is $\lbrace |0\rangle, |1\rangle, \ldots, |q-1\rangle\rbrace$.  In the special case of $q=2$ (qubits or spin-$\frac{1}{2}$ degrees of freedom), we define the Pauli operators by $Z|0\rangle = |0\rangle$, $Z|1\rangle = -|1\rangle$, $X|0\rangle = |1\rangle$ and $X|1\rangle = |0\rangle$.

We place each qudit on one vertex of a graph with vertex set $\Lambda=\{i:i=1,2,\cdots,N\}$.    This graph is naturally imbued with a Manhattan distance function $\mathsf{d}(i,j)$ between any pairs of sites $i,j$, which counts the length of the shortest path between $i$ and $j$ in the graph.\footnote{This path is defined along the edges of the graph, but since we will not need to reference the edge set we do not bother to give it an explicit label.}  The distance function induces a natural definition for distance between a site and a subset $S\subset\Lambda$: $\mathsf{d}(i,S)=\min_{j\in S}\mathsf{d}(i,j)$.  The distance between two sets $\mathsf{d}(S,S')=\min_{i\in S,j\in S'}\mathsf{d}(i,j)$.  The diameter of a subset $\mathrm{diam}S:=\max_{i,j\in S}\mathsf{d}(i,j)$.  We assume the graph is $d$-dimensional, in the sense that there exists a finite constant $c_d$ such that for any connected subset $S\subset \Lambda$, \begin{equation}\label{eq:cd}
    |\partial S| \le c_d \cdot \max(1,\mathrm{diam}S)^{d-1}, \quad \mathrm{and} \quad |S| \le c_d \cdot \max(1,\mathrm{diam}S)^{d},
\end{equation}
where $\partial S:=\{i\in S: \mathsf{d}(i,S^{\rm c})=1 \}$ is the boundary of $S$.  In the equations above, $|S|$ is the number of sites contained in $S$, and $S^{\rm c}$ is the complement of set $S$ in $\Lambda$, or elements in $\Lambda$ not contained in $S$. 
For the sake of efficiency, we will say that a constant depends on $d$ if it depends on $d$ along with the constant $c_d$ in (\ref{eq:cd}).   Note that the graph does not need to be a (subset of a) regular lattice; nevertheless, we will sometimes refer to the graph as the ``lattice''.

Let us now briefly remind the reader of some additional common mathematical notations that will prove useful for us.  The subtraction of two sets is defined as $S\setminus S' := \{i\in S: i\notin S'\}$. For example, $S^{\rm c}:=\Lambda\setminus S$ is the complement of $S$.  A subset $S$ is connected, if $\forall i,j\in S$, there exists a path in $S$ that connects the two sites; i.e., there exists a set of edges $(i,i_1),(i_1,i_2),\cdots,(i_p,j)$ of the graph where $i_1,\cdots,i_p\in S$.  We define the ball $B(S,r):=\{i\in\Lambda: \mathsf{d}(i,S)\le r \}$ which contains all sites with distance $\le r$ to a given set $S$. \revise{We use $\kzero_S$ to denote the all-zero state in qudit set $S$, with $\bra{\bm{0}}_S$ being its bra and $\kzero:=\kzero_\Lambda$ for the whole system. We use shorthand notation $\ket{\phi}_S\bra{\psi}$ for operator $\ket{\phi}_S\bra{\psi}_S$ supported in $S$. } The floor function $\lfloor a\rfloor$ is defined as the largest integer no larger than $a$. We use big-O notations, where
$f=\mathrm{O}(g)$ means there exists a constant $c$ such that $f\le cg$. $f=\mathrm{\Omega}(g)$ similarly means $f\ge cg$. $f=\Theta(g)$ means both $f=\mathrm{O}(g)$ and $f=\mathrm{\Omega}(g)$.

To these $N$ qudits, we associate a Hamiltonian \begin{equation}\label{eq:H=HS}
    H = \sum_S H_S.
\end{equation}
The sum $S$ runs over subsets of $\Lambda$.  We demand that all subsets $S$ appearing in the sum are connected (unless explicitly stated, this will be an implicit assumption for the remainder of the paper), and that each local term $H_{S}$ is supported in set $S$, i.e. $H_{S}=I_{S^{\rm c}}\otimes \mathcal{O}_S$ where $I_{S^{\rm c}}$ is the all-identity operator on $S^{\mathrm{c}}$, and $\mathcal{O}_S$ acts non-trivially in $S$. The decomposition \eqref{eq:H=HS} is not unique:  e.g. $Z_1+Z_2$ (Pauli-$Z$ operators on qubits $i=1,2$) can be either two terms supported in $S=\{1\},\{2\}$, or a single term supported in $S=\{1,2\}$. It will suffice for us to find one valid decomposition to analyze; all things considered, we wish to associate terms in $H$ with the smallest $S$ possible.\footnote{
This is common in the literature of mathematical statistical physics \cite{barry_simon_book}, where a decomposition like \eqref{eq:H=HS} is sometimes referred to as choosing ``potentials''.} Given decomposition \eqref{eq:H=HS}, we assume the Hamiltonian is sufficiently local in the $d$-dimensional lattice: \begin{equation} 
    \label{eq:H<h}
    \max_{i \in \Lambda} \sum_{S \ni i} \mathrm{e}^{2\mu\, \mathrm{diam}S} \norm{H_{S}} \le h,
\end{equation}
for some constants $\mu,h$. \eqref{eq:H<h} implies the local terms $\norm{H_S}\le h\ee^{-2\mu \diam S}$ decay no slower than an exponential of the support diameter. Clearly, any finite-range Hamiltonian can be normalized to satisfy \eqref{eq:H<h}.

As \eqref{eq:H<h} effectively bounds the ``local norm'' of $H$, such a local $H$ generates dynamics with a ``linear light-cone structure", as guaranteed by the following Lieb-Robinson bound (LRB) \cite{Lieb1972}: 
\begin{thm}[Lieb-Robinson Theorem]\label{thm:LRB}
For any $H$ satisfying \eqref{eq:H<h}, 
there exists a constant $u$ determined by $d,h,\mu$ such that
 \begin{equation}\label{eq:LRB_meta}
    \norm{\mlr{\ee^{\ii tH}\OO_S\ee^{-\ii t H}, \OO'_{S'}}} \le 2\norm{\OO_S} \norm{\OO'_{S'}} \min\mlr{1,  c_{\rm LR}\min\lr{|\partial S|, |\partial S'|} \ee^{\mu (u|t|-\mathsf{d}(S,S') )} },
\end{equation}
for any pair of local operators $\OO_S, \OO'_{S'}$. The constant $c_{\rm LR}$ is determined by $d$ and is independent of $H$.
\end{thm}
The proof is found in Theorem 3.7 and 3.11 of a recent review  \cite{ourreview}. This locality bound is a crucial ingredient for our results.  We will say $H$ has a $(\mu,u)$-LRB if \eqref{eq:LRB_meta} holds.

\revise{Lastly, we will sometimes also need a slightly different local norm than \eqref{eq:H<h}. For any operator $\OO=\sum_S \OO_S$ decomposed like \eqref{eq:H=HS}, we define a local norm parametrized by $\kappa$ and $\alpha\in (0,1)$: \begin{equation}\label{eq:kappa_norm}
    \norm{\OO}_{\kappa}:=\max_{i \in \Lambda} \sum_{S \ni i} \ee^{\kappa\, (\mathrm{diam}S)^\alpha} \norm{\OO_{S}},
\end{equation}
where we do not show the $\alpha$-dependence explicitly. Throughout the paper, $\alpha$ can be viewed as a fixed number in $(0,1)$, and our results will be asymptotically tightest in the limit $\alpha\rightarrow 1$, where operators have an almost exponential decay with the support diameter.
}

\subsection{Metastability is an energy barrier to local excitations}\label{sec:define_meta}
Given the formal setup, we are now ready to introduce what a metastable state is.  Intuitively, metastability should be associated to ``local minima in the energy landscape".  In classical statistical physics, such a statement need not be in quotes: the Hamiltonian of each state in state space is a well-defined function.  In quantum many-body systems however, a Hamiltonian $H$ has complex and entangled \emph{eigenstates} which are not calculable in any thermodynamic limit.  To give a precise definition for metastability, therefore, we need to be careful about what ``local minimum" refers to. Crucially, we do \emph{not} want our definition of metastability to relate in any way to exact eigenstates of a Hamiltonian; in contrast, checking that a state is metastable should be computationally straightforward. 

To overcome the issue raised above, we say that state $\kpsi$ is a local minimum when any sufficiently local Hermitian operator $\OO=\OO_S$ acting on $\kpsi$ either stabilizes the state $\OO\kpsi=\alr{\OO}_{\psi_0}\kpsi$,\footnote{If this equation holds, we say $\OO$ is a stabilizer of $\kpsi$: $\OO$ need not  be a Pauli string.} or raises the energy by at least a constant $\Delta$. More precisely, 
\begin{equation}\label{eq:H0>D}
    H\mlr{ \lr{\OO - \expval{\OO}_{\psi_0}}\kpsi}-H[\kpsi] \ge \Delta, \quad \text{for any } \OO=\OO^\dagger=\OO_S \text{ acting on connected $S$ with } \diam S \le R,
\end{equation}
where $R$ is an O(1) -- but, ideally, large -- user-defined constant. Here, \begin{equation}
    \expval{\OO}_{\psi_0} :=\frac{ \bra{\psi_0} \OO \kpsi}{\langle \psi_0|\psi_0\rangle},
\end{equation}
is the component that stabilizes the state, and \begin{equation}
    H[\ket{\psi}]:= \left\{\begin{array}{cc}
    \displaystyle \dfrac{\bra{\psi} H \ket{\psi}}{\expval{\psi|\psi}},   & \text{ if }\ket{\psi}\neq 0 \\
     +\infty,    & \text{ if } \ket{\psi}= 0
    \end{array}\right.
\end{equation}
is the average energy of an unnormalized state $\ket{\psi}$. Here the $+\infty$ energy for a vanishing state is chosen such that \eqref{eq:H0>D} is satisfied if $\OO$ stabilizes $\kpsi$. This is intuitive because the state is ``stable'' with respect to its stabilizers.  We arrive at:
\begin{defn}\label{def:meta}
We say the pair $(H,\kpsi)$ is $(\Delta, R)$-metastable, if \eqref{eq:H0>D} holds with $R\ge 1$.
\end{defn}

A state which is metastable under $H$ is not metastable under $-H$ by this definition.  However, as it is straightforward to simply reverse the sign of $H$ for the remainder of the paper, we remark that any state which is a local maximum will exhibit all of the same phenomena that metastable states have.

We can relate this definition to the quantum Ising models reviewed in the introduction. 

(\emph{1}) If $R$ is chosen as the system size $\diam\Lambda$, this definition is equivalent to $\kpsi$ having large overlap (see \eqref{eq:gsoverlap}) with a nondegenerate gapped ground state, because $\OO$ can be any operator in \eqref{eq:H0>D}; such a state would be metastable for any gapped system.


(\emph{2}) For $H=H_{\rm TFIM}$ being the ferromagnetic transverse-field Ising Hamiltonian -- i.e., the first two terms in \eqref{eq:ZZ-X-Z} --  and $\ket{\psi_{\pm}}$ being its two ground states with negative/positive polarization\footnote{Strictly speaking, in finite volume $\Lambda$ the two lowest energy states are the cat states $\frac{1}{\sqrt{2}}(\ket{\psi_+}\pm \ket{\psi_-})$, with an energy splitting $\sim \exp(-\Lambda)$.} for any finite $R$ each of  $\ket{\psi_{\pm}}$ is $(\Delta_{\rm TFIM}(1-c), R)$-metastable for some small constant $c$ if the volume $\Lambda$ is sufficiently large. The reason is that any local operator $\OO$ (that acts on a finite number of sites) has a matrix element $\bra{\psi_+}\OO \ket{\psi_-}=  O(\exp(-\Lambda))$ in the limit of large $\Lambda$, so the state $(\OO-\alr{\OO}_{\psi_+})\ket{\psi_+}$ (and similarly for $\ket{\psi_-}$) will be mostly supported on excited states, and $c$ can be taken as small as $O(\exp(-\Lambda))$.

(\emph{3}) For $H$ being the longitudinal-field Ising Hamiltonian \eqref{eq:H=ZZ-Z}, the all-zero state $\kzero$ is the unique ground state. On the other hand, $(H,\ket{\bm{1}})$ is $(\Delta/2, c\Delta /\epsilon)$-metastable with a constant $c$, where $\ket{\bm{1}}$ is the all-one state. In this case, one still has a parametrically large $R\sim 1/\epsilon$ at small $\epsilon$, because one has to flip (roughly speaking) a radius-$R$ ball of qubits to get a state with similar energy, so that the energy gain in the bulk $\sim \epsilon R^d$ compensates the energy loss at the boundary $\sim \Delta R^{d-1}$. (\emph{3}) When adding a transverse-field so that $H$ is given by \eqref{eq:ZZ-X-Z}, one expects $(H,\kpsi_{\rm TFIM})$ is still $(\Delta_{\rm TFIM}/2, c\Delta /\epsilon)$-metastable for some constant $c$. Here $\kpsi_{\rm TFIM}$ is one of the ferromagnetic ground states of $H$ without the longitudinal field defined above. The physical argument (although not rigorous) is that the excitations of $H_{\rm TFIM}$ are smoothly connected to those of the solvable point $H_{\rm Ising}=-\Delta \sum_{i\sim j}Z_i Z_j$, because they are in the same ferromagnetic phase. Therefore, the metastability radius $R$ of $H$ should be similar to that of \eqref{eq:H=ZZ-Z}.  Of course, given a suitable state, our definition can be explicitly checked.

Our first goal is to prove that, for any system satisfying the metastability condition Definition \ref{def:meta} with a large $R$ and a short-range entangled state $\kpsi$ (to be rigorously defined shortly),  $\kpsi$ exhibits the \emph{dynamical} phenomena associated with metastability. Namely, when evolved by $H$, 
up to a long time-scale $t_*$ that roughly scales exponentially with $R$, that there exists a ``dressed" version of $\kpsi$, denoted as $\tkpsi$, for which local correlation functions of $\ee^{-\ii tH}\tkpsi$ remain nearly the same as the initial $\tkpsi$ for $|t|<t_*$.\footnote{While our physical intuition is that this statement should also hold for $\kpsi$, one has to be somewhat careful: there can be short-time transient dynamics in state $\kpsi$, but this transient dynamics should not e.g. wholly decay a correlation function from 1 to 0.  We were unable to prove this intuition in complete generality, while in contrast we could prove very strong statements about the decay of $\tkpsi$.}  Hence, we need to show that Definition \ref{def:meta} implies $\kpsi$ is close to a state $\tkpsi$. This will not depend on whether $\kpsi$ is the \emph{global} minimum: as in the introduction, we will see that even if $\kpsi$ is a highly excited (superposition of exponentially many eigenstates of $H$), the fact that it is a local minimum will strongly suppress the decay of local correlation functions, as any ``path'' in the Hilbert space connecting $\kpsi$ to any such $\ket{\phi}$ requires ``tunneling through" a large energy barrier. Crucially, in the many-body setting, when the Hamiltonian only contains few-body operators, we will be able to bound this ``tunneling" through radius $R$, since the state is evolved only by local operators and thus any perturbation of size $R$ must grow in an ``off-shell" way through many orders in perturbation theory.

As long as $\OO\kpsi$ (and their inner products) can be computed easily for a given $\OO$ (which is the case for short-range-entangled $\kpsi$, which we will mostly focus on),
Definition \ref{def:meta} can be verified numerically with complexity that grows exponentially with $R^d$, but at most linearly with system size $N$; with translation invariance, the metastability criterion can be checked in time independent of $N$. Our results therefore provide a rigorous way of using moderate-size classical computation to predict long-time dynamics of very large quantum many-body systems, as we will do in Section \ref{sec:models}.

We expect that Definition \ref{def:meta} does not capture all kinds of quantum many-body metastable phenomena.  The most clear shortcoming of our definition is that \eqref{eq:H0>D} requires the energy to raise by a finite amount when acting with any local operator, but the more physical condition -- which we will indeed make heavy use of in Section \ref{sec:falsevacuum} -- is that the \emph{minimum barrier height} is finite.  The barrier height may be largest when acting with operators of size $R$, not size 1.  We expect this ``shortcoming" of our definition to be mainly technical in nature; as our theory already requires somewhat involved proofs, we have elected to focus on understanding the implications of the simpler Definition \ref{def:meta} for this work.  It is also likely, of course, that the dynamical implications of such a more relaxed notion of metastability will be weaker than those proved in this paper.  We will nevertheless see that our definition of quantum metastability encapsulates a large number of examples of interest across multiple subfields of physics.  

Another shortcoming of our definition is that it does not account for the metastability of mixed states, such as the thermal ensemble of supercooled water which is metastable to the solid phase.  Such a model should in principle have a quantum mechanical formulation, although we expect that this type of metastability is largely captured by classical models.  The metastability of mixed states was recently discussed from a quantum computer science perspective  \cite{Chen:2023idc}.

\subsection{Metastability as a local gap}

This paper mostly focuses on $\kpsi$ which are short-range entangled (SRE): $\kpsi = U_0 \kzero$,
where $\kzero$ is the all-$0$ state, and unitary $U_0$ represents finite-time evolution generated by a quasi-local Hamiltonian (quasi-local in the sense that its local strength is bounded similarly as \eqref{eq:H<h}).  It is reasonable to focus on such states for two reasons.  Firstly, this SRE condition is satisfied by the canonical models with false vacua described above. As we will see shortly, the physical meaning of the metastability condition also becomes more clear for SRE states, so it will be easy for us to cook up new models of metastability in Section \ref{sec:models}. Secondly, long-range entanglement can introduce a number of complications: for example, the long-range entangled state can be \emph{completely stable} (not metastable) in the thermodynamic limit, e.g. in theories with topological order, or highly sensitive to perturbations (and thus not even metastable).  This sensitivity is well-illustrated by the GHZ state $\kpsi=(\kzero+\ket{\bm{1}})/\sqrt{2}$, which is readily dephased by a small $Z$ field \emph{on one site} to $\kzero-\ket{\bm{1}}$, which seems to violate our intuition that metastable states ought to be locally robust.  Insofar as we only care about local correlation functions, the GHZ state should instead be viewed as a classical mixture of two SRE states $\kzero,\ket{\bm{1}}$, and both of these states may be metastable.  
However, we will show that some of our results about the phenomenology of metastable states do generalize to long-range entangled metastable states, in Appendix \ref{app:B}.

The quasi-local unitary $U_0$ can be viewed as a change of local basis/frame.
From now on, we just work in the frame dressed by $U_0$, so that \begin{equation}\label{eq:psi=0}
    \kpsi = \kzero,
\end{equation} 
and $H$ is actually $U_0^\dagger H_{\rm original} U_0$ where $H_{\rm original}$ is the Hamiltonian in the original frame. If $H_{\rm original}$ is quasi-local in the sense of obeying \eqref{eq:H<h}, $H$ remains quasi-local due to Theorem \ref{thm:LRB}; the only price to pay is to modify constants such as $\mu$, $h$, $\Delta$ and $R$ by O(1) factors that depend on the ``depth" of $U_0$. Therefore, we still assume \eqref{eq:H<h} in the dressed frame without loss of generality. Due to the same argument, the results obtained in the dressed frame can be automatically translated back to the original frame by just adjusting some constants. As an example of this local-frame change, consider the metastable system $(H_{\rm original}, \kpsi_{\rm TFIM})$ with $H_{\rm original}$ being \eqref{eq:ZZ-X-Z}. Here, $\kpsi_{\rm TFIM}$ is in the same ferromagnetic phase as $\ket{\bm{1}}$, so $U_0$ is just the finite-time (quasi-)adiabatic unitary evolution connecting them\footnote{Strictly speaking, the quasi-adiabatic unitary may be generated by a Hamiltonian that decays slower than exponential with the support diameter (see e.g. \cite{bachmann2012automorphic}). Then $H=\widetilde{H}+\widetilde{V}$ where $\widetilde{H}$ is $H$ truncated at some support diameter and thus quasi-local in the sense of \eqref{eq:H<h}, while the remaining $\widetilde{V}$ has very small local norm but decays slower than exponential. We expect our results also apply to this case, since we will treat slower-decaying perturbations in Theorem \ref{thm:main} anyway, but again we have kept the assumptions a bit simpler for the sake of clarity. }, followed by a trivial relabeling $0\leftrightarrow 1$ of the local states. $H=U_0^\dagger H_{\rm original} U_0$, on the other hand, is still quasi-local containing the dressed $Z$ field $-\epsilon\sum_i U_0^\dagger Z_i U_0$. The $Z$ field after dressing can have terms like $X_i$, which no longer preserve the $\kzero$ state (here we have relabeled $0\leftrightarrow 1$) and causes this false vacuum to decay.

After restricting to the all-zero state (or SRE states accompanied by dressing), the meaning of the metastability condition \eqref{eq:H0>D} becomes more clear by the following Proposition:
\begin{prop}\label{prop:H0meta}
The pair $(H,\kzero)$ is $(\Delta, R)$-metastable, if and only if
\begin{equation}\label{eq:H0>Delta1}
    H[\ket{\phi}_S\otimes \kzero_{S^{\rm c}}]-H[\kzero]\ge \Delta, \quad \text{for any } \ket{\phi}_S\perp \kzero_S \text{ with connected $S$ and } \diam S \le R.
\end{equation}
\end{prop}

\begin{proof}
For any $\ket{\phi}_S\perp \kzero_S$, one can define $\OO=\ket{\phi}_S\bra{\bm{0}}+\mathrm{H.c.}$ so that \eqref{eq:H0>D} implies \eqref{eq:H0>Delta1}. On the other hand, any state $(\OO-\alr{\OO}_{\bm{0}})\kzero$ is proportional to the form $\ket{\phi}_S\otimes \kzero_{S^{\rm c}}$ for some $\ket{\phi}_S\perp \kzero_S$, so \eqref{eq:H0>Delta1} also implies \eqref{eq:H0>D}.
\end{proof}

Proposition \ref{prop:H0meta} implies that, when fixing the quantum state (i.e. boundary conditions) in $S^{\mathrm{c}}$, a state is metastable if and only if it has sufficient overlap with the ground state.  If the gap between the fixed-boundary-condition Hamiltonian's ground state and first excited state is $\widetilde{\Delta}$, then the state can\footnote{However, the state is not guaranteed to be metastable: we could consider the metastable state to be $\cos\theta |\mathbf{0}\rangle_S + \sin\theta |\text{max}\rangle_S$ where $|\text{max}\rangle_S$ is the highest energy state in the $S$-restricted Hamiltonian. This state is less metastable than \eqref{eq:gsoverlap}.} be $(\Delta,R)$-metastable so long as the overlap between $|\mathbf{0}\rangle$ and the true ground state $|\text{gs}\rangle$ in region $S$ obeys 
\begin{equation}\label{eq:gsoverlap}
    |\langle \mathbf{0}|\text{gs}\rangle_S|^2  \ge \frac{1}{2}\left(1+\frac{\Delta}{\widetilde{\Delta}}\right).
\end{equation}
We can, alternatively, bound the local gap $\Delta$: \begin{equation}\label{eq:Delta<h}
    \Delta \le H\mlr{\ket{1}_i\otimes \kzero_{i^{\rm c}}} - H\mlr{ \kzero} \le 2\sum_{S\ni i} \norm{H_S} \le 2h.
\end{equation}
Flipping one spin increases the energy by at most an $\mathrm{O}(1)$ amount. The key point is that acting a local operator $\OO_S$ on $\kzero$ does not change the expectation value of $H_{S'}$ for any $S'$ that does not intersect with $S$. 
These ideas will play an important role in our proofs.

\subsection{Metastable systems as perturbations of stable systems}
Our definition of metastability does not depend in any way on the form of $H$. 
 However, our motivating example -- the false vacuum problem -- does have special structure in the Hamiltonian $H$, which can be expressed as \begin{equation}\label{eq:H=H0+V}
    H = H_0 + V,
\end{equation}
where $H_0$ is the ``ideal'' part of $H$ for which $\kpsi$ is an exact eigenstate, and $V$ is a generic small perturbation such that $\kpsi$ is not an eigenstate of the full $H$. For example, for the metastable system $(H,\kpsi_{\rm TFIM})$ with $H$ being \eqref{eq:ZZ-X-Z}, it can be separated to $H_0=H_{\rm TFIM}$ and $V=-\epsilon \sum_i Z_i$, where $\kpsi_{\rm TFIM}$ is one gapped ground state of $H_0$, and $V$ has small $\sim \epsilon$ local strength. In this case, the gap of $H_0$ guarantees that under a generic perturbation $V$, $\ee^{-\ii t H}U\kpsi_{\rm TFIM}$, where $U$ is an appropriate small local rotation (see Section \ref{sec:consequence}), is effectively constrained in the two-dimensional ground state subspace of $H_0$ (probed by local correlation functions) up to a long time-scale $t_*\sim \ee^{(\Delta/\epsilon)^a}$. Since $\kpsi_{\rm TFIM}$ cannot be mapped to the other ground state by local operators, this gap condition already proves the slow decay of the false vacuum, albeit with a loose exponent $a\approx 1/(2d-1)$ \cite{our_preth}.

Generalizing the \emph{global} gap condition to \emph{local} minimum condition \eqref{eq:H0>D}, one expects that perturbing a stable system $(H_0,\kpsi)$ where $\kpsi$ is an eigenstate and local minimum of $H_0$ generically leads to metastability. Indeed, it is straightforward to show that for SRE states, the metastability condition is robust under arbitrary small perturbation to the Hamiltonian:
\begin{prop}\label{prop:meta_robust}
    If $(H_0,\kzero)$ is $(\Delta, R)$-metastable, and \begin{equation}
        \max_{i\in\Lambda}\sum_{S\ni i}\norm{V_S} \le \epsilon,
    \end{equation}
    then $(H,\kzero)$ is $(\Delta/2, R')$-metastable, where $H$ is given by \eqref{eq:H=H0+V} and $R'=\min(R, \lr{\frac{\Delta}{2c_d\epsilon}}^{1/d})$.
\end{prop}
\begin{proof}
    We use the simpler condition \eqref{eq:H0>Delta1}.
    For any $\ket{\phi}_S\perp \kzero_S$ with connected $S$ and $\diam S \le R'$, \begin{equation}
        H[\ket{\phi}_S\otimes \kzero_{S^{\rm c}}]-H[\kzero] = H_0[\ket{\phi}_S\otimes \kzero_{S^{\rm c}}]-H_0[\kzero] + V[\ket{\phi}_S\otimes \kzero_{S^{\rm c}}]-V[\kzero]\ge \Delta - \epsilon |S|\ge \Delta-\epsilon c_d (R')^d \ge \Delta/2,
    \end{equation}
    where the term $\epsilon|S|$ comes from a straightforward generalization of \eqref{eq:Delta<h}.\footnote{A similar argument shows that $(H,\kzero)$ is $((1-\gamma) \Delta, R')$-metastable for any $0<\gamma<1$ if $R^\prime = \min (R, (\gamma \Delta/c_d \epsilon)^{1/d})$.}
\end{proof}

Note that Proposition \ref{prop:meta_robust} does not assume $\kzero$ is an eigenstate of $H_0$. 
In many examples, it will still be useful to start with an eigenstate $|\psi_0\rangle$ of $H_0$.   In this setting, if $|\psi_0\rangle$ is a gapped ground state of $H_0$ (which, in a nontrivial model, is metastable with thermodynamically large $R$), Proposition \ref{prop:meta_robust} shows that perturbing this stable system leads to metastability.
Intriguingly, the converse is almost also true: any SRE metastable system $(H,\kzero)$ has a decomposition \eqref{eq:H=H0+V} with $H_0$ stabilizing the state (which is still metastable with respect to $H_0$) and small local norm for $V$; the price to pay is that the metastability range $R$ may shrink for $(H_0,\kzero)$ to $R_0\sim \sqrt{R}$:

\begin{thm}\label{thm:H=H0+V}
    Suppose $H$ is quasi-local by \eqref{eq:H<h} and $(H,\kzero)$ is $(\Delta, R)$-metastable. Suppose $R\ge c_{\rm R}$ with $c_{\rm R}$ determined by $d,h/\Delta,\mu$, and the system is translation-invariant on a $d$-dimensional square lattice. Then for any constants $\alpha,\delta \in (0,1)$, $H$ can be decomposed as \eqref{eq:H=H0+V}, where  $(H_0,\kzero)$ is $(\Delta/2, R_0)$-metastable with $R_0=c'_\delta R^{1/2-\delta/d}$ and \begin{equation}\label{eq:H0psi0=0}
    \revise{H_0 \kzero = E_0\kzero},
\end{equation}
    for some number $E_0$, while the local norms are bounded by \begin{subequations} \label{eq:H0Vlocality}\begin{align}
    \label{eq:H0<1}
    \max_{i \in \Lambda} \sum_{S \ni i} \mathrm{e}^{\frac{2}{5}\mu\, \mathrm{diam}S} \norm{H_{0,S}} &\le c_{\rm H} h, \\
    \label{eq:V1<eps}
    \norm{V}_{\kappa_1} &\le c_\delta h R^{\delta-d/2},
\end{align}
\end{subequations}
where $\kappa_1=\mu$, and constants $c_{\rm H}, c_\delta,c_\delta'$ are determined by $d,\alpha,h/\Delta,\mu,\delta$.
\end{thm}

The proof is given in Appendix \ref{app:H=H0+V}, and proceeds by explicitly constructing the Hamiltonians $H_0,V$ by their local decompositions \begin{equation}\label{eq:H=HSV=VS}
    H_0=\sum_{S\subset \Lambda} H_{0,S}, \quad V=\sum_{S\subset \Lambda} V_{S}.
\end{equation}
There we will also show (\emph{1}) the assumption of translation invariance is not fundamental, but it does simplify the statement of the theorem; (\emph{2}) an explicit example which (at least naively) saturates the scaling $R_0\sim \sqrt{R},\norm{V}_{\kappa_1}\sim R^{-d/2}$ at $\delta\rightarrow 0$.

Because of this reduction, starting from the next section we take the perturbation scenario \eqref{eq:H=H0+V} as the starting point, and aim to show long life-time of the metastable state.
\revise{The local norm $\norm{V}_{\kappa_1}$ is defined in \eqref{eq:kappa_norm} and requires the interaction strength of $V$ to decay only stretched-exponentially with the support diameter of its local terms. However, for latter technical reasons, we do require that $H_0$ decays exponentially in \eqref{eq:H0<1} and has an exponential LRB \eqref{eq:LRB_meta} as a consequence. }

Before proceeding, we remark that the robustness Proposition \ref{prop:meta_robust} also implies robustness of metastability when dressing the state locally:  

\begin{prop}\label{prop:U_robustness}
    Let $U=\mathcal{T}\ee^{\int^1_0\dd s A(s)}$ ($\mathcal{T}$ means time-ordering) be any quasi-local unitary where the anti-Hermitian $A(s)$ has local norm bounded by a sufficiently small $\epsilon>0$: \begin{equation} \label{eq:U_robustness_A}
    \max_s\max_{i \in \Lambda} \sum_{S \ni i} \mathrm{e}^{2\mu\, \mathrm{diam}S} \norm{A_{S}(s)} \le \frac{\epsilon}{c_- h} (2\mu)^{4d-2} < (2\mu)^{6d+2},
\end{equation}
with constant $\mu\le 1/(2\ee)$ and $c_-$ determined by $d$.
If $(H,\kzero)$ is $(\Delta, R)$-metastable with $H$ satisfying \eqref{eq:H<h}, then $(U^\dagger HU, \kzero)$ is $(\Delta/2, R')$-metastable with $R'$ the same as in Proposition \ref{prop:meta_robust}.
\end{prop}
\begin{proof}
    Using the $\kappa$-norm notation in \eqref{eq:kappa_norm} with $\alpha=0.5$, we have e.g. $\norm{H}_{2\mu}\le h$ from \eqref{eq:H<h} and $\diam S\ge (\diam S)^\alpha$. We then invoke Proposition \ref{prop:AO-O} in Appendix \ref{app:B},\footnote{Strictly speaking the proof of this result does not include time-dependent $A(s)$, but it is straightforward to generalize to this case.} where the condition $\max_s\norm{A(s)}_{2\mu}\le (2\mu)^{6d+2}$ holds from \eqref{eq:U_robustness_A}, so that \begin{equation}
        \norm{U^\dagger HU-H}_0 \le c_-(2\mu)^{-2(2d-1)} \norm{H}_{2\mu}\max_s\norm{A(s)}_{2\mu} \le \epsilon.
    \end{equation}
    Since $U^\dagger HU = H + (U^\dagger HU-H)$, this reduces to the setting in Proposition \ref{prop:meta_robust}, and metastability of $(U^\dagger HU, \kzero)$ follows.
\end{proof}

Alternatively, we can study the metastability of $(H,U\kzero)$. Ignoring the dependences on constants $d,\mu$, Proposition \ref{prop:U_robustness} implies that given one metastable SRE state $\kpsi$, any state of the form $U\kpsi$ is also metastable, as long as the quasi-local unitary $U$ has local strength $\epsilon\lesssim h,\Delta$. 

\section{Nonperturbatively slow decay of metastable states}
\label{sec:consequence}

The most important consequence of metastable states is the non-perturbatively slow decay of local correlation functions.  In this section, we outline the proof of this, and related, statements.

\subsection{Long life-time from local diagonalizability} 
We will prove the long ``lifetime" of $\kpsi$ under $H$ by showing that $H$ can be almost ``locally diagonalized'' to stabilize $\kpsi$. Here ``local'' is crucial: one does not want to fully diagonalize $H$ using a global unitary because no true eigenstates of $H$ will be close to $\kpsi$ due to the orthogonality catastrophe. Instead, one asks if any quasi-local unitary $U$ approximately diagonalizes $H$ up to a small but still extensive error term \cite{our_preth}.  Formally, we define local diagonalizability by the following structure for $H_0, V$:
\begin{defn}\label{def:localdiag}
Suppose $H_0,V$ are two Hamiltonians and $\kpsi$ is a quantum state.
We say the triple $(H_0,V,\kpsi)$ is $(v_*,\mathbbm{d}_*,a_*, \kappa_*)$-locally-diagonalizable, if the following statements hold:

There exists a quasi-local unitary $U=\mathcal{T}\ee^{\int^1_0\dd s A(s)}$ (recall that $\mathcal{T}$ means time-ordering) with anti-Hermitian $A(s)$ of local norm \begin{equation}\label{eq:A<eps0}
    \norm{A(s)}_{\kappa_*} \le a_*,
\end{equation}
such that the rotated Hamiltonian \eqref{eq:H=H0+V} \begin{equation}\label{eq:UHU=}
    H_*+V_*:=U^\dagger (H_0+V) U,
\end{equation}
($H_*,V_*$ both Hermitian)
satisfies \begin{equation}\label{eq:H*=E*}
    H_* \kpsi = E_*\kpsi,
\end{equation}
for some number $E_*$,
and \begin{subequations}
    \begin{align}
    \norm{H_*+V_*- H_0}_{\kappa_* } &\le \mathbbm{d}_*, \label{eq:H*<0} \\
    \norm{V_*}_{0} &\le v_*. \label{eq:V<exp0}
\end{align}
\end{subequations} 

\end{defn}

Here $U$ is a quasi-local unitary because it is generated by finite-time evolution of some time-dependent quasi-local Hamiltonian $\ii A(s)$.
The parameter $a_*$ describes the local strength of $U$ that brings $H$ to an almost diagonalized form (here diagonalization is only with respect to the particular state $\kpsi$). After the rotation, $\mathbbm{d}_*$ quantifies how far the new Hamiltonian is with respect to the unperturbed $H_0$. $v_*$ is the most important parameter that quantifies the local strength of the remaining undiagonalized part $V_*$.

This local diagonalizability implies long life-time $t_*\sim v_*^{-1/(d+1)}$ of the system, among other things:
\begin{prop}[Corollary 4 in \cite{our_preth}]\label{cor1}
If the triple $(H_0,V,\kpsi)$ is $(v_*,\mathbbm{d}_*,a_*, \kappa_*)$-locally-diagonalizable with corresponding $U,H_*$ defined above, and $\norm{H}_{\kappa_1}\le h$ for \eqref{eq:H=H0+V}, there exists constants $c_U,c_{\rm oper}$ determined by $\kappa_1,h$ such that the following statements hold. 
\begin{enumerate}
    \item \textbf{ Locality of $U$:} Given local operator $\mathcal{O}=\OO_S$ supported in connected set $S$, $U^\dagger \OO U$ is quasi-local in the sense that
    \begin{equation}\label{eq:cU}
        U^\dagger \OO U = \OO+ \sum_{r=0}^\infty \OO_r,
    \end{equation}
    where $\OO_r$ is supported in $B(S,r)$, but not in $B(S,r-1)$, and $\OO_r$ decays rapidly with $r$: \begin{equation}\label{eq:Or<S}
        \norm{\OO_r}\le c_U a_* |S|\norm{\OO} \ee^{-\kappa_* r^\alpha}.
    \end{equation}
    
    \item \textbf{Local operator dynamics is approximately generated by $H_*$:} Given local operator $\mathcal{O}=\OO_S$,
\begin{equation}\label{eq:coper}
    \norm{\ee^{\ii t H}\mathcal{O} \ee^{-\ii t H} - U\ee^{\ii t H_*}U^\dagger \mathcal{O} U\ee^{-\ii t H_*}U^\dagger } \le c_{\mathrm{oper}} v_* |S|(|t|+1)^{d+1} \norm{\OO}.
\end{equation}
    \revise{As a result, for $|t|\ll t_*\sim v_*^{-1/(d+1)}$, one can effectively evolve a local operator by first going to the dressed frame $U^\dagger \OO U$, evolving by $H_*$, and finally transforming to the original frame. }

    \item \textbf{The dressed $\kpsi$ is locally preserved for $|t| \ll t_*\sim v_*^{-1/(d+1)}$:} Define the dressed metastable state \begin{equation}\label{eq:rho0=dress}
        \tkpsi=U \kpsi,
    \end{equation}
 along with its time-evolved reduced density matrix on any set $S$: \begin{equation}\label{eq:rhoS}
        \rho_S(t) := \mathrm{Tr}_{S^c}\lr{ \ee^{-\ii t H}\tkpsi \bra{\widetilde{\psi}_0} \ee^{\ii t H}}.
    \end{equation} 
 $\rho_S(t)$ is close to the initial value in trace norm \begin{equation}\label{eq:rho-rho}
        \norm{\rho_S(t)-\rho_S(0)}_1 \le c_{\mathrm{oper}} v_* |S|(|t|+1)^{d+1}.
    \end{equation}
\end{enumerate}
\end{prop}

The proof follows verbatim from that of Corollary 4 in  \cite{our_preth},\footnote{Definition \ref{def:localdiag} provides the implications (S23) and (S25) from Theorem 3 (references refer to \cite{our_preth}) which go into the proof of Corollary 4 in \cite{our_preth}.  In this paper, we will prove that metastability can also imply these properties, in Theorem \ref{thm:main}.} \revise{where \eqref{eq:rho-rho} follows from \eqref{eq:coper} by a state tomography argument, together with the fact that $H_*$ preserves $\kpsi$.} \eqref{eq:rho-rho} is what we mean by the long life-time of $\kpsi$: If $v_*$ is non-perturbatively small, there exists a locally dressed version \eqref{eq:rho0=dress} of $\kpsi$, such that all local correlation functions in the dressed $\tkpsi$ are preserved for a non-perturbatively long time. This is because any local correlation function at time $t$ is of the form $\mathrm{Tr}(\OO_S \rho_S(t))\approx \mathrm{Tr}(\OO_S \rho_S(0))$ for some small set $S$. Although the rigorous life-time bound only holds for the dressed state, physical arguments suggest that the undressed $\kpsi$, as well as other locally-rotated states, should share similar life-times \cite{our_preth}. 
One also does not need to prepare the exact $N$-qudit state \eqref{eq:rho0=dress} to see this effect (which is practically impossible at large $N$ in an experiment), because we are probing local correlation functions: The result \eqref{eq:rho-rho} holds for any initial state whose reduced density matrix in the backwards light-cone region of $S$ approximates that of $\tkpsi$.  The size of the allowed region $S$ does not depend on the total number of degrees of freedom $N$, although it does scale with $v_*$.  

\subsection{Metastability implies local diagonalizability}
Our goal is now to show that metastability leads to local diagonalizability with non-perturbatively small $v_*$, which automatically implies that metastable states are long-lived in the sense of Proposition \ref{cor1}. 
This is established for general metastable systems with SRE metastable states by the following key theorem, utilizing the perturbation decomposition $H=H_0+V$ guaranteed by Theorem \ref{thm:H=H0+V}. 

\begin{thm}\label{thm:main}
Suppose $H_0=\sum_{S}H_{0,S}$ is a sum of operators with finite norm: \begin{equation}\label{eq:H0S<h0}
    \max_i \sum_{S\ni i} \norm{H_{0,S}}\le h_0,
\end{equation}
$H_0$ has $(\mu,u)$-LRB \eqref{eq:LRB_meta}, and 
 $(H_0,\kpsi)$ is $(\Delta,R)$-metastable where $H_0$ stabilizes product state $\kpsi$ as in \eqref{eq:H0psi0=0}.
 \comment{, and either of the two conditions hold: \begin{enumerate}
    \item $\kpsi$ is a product state.
    \item $H_0$ is $k_0$-local.
\end{enumerate}
}
For any $\alpha\in(0,1)$ and \begin{equation}\label{eq:kap<mu}
    \kappa_1 \le \min\lr{\frac{\mu}{5}, \frac{1}{\ee}},
\end{equation}
there exist constants $c_{\rm R},c_*,c_V,c_A,c_D$ determined by $\alpha,d,\mu,\Delta/h_0,\Delta/u,\kappa_1$ (where $c_{\rm R}$ does not depend on $\alpha,\kappa_1$), such that if $R\ge c_{\rm R}$, for any perturbation $V$ with small local norm \begin{equation}\label{eq:e<Delta}
    \epsilon:=\norm{V}_{\kappa_1} \le c_V \Delta ,
\end{equation} 
$(H_0,V,\kpsi)$ is $(v_*, c_D  \epsilon, c_A  \epsilon,\revise{\kappa_1/2})$-locally-diagonalizable, where
$v_*$ is small so that $\kpsi$ has life-time \begin{equation}\label{eq:tstarsec3}
    t_* = \left\{ \begin{array}{ll} \displaystyle
       \epsilon^{-\frac{2(d-\alpha)}{(2d-1)(d+1)}} \exp\mlr{\Omega\lr{\dfrac{R^\alpha}{\ln R}}},  & \dfrac{\epsilon}{\Delta}=\mathrm{O}\lr{ R^{-\frac{\alpha(2d-1)}{2\alpha-1}} } \\
       \displaystyle \exp\mlr{\Omega\lr{\lr{\dfrac{\Delta}{\epsilon}}^{\frac{2\alpha-1}{2d-1}}}},  & \dfrac{\epsilon}{\Delta} =\Omega\lr{ R^{\delta-\frac{\alpha(2d-1)}{2\alpha-1}} }
    \end{array}\right.
\end{equation}
in the sense of Proposition \ref{cor1}. Here \eqref{eq:tstarsec3} only
keeps track of the two large parameters $\Delta/\epsilon, R\gg 1$, and $\delta>0$ is any constant. 
\end{thm}

The condition \eqref{eq:kap<mu} is always achievable by freely reducing $\kappa_1$ because $\norm{V}_{\kappa}\le \norm{V}_{\kappa'}$ for any $\kappa<\kappa'$.
Intuitively, at large $R$ the second line of \eqref{eq:tstarsec3} recovers the nonperturbative prethermalization lifetime \eqref{eq:tgap=2d-1} for perturbed gapped systems \cite{our_preth}; while for extremely small perturbation $\epsilon/\Delta$, the perturbation theory works at very high order $k_*$ and involves large operators that feel the presence of the metastability length scale $R$. In the latter case, $R$ determines an effective cutoff order that leads to $t_*\sim \ee^{R}$.

We prove Theorem \ref{thm:main} in Appendix \ref{app:B}, and sketch the idea here for $\kpsi=\kzero$ (see a cartoon illustration in Fig.~\ref{fig:proof_idea}(a)). 
In the limit $R\rightarrow\infty$, the problem reduces to perturbing a gapped $H_0$, where the gapped ground state $\kzero$ needs to gain energy $\ge \Delta$ from perturbation $V$ in order to become resonant with excited states of $H_0$. Since the local strength of $V$ is $\epsilon\ll \Delta$, the resonant decay process happens only at high enough perturbation order $k \gtrsim (\Delta/\epsilon)^a$ with $a=1$ (naively). This high-order perturbation theory is usually formalized by a sequence of Schrieffer-Wolff transformations (SWT). 

As an example of how and why to use SWTs, consider the Fermi-Hubbard model $H_0+V$ where $H_0=\Delta \sum_i n_{i\uparrow}n_{i\downarrow}$ is on-site repulsion and $V=\epsilon\sum_{i\sim j,\sigma}c^\dagger_{i,\sigma} c_{j,\sigma}+\mathrm{H.c.}$ is nearest-neighbor hopping with $\sigma=\uparrow,\downarrow$. At half-filling and strong repulsion, $H_0$ prefers one fermion per site with a dangling spin-$1/2$, and the perturbation $V$ induces an effective antiferromagnetic Heisenberg interaction acting on the spin degrees of freedom. This superexchange process comes from a SWT \begin{equation}
    \ee^{-A_1}(H_0+V_1)\ee^{A_1} = H_0 +D_1+V_2,
\end{equation}
where $A_1$ is an anti-Hermitian operator that generates the SWT unitary $\ee^{A_1}$, and the subscripts indicate the order of $\epsilon$: $V_1:=V\propto \epsilon^1$ for example. To determine $A_1,D_1,V_2$, one expands the left hand side $\ee^{-A_1}(H_0+V_1)\ee^{A_1}=H_0+V_1 + [H_0+V_1,A_1]+\cdots$ and matches with the right hand side order-by-order. The zeroth order $H_0=H_0$ is trivial. The first order is \begin{equation}\label{eq:SWT_idea}
    V_1 + [H_0,A_1] = D_1,
\end{equation}
while the higher order terms are all included in $V_2=[V_1,A_1]+\frac{1}{2}[[H_0,A_1],A_1]+\mathrm{O}(\epsilon^3)$. To solve \eqref{eq:SWT_idea}, one demands that $D_1$ is block-diagonal $P^>D_1 P^<=0$ with respect to the two subspaces separated by gap (projectors $P^<$ and $P^>$ project onto states below and above said gap), so that $D_1$ is the effective Hamiltonian inside the subspace, up to $\mathrm{O}(\epsilon^2)$ corrections. Since $V_1$ always maps out of the subspace, $D_1 = 0$ in this example, and \eqref{eq:SWT_idea} is solved by \begin{equation}
    A_1=\frac{\epsilon}{\Delta} \sum_{i\sim j,\sigma} P^>_{ij}(c^\dagger_{i,\sigma} c_{j,\sigma}+\mathrm{H.c.})P^<_{ij}-\mathrm{H.c.}
\end{equation}
where $P^>_{ij}$ ($P^<_{ij}$) denotes the subspace on $i,j$ where one site is occupied by two fermions and the other site is empty (each site is occupied by one fermion). One can verify that the resulting $V_2=\frac{1}{2}[V_1,A_1]+\mathrm{O}(\epsilon^3)$, when restricted to the gapped subspace $P^<$, is dominated by the superexchange interaction of order $\epsilon^2$, serving as an effective Hamiltonian in the frame ``rotated'' by unitary $\ee^{A_1}$. Furthermore, since the SWT generator $A_1$ is local, even if $V_2$ contains higher orders that may not preserve subspace $P^<$, it is still quasi-local from evolution by $A_1$. One can further perform a second-order SWT $\ee^{A_2}$ to rotate $H_0 +D_1+V_2$ viewing $V_2$ as perturbation, so that the remaining perturbation becomes $V_3\propto \epsilon^3$, and so on.

In the above example, $H_0$ is commuting and has integer spectrum. Such cases are closely related to Floquet prethermalization \cite{Floq_KS_16,Floq_KS_PRL,Floq_PRB,abanin2017rigorous}, and the perturbation theory is proven to be convergent $\norm{V_k}_\kappa \lesssim \epsilon^k$ up to order $k\sim \Delta/\epsilon$ where resonant decay process kick in. For general gapped $H_0$ like the transverse-field Ising model (the first two terms in \eqref{eq:ZZ-X-Z}), however, technicality arises already at the first step of solving \eqref{eq:SWT_idea}: The solutions for $A_1,D_1$ are no longer strictly local: \begin{subequations} \begin{align}
    A_1 &= \ii \int^\infty_{-\infty}\dd t\, W(t)\ee^{\ii t H_0} V_1 \ee^{-\ii t H_0} ,\label{eq:A1=idea} \\
    D_1 &= \int^\infty_{-\infty}\dd t\, w(t)\ee^{\ii t H_0} V_1 \ee^{-\ii t H_0} ,\label{eq:D1=idea}
\end{align}
\end{subequations}
involve long-time Heisenberg evolution of $V_1$ by $H_0$, where the filter function $w(t)$ is (roughly speaking) the derivative of $W(t)$ so that \eqref{eq:SWT_idea} holds. \revise{More precisely, $W(t)=\pm 1/2$ is (only) discontinuous at $t=0$, with $W(0^\pm) = \pm 1/2$, and decays to zero at $|t|\rightarrow\infty$, so that \begin{equation}
    [H_0,A_1] = \int^\infty_{-\infty}\dd t\, W(t) \frac{\dd}{\dd t}\lr{\ee^{\ii t H_0} V_1 \ee^{-\ii t H_0}} = -\int^\infty_{-\infty}\dd t\, \frac{\dd W(t)}{\dd t} \ee^{\ii t H_0} V_1 \ee^{-\ii t H_0} - \lr{W(0^+)-W(0^-)}V_1,
\end{equation}
from integration by parts, where the two contributions equal to $D_1$ and $-V_1$ respectively by choosing $\frac{\dd W(t)}{\dd t}=-w(t)$ at $t\neq 0$ (see \eqref{eq:W=int_w} for the explicit formula).} A good filter function $w(t)$ exists such that (\emph{i}) $w(t)$ has a compact Fourier transform \begin{equation}\label{eq:filter_idea}
    \hat{w}(E)=0,\quad \forall |E|> \Delta/2,
\end{equation}
so that $D_1$ \revise{(and $A_1$)} is still block-diagonal: for any pair of eigenstates $\ket{E^{>,<}}$ of $H_0$ in the two subspaces, \begin{equation}
    \alr{E^>|D_1|E^<} = \int^\infty_{-\infty}\dd t\, w(t) \ee^{\ii t (E^>-E^<)} \alr{E^>|V_1|E^<} = \hat{w}(E^>-E^<) \alr{E^>|V_1|E^<} = 0;
\end{equation}
(\emph{ii}) $w(t)$ decays nearly exponentially: $w(t) \sim \ee^{-|t|/\ln^2|t|}$ at large $|t|$, so $A_1,D_1$ are still ``sufficiently local'' in the sense that their $\kappa$-norm \eqref{eq:kappa_norm} is bounded, because $\ee^{\ii t H_0} V_1 \ee^{-\ii t H_0}$ contains terms of size $\lesssim ut$ from the Lieb-Robinson bound. Although the evolution generated by $A_1$ does not have the standard Lieb-Robinson bound \eqref{eq:LRB_meta} for strictly exponentially-decaying interactions, one can nevertheless prove that \cite{our_preth} $V_2=\ee^{-A_1}V_1\ee^{A_1}-V_1+\cdots$ is also sufficiently local with a bounded $\norm{V_2}_\kappa$ for a suitably chosen $\kappa$ (recall that in the above commuting case, $V_2$ is quasi-local with strictly exponential tails). Based on these locality estimates, one can iterate the SWT procedure like the commuting case, and show convergence $\norm{V_k}_\kappa \lesssim \epsilon^k$ up to $k=k_*=(\Delta/\epsilon)^a$, where $a\approx 1/(2d-1)$ is no longer $1$ due to the weaker locality bounds. Nevertheless, the remaining perturbation $V_*=V_{k_*}\sim \epsilon^{k_*}$ at this order is already small enough to yield the second line of \eqref{eq:tstarsec3}.

We have briefly reviewed the proof for $R=\infty$, i.e. when $H_0$ is gapped \cite{our_preth}. 
Observe that SWT generated by \eqref{eq:A1=idea} is defined even if $H_0$ has no gap: One just evolves $V_1$ by $H_0$ and integrates with the filter function $W(t)$. The corresponding $D_1$ in \eqref{eq:D1=idea} is still sufficiently local, but does not have the block-diagonal structure anymore because there may be eigenstates of $H_0$ with energy close to $\kzero$, between which $D_1$ can have nonzero matrix elements.  Nevertheless, notice that \emph{local terms} in $D_1$ with diameter $\le R$ see an ``effective gap'' above $\kzero$ from the metastability condition, as shown in Fig.~\ref{fig:proof_idea}(a); hence, they nearly (instead of exactly) stabilize $\kzero$. More precisely, suppose (for now) $V_1=V_{1,\{i\}}$ contains only one local term on site $i$, and suppose $D_1$ is \emph{exactly} supported in the ball $B(i,R/2)$ centered at $i$. The metastability condition implies \begin{equation}
    H_0[\lr{D_1-\alr{D_1}}\kzero] - H_0[\kzero] \ge \Delta.
\end{equation}
However, this is impossible unless $\lr{D_1-\alr{D_1}}\kzero=0$, because $D_1\kzero$ has no support on eigenstates of energy $>\Delta/2$ due to \eqref{eq:filter_idea}. In reality, $D_1$ has tails $\lesssim \epsilon\, \ee^{-R}$ supported outside of $B(i,R/2)$. It turns out that the above argument is robust and yields $\norm{\lr{D_1-\alr{D_1}}\kzero} \lesssim \epsilon\, \ee^{-R}$ in general.
Now let us consider extensive $V_1$: The linearity of \eqref{eq:D1=idea} with respect to $V_1$ implies the existence of an extensive operator $D_1^{\rm P}$ that (\emph{i}) approximates $D_1$: \begin{equation}\label{eq:D1-D1P_idea}
    \norm{D_1-D_1^{\rm P}}_\kappa \lesssim \epsilon\, \ee^{-R},
\end{equation} 
and (\emph{ii}) stabilizes $\kzero$ \emph{exactly}. We repeat such arguments at higher orders. 

To summarize, making $R$ finite in the SWTs merely produces an extra error term $D_k-D_k^{\rm P}\propto \epsilon^k \ee^{-R}$ that does not stabilize $\kzero$ at order $k$. Note that $D_1-D_1^{\rm P}$ cannot be rotated away by higher-order SWTs like typical terms in $V_2$, because $D_1-D_1^{\rm P}$ is already dominated by large-size operators that do not see the effective gap. The dominant error for $v_*$ is then just \eqref{eq:D1-D1P_idea} at first order, which leads to \begin{equation}\label{eq:vstar_idea}
    v_* \lesssim \epsilon\, \ee^{-R} + \ee^{-(\Delta/\epsilon)^a},
\end{equation}
where the second term, present for gapped $H_0$, comes from the remaining perturbation at the terminating order $k_*$. The magnitude of $v_*$ leads to  \eqref{eq:tstarsec3}, wherein we have further estimated which of the two terms in \eqref{eq:vstar_idea} dominates. 

In Appendix \ref{app:B}, we also prove that much of the critical phenomenology of metastability can be generalized to states with long-range entanglement, so long as the Hamiltonian has strictly few-body interactions, and $H=H_0+V$ with the metastable state an eigenstate of $H_0$.  We do not have an analogue of Theorem \ref{thm:H=H0+V} for long-range entangled metastable states, so this latter assumption becomes more nontrivial.

\section{False vacuum decay}\label{sec:falsevacuum}
\begin{figure}
    \centering
    \includegraphics[width=0.5\linewidth]{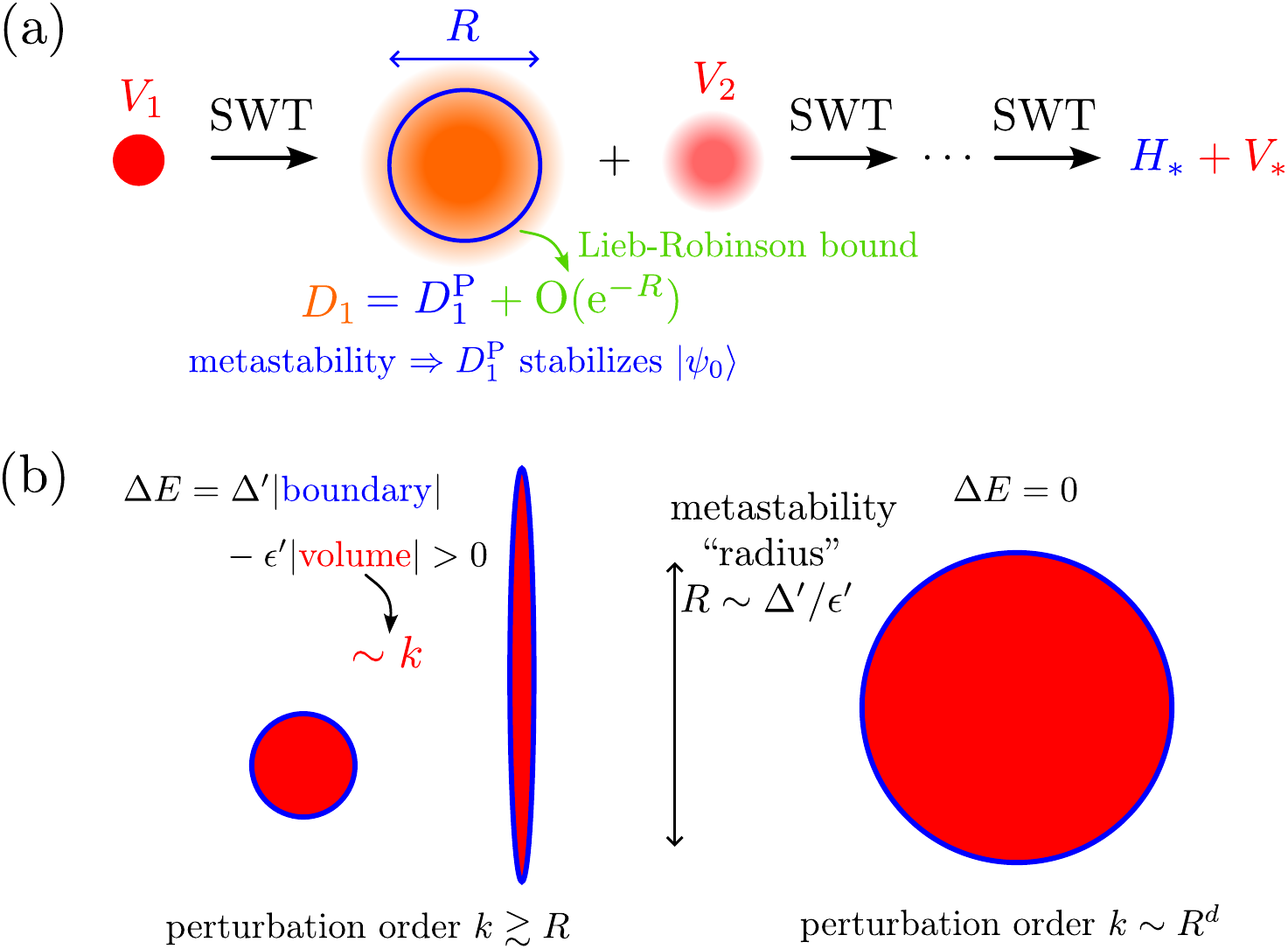}
    \caption{Proof ideas for (a) Theorem \ref{thm:main} and (b) Theorem \ref{thm:Ising_t_d}. (a) We perform Schrieffer-Wolff transformations to locally diagonalize the Hamiltonian as if $H_0$ was gapped, because metastability implies local operators feel an effective gap. (b) By keeping track of the operator volume instead of diameter, resonance happens at a much higher order $k\sim R^d$ instead of $k\sim R$.}
    \label{fig:proof_idea}
\end{figure}

Thus far, we have proved metastability implies nonperturbatively (in $R$ or $\Delta/\epsilon$) slow relaxation of local correlation functions in complete generality.   For the remainder of the paper, we will discuss applications and extensions of this general framework for metastability.  First, we revisit carefully the question of false vacuum decay in certain lattice models.   Using the physical intuition from our theory of metastability, we will prove a bound on the ``lifetime" of the metastable state that matches the path integral calculation \eqref{eq:thitq}, far stronger in $d>1$ spatial dimensions than any rigorous results \eqref{eq:tgap=2d-1} or \eqref{eq:prethermalscaling} in the literature. 

More precisely, we consider generic perturbations to the solvable Ising model in $d$ dimension $H'=H_0'+V'$ where \begin{equation}\label{eq:H0'=Ising}
    H_0'=\frac{\Delta'}{2} \sum_{i\sim j} (1-Z_iZ_j),
\end{equation}
and \begin{equation}\label{eq:V'=maintext}
    V'= \sum_{S:\diam(S)\le \ell_V} V'_S,
\end{equation}
has finite range $\ell_V$.
Here we use primes because we will construct a different Hamiltonian $H_0$ that corresponds to the previous metastable one. We prove that:

\begin{thm}\label{thm:Ising_t_d}
Consider $H_0'$ \eqref{eq:H0'=Ising} that has two ground states $\kzero,\ket{\bm{1}}$. Let $\ell_V$ be a finite number. There exist constants $c_V,c_A,c_D,c_*,\tilde{c}_*$ determined by $d,\ell_V$ such that the following holds: For any perturbation $V'$ \eqref{eq:V'=maintext}
with finite range $\ell_V$ and sufficiently small local norm \begin{equation}
    \epsilon':=\norm{V'}_0 \le c_V \Delta',
\end{equation}
$(H_0',V', \kpsi)$ for either $\kpsi\in\{\kzero,\ket{\bm{1}}\}$ is $(\tilde{c}_* \epsilon' 2^{-k_*}, c_D\epsilon',c_A\epsilon',\revise{1/6})$-locally-diagonalizable (see Definition \ref{def:localdiag}), where \begin{equation}\label{eq:k*cor_maintext}
    k_* =\left\lfloor c_* \revise{(\Delta'/\epsilon')^d} \right\rfloor.
\end{equation}
\end{thm}

Applying Proposition \ref{cor1} then implies that the false vacuum that is either $\kzero$ or $\ket{\bm{1}}$, relaxes slowly with lifetime \begin{equation}\label{eq:tFVD>d}
    t_{\mathrm{FVD}} \gtrsim \exp\mlr{\revise{(\Delta'/\epsilon')^d}}.
\end{equation} 
In comparison, the previous best bound for this setting is the prethermalization bound $t_{\mathrm{FVD}} \gtrsim \exp(\Delta'/\epsilon')$ \eqref{eq:prethermalscaling}, using the fact that $H_0'$ is commuting and has integer spectrum; this is much weaker than \eqref{eq:tFVD>d}. Moreover, \eqref{eq:tFVD>d}
matches the path integral calculation \eqref{eq:thitq}, putting the saddle-point approximation on a firmer ground.

Appendix \ref{app:fvd} contains the proof of Theorem \ref{thm:Ising_t_d}, where we in fact prove a stronger result: $H_0'$ can be any commuting and frustration-free qudit Hamiltonian, with a condition that its ground state degeneracy does not come from trivial effects (e.g. flipping a single qubit may not change the energy \revise{of a system consisting of non-interacting qubits; see Definition \ref{def:nondegen}  for formal statements}).  In the main text, we focus on the simplest Ising example with $\kpsi = \kzero$ and sketch the proof idea.

Since $H_0'$ is gapped, one can perform a SWT (reviewed in the previous section) up to some finite order $k_d \sim 3d/2$, so that 
\begin{equation}\label{eq:finiteSWT_idea}
    U_{\rm finite}^\dagger (H_0'+V')U_{\rm finite} = H_0+V.
\end{equation}
We now consider the Hamiltonian $H_0+V$ in the frame rotated by the SWT unitary $U_{\rm finite}$. Here $H_0=H_0'+D'$ where $D'$ is the block-diagonal effective Hamiltonian (i.e. $D_1$ in \eqref{eq:SWT_idea} for first-order SWT), and the remaining perturbation $V\sim \epsilon \propto (\epsilon')^{d+1}$.  Once $V$ is this small, the metastability radius should control the lifetime of the false vacuum (analogous to \eqref{eq:tstarsec3}). 
The most naive way is to use Proposition \ref{prop:meta_robust} that implies $(H_0, \kzero)$ is $(\Delta'/2, R)$-metastable with $R\sim (\Delta'/\epsilon')^{1/d}$, since the local norm of $D'$ is $\sim \epsilon'$. However, this is too weak to yield \eqref{eq:tFVD>d}. Physically, as sketched in Section \ref{sec:define_meta}, the metastability radius should scale as \begin{equation}\label{eq:RFVD=}
    R\sim \frac{\Delta'}{\epsilon'},
\end{equation}
because this is the shortest length scale of a true-vacuum bubble where the bulk energy gain $\sim \epsilon' R^d$ overcomes the boundary energy penalty $\sim \Delta' R^{d-1}$. It turns out that such a metastability \emph{radius} is still not sufficient, because applying Theorem \ref{thm:main} directly would yield $t_{\rm FVD}\gtrsim \exp(\Delta'/\epsilon')$. 

Instead of Definition \ref{def:meta} and \ref{prop:H0meta},  we will need  a slightly stronger notion of metastability that counts the \emph{volume} of local operators: $(H_0, \kzero)$ has metastability volume $\sim R^d$ with $R$ being \eqref{eq:RFVD=}, in the sense that 
\begin{equation}\label{eq:volume_idea}
    H_0[\ket{\phi}_S\otimes \kzero_{S^{\rm c}}]-H_0[\kzero]\ge \Delta, \quad \text{for any } \ket{\phi}_S\perp \kzero_S \text{ with connected $S$ and } |S| \le R^d.
\end{equation}
\eqref{eq:volume_idea} implies metastability radius \eqref{eq:RFVD=} for Definition \ref{def:meta} due to \eqref{eq:cd}, and is stronger than that because \eqref{eq:volume_idea} also allows for $S$ with $\diam S\gg R$, like a long-stripe region. This is the key to obtain the tight exponent in \eqref{eq:tFVD>d}, illustrated in Fig.~\ref{fig:proof_idea}(b). Intuitively, \eqref{eq:volume_idea} comes from the geometrical fact that there exists a constant $c'_d$ such that \begin{equation}
    |\partial S|\ge c'_d |S|^{1-1/d},
\end{equation}
for any subset $S$. As a result,
any true-vacuum bubble that flips spins in $S$ with volume bounded in \eqref{eq:volume_idea} has larger boundary energy penalty than the bulk energy gain. A technical complication we overcome to show this result is that $H_0$ can map between computational basis states, so long as they do not (locally) look like $|\mathbf{0}\rangle$ or $|\mathbf{1}\rangle$; therefore, \eqref{eq:volume_idea} does not directly follow from the form of $H_0^\prime$.

We then follow the idea of the previous section to show that when doing SWT for $H_0+V$, the obtained effective Hamiltonian (e.g. $D_1$ \eqref{eq:D1=idea} with $V_1=V$) stabilizes $\kzero$ approximately with local error $v_*\sim \exp(-R^d)$. Given \eqref{eq:volume_idea}, it suffices to show that local terms in e.g. $D_1$ that act on a region of volume $\gtrsim R^d$ are highly suppressed $\sim v_*$. However, this is hard to prove using \eqref{eq:D1=idea}, because $\ee^{\ii t H_0} V \ee^{-\ii t H_0}$ grows to diameter $\gtrsim R$ and thus volume $\gtrsim R^d$ at time $t\gtrsim R$ from the constant Lieb-Robinson velocity of $H_0$, but at that timescale, $w(t)\sim \ee^{-t}$ is only suppressed by $\sim \ee^{-R}\gg v_*$.

This technical issue is the reason we restrict the original unperturbed part $H_0'$ to be commuting. For this special case, although $H_0$ is not commuting anymore, it is close to a commuting $H_0'$, and we have more control over the operator locality than \eqref{eq:D1=idea}. A similar technique has been used in \cite{topo_Hastings} to prove stability of topological order given a commuting $H_0$:  we can write $H_0=H_0'+\epsilon' D'$, and solve the SWT equation \eqref{eq:SWT_idea} order-by-order in $\epsilon'$. The zeroth order is just SWT for the commuting $H_0'$, so that the corresponding zeroth-order solutions for $A_1,D_1$ are strictly local for a local $V_1$, as in the Hubbard model example; similarly, each order higher just grows the support of operators by a strictly finite amount. It turns out that this method avoids the dangerous operator growth that would arise in \eqref{eq:D1=idea}, and leads to the desired bound $v_*$ for local terms with large volume $\gtrsim R^d$. 

To summarize the idea of Theorem \ref{thm:Ising_t_d}, one first performs a finite round of SWTs to suppress the remaining perturbation $V$ to sufficiently high but constant (in $\epsilon'$) order \eqref{eq:finiteSWT_idea}, where the effective Hamiltonian $H_0$ is metastable with volume $\sim R^d$. One then performs SWTs for this new perturbation problem $H_0+V$ up to high order \eqref{eq:k*cor_maintext}, following the spirit of Theorem \ref{thm:main}, whereas the metastability volume (instead of diameter) together with the solvability of $H_0'$ lead to a much stronger bound \eqref{eq:tFVD>d} than the more general theory of Theorem \ref{thm:main}.
In practice, for the second stage of perturbing $H_0+V$, we do not follow directly the proof structure of \cite{topo_Hastings}, which has two layers of perturbation series: the outer layer performs SWTs, and each SWT step is embeded by an inner-layer perturbation series for solving \eqref{eq:SWT_idea}. We find it more convenient to reorganize the orders of $\epsilon'$ to have just a single layer of perturbation series, where we always perform SWT using $H_0'$ instead of $H_0$. This proof structure looks more like the original prethermalization proof \cite{abanin2017rigorous}; see Appendix \ref{app:fvd} for details.

\section{Constructing metastable states}\label{sec:models}
Next, we show that our metastability formalism can also be applied to systems and states which are not a symmetry-breaking false vacuum.  Indeed, our theory of metastability leads to a mechanism to systematically generate ``prethermal scars" and explain their long lifetime.

\subsection{One-dimensional (generalized) helical models}\label{sec:helix}
We first turn our attention to one-dimensional models, where we can present a simple and explicit construction for finding \emph{exponentially many finite energy-density} metastable states in system size.  Consider the following model acting on a one-dimensional chain of qudits with local Hilbert space dimension $q$, which we will dub the ``unperturbed helical model":  \begin{equation}
    H_0 = \sum_i (1-P_i) K_i (1-P_i),
\end{equation}
where $P_i$ and $K_i$ act on three sites $\lbrace i-1,i,i+1\rbrace$, $K_i$ is an arbitrary operator with all positive eigenvalues: $K_i \ge \Delta > 0$, and $P_i$ is a projector Hamiltonian defined as \begin{equation}
    P_i = |000\rangle \langle 000|_i + |001\rangle\langle 001|_i + |(q-1)00\rangle\langle (q-1)00|_i + \sum_{k=0}^{q-2} |k(k+1)(k+2)\rangle\langle k(k+1)(k+2)|_i
\end{equation} 
where we define $|q\rangle = |0\rangle$.  Here the $|abc\rangle_i$ is shorthand for the state $|a_{i-1}b_i c_{i+1}\rangle$. 

Observe that $H_0$ has exponentially many ground states of the form \begin{equation}
    |\text{gs}\rangle = |0_10_2\cdots 0_a 1_{a+1}\cdots (q-1)_{a+q-1}0_a 0_{a+1}\cdots 0_b 1_{b+1}\cdots \rangle .
\end{equation}
It is easy to see that in a chain of length $L$ there are at least $2^{\lfloor L/(q+1)\rfloor}$ such ground states -- in each block of length $q+1$ we pick between motif $01\cdots (q-1)0$ or $00\cdots 00$.  By construction, any state which is not one of these helical ground states has energy at least $\Delta$, since it must be composed of orthogonal states to these ground states, all of which by construction must (in the computational basis) have at least one triplet of sites which is annihilated by $P_i$, and which thus has energy of at least $\Delta$.

To obtain metastable states, we simply need to modify $H_0$ by a weak perturbation, such as \begin{equation}\label{eq:H=helix-00}
    H = H_0 - \mu \sum_{i=1}^N |0\rangle\langle 0|_i
\end{equation}
which picks out the helix-free ground state as unique.  Any state with a finite helix density now has finite energy density, yet it is metastable because there is no way to locally flip the helix motif to 0s without an operator that acts on at least $q-1$ sites.  At small $\mu$, we thus deduce that this model is $(\Delta - (q-1)\mu, q-1)$ metastable, and any further perturbation to \eqref{eq:H=helix-00} leads to prethermalization proven in Theorem \ref{thm:main}.

Notice that this is not the only valid construction.  We can in fact look for much longer metastable strings.   One explicit model has helix-antihelix motifs of the type $01\cdots (q-1)(q-1)\cdots 10$. Let us consider for example the case $q=3$ and define the Hamiltonian\footnote{Formally speaking, $H_0$ here is also covered by the assumptions of the generalized prethermal bounds of \cite{roeck2019weak_driven} since all terms in $H_0$ commute.  However, we could also study modified $H_0$ where this commutativity assumption is relaxed, or where the energy levels are highly incommensurate, and our metastability bounds continue to hold.}
\begin{equation}
    H_0=\sum_{i}  \left( 1-\sum_{\mathbf{s}\in B}P_i^{(\mathbf{s})}\right)+\sum_i (\mu_1 |1\rangle \langle 1|_i+\mu_2 |2\rangle \langle 2|_i)
\end{equation}
where $P_{i}^{(\mathbf{s})}=|\mathbf{s}\rangle \langle \mathbf{s}|_{i-1,i,i+1}$ is the (three-site) projector on the $\mathbf{s}$ configuration and $B=\{000, 001,012,122,\allowbreak 221,210,100,101,020,002,200\}$. The set $B$ is chosen such that for $\mu_1=\mu_2=0$ there is an exponential number of degenerate ground states, consisting of blocks of motifs of the types $012210, 00200$ interspersed with sequences of $0$s.  For $\mu_1, \mu_2>0$, the uniform 0 state is the only ground state, while long motifs of the type $012210$ can be metastable depending on the parameters $\mu_1, \mu_2$. We study this model numerically for the case $\mu_1=0.2$, $\mu_2=0.12$: in this case the repeated $012210$ motif is $(\Delta_R, R)$ metastable for $R=1,2$ and 
\begin{equation}
\label{eq:DR}
\Delta_1=1-\mu_1-\mu_2, \qquad \Delta_2=1-\mu_1-2\mu_2.
\end{equation}
More generally, any state consisting of repeated motifs of the types $01221$ and $012210$ is similarly $(\Delta_R, R)$ metastable. Note that the number of these metastable states is exponentially large in the system size. We also observe that the ground state has a smaller local gap than these metastable states, since applying a single flip from $\dots 00000\dots $ to $\dots 00200\dots $ only costs an energy $\Delta=\mu_2$.
 Similarly, a state with a certain number of $01221$ motifs interspersed with long strings of zeros is therefore not $(\Delta_R, R)$ metastable (with the same $\Delta_R$ $R=1,2$ as the motifs-only states), because the $00200$ flip is ``easy". However, we do expect that physically the helix-antihelix motifs are long-lived because they are locally protected by a large gap, and therefore the low-energy excitations represented by $00200$ flips should not  thermalize them. 

We study the spectrum of the model for a chain of $L$ qudits with periodic boundary conditions in the presence of a perturbation $V$: $H=H_0+V$ with
\begin{equation}
    V=\epsilon\sum_i (\ket{0}\bra{1}_i+\ket{1}\bra{2}_i+\ket{2}\bra{0}_i+\text{H.c.}).
\end{equation}
The exact diagonalization results are shown in Fig.~\ref{fig:doublehelix}-(a).

\begin{figure}[t]
    \centering
    \includegraphics[scale=0.55]{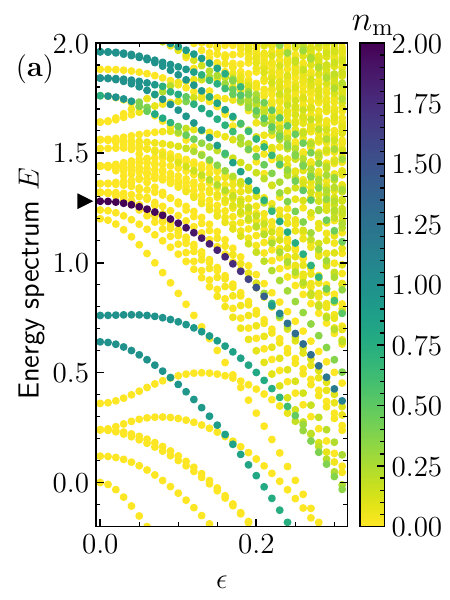}
    \includegraphics[scale=0.55]{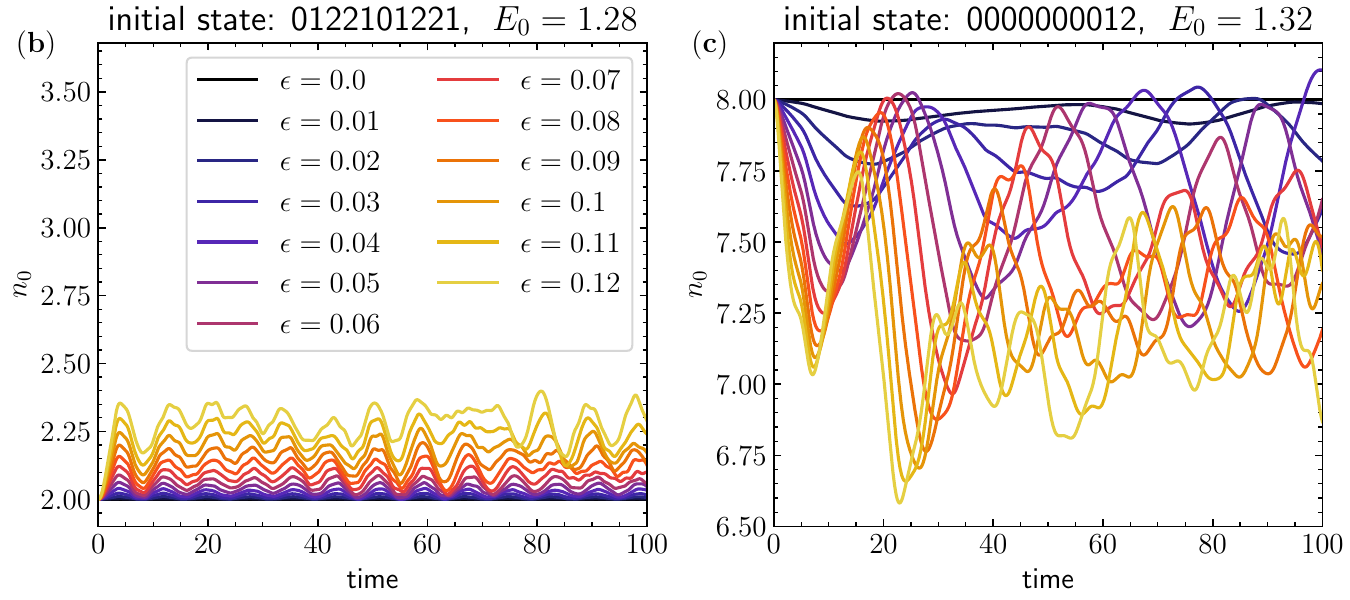}
    \caption{(a) Energy spectrum vs perturbation strength $\epsilon$ in the zero momentum, inversion even sector. The color indicates, for each energy eigenstate, the expectation value of the number of motifs $n_m=\sum_i \ket{012210}\bra{012210}_{i,\dots, i+5}$. The arrow points at the metastable state $\ket{0122101221}$. (b) Expectation value of the number of $0$s $n_0=\sum_i \ket{0}\bra{0}_i$ as a function of time for different values of the perturbation strength $\epsilon$. The initial state is the metastable state $\ket{0122101221}$. (c) Same as (b), but for the non-metastable initial state $\ket{0000000012}$. All the results were obtained for $L=10$, $\mu_1=0.2$, $\mu_2=0.12$. }
    \label{fig:doublehelix}
\end{figure}

Even for relatively small system size ($L=10$), we see that for $\epsilon=0$ the state $\ket{0122101221}$ is in a rather dense region of the spectrum, and is very close in energy with states that have no $012210$ motifs. When the perturbation strength $\epsilon$ is non-zero, we still observe a distinct eigenstate in the spectrum that seems smoothly connected with the $\epsilon=0$ eigenstate $\ket{0122101221}$. This persists to values of the perturbation strength that are much larger than the characteristic level spacing in the region of the spectrum around this eigenstate. These data show the effect of metastability on the robustness of energy eigenstates: since $\ket{0122101221}$ is $(\Delta_R, R)$ metastable for $\epsilon=0$, it (or better, a state obtained from it through a quasi-local unitary transformation) has long lifetime and small energy variance. In other words, it is an approximate eigenstate, and shows very little hybridization with the rest of the spectrum when perturbed. To test this long-lived robustness to perturbations we also study the time evolution (quench) taking two different initial states: the metastable $\ket{0122101221}$ (Fig.~\ref{fig:doublehelix}-(b)) and the state $\ket{0000000012}$ (Fig.~\ref{fig:doublehelix}-(c)), which is not metastable (a single-site flip from $1$ to $0$ can lower its energy) but has a comparable energy to the $\ket{0122101221}$ state. We consider the expectation value of a local observable as a function of time after the quench for different values of the perturbation strength. For both initial states, we observe some oscillations, but the amplitudes are much smaller for the metastable state compared to the non-metastable one. The small oscillations observed for the former can be understood as the effect of the quasi-local unitary transformation, and should be accompanied by small spreading of correlations. 

To verify this, we perform a numerical simulation of the quench for larger system sizes using the time-dependent variational principle (TDVP) algorithm for matrix product states \cite{TDVP1,TDVP2,tenpy}. We consider a system of $L=30$ sites with open boundary conditions (including some boundary terms that mimic a prolonged chain with '0's for the external qudits) and study the time evolution starting from two initial states with the same energy $E=3.2$: \emph{(1)} a $(\Delta_R, R)$-metastable state (with $R=1,2$ and $\Delta_R$ as in Eq.~(\ref{eq:DR})), with '012210' motifs and \emph{(2)} a generic, non-metastable state. For the timescale considered, the evolution from the metastable state shows no significant changes in the expectation values of local observables [Fig.~\ref{fig:quenchTDVP}-(a)] and an almost undetectable growth of the entanglement entropy [Fig.~\ref{fig:quenchTDVP}-(c)]. The generic state, on the other hand, shows a much faster evolution of local observables [Fig.~\ref{fig:quenchTDVP}-(b)] and faster growth of the entanglement entropy [Fig.~\ref{fig:quenchTDVP}-(d)]. In Fig.~\ref{fig:quenchTDVP} we also observe that, while the half-chain entanglement entropy of the generic state shows a saturation at long times due to the finite bond dimension $\chi$ utilized in the simulations, for comparable times the entanglement entropy of the metastable state remains small and independent on $\chi$.

\begin{figure}
    \centering
    \includegraphics[scale=0.55]{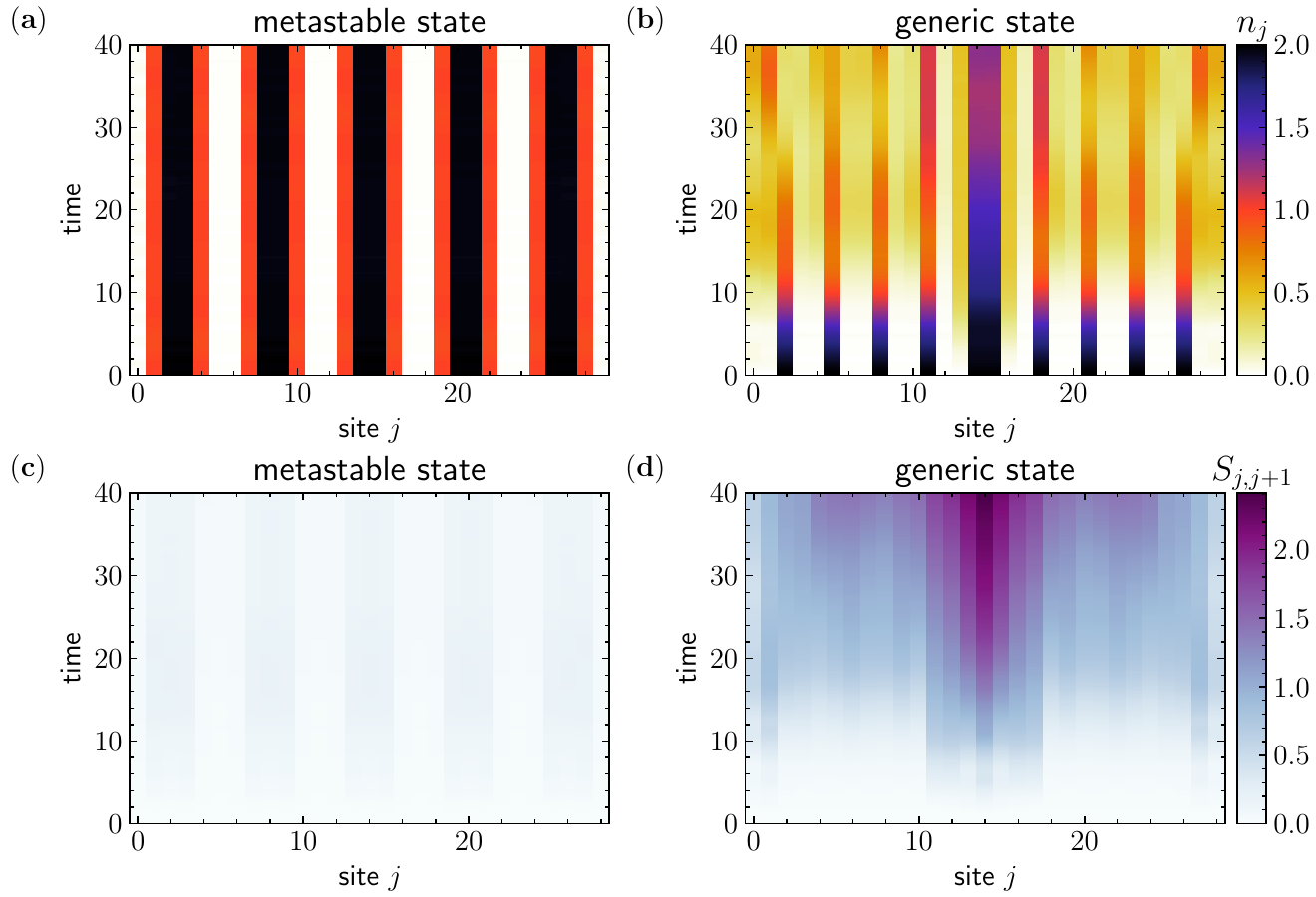}
\includegraphics[scale=0.55]{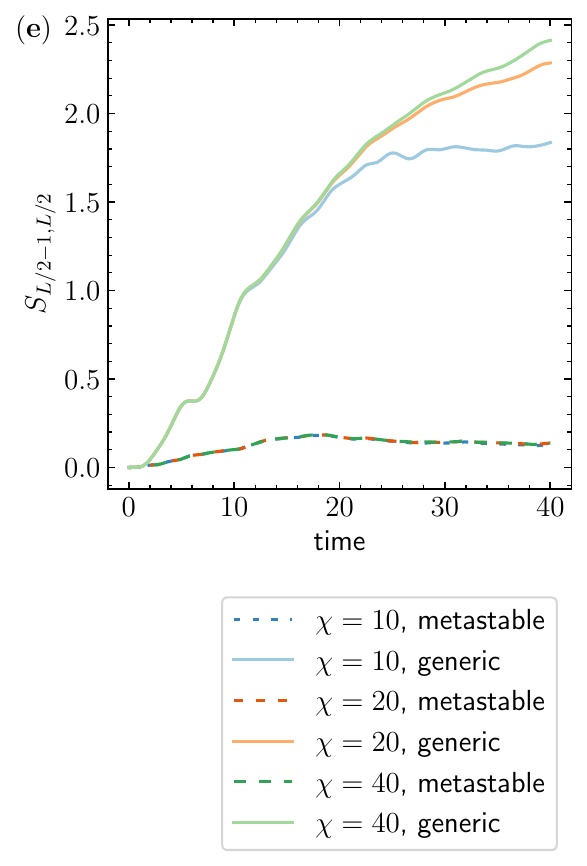}
    \caption{Time evolution of the local observable $n_j=\ket{1}\bra{1}_j+2\ket{2}\bra{2}_j$ for two initial states with the same energy $E=3.2$: (a) a metastable state $\ket{0122100122100...}$ and (b) a non-metastable state $\ket{00200200...00220020...}$. Entanglement entropy $S_{j,j+1}$ at a cut between site $j$ and $j+1$ in the evolution from the same (c) metastable and (d) non-metastable states as in (a), (b) respectively.  All the results were obtained for $L=30$, $\mu_1=0.2$, $\mu_2=0.12$, and maximal bond dimension $\chi=40$. (e) Half chain entanglement entropy for different values of the bond dimension, which demonstrates that the slow dynamics of the metastable state is not limited by the finite bond dimension.}
    \label{fig:quenchTDVP}
\end{figure}

\subsection{The PXP model}
Next, we address the extent to which our theory of metastability can provide insight to the unexpectedly slow thermalization of the one-dimensional PXP model \cite{PXP_NP18}, whose Hamiltonian is 
\begin{equation}\label{eq:PXPH}
    H = \sum_{i=1}^{N} P_{i-1}X_i P_{i+1}.
\end{equation}
Here, each site of a one-dimensional lattice with periodic boundary conditions contains a qubit with Pauli matrices $X,Y,Z$, and $P_i=(1+Z_i)/2$. $H$ is able to excite a local $\ket{0}$ to $\ket{1}$ only if its two neighbors are all zero, and we focus on the invariant subspace $\mathcal{H}_{\rm PXP}$ containing no neighboring excitations $\ket{11}$.
The PXP model has been realized experimentally in Rydberg atom arrays \cite{PXP_exp17} and studied extensively as a model with approximate quantum many-body scars \cite{PXP_NP18,PXP_Papic18,PXP_Khemani19,PXP_SU2_19,PXP_GS19,PXP_exact19,PXP_TDVP19,PXP_SU2_20,PXP_TDVP21,PXP_all2all22,PXP_composite22,PXP_spin1_23,PXP_rev_Papic,PXP_rev_Kehmani}. 

Most notably, starting from a charge density wave (CDW) state $\ket{010101\cdots}$, the system exhibits long-lived oscillations probed by e.g. the CDW order parameter, even if the dynamics is not perturbatively close to any known form of solvability. Although heuristic arguments have been made to explain this surprisingly slow thermalization process, a physically satisfactory theory is still missing.

Here, we apply our rigorous theory of metastability to the PXP model. First, we ask whether the CDW state is metastable: The answer is no, because flipping a single qubit already maps the state to an orthogonal one $\ket{01000101\cdots}$ with the same energy $0$. In other words, $(H,\ket{010101\cdots})$ is not $(\Delta,R)$-metastable for any $\Delta>0, R\ge 0$\footnote{Note that $R$ counts the set diameter defined in Section \ref{sec:define};  $R=0$ for a single-site operator.}. Nevertheless, one can ask if there is any metastable state for $H$ beyond the ground state (which numerical simulations find to be gapped \cite{PXPspec}). Intriguingly, we find \begin{equation}\label{eq:state_0-}
    \ket{0-}:=\ket{0-0-0-\cdots},
\end{equation}
to be metastable, together with its translated cousin $\ket{-0}:=\ket{-0-0-0\cdots}$, which share a similar form with the CDW states.  Here $X|\pm\rangle =\pm |\pm\rangle$ on a single site. Note that the overlap $\langle -0|0-\rangle = 2^{-N/2}$ vanishes quickly in the thermodynamic limit.   We fix $N$ to be even for simplicity. 

The intuition for the product state $\ket{0-}$ being metastable is as follows. Observe that $\ket{0-}$ minimizes all terms with even $i$ in \eqref{eq:PXPH}, and the odd terms have zero expectation. To gain (lower) energy from the odd terms, one needs to add $\ket{1}$ components to odd sites. However, the Hamiltonian \eqref{eq:PXPH} is ``constrained'': To excite an odd site, one needs to first put its two even site neighbors from $\ket{-}$ to $\ket{0}$, which requires an operator of size $R+1\ge 3$. Furthermore, this process necessarily loses a large energy $\approx 2$ from the two even PXP terms, with $\ket{0-0-00-00-0-}$ being a concrete example. So to gain enough energy that compensates the loss, one needs to apply a very nonlocal change to the state $\ket{0-}$, leading to a large metastability range $R$.

Although Definition \ref{def:meta} formally does apply to \eqref{eq:state_0-}, the name ``metastability" might be misleading: although this state is excited and has a finite energy density relative to the ground state, it is locally close to the ground state, so in the sense of Proposition \ref{prop:meta_robust}, the ``metastability" of the gapped ground state can in turn imply ``metastability" of $|0-\rangle$.  This is not, then, necessarily a metastable state which is ``locally far" from the ground state, like our previous examples of the false vacuum or helical state.  Nevertheless, since we have shown that our theory of metastability has a number of useful implications, it is worth investigating the dynamics of $|0-\rangle$ a little further.

Exact diagonalization (ED) verifies in Fig.~\ref{fig:PXP}(a) that $\ket{0-}$ has metastability range $R=7$ in the bulk ($\ket{-0}$ is the same due to symmetry). More precisely, for each $R$, we diagonalize a ``metastable Hamiltonian'' on $R+1$ sites with boundary conditions determined by the metastable product state, to minimize the energy difference in \eqref{eq:H0>Delta1} (with $\kzero$ replaced by $\ket{0-}$); the result is labeled by $\Delta$.   The maximal metastability range is the largest $R$ with a positive $\Delta$.\footnote{Of course, if the value of $\Delta$ is quite small at this value, better bounds may arise from treating the model as $(\Delta^\prime,R^\prime)$-metastable with $R^\prime<R$ but $\Delta^\prime>\Delta$.} $\Delta$ for odd $R$ turns out to equal that for the even $R-1$; the reason is that the left boundary condition for $\ket{0-0-\cdots}$ fixes the leftmost qubit to $0$ to minimize energy. In Fig.~\ref{fig:PXP}(a), we also plot the minimal energy difference when changing the boundary condition from the metastable Hamiltonian to the usual open and periodic boundaries, where $\Delta$ becomes negative already before $R<4$. This comparison emphasizes the fact that by treating the interactions among different parts (i.e. boundary conditions) more carefully, a system can be much more metastable than what finite-size ED would naively suggest.

\begin{figure}[t]
    \centering
    \includegraphics[width=\linewidth]{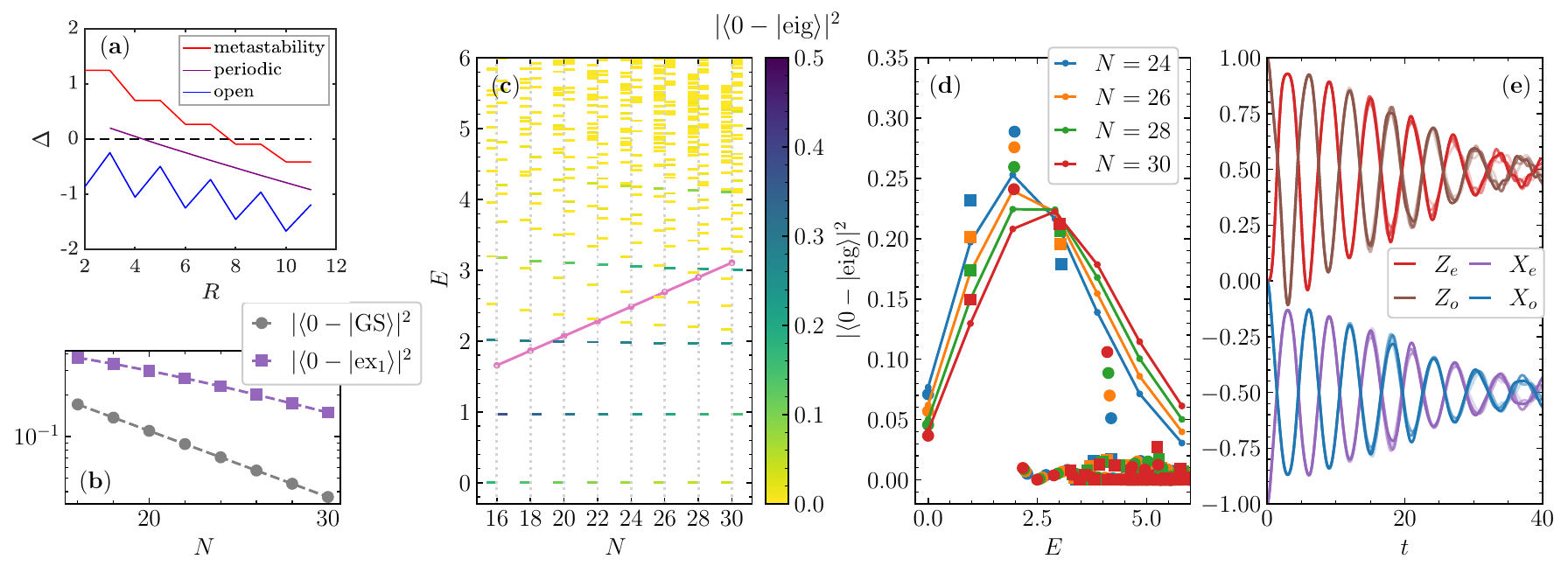}
    \caption{Metastability of state $\kpsi=\ket{0-}$ in the PXP model with periodic boundaries from ED. (a) For a given $R$, the energy difference \eqref{eq:H0>Delta1} is minimized to obtain $\Delta$, which shows the metastability range $R=7$. To this end, an $(R+1)$-site metastable Hamiltonian with suitable boundary conditions is diagonalized. The data is compared with more usual boundary conditions, where the periodic one only applies to odd $R$. (b) Overlap of the $\ket{0-}$ state with the ground state and first excited state. Both overlaps decay exponentially with the system size $N$.  (c) Low energy spectrum (dots) $E$ of $H$ in the constrained subspace $\mathcal{H}_{\rm PXP}$, with the ground state energy set to $0$. For each $N$, the energy levels in the $\mathcal I =+1, k=0$ sector ($\mathcal I =-1, k=\pi$  sector) are plotted on the left (right) of the dashed line. The energy of $\ket{0-}$ (pink solid line) is deep in a continuum of excited states at large $N$. The color of each eigenstate represents its overlap with $\ket{0-}$, showing that $\ket{0-}$ is approximately supported on a minority of 
    eigenstates. 
    (d) An alternative view of the eigenstate overlap of $|0-\rangle$, showing that the state is predominantly supported on a handful of eigenstates of approximately equal energy spacing (circle markers for $\mathcal I =+1, k=0$ sector, square markers for $\mathcal I =-1, k=\pi$, no line).  For comparison, we have shown (solid line) what the spectrum would look like if the state was made of a Poisson-distributed collection of non-interacting ``quasiparticles" $\sum \alpha^{n/2} \mathrm{e}^{-\alpha n/2}(n!)^{-1/2} \cdot |n\rangle $, where $|n\rangle$ denotes the state with $n$ quasiparticles at momentum $k\approx \pi$.  The similarity in the spectra is  consistent with a quasiparticle interpretation for the metastable state. (e) We observe slow thermalization and persistent oscillations in a quench from the metastable state, as evidenced by $Z_e = \langle 2N^{-1}\sum_{i \; \text{even}} Z_i\rangle$, for example. Data for $N=24,26,28,30$ are plotted (they almost exactly overlap), with color intensity increasing with $N$.
    }
    \label{fig:PXP}
\end{figure}

The metastable state $\ket{0-}$ has energy density $-0.5$ at large $N$, which is noticeably larger than the ground state energy density $-0.6034$ that has been estimated numerically using the density matrix renormalization groups \cite{PXPspec,PXP_GS19}. Therefore, although it is at relatively low energy comparing to the CDW state (which lies exactly in the middle of the spectrum -- i.e. at infinite temperature), it is still probing physics at a nonvanishing energy density, or the physics at a low but positive temperature.
We confirm this by diagonalizing the lowest energy levels of the PXP Hamiltonian up to $N=30$ in the two symmetry sectors where $\ket{0-}$ is supported (the inversion-even  zero-momentum sector $(\mathcal I=+1,k=0)$ and the inversion-odd $\pi$-momentum sector $(\mathcal I=-1,k=\pi)$), and show in Fig.~\ref{fig:PXP}(b) that the overlaps of the metastable state with the ground state and the first excited state decay exponentially with $N$. Moreover,  the single particle gap is $\tilde \Delta\approx 0.9682$ in the $k=\pi$ momentum sector \cite{PXPspec}, so for $N\ge 20$ the energy of the $\ket{0-}$ state is above the two-particle threshold, in a ``dense'' continuum of excited states. 

Numerically, we construct the exact low energy eigenstates for PXP chains up to $N=30$, and find that the overlap of the metastable $|0-\rangle$ state is concentrated on only a handful of states: Fig.~\ref{fig:PXP}(c).  Intriguingly, one possible explanation for this is that the metastable state is ``close to" a superposition of weakly interacting quasiparticles at lattice momentum $k=\pi$.
A similar quasi-particle picture is expected to be quite general and has been applied to low-energy-density states in other models, demonstrating long-lived non-thermal behavior \cite{falseVac_quasip,Robertson2024}. The underlying intuition is that states with low energy density correspond to a low density of quasiparticles, which, being dilute, interact only weakly. Although the system will eventually thermalize due to quasiparticle scattering, this process occurs over a long timescale. Our framework offers a systematic method to justify and rigorously establish such intuitive quasi-particle descriptions.
While our approach only applies to a low energy density state (the $\ket{0-}$ state), in the PXP model such quasiparticles have been argued to be responsible for the slow dynamics even at very high energy density (i.e., starting from the infinite temperature CDW state)\cite{PXP_exact19,PXP_GS19,PXP_spin1_23,PXP_rev_Kehmani}. 

To argue that the quasi-particle interpretation applies to the $\ket{0-}$ state, in
Fig.~\ref{fig:PXP}(d) we plot the energies of the eigenstates with which $|0-\rangle$ has high overlap, and compare them to what we would predict if we considered a cartoon model where the metastable state took the form \begin{equation}
    |0-\rangle \sim \sum_{n=0}^\infty \sqrt{\frac{\alpha^n \mathrm{e}^{-\alpha n}}{n!}} |\text{$n$ quasiparticles at $k=\pi$}\rangle, \label{eq:cartoonstate}
\end{equation}
where the ``$|\text{$n$ quasiparticles}\rangle$" state is a cartoon (not actual) quantum state which is approximately a superposition of $n$ noninteracting $k=\pi$ quasiparticles, and which is assumed to obey \revise{$H|n\rangle = \tilde \Delta n |n\rangle$ where $\tilde \Delta$ is the single particle gap above}.   The $\alpha$ parameter is chosen so the average energy of the cartoon state (\ref{eq:cartoonstate}) is equal to $\langle 0-|H|0-\rangle$.   There is a surprisingly good fit between this crude model and the numerical data, which suggests that our metastable state is close to a state built out of these quasiparticles, and that there exist exact eigenstates of the PXP model which do admit (approximately) an interpretation in terms of these quasiparticles.

From this perspective, our theory of metastability can then provide a rigorous justification for the intuitive argument that these weakly interacting quasiparticles' slow decay is responsible for the slow thermalization of the $\ket{0-}$ state in the PXP model.
We investigate the slow thermalization from the initial state $\ket{0-}$ in
Fig.~\ref{fig:PXP}(e), where we show the presence of persistent oscillations in the expectation value of operators $Z$ and $X$ on even/odd sublattices -- in the $|0-\rangle$ (or $|-0\rangle$) states, such expectation values show large differences between even and odd sites, whereas the differences will be zero in a typical thermal state.  The persistence of large oscillations confirms that the metastable state appears athermal over long time scales.   In contrast to the expectations from Theorem \ref{thm:main}, we do \emph{not} see these correlation functions approximately constant at short time scales.  This is not, of course, a contradiction to our rigorous result, but a consequence of finite $R=7$ and -- far more importantly -- the very tiny metastability gap $\Delta$, which is not large enough to guarantee slow dynamics.
In principle, we could adjust the couplings in the PXP Hamiltonian on odd and even sites to make the $|0-\rangle$ state a metastable state with as large of an $R$ as we wish.  For sufficiently large $R$ (albeit small enough $R$ that  $\Delta/\epsilon$ is not too small), Theorem \ref{thm:main} can provide strong bounds on the thermalization time scale that differ from Fermi's Golden Rule.  

Ultimately, the theory of metastability guides us to discover
initial states for which we can argue, a little more ``carefully", for slow dynamics in the strict thermodynamic limit.  As noted above, we can certainly modify the Hamiltonian to make such metastability-derived bounds increasingly strict, if desired. Hence,  our methodology provides a new path towards obtaining controlled models of slow thermalization and constrained dynamics, which we expect could be applied to many other models, in which exact diagonalization for medium-size ($R$) systems could help to identify metastable short-range entangled states, for which Theorem \ref{thm:main} can imply slow thermalization.

In some cases, a quasiparticle interpretation can accurately describe the quench dynamics also away from the ground state, as was shown for example for an initial state close to the false vacuum in the mixed field quantum Ising chain  \cite{lagnese2023detecting}. Our results hint that, more broadly, metastable
states with sufficiently low energy density might have long-lived quasiparticles as well, irrespective of whether they are close to the global ground state.  These quasiparticles may be of very different nature from the low energy ones.  This is an intriguing possibility as it is far from obvious that such a quasiparticle description would be valid so far from the ground state. Although such quasiparticle interpretations have been proposed to describe the physics of quantum many-body scars in the PXP model by many authors \cite{PXP_exact19,PXP_GS19,PXP_spin1_23,PXP_rev_Kehmani}, there is not a rigorous understanding of how and why such quasiparticles arise, especially as they seem to survive to very high energies. Our observation that $\ket{0-}$ is metastable does not answer this question, yet our numerical results in Fig.~\ref{fig:PXP}(c,d) suggest that our theory of metastability may provide valuable new insight into the origin of quasiparticles.  

In the literature, 
nonthermal dynamics in the PXP model is often justified by trying to deform $H$ to a Hamiltonian $H_{\rm deform}$ that has exact scars \cite{PXP_rev_Kehmani}, with $H-H_{\mathrm{deform}}$ a sum of small local terms.  However, we are not aware of any systematic algorithm to obtain $H_{\mathrm{deform}}$ for general models, and further whether there would be any guarantee that the dynamics of $H_{\rm deform}$ would be robust when deforming back to PXP \cite{PXP_pert20,PXP_pert21}; indeed, one might expect Fermi's Golden Rule to lead to perturbative decay rates at second order in the size of $H-H_{\mathrm{deform}}$. Our Theorem \ref{thm:H=H0+V} provides useful insights in this regard, which guarantees e.g. the proximity of $H$ to some other Hamiltonian that stabilizes $\ket{0-}$, along with robustness of dynamics to far higher orders in perturbation theory. 


\subsection{Two-dimensional Ising model}

\begin{figure}
    \centering
    \includegraphics[width=0.5\linewidth]{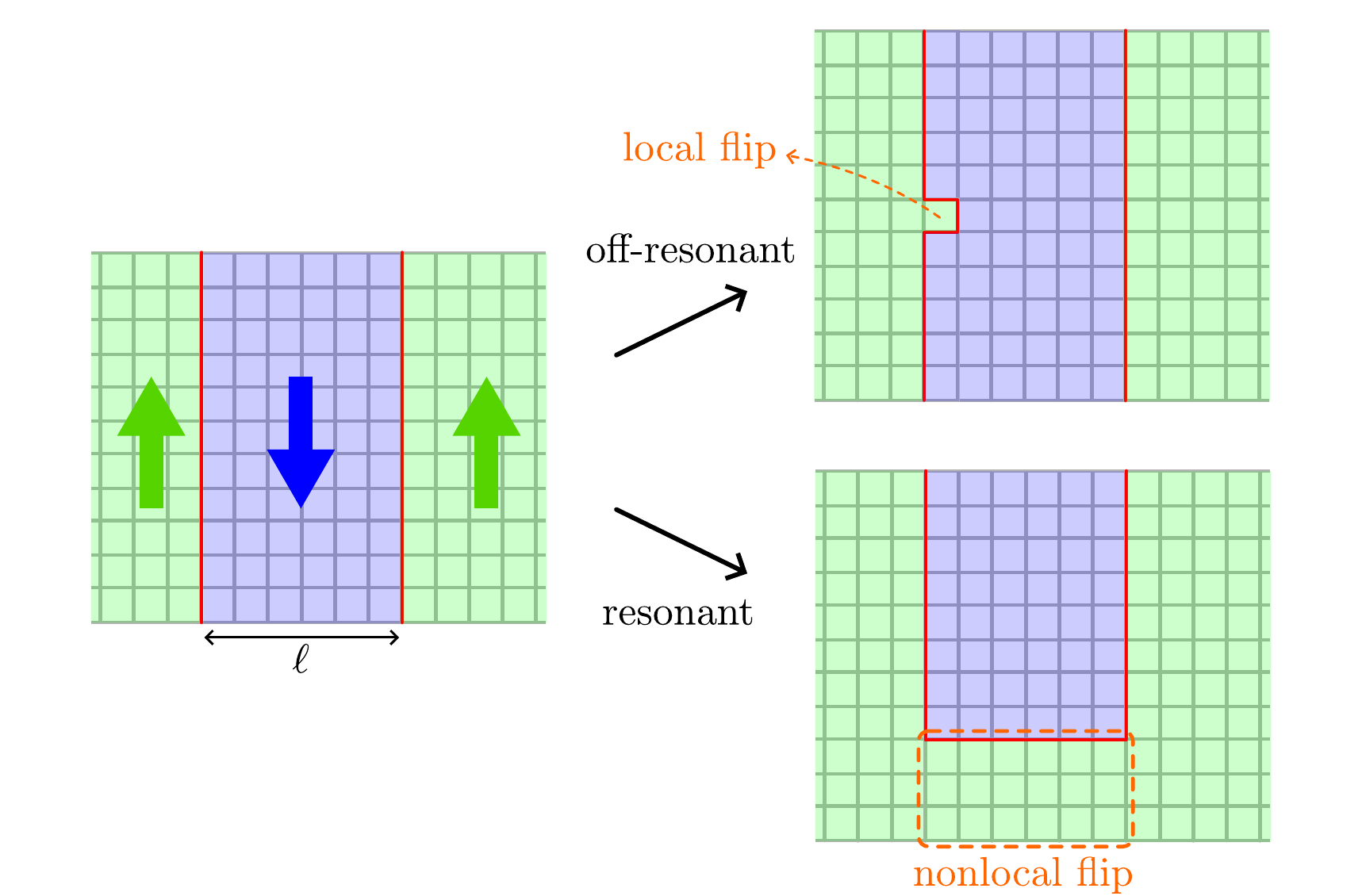}
    \caption{\emph{Left:} A metastable state in the 2d Ising model. 
 \emph{Right:} Local operators always raise the energy, consistent with our definition of metastability.  A nonlocal operator of radius $\ell$ is required to find any many-body resonance.  }
    \label{fig:2dIsing}
\end{figure}

In Section \ref{sec:helix}, we constructed one-dimensional models with exponentially many finite energy-density metastable states. However, we need to tune the local Hilbert space dimension to get an asymptotically long metastable timescale.  We do not have an explicit construction of such a proliferation of metastable states in one dimension in a local qubit model.  In contrast, in two or more dimensions, the canonical Ising model $H_0'$ \eqref{eq:H0'=Ising} is itself a system with exponentially many metastable states with asymptotically long lifetimes.

To show this explicitly, we place the qubits on the faces of a two-dimensional square lattice of linear size $L$ (with open boundary conditions), and consider a state $\kpsi$ consisting of stripes of all-zero and all-one regions, as shown in  Fig.~\ref{fig:2dIsing}. The stripes all have width $\ge \ell$ with a tunable parameter $\ell$, and neighboring stripes are separated by a straight line of domain walls (DWs). $H_0'$ conserves the number of DWs, which becomes a crucial dynamical constraint for the system under perturbation, indicated by numerics and perturbative arguments \cite{2dIsing_PRL22,2dIsing_Hart22,2dising_DW24,Ising2D_1,Ising2D_2,pavevsic2024constrained}. Here we provide rigorous lifetime bounds for the particular stripe states. Note that the set up looks similar to inert states in Hilbert space shattering with dipole conservation \cite{shatter2}. 

With respect to $H_0'$, $\kpsi$ is metastable with range $R=\ell$: As shown in Fig.~\ref{fig:2dIsing}, any local operator will always create more DWs and increase the energy; in order to preserve energy, one has to flip a region that connects two stripes of the same polarization, whose diameter is at least $\ell$. Theorem \ref{thm:main} then implies that under a generic perturbation of local norm $\epsilon$, $\kpsi$ has lifetime $t_*= \min(\exp(\Omega(\ell)), \exp\mlr{\Omega\lr{(\Delta/\epsilon)^{1/3}}} )$. 

In fact, since $H_0'$ is commuting, we expect Theorem \ref{thm:Ising_t_d} can be generalized to show that the stripe state has $t_*=\exp\mlr{\Omega(\ell_{\rm resonant}^2)}$ with $\ell_{\rm resonant}=\min(\ell, \Delta/\epsilon)$. This is because the smallest resonant bubble flips $\sim \ell_{\rm resonant}^2$ qubits: If $\Delta/\epsilon\ll \ell$, the smallest bubbles are false-vacuum-like that melt the false-vacuum stripes from inside or by flipping regions adjacent to the walls; otherwise if $\Delta/\epsilon\gg \ell$, the dynamics would be dominated by processes in the bottom right of Fig.~\ref{fig:2dIsing}, which connects neighboring stripes of the same polarization by flipping rectangles. Furthermore, if we require that the perturbation preserves the Ising $\mathbb{Z}_2$ symmetry, there is no longer a  false vacuum decay as in Section \ref{sec:falsevacuum}, so we expect (but did not prove) that $t_*=\exp\mlr{\Omega(\ell^2)}$ coming from the latter channel alone, so long as one starts from the genuine metastable state where each stripe corresponds to a ``true ground state" of the symmetric Hamiltonian.

If $\ell$ divides $L$, we can consider a cartoon where each set of $\ell$ adjacent columns in the lattice is chosen to be either all 0 or all 1.  The metastability radius is at least $\ell$, as above, for every such state, and clearly there are $2^{L/\ell}$ such metastable states $\kpsi$.  This is an undercount of the number of actual metastable states of this kind, as we could consider a more complex mixture of domains, with widths that are not constrained to integer multiples of $\ell$.   Larger $\ell$ thus leads to fewer (but still exponential in the linear system size $L$) metastable states. Unlike the case of shattering from dipole conservation \cite{shatter2}, here the DWs have to form straight lines in $\kpsi$ and cannot be roughened; otherwise the DWs are free to move locally and thus by our definition, the state is not metastable. There are exponential-in-volume $\sim 2^{L^2}$ number of such roughened states, and they may belong to a ``metastable manifold",  but again, such a generalization is beyond the scope of this work.

\section{Conclusion}\label{sec:generalize}
In this paper, we have developed a mathematical framework for defining metastable pure states in many-body quantum systems with finite-dimensional Hilbert space.  Definition \ref{def:meta} says, loosely, that a state $|\psi\rangle$ is metastable with respect to Hamiltonian $H$ if all few-body perturbations acting on $\le R$ sites either stabilize $|\psi\rangle$, or raise the average energy of a state by a finite amount.  This definition is inspired by a natural definition of metastability in classical spin systems.  We showed that this definition of metastability implies a number of remarkable properties.  Theorem \ref{thm:H=H0+V} proves that all short-range entangled metastable states, obtained from product states in sufficiently short time evolution, are exact eigenstates of a Hamiltonian $H_0 = H-V$, where $V$ is perturbatively small in the metastability radius $R$.  Theorem \ref{thm:main} proved that all metastable states exhibit prethermalization, and in particular that local correlation functions in the metastable state do not decay until times $t\sim \exp[R^\alpha]$, which are non-perturbatively large in the small parameter $R^{-1}$.

This theory of metastability has a number of profound implications.  Theorem \ref{thm:Ising_t_d} proves that the lifetime of the false vacuum of an Ising-like model in $d$ dimensions is $\exp(\epsilon^{-d})$, which saturates the famous  semiclassical prediction of \cite{coleman}, using quantum field theory.  This provides a mathematically rigorous justification for the long lifetime of the false vacuum and establishes a precise quantum mechanical theory of nucleation.  We also saw a number of one-dimensional models with finite energy density metastable states, such as the PXP model and a ``helical" spin chain where we in fact proved the existence of exponentially many metastable states.  

There are a number of exciting possible extensions and generalizations of our metastability theory.  Firstly, our metastability theory applies only to pure states, but there are certainly metastable mixed states, such as the thermal state describing a supercooled liquid.  While such states are presumably well-described in practice by the classical theory of metastability \cite{langer}, it could be formally desirable to understand a quantum theory of such mixed states, to understand whether there are any intrinsically quantum notions of mixed state metastability.  We note that recent work \cite{Chen:2023idc} may provide a useful starting point for such analyses.  A theory of mixed state metastability might provide a rigorous framework to study quantum dielectric breakdown (see e.g. the cartoon models \cite{kawakami,biaolian}).

Secondly, to prove the slow false vacuum decay, it was important to use a stronger notion of metastability than Definition \ref{def:meta}, which accounts for the ``surface tension" of the domain wall between true/false vacua.  More generally on physical grounds, we expect that a state deserves to be called metastable if it lies in a ``metastable basin" of states, even when low-energy single-spin flips do exist.  Generalizing our formalism to account for such states is a worthwhile endeavor, and it would be interesting if it is relevant for understanding infinite-temperature (approximate) quantum scarring in the PXP model. 

Another interesting motivation for extending the theory of metastability along these lines comes from classical spin glasses (and related models) perturbed by quantum fluctuations.   It has recently been proved that such models have localized many-body eigenstates \cite{MBL_LDPC24} (see also \cite{Altshuler_2010,MBL_glass17}), which ensure the complete breakdown of thermalization; states are no longer just metastable, but genuinely stable.  In the proof of \cite{MBL_LDPC24}, macroscopic energy barriers between different bitstings play an essential role, even though it takes many bit flips to reach the top of the energy barrier.  A more general metastability theory that takes into account the paths through configuration space needed to escape local energy barriers could perhaps play an intriguing role in building a more systematic theory for slow dynamics in finite-dimensional quantum glasses, and in ergodicity-breaking systems robust to perturbations up to a system-size dependent order \cite{loop_David24,loop_Charlie24}.  

Thirdly, our theory of metastability did not invoke any notions of symmetry, yet symmetry and/or perturbatively weak symmetry breaking do often play important roles in understanding why metastability arises in nature.  For example, the slow false vacuum decay arises because the Hamiltonian $H$ is close to one with a spontaneously broken discrete symmetry.  Can such symmetry be used to prove stronger bounds?  Do models with a ``symmetry-derived" metastability have a parametrically stronger notion of metastability than more general models, such as the PXP or helical models we studied? How much of our theory persists if one only has symmetry protected metastable states (i.e., where we require in the definition of metastability that $\mathcal{O}$ is invariant under a certain symmetry group)?

Fourth, one intriguing result of \cite{abanin2017rigorous} is the approximate conservation of quasiparticle number in the Hubbard model with strong on-site interactions, and weak kinetic interactions.  We expect that an analogous story exists in models with weakly interacting quasiparticle excitations, in which the single quasiparticle dispersion relation is gapped:  whether or not interactions formally preserve quasiparticle number, in sufficiently low particle/energy density states, quasiparticles should be long-lived.  Indeed, we saw some numerical evidence for this in the PXP model in Fig.~\ref{fig:PXP}. Such metastability of quasiparticles was also used recently to understand slow thermalization dynamics in Floquet driven systems \cite{Floq_GS24}, where certain states are extremely hard to heat up, in contrast to typical states. A simple obstruction to using Definition \ref{def:meta} directly is that it is not obvious how to guarantee that no Hermitian operator exists which can remove quasiparticles from a candidate metastable state (which is likely to lower the energy, and thus render the state not metastable).\footnote{As a simple example, consider the Dicke state with two particles: $|1100\rangle + |1010\rangle + |0110\rangle + \cdots $; here `1' represents a quasiparticle-occupied site.  A two-site operator $|11\rangle (\langle 01|-\langle 10|)+\mathrm{H.c.}$  will, even on average, remove one quasiparticle from this state.} We expect that operators which change the quasiparticle number either raise or lower the energy by a finite amount, however, and so it may be possible to generalize the spectral filtering techniques we developed for metastability to tackle this problem as well.

It will be desirable to extend our metastability theory to systems with long-range interactions.  The presence of sufficiently long-range interactions is known to drastically modify the theory of phase transitions \cite{dyson,israel} and nucleation \cite{mccraw,van_Enter_2019}.  Moreover, intriguing slow dynamics enabled by long-range interactions has recently been proposed in many-body quantum systems \cite{lerose,longrange_metastab_Ising,longrange_metastab_fluc,longrange_metastab_confine,longrange_metastab_spec,longrange_metastab_violent}. Despite the important progress in developing Lieb-Robinson bounds for systems with power-law interactions \cite{chen2019finite,Kuwahara:2019rlw,Tran:2021ogo}, generalizing the methods we have developed here to power-law interacting models appears to be an extremely challenging technical problem in full generality. \revise{Nevertheless, for special cases where $H_0$ is commuting, we expect known results on Floquet prethermalization in long-range interacting systems \cite{Floq_KS_16,Floq_KS_PRL,Floq_power} could be adapted to establish metastability (see \cite{abanin2017rigorous,quasiperiodic} for relations between static and Floquet prethermalization), where the crucial step of using Lieb-Robinson bounds to control SWT from evolution under $H_0$, becomes much simpler. }

\revise{
Finally, an interesting open questions concerns the extent to which our results might generalize to models of interacting bosons.  Like the case of power-law interactions, there are recent Lieb-Robinson bounds \cite{schuch,boson_kuwahara,boson_yin,boson_lemm}, especially in one dimension, which are strong enough to suggest that some of our key ideas may generalize to bosonic models, so long as the boson number is a conserved quantity.  For more general bosonic models, strong locality bounds are not possible \cite{gross}, and we expect entirely new techniques would be needed to develop a theory of metastability for these systems.
}


\section*{Acknowledgements}
We thank Alessio Lerose, Olexei Motrunich, Pablo Sala and Dong Yuan for useful discussions. Tensor network calculations were performed using the TeNPy Library (version 1.0.0) \cite{tenpy}. This work was supported by the Alfred P. Sloan Foundation under Grant FG-2020-13795 (AL), the Heising-Simons Foundation under Grant 2024-4848 (AL), the Department of Energy under Quantum Pathfinder Grant DE-SC0024324 (CY, AL). 
FMS acknowledges support provided by the U.S.\ Department of Energy Office of Science, Office of Advanced Scientific Computing Research, (DE-SC0020290); DOE National Quantum Information Science Research Centers, Quantum Systems Accelerator; and by Amazon Web Services, AWS Quantum Program.  FMS and AL also thank the Kavli Institute for Theoretical Physics, which is supported by the National Science Foundation under Grant PHY-1748958, for hospitality as this work was initiated.

\begin{appendix}
\renewcommand{\thesubsection}{\thesection.\arabic{subsection}}
\section{Finding a ``prethermal decomposition"}\label{app:H=H0+V}
This appendix proves Theorem \ref{thm:H=H0+V}. To set the stage for the proof, given the metastable state $\kzero$ and $q=2$ local Hilbert space dimension, we can take the Hamiltonian $H$ and separate it into Pauli $Z$-strings and others that involve Pauli-$X$. It is natural to put $Z$-strings in $H$ into $H_0$ because they stabilize $\kzero$, while the others into $V$ so that $H=H_0+V$. The challenge is to show $V$ is small locally. 

Here we illustrate the idea behind our proof with the simple 1d example $H_0=-\sum_i Z_i Z_{i+1}, V=-\epsilon \sum_i X_i$ (the following does not rely on the integrability of $H=H_0+V$). Due to the $X$ field, $\kzero$ is not the exact ground state of $H$, and one \emph{gains} energy from $V$ by slightly rotating a single site $i$ by angle $\theta\ll 1$: $V[\ket{\theta}\otimes \kzero_{\rm rest}]-V[\kzero]\approx -2\epsilon\theta$, where \begin{equation}\label{eq:kettheta}
    \ket{\theta}:=\cos(\theta)\ket{0} +\sin(\theta) \ket{1}.
\end{equation}
This rotation also \emph{costs} energy from $H_0$: $H_0[\ket{\theta}\otimes \kzero_{\rm rest}]-H_0[\kzero]\approx 4\theta^2$. Therefore, in order to gain net energy for $H$, one should choose \begin{equation}\label{eq:theta=eps}
    \theta\approx \epsilon/4,
\end{equation}
so that $H[\ket{\theta}\otimes \kzero_{\rm rest}]-H[\kzero]\approx -\epsilon^2/4$. This does not violate the metastability condition, because $\ket{\theta}\otimes \kzero_{\rm rest}$ is not orthogonal to $\kzero$: they actually have large overlap $\alr{0|\theta}\approx 1-\theta^2/2$. However, in a size-$R$ region, one can slightly rotate every other qubit and consider the state $\ket{\phi}=(\ket{\theta}_1\otimes\ket{\theta}_3\otimes\cdots\otimes \ket{\theta}_{R-1})\otimes \kzero_{\rm rest}$, which now have small overlap \begin{equation}\label{eq:0overlap_phi}
    \alr{\bm{0}|\phi}\approx(1-\theta^2/2)^{R/2}.
\end{equation}
On the other hand, $\ket{\phi}$ gains energy from every length-3 motif $\ket{0\theta 0}$, so has energy smaller than $H[\kzero]$. Therefore, the metastability condition implies that $R$ cannot be too large because otherwise $\ket{\phi}$ becomes almost orthogonal to $\kzero$. According to \eqref{eq:0overlap_phi}, this requires $R\lesssim \theta^{-2}$, which yields the desired $\epsilon\lesssim R^{-1/2}$ due to \eqref{eq:theta=eps}.

\subsection{Proof of Theorem \ref{thm:H=H0+V}}
\begin{proof}[Proof of Theorem \ref{thm:H=H0+V}]
For any local term $H_S$, it can be written as \begin{equation}\label{eq:HS=}
    H_S = \Pzero{S} H_S \Pzero{S} + \lr{I-\Pzero{S}} H_S \lr{I-\Pzero{S}} + \mlr{ \lr{I-\Pzero{S}} H_S \Pzero{S} + \text{H.c.} },
\end{equation}
where the projector $\Pzero{S}$ is a local operator that only acts inside $S$. The first two terms in \eqref{eq:HS=} stabilize $\kzero$. Furthermore, we can expand $\lr{I-\Pzero{S}} H_S \kzero_S= \sum_{F\subset S: F\neq \emptyset}\epsilon_{S,F}\ket{\varphi_{S,F}}_{F}\otimes \kzero_{S\setminus F}$, where $F$ can be disconnected, $\epsilon_{S,F}$ is a complex number, and the normalized $\ket{\varphi_{S,F}}_{F}=\sum_{\bm{z}\in\{1,\cdots,q-1\}^F} \widetilde{\epsilon}_{S,F,\bm{z}}\ket{\bm{z}}_F$ contains computational basis states that are not $\ket{0}_i$ if and only if $i\in F$. The third term in \eqref{eq:HS=} then becomes \begin{align}\label{eq:HS=1F}
    &\lr{I-\Pzero{S}} H_S \Pzero{S}+ \text{H.c.} = \sum_{F\subset S: F\neq \emptyset}\epsilon_{S,F}\ket{\varphi_{S,F}}_{F}\bra{\bm{0}}\otimes\kzero_{S\setminus F} \bra{\bm{0}} + \text{H.c.} \nonumber\\
    &\quad = \sum_{F\subset S: F\neq \emptyset}\epsilon_{S,F}\ket{\varphi_{S,F}}_{F}\bra{\bm{0}}-\epsilon_{S,F}\ket{\varphi_{S,F}}_{F}\bra{\bm{0}}\otimes\lr{1-\kzero_{S\setminus F} \bra{\bm{0}}} + \text{H.c.} \nonumber\\
    &\quad = \sum_{F\subset S: F\neq \emptyset}\epsilon_{S,F}\ket{\varphi_{S,F}}_{F}\bra{\bm{0}}\otimes\kzero_{\ball_{i(F)}\setminus F} \bra{\bm{0}}-\epsilon_{S,F}\ket{\varphi_{S,F}}_{F}\bra{\bm{0}}\otimes\lr{\kzero_{\ball_{i(F)}\setminus F}\bra{\bm{0}}-\kzero_{S\setminus F} \bra{\bm{0}}} + \text{H.c.} .
\end{align}
The last line here comes from the following manipulation: 
For each $F$ that can appear in \eqref{eq:HS=1F} for any $S$ with $\diam S\le r$ (a parameter that will be chosen as $r\sim \ln R$ later in \eqref{eq:r=logR}), $\diam F\le r$ so it is included in some ball $\ball_{i}:=B(i,r)$ centered at $i$ with radius $r$. We choose a function $i(F)$ (independent of $S$) such that $F\subset \ball_{i(F)}$ for any such $F$.  This grouping is for later convenience, as we will explain.

We can then combine the first terms (and their conjugates) in \eqref{eq:HS=1F} into such balls by defining \begin{equation}\label{eq:Vballj=}
    V_{\ball_j} := \sum_{S:\diam S\le r}\sum_{F\subset S: F\neq \emptyset,i(F)=j}\epsilon_{S,F}\ket{\varphi_{S,F}}_{F}\bra{\bm{0}}\otimes\kzero_{\ball_{i(F)}\setminus F} \bra{\bm{0}}+\text{H.c.} =:\epsilon_j\lr{\ket{\varphi_j}_{\ball_j}\bra{\bm{0}} + \text{H.c.}},
\end{equation}
which is supported in $\ball_j$, and $\ket{\varphi_j}$ is normalized and orthogonal to $\kzero$, whose phase makes $\epsilon_j\ge 0$. See Fig.~\ref{fig:H0V}(a) for an illustration.
Note that we will apply translation invariance at the very end of this proof, so for now $\epsilon_j$ can be non-uniform. We combine terms in \eqref{eq:Vballj=} so that: (\emph{1}) The number of local terms in the final perturbation $V=\sum_S V_S$ is significantly reduced (number of size-$r$ local terms overlapping a site is $\sim r^d$ instead of $\sim 2^r$). (\emph{2}) Because $\ket{\varphi_{S,F}}_F\otimes \kzero_{F^{\rm c}}$ for different $F$ are orthogonal to each other, so are the $\ket{\varphi_j}$ for different $j$: \begin{equation}\label{eq:phiiphij=0}
    \lr{\bra{\varphi_i}_{\ball_i}\otimes \bra{\bm{0}}_{\ball_i^{\rm c}}} \lr{\ket{\varphi_j}_{\ball_j}\otimes \kzero_{\ball_j^{\rm c}}}=0,\quad \forall i\neq j.
\end{equation} 
Indeed, since each set $F$ was associated to a unique ball $\mathcal{B}_i$, although $\bra{\varphi_i}_{\ball_i}$ is generically an entangled state, it will have no basis vectors in common with $\bra{\varphi_j}_{\ball_j}$ in the computational basis.

Having defined $V_S$ for the case $S=\ball_j$ (for some $j$), the remaining local terms of $V$ are defined by
 \begin{equation}\label{eq:VS_largeS}
    V_S = \left\{\begin{array}{ll}
    \lr{I-\Pzero{S}} H_S \Pzero{S} +\text{H.c.} ,& S\neq \ball_j' (\forall j). \\
    \lr{I-\Pzero{S}} H_S \Pzero{S}+\lr{I-\Pzero{\ball_j}} H_{\ball_j} \Pzero{\ball_j} +\text{H.c.},     & S=\ball_j'
    \end{array}\right. 
\end{equation}
for any $S$ with $\diam S>r$ and $S\neq \ball_j$. Here we have defined $\ball'_j$ for each $j$ such that $\ball'_j$ consists of $\ball_j$ and one extra neighboring site (so $\diam \ball'_j\le 2r+1$). For concreteness, we choose the extra site to be the one right below $\ball_j$'s ``south pole'', as shown in Fig.~\ref{fig:H0V}(a).  We will also dilate the original $S=\ball_j$ term which can, in principle, arise in \eqref{eq:HS=} to $\ball'_j$, so that $S=\ball_j$ do not appear twice in $V=\sum_S V_S$ ($V_{\ball_j}$ is already defined in \eqref{eq:Vballj=}).

\begin{figure}
    \centering
    \includegraphics[width=0.9\textwidth]{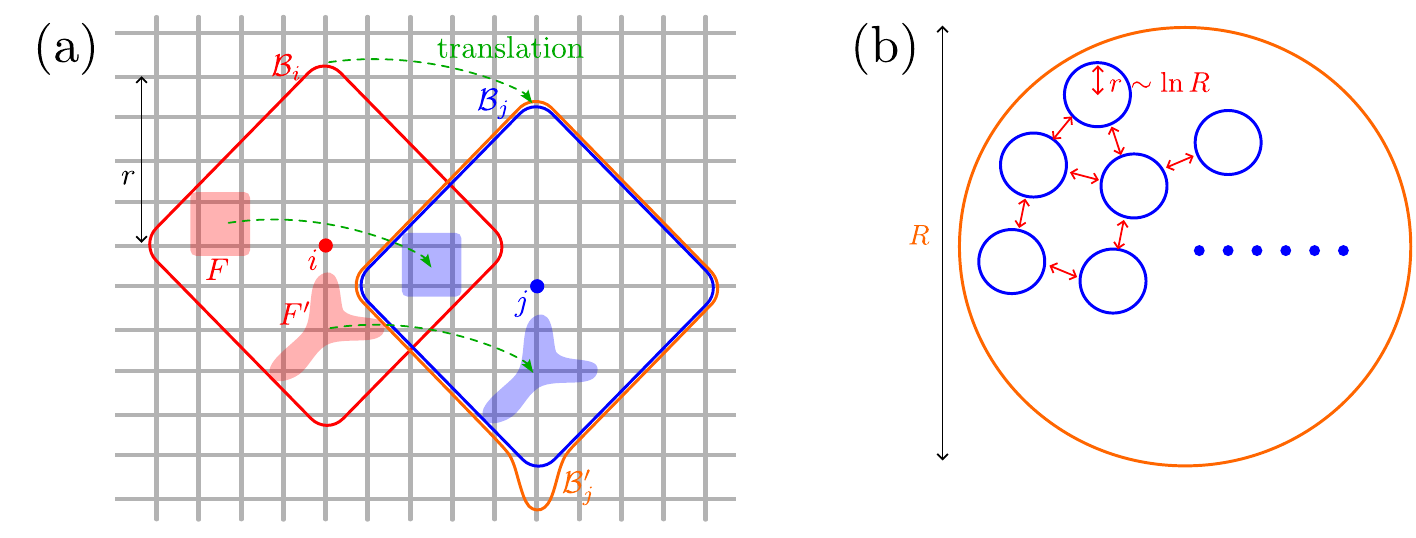}
    \caption{(a) Terms in $H$ that do not stabilize $\kzero$ are $\propto \ket{\varphi}_{F}\bra{\bm{0}}$ (and their conjugates) for some set $F$, where the string $|\varphi\rangle$ contains no 0 elements. We combine such terms (shaded regions) with sufficiently small support $F$ into balls $V_{\ball_j}$ of radius $r$, centered at site $j$. The red and blue colors represent two balls and their corresponding representative $F$s. This procedure can be done in a way respecting translation symmetry. In particular, since the red square $F$ is assigned to $\ball_i$, its translated version, the blue square, is assigned to $\ball_j$ even if it is also a subset of $\ball_i$. We also assign a neighboring vertex to each $\ball_j$ so that they combine to $\ball_j'$, in a translation-invariant way.  (b) For a ball-like region of size $R$, we can fill it by smaller balls of radius $r$ where the balls have distance $\ge r$ to each other. The number of balls can be chosen to scale as $\gtrsim (R/r)^d$ while respecting these constraints. }
    \label{fig:H0V}
\end{figure}

The above procedure of organizing terms in $V$ respects translation symmetry, so that if $H$ is translation-invariant (i.e. $H_S = H_{\mathsf{T}_x(S)}$ for any $S$ and any translation $\mathsf{T}_x$ by vector $x$), so is $V$: Any $F$ appearing in \eqref{eq:Vballj=} can be labeled by its shape $\widetilde{F}$ and coordinate $x$, so $F$ is just $\mathsf{T}_x(\widetilde{F})$. We first assign a ball of center $i_{\widetilde{F}}$ for each shape $\widetilde{F}$, then $F$ is assigned to the ball centered at the $\mathsf{T}_x(i_{\widetilde{F}})$. The way to assign $\ball_j'$ for $\ball_j$ is also translation invariant.  The remaining part $H_0=H-V$ with decomposition $H_{0,S}=H_S-V_S$ then satisfies \eqref{eq:H0psi0=0}, because it contains the first two terms in \eqref{eq:HS=} along with the second term in \eqref{eq:HS=1F} that all stabilizes $\kzero$.

Our main goal is to bound the local norm of $V$: For any $\kappa<2\mu$,
\begin{align}\label{eq:V<Rd}
    \norm{V}_{\kappa} &= \max_i \lr{\sum_{j:\ball_j\ni i} \mathrm{e}^{\kappa\, (\mathrm{diam}\ball_j)^\alpha} \norm{V_{\ball_j}} + \sum_{j:\ball'_j\ni i} \mathrm{e}^{\kappa\, \mathrm{diam}\ball'_j} \norm{V_{\ball'_j}} +
    \sum_{S \ni i:\diam S>r:S\neq \ball_j,\ball'_j} \mathrm{e}^{\kappa\, \mathrm{diam}S} \norm{V_{S}}} \nonumber\\
    &\le \max_i\lr{ \sum_{j\in \ball_i} \ee^{\kappa (2r)^\alpha} \epsilon_j + \ee^{\kappa}\sum_{S \ni i:\diam S>r} \mathrm{e}^{\kappa\, \mathrm{diam}S} \norm{H_{S}} } \nonumber\\
    &\le c_d(2r)^d \ee^{\kappa (2r)^\alpha}\max_j \epsilon_j + \ee^{\kappa}\max_i\sum_{S \ni i:\diam S>r}\mathrm{e}^{(\kappa-2\mu)(r+1)} \mathrm{e}^{2\mu \, \mathrm{diam}S} \norm{H_{S}} \nonumber\\
    &\le c_d(2r)^d \ee^{\kappa (2r)^\alpha}\max_j \epsilon_j +\ee^{\kappa} \mathrm{e}^{(\kappa-2\mu)(r+1)}h.
\end{align}
Here we treated the two kinds of terms \eqref{eq:Vballj=} and \eqref{eq:VS_largeS} separately, and used $i\in \ball_j$ is equivalent to $j\in \ball_i$ together with $|\ball_i|\le c_d(\diam \ball_i)^d\le c_d(2r)^d$ for the first kind. The second kind is bounded using $\norm{V_S}\le \norm{H_S}$ for \eqref{eq:VS_largeS} with $S\neq \ball'_j$, and the fact that it only contains large support terms in $H$ that decay with exponent $2\mu$. The special $S=\ball'_j$ case is bounded by adding an extra factor $\ee^\kappa$ because $\mathrm{e}^{\kappa\, \mathrm{diam}\ball'_j}\le \ee^\kappa \mathrm{e}^{\kappa\, \mathrm{diam}\ball_j}$. As a result, we choose
\begin{subequations}\begin{align}
    r&=\lfloor (2\mu)^{-1}\ln (4hR^d/c_d\Delta)\rfloor+1, \label{eq:r=logR} \\
    \kappa&=\mu, \label{eq:kappa=mu}
\end{align}\end{subequations}
so that the second term in \eqref{eq:V<Rd} is upper bounded by $\frac{1}{2}\sqrt{c_dh\Delta} R^{-d/2}\le \frac{1}{2}\sqrt{c_d}h R^{-d/2}$ using \eqref{eq:Delta<h}. Therefore, to show \eqref{eq:V1<eps} it remains to deal with the first term. 

To bound $\max_i\epsilon_i$, for any $j$, consider the ball $B(j,R/2)$ of diameter $\le R$, which contains many radius-$r$ balls $\ball_i\subset B(j,R/2)$. Pick a subset of these balls $\{\ball_{i_m}\subset B(j,R/2):m=1,\cdots,M\}$ that are mutually faraway with distance $\gtrsim r$: \begin{equation}
    \mathsf{d}(\ball_{i_m}, \ball_{i_{m'}}) \ge r,\quad \forall m\neq m'.
\end{equation}
There always exists such a set of balls of number \begin{equation}
    M\ge c_{\rm ball} R^d/[c_d(2r)^d], \label{eq:boundonM}
\end{equation}
with constant $0<c_{\rm ball}<c_d$ determined by $d$, see Fig.~\ref{fig:H0V}(b) for illustration. 

We aim to show that $\sum_m \epsilon_{i_m}^2$ (roughly speaking) cannot be large, otherwise the metastability condition on $H$ would be violated. In particular, consider state \begin{equation}
    \ket{\phi} := \kzero_{S^{\rm c}} \otimes_{m=1}^M \ket{\phi_m},
\end{equation}
where $S=\cup_{m=1}^M \ball_{i_m}$ is disconnected, and \begin{equation}\label{eq:phim}
    \ket{\phi_m} := C_m \lr{ \kzero_{\ball_{i_m}} - \frac{\epsilon_{i_m}}{\Delta_m}\ket{\varphi_{i_m}}_{\ball_{i_m}} }, 
\end{equation}
where \begin{equation}
    \Delta_m= H\mlr{\ket{\varphi_{i_m}}_{\ball_{i_m}}\otimes \kzero_{\ball_{i_m}^{\rm c}} }\ge \Delta.
\end{equation}
Here $C_m$ is an overall normalization constant \begin{equation}
    C_m=\lr{1+\left(\frac{\epsilon_{i_m}}{\Delta_m}\right)^2}^{-1/2},
\end{equation} and we have assumed $H\mlr{\kzero}=0$ without loss of generality. \eqref{eq:phim} is approximately the ground state in the two-dimensional subspace $\mathrm{span}(\kzero_{\ball_{i_m}},\ket{\varphi_{i_m}}_{\ball_{i_m}})\otimes \kzero_{\ball_{i_m}^{\rm c}}$ at small $\epsilon_{i_m}$: $H$ restricted to this subspace is \begin{equation}
 H_{\mathrm{restricted}} =    \begin{pmatrix}
0 & \epsilon_{i_m} \\
\epsilon_{i_m} & \Delta_m 
\end{pmatrix} + V_{>r,m}.
\end{equation}
Here all $V_{\ball_i}$ with $i\neq i_m$ do not contribute $\bra{\bm{0}}V_{\ball_i}\ket{\varphi_{i_m}}_{\ball_{i_m}}\otimes \kzero_{\ball_{i_m}^{\rm c}}=0$, thanks to the orthogonality condition \eqref{eq:phiiphij=0}. $V_{>r,m}$ is the contribution from large-support terms \eqref{eq:VS_largeS}, bounded by \begin{equation}
    \norm{V_{>r,m}} \le \sum_{S'\neq \ball_i (\forall i):S'\cap \ball_{i_m}\neq \emptyset, \diam S'>r} \norm{V_{S'}} \le \sum_{i\in \ball_{i_m}} \sum_{S'\ni j:\diam S'>r} \norm{H_{S'}} \le c_d (2r)^dh \ee^{-2\mu (r+1)},
\end{equation}
similarly as the second term in \eqref{eq:V<Rd}.
Therefore, \begin{equation}
    H\mlr{\ket{\phi_m} \otimes \kzero_{\ball_{i_m}^{\rm c}} } \le C_m^2 \begin{pmatrix}
1 & -\frac{\epsilon_{i_m}}{ \Delta_m }
\end{pmatrix}
\begin{pmatrix}
0 & \epsilon_{i_m} \\
\epsilon_{i_m} & \Delta_m 
\end{pmatrix} \begin{pmatrix}
1 \\ -\frac{\epsilon_{i_m}}{ \Delta_m }
\end{pmatrix} + \norm{V_{>r,m}} \le 0+\norm{V_{>r,m}}\le c_d (2r)^dh \ee^{-2\mu (r+1)}.
\end{equation}

Since the balls $\ball_{i_m}$ are chosen to be far $\gtrsim r$ from each other, we expect the flips $\kzero\rightarrow \ket{\phi_m}$ in them change the energy in an almost \emph{additive} way: \begin{equation}
    H[\ket{\phi}]\approx \sum_m H\mlr{\ket{\phi_m} \otimes \kzero_{\ball_{i_m}^{\rm c}} }\le Mc_d (2r)^d h\ee^{-2\mu (r+1)}\le c_{\rm ball}R^d h\ee^{-2\mu (r+1)}\le \frac{\Delta}{4},
\end{equation}
where we have used \eqref{eq:r=logR}.
The error comes from terms in $H$ that act on more than one balls: \begin{align}
    H[\ket{\phi}]- \sum_m H\mlr{\ket{\phi_m} \otimes \kzero_{\ball_{i_m}^{\rm c}} } &\le \sum_{S':S'\text{ intersects }\ge 1 \text{ balls } \ball_{i_m}} \norm{H_{S'}} \le \sum_{i\in S} \sum_{S'\ni i:\diam S'>r} \norm{H_{S'}} \notag \\
    &\le c_d R^d h \ee^{-2\mu (r+1)} \le \frac{\Delta}{4}.
\end{align}
For technical simplicity, we have overcounted in the above expression by including $H_{S'}$ that only intersect one ball,  as long as $\diam S'>r$; we have also counted each $\lVert H_{S^\prime}\rVert$ in every ball that it intersects.  As a result \begin{equation}\label{eq:H<D2}
    H[\ket{\phi}]\le \frac{\Delta}{2}.
\end{equation}
So in order for metastability to hold, the overlap between $\ket{\phi}$ and $\kzero$ should be large.

We now write \begin{equation}
    \ket{\phi} = C_\phi \kzero + \sqrt{1-C_\phi^2} \kzero^\perp,
\end{equation}
where $\kzero^\perp$ is orthogonal to $\kzero$, and $C_\phi=\prod_m C_m$. According to \eqref{eq:H0>Delta1}, \begin{equation}
    H\mlr{\kzero^\perp}\ge \Delta.
\end{equation}
Combining this with \eqref{eq:H<D2} and $H\mlr{\kzero}=0$, \begin{align}\label{eq:Delta>C}
    \frac{\Delta}{2} &\ge H[\ket{\phi}] = 0+\lr{1-C_\phi^2}H\mlr{\kzero^\perp} +C_\phi \sqrt{1-C_\phi^2} \lr{ \bra{\bm{0}}H\kzero^\perp+ {\rm H.c.} } \nonumber\\
    &\ge \lr{1-C_\phi^2}\Delta - 2c_d h C_\phi R^d,
\end{align}
where we have used \begin{equation}\label{eq:0H0perp}
    \abs{\bra{\bm{0}}H\kzero^\perp} \le \norm{ \sum_{S': S'\cap S\neq \emptyset}   H_{S'}} \le \sum_{i\in S} \sum_{S'\ni i} \norm{H_{S'}} \le c_d R^d h.
\end{equation}
Obviously \eqref{eq:Delta>C} cannot hold for $C_\phi=0$; solving a quadratic equation, we see that any solution obeys \begin{equation}\label{eq:Cphi>R}
    C_\phi \ge \sqrt{\eta^2+\frac{1}{2}} - \eta=\eta\lr{\sqrt{1+\frac{1}{2\eta^2}}-1}\ge \eta\times\frac{2}{5}\times \frac{1}{2\eta^2}=\frac{\Delta}{5c_dhR^d} ,
\end{equation}
where \begin{equation}
    \eta=c_d hR^d/\Delta\ge c_d R^d\ge 1,
\end{equation}
where $2/5$ is from monotonicity of function $(\sqrt{1+x}-1)/x\ge 2(\sqrt{1+1/2}-1)>2/5$ at $x\le 1/2$.
Using $C_m$ in \eqref{eq:phim}, \eqref{eq:Cphi>R} becomes \begin{align}\label{eq:>eout}
    2\ln \frac{5c_dhR^d}{\Delta} &\ge \sum_{m=1}^M \ln\mlr{1+ \lr{\frac{\epsilon_{i_m}}{\Delta_m}}^2} \ge \sum_{m=1}^M  \ln\lr{1+ \frac{\epsilon_{i_m}^2}{[c_d(2r)^d h]^2}}, 
\end{align}
where we have used $\Delta_m\le c_d (2r)^d h$ because only $H_{S'}$s that intersect with $\ball_{j_m}$ contribute similar to \eqref{eq:0H0perp}.

If the system is translation invariant (the procedure to define $V$ respects translation invariance, as we have explained), $\epsilon_{i_m}:= \epsilon$ with \begin{align}
\ln\lr{1+ \frac{\epsilon^2}{c_d^2(2r)^{2d} h^2}} &\le \frac{2c_d(2r)^{d}}{c_{\rm ball}R^d} \ln \frac{5c_dhR^d}{\Delta}\le \ln 3
.
\end{align}
Here we used the bound (\ref{eq:boundonM}) on $M$, along with the fact that $\sim R^{-d}(\ln R)^{d+1}$ is bounded by $\ln 3$ at sufficiently large $R\ge c_{\rm R}$, with $c_{\rm R}$ determined by $d,h/\Delta,\mu$.  We then deduce \begin{equation}
        \epsilon^2 \le 2c_d^2(2r)^{2d} h^2 \frac{2c_d(2r)^{d}}{c_{\rm ball}R^d} \ln \frac{5c_dhR^d}{\Delta} \le c_{\rm TI}h^2 (\ln R)^{3d+1}R^{-d}, \label{eq:e<cTI}
\end{equation} where we have used $\ln(1+x)\ge x/2$ if $x\le 2$, with constant $c_{\rm TI}$ determined by $d,h/\Delta,\mu$. Plugging \eqref{eq:e<cTI} into \eqref{eq:V<Rd} with \eqref{eq:kappa=mu} leads to \eqref{eq:V1<eps}, because we have bounded the second term of \eqref{eq:V<Rd} above, while the first term uses $(\ln R)^{(5d+1)/2}\ee^{c (\ln R)^\alpha}=\mathrm{O}( \ee^{\delta \ln R})$ for any constant $\delta>0$. This also bounds the local norm of $H_0$: \begin{align}
    \sum_{S\ni i} \norm{H_{0,S}}\ee^{2\mu_0\diam S} &\le \sum_{S\ni i} \norm{H_{S}}\ee^{2\mu_0\diam S} + \sum_{S\ni i} \norm{V_{S}}\ee^{2\mu_0\diam S} \nonumber\\
    &\le h+ c_d(2r)^d \ee^{2\mu_0 (2r)}\max_j \epsilon_j + h \le c_{\rm H} h,
\end{align}
for some constant $c_{\rm H}$ determined by $d,h/\Delta,\mu$,
where we have generalized \eqref{eq:V<Rd} to bound the $V$ contribution, and used $\mu_0=\mu/5$ so that the second term in the second line scales as $(\ln R)^{(5d+1)/2} R^{2d/5}R^{-d/2}=\mathrm{O}(1)$. Finally, $(H_0, \kzero)$ is $(\Delta/2, c'_\delta R^{1/2-\delta/d})$-metastable for $c'_\delta=(\Delta/2c_dc_\delta h)^{1/d}$, because for any $S$ with $\diam S\le c'_\delta R^{1/2-\delta/d}$, and any orthogonal state $\ket{\phi}_S$,
    \begin{equation}
        H_0[\ket{\phi}_S\otimes \kzero_{S^{\rm c}}]\ge H[\ket{\phi}_S\otimes \kzero_{S^{\rm c}}] - V[\ket{\phi}_S\otimes \kzero_{S^{\rm c}}]\ge \Delta- \sum_{S':S'\cap S\neq \emptyset} \norm{V_{S'}} \ge \Delta- |S| \norm{V}_\mu \ge \frac{\Delta}{2},
    \end{equation}
    where we have used $|S|\le \Delta/(2c_\delta h) R^{d/2-\delta}$, and \eqref{eq:V1<eps}.
\end{proof}

\subsection{Relaxing the assumption of translation invariance}

In the above proof, translation invariance is only used to simplify \eqref{eq:>eout}. Therefore, the conclusion generalizes directly as long as the \emph{strength} is uniformly bounded $\epsilon_{i_m}\le \epsilon$, even if the form of perturbation can be highly translational asymmetric.  Without translation invariance, it is not necessarily clear how to guarantee such uniform boundedness.  Still, even without such uniform boundedness, \eqref{eq:>eout} guarantees that although some $\epsilon_{i_m}\gtrsim \Delta$ can be large, the number of these large ones in the size-$R$ ball is bounded by $\mathrm{O}(\ln R)\ll R^d$. In other words, any perturbation to $H_0$ that is strong must itself be isolated from other strong perturbations. 

We can consider a contrived example in $d=1$ to illustrate the points above.  Consider \begin{equation}
    H=-\sum_{i=1}^{NR-1} Z_i Z_{i+1} - \epsilon \sum_{i=1}^N X_{Ri},
\end{equation}
where we do not require that $|\epsilon| \ll 1$, since it acts once every $R$ sites. As this Hamiltonian commutes with $Z_i$ if $i$ is not divisible by $R$, it is exactly solvable.  Proceeding without this observation, we observe that $(H,\kzero)$ is $(1,R)$-metastable -- focusing e.g. on $i=R$, while a rotation by $\mathrm{O}(\epsilon)$ angle clearly lowers the energy, any \emph{orthogonal} state raises the energy by O(1) -- recall the discussion around (\ref{eq:0overlap_phi}). Any other orthogonal states are also of higher energy, because $\kzero$ has $>1/2$ fidelity with the true gapped ground state $\ket{\theta}_i\otimes \ket{0}_{\rm other}$.  However, dynamically the qubit at $R$ can be flipped to negative polarization due to the large $X$ field. In this case, one should use $(\ket{\theta}_R\otimes \ket{\theta}_{2R}\otimes\cdots)\otimes \kzero_{\rm rest}$ for the metastable state in the first place, where the parameter $\theta$ for \eqref{eq:kettheta} minimizes the ``local energy'' of the 3-qubit state $\ket{0\theta 0}_{R-1,R,R+1}$, measured by $-h(Z_{R-1}Z_R+Z_RZ_{R+1})+\epsilon X_R$. Indeed, this state is $(2,R)$-metastable, and also dynamically stable because it is an eigenstate of $H$. In other words, we expect that beyond uniform boundedness, more general cases of metastable $H$ can be captured by the perturbation scenario $H=H_0+V$ with small $V$, by choosing carefully the metastable state. 

\subsection{On the metastability radius in Theorem \ref{thm:H=H0+V}}
In Theorem \ref{thm:H=H0+V}, the metastability radius for $H_0$ is shrunk from the metastability radius $R$ of the original $H$, to a smaller $R_0 \sim \sqrt{R}$ (setting $\delta\rightarrow 0$). There is a similar square root in the perturbation strength \eqref{eq:V1<eps}. Here we show that such shrinking is, unfortunately, physical by presenting a simple example that saturates this scaling. 

Consider the qubit-chain Hamiltonian 
\begin{equation}\label{eq:H=P00++}
    H = \sum_{i=1}^{N-1} (1- P_{00,++})_{i,i+1} + \epsilon \sum_{i=1}^N (Z_i+X_i),
\end{equation}
where \begin{equation}
    P_{00,++} := \ket{00}\bra{00} + \frac{(|01\rangle +|10\rangle + |11\rangle)(\langle 01| + \langle 10| + \langle 11|)}{3},
\end{equation}
is a two-qubit projector onto the 2-dimensional subspace $\mathrm{span}(\ket{00}, \ket{++})$, where $|+\rangle = \frac{1}{\sqrt{2}}(|0\rangle + |1\rangle)$. If $\epsilon=0$, $H$ is frustration-free and the ground states of $H$ can be found exactly (see Theorem 3 in \cite{FFspinchain_classify15}) to be $\mathrm{span}(\kzero, \ket{+\cdots +})$. Moreover, the 2-dimensional ground subspace is gapped, featuring a spontaneously broken $\mathbb{Z}_2$ symmetry generated by Hadamard gates (i.e. flipping $Z_i\leftrightarrow X_i$).

We focus on the metastable state $\kzero$ when turning on $\epsilon$. The metastability condition \eqref{eq:H0>Delta1} can be thought of as restricting Hamiltonian \eqref{eq:H=P00++} to a length-$(R+2)$ chain with a fixed boundary condition that the two boundary qubits are set to $0$, and asking whether $\kzero$ has large overlap with the ground state in this subspace, and whether there is a gap above.

Since the $Z+X$ perturbation preserves the $\mathbb{Z}_2$ symmetry, for periodic boundaries the gap and the degenerate ground states are expected to be stable if $\epsilon$ is smaller than some constant. More precisely, the perturbed ground states are $\ee^A \kzero$ and $\ee^A \ket{+\cdots+}$ where $A$ is quasi-local with local norm $\sim \epsilon$. The fixed boundary condition for the length-$(R+2)$ chain picks the single ground state $\ee^A \kzero$, because the other one violates the boundary terms. The metastability radius, i.e. the maximal $R$ such that $\kzero$ has large overlap with $\ee^A \kzero$, is then expected to scale as \begin{equation}\label{eq:R=eps2}
    R\sim \epsilon^{-2},
\end{equation}
where the power of 2 is analogous to the coefficient of $\theta^2$ in \eqref{eq:0overlap_phi}. In Fig.~\ref{fig:sqrtRvsR}(a), we verify this scaling numerically by exact diagonalization up to $R=22$. 

Since only term in \eqref{eq:H=P00++} that does not stabilize $\kzero$ is the $X$ term, it is natural to separate $H=H_0+V$ with $V=\epsilon \sum_i X_i$, so that $H_0$ stabilizes $\kzero$. Indeed, the procedure in the above proof finds this separation with $r=0$ because the $X$ term is just one-local. Therefore $\norm{V}_\kappa\sim 1/\sqrt{R}$, the scaling in \eqref{eq:V1<eps}, is saturated.
On the other hand,
$H_0$ is perturbing the $\epsilon=0$ Hamiltonian with $\epsilon Z$ that explicitly breaks the $\mathbb{Z}_2$ symmetry, so that $\kzero$ becomes the false vacuum. The metastability radius of $H_0$ is then \begin{equation}\label{eq:R=eps1}
    R_0 \sim \epsilon^{-1},
\end{equation}
because bubbles containing all-$+$ of this length scale become resonant to $\kzero$. This satisfies $R_0\sim \sqrt{R}$ and is verified numerically in Fig.~\ref{fig:sqrtRvsR}(b).

To summarize this example, the shrinking of $R_0 \sim \sqrt{R}$ comes from the fact that, the part of the perturbation $V$ ($\epsilon X$) that would naively evolve $\kzero$, actually restores the $\mathbb{Z}_2$ Hadamard symmetry, such that this perturbation actually increases the metastability radius from $R_0$ to $R$.
This example also highlights that our Theorem \ref{thm:main} does not capture stability in the thermodynamic limit: It proves $t_*\sim \ee^{-R_0}\sim \ee^{-1/\epsilon}$ for this example, while the physical expectation is actually $t_*=\infty$ because of the robustness against the symmetry-preserving perturbation $X+Z$.

\begin{figure}[tbp]
    \centering
    \includegraphics[width=0.5\textwidth]{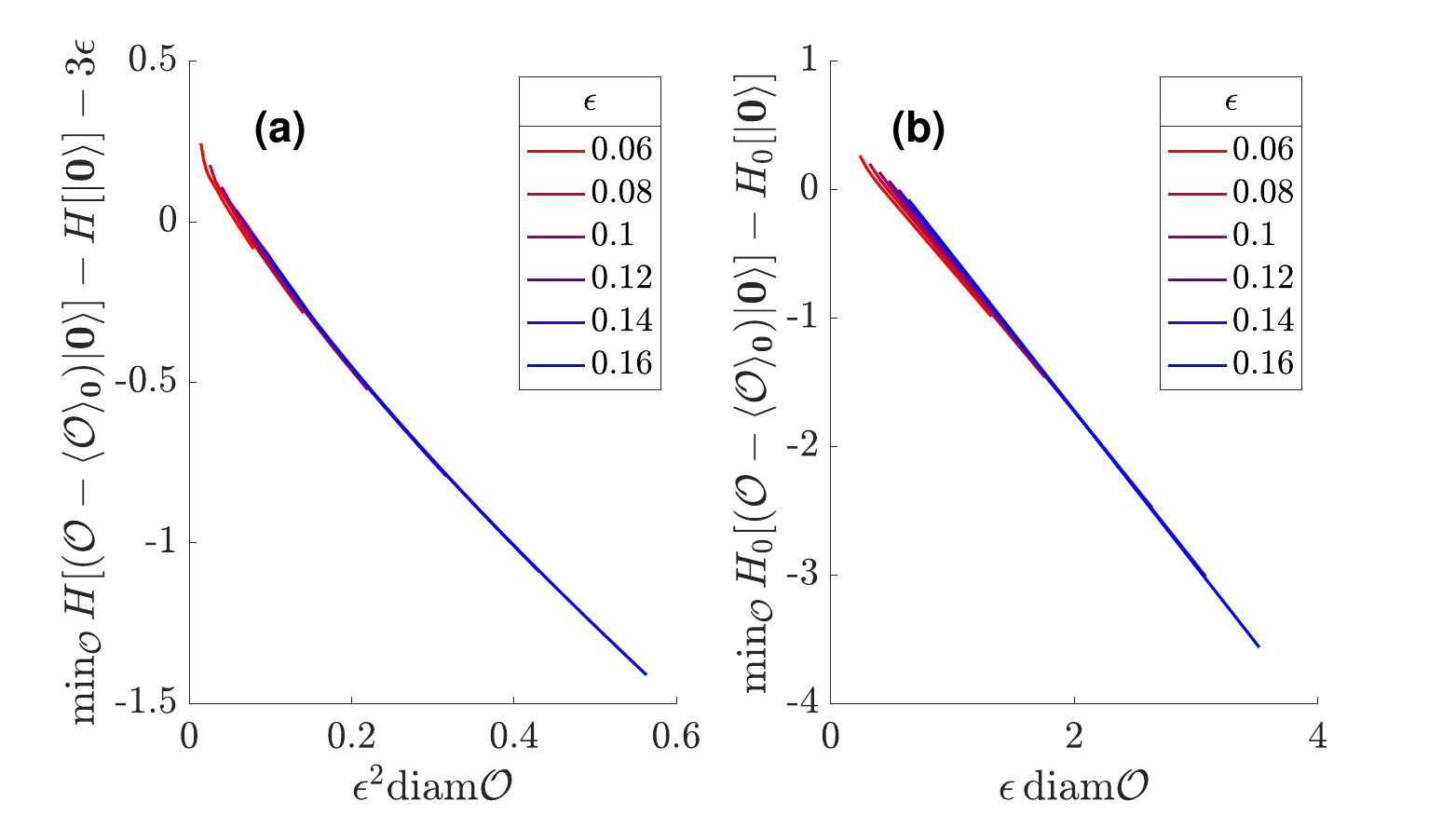}
    \caption{An 1d example \eqref{eq:H=P00++} saturating $R_0\sim \sqrt{R}$ and $\norm{V}_\kappa \sim R^{-d/2}$ when decomposed by $H=H_0+V$ according to Theorem \ref{thm:H=H0+V}. In particular, each color represents a different $\epsilon$, and the collapse of the different colors indicate scalings \eqref{eq:R=eps2} (a) and \eqref{eq:R=eps1} (b). To make the collapse more visible in (a), the energy difference has offset $-3\epsilon$, which vanishes at small $\epsilon$ so does not change the conclusion. }
    \label{fig:sqrtRvsR}
\end{figure}

\section{Diagonalizing metastable systems via Schrieffer-Wolff transformations}\label{app:B}
This appendix is organized around the proof of Theorem \ref{thm:main}, which proves that the metastability condition in Definition \ref{def:meta} implies non-perturbatively slow dynamics in local correlation functions, where the metastable state is a product state. 
 As discussed in Section \ref{sec:define}, this is equivalent (up to a quasilocal change of basis) to any model where $|\psi_0\rangle$ is SRE. We first present the proof, with subsequent subsections used to prove technical propositions for some intermediate steps. In the final subsection, we generalize the results to general $|\psi_0\rangle$, which requires further technical modifications to a few steps.

\subsection{Proof of Theorem \ref{thm:main}}

\begin{proof}[Proof of Theorem \ref{thm:main}]
We assume $\kpsi=\kzero$ without loss of generality.
We use an iterative Schrieffer-Wolff transformation (SWT) procedure to find a quasi-local unitary $U$ that makes $(H_0,V,\kzero)$ $(v_*, \mathbbm{d}_*,a_*,\kappa_*)$-locally-diagonalizable, where \begin{subequations}
    \begin{align}
    v_*&=\tilde{c}_*  \epsilon\, \lr{ 2^{-k_*} + \sum_{k=1}^{k_*-1} 2^{-k} \ee^{-\frac{\kappa_k}{2}(R/2)^\alpha } }, \label{eq:V<exp} \\
    \mathbbm{d}_* &= c_D  \epsilon, \label{eq:H*<} \\
    a_* &=c_A  \epsilon \label{eq:A<eps}.
\end{align}
\end{subequations}
\revise{Here \begin{equation}\label{eq:k*}
    k_*:= \left\lfloor c_* \lr{\frac{\Delta}{\epsilon}}^{\frac{2\alpha-1}{2d-1}} \right\rfloor,
\end{equation} 
and $\kappa_k$ is defined below.}
In fact, the unitaries constructed below are exactly the same as those in \cite{our_preth} that proves prethermalization for perturbing a gapped $H_0$; the crucial observation here is that such SWT is useful for the weaker metastable condition instead of a global gap. To be self-contained, we describe the iterative procedure as follows.

Suppose at the $k$-th step, we have constructed a unitary $U_{k-1}$ such that the Hamiltonian is rotated to \begin{equation}\label{eq:UkHUk=}
    U^\dagger_{k-1} H U_{k-1}=H_0+D_k+V_k.
\end{equation}
Here $D_k$ is the effective Hamiltonian at this order that roughly stabilizes $\kzero$ (as we will see), and $V_k$ is the remaining perturbation that is roughly suppressed to order $\epsilon^k$.
For example, $k=1$ corresponds to the original Hamiltonian with $D_1=0,U_0=I,V_1=V$. For $k=k_*$, we have $U=U_{k_*-1}$ and $H_*+V_*=H_0+D_{k_*}+V_{k_*}$. Note that $V_*\neq V_{k_*}$ as we will move terms around in the end of the proof.

Defining superoperator $\LL_0$ by \begin{equation}\label{eq:L0=H0}
    \LL_0 \OO :=\ii [H_0, \OO],
\end{equation}
we choose \begin{equation}\label{eq:Ak+1=}
    A_{k} = \mathbb{A} V_{k}:= \ii \int^\infty_{-\infty} \dd t\, W(t)\ee^{ t \LL_0} V_{k},
\end{equation}
($\mathbb{A}$ is a superoperator that acts on operators, while the function $W(t)$ will be defined shortly
to rotate to the next order so that $U_k = U_{k-1}\ee^{A_k}$, and \begin{equation}\label{eq:Hk+V=Hk+1}
    \ee^{-\mathcal{A}_k}(H_0 + D_k+ V_k) = H_0 + D_{k+1}+ V_{k+1}, \where \mathcal{A}_k = [A_k, \cdot].
\end{equation}
Here \begin{equation}
    D_{k+1} := D_k + \PP V_k,
\end{equation}
and $V_{k+1}$ is defined by \eqref{eq:Hk+V=Hk+1}. Here the superoperator $\PP$ acting on any operator $\OO$ is defined similarly as $\mathbb{A}$ in \eqref{eq:Ak+1=}: \begin{equation}\label{eq:PO=w}
    \PP \OO:= \int^\infty_{-\infty} \dd t\, w(t)\ee^{t \LL_0} \OO .
\end{equation}

The real kernel functions $w(t),W(t)$ are the same as \cite{our_preth} with the only difference of rescaling $\Delta\rightarrow \Delta/2$; for example, \begin{align}
    w(t) := c_\Delta \frac{\Delta}{2} \prod_{n=1}^\infty \lr{\frac{\sin a_n t}{a_n t}}^2,\quad \mathrm{where}\quad a_1=c_1\Delta/2, \quad a_n=\frac{a_1}{n\ln^2n}, \quad\forall n\ge 2.
\end{align}
Here $c_1\approx 0.161$ is chosen such that $\sum_{n=1}^\infty a_n=\Delta/4$, and $c_\Delta\in (1/(2\pi), 1/\pi)$ is a pure number chosen so that the function is normalized $
    \int_{-\infty}^\infty \mathrm{d}t \; w(t) = 1$.
$w(t)$ decays almost exponentially at large $t$: \begin{equation}\label{eq:w<}
    w(t)\le \frac{\ee^2 \Delta^2}{2} t \ee^{-\frac{2}{7} \frac{\Delta t/2}{\ln^2 (\Delta t/2)}}, \quad \mathrm{if}\quad t \ge 2\ee^{1/\sqrt{2}}/\Delta,
\end{equation}
and
has a compact Fourier transform $\hat{w}(E)$: \begin{equation}\label{eq:w<w0}
    \hat{w}(E) =0, \quad \forall |E|\ge \Delta/2.
\end{equation}
The odd function $W(t)$ is then defined as an integral over $w(t)$: \begin{equation}\label{eq:W=int_w}
    W(t)=-W(-t) := \int\limits^\infty_t \dd s\, w(s),\quad (t>0).
\end{equation}

As one can verify from \eqref{eq:PO=w} and \eqref{eq:w<w0}, $\PP \OO$ annihilates all matrix elements $\alr{E|\OO|E'}$ of $\OO$ that connect eigenstates whose energies $|E-E'|\le \Delta/2$. We refer to $\PP$ as a superprojector. 
On the other hand, \eqref{eq:W=int_w} leads to \begin{equation}\label{eq:H0A+V=PV}
    [H_0, A_k] + (1-\mathbb{P})V_k = 0,
\end{equation}
namely at first order, the rotated $H_0$, $(\ee^{-\cA_k}-1)H_0 \approx [H_0,A_k]$, cancels the off-resonant perturbation $(1-\PP)V_k$.

The perturbation at next order is then \begin{align}\label{eq:Vk+1}
    V_{k+1} &= \ee^{-\cA_k}(H_0 + D_k + V_k) - H_0 - (D_k+ \PP V_k) \nonumber\\ 
    &= (\ee^{-\cA_k} - 1)H_0 + (1-\PP)V_k + (\ee^{-\cA_k}-1)(D_k+V_k)  \nonumber\\
    &= \int_0^1 \dd s \lr{ \ee^{-s\cA_k}-1}(\PP-1)V_k +(\ee^{-\cA_k}-1)(D_k+V_k),
\end{align}

To bound the locality of $V_k$ etc in the above procedure, we invoke the following two propositions adapted from \cite{our_preth}.
\begin{prop}[Proposition 5 in \cite{our_preth}]\label{prop:POAO<}
Suppose $H_0$ has $(\mu,u)$-LRB.
    Let $0<\kappa^\prime<\kappa$, with $\delta\kappa = \kappa-\kappa'$. If \begin{equation}\label{eq:kappa<u}
        \kappa \le \min\lr{\mu/5, 1/\ee},
    \end{equation}
then 
\begin{equation}\label{eq:PO<}
    \norm{\PP \OO}_{\kappa'}, \Delta \norm{\bA \OO}_{\kappa'} \le c_w (-\ln \delta\kappa)^{d-1} \norm{\OO}_\kappa,
     \end{equation}
where $c_w$ is determined by $\alpha,d,\mu,u/\Delta$.
\end{prop} 
Below when applying this proposition, the condition \eqref{eq:kappa<u} will always be satisfied because we only use $\kappa\le\kappa_1$ and $\kappa_1$ is assumed to be bounded \eqref{eq:kap<mu}.

\begin{prop}[Proposition 6 in \cite{our_preth}]\label{prop:AO-O} If \begin{equation}\label{eq:A<C}
    \norm{A}_\kappa \le \lr{\delta\kappa}^{\frac{3d+1}{\alpha}},
\end{equation} 
with $\delta\kappa<\kappa\le 1/\ee$,
then
\begin{equation}\label{eq:AO-O}
    \norm{\ee^{-\mathcal{A}} \OO -\OO}_{\kappa'} \le c_- (\delta \kappa)^{-\frac{2d-1}{\alpha}} \norm{A}_{\kappa} \norm{\OO}_\kappa,
\end{equation}
where $\kappa'=\kappa-\delta\kappa$, and $c_-$ is determined by $\alpha,d$.
\end{prop} 
Note that $c_\omega$ and $c_-$ do not depend on $\kappa_1$ anymore comparing to the statements in \cite{our_preth}, because we have explicit upper bound $\kappa\le 1/\ee$, and $c_\omega,c_-$ do not blow up at small $\kappa_1$ by tracing the proof. One can see this latter point for $c_-$ more directly just from the statement \eqref{eq:AO-O} in \cite{our_preth}: it holds regardless of how small $\kappa$ is with respect to $\kappa_1$, which does not appear on both sides.

With the above two propositions, we follow \cite{our_preth} to define 
\begin{equation}\label{eq:dkvk=}
    \kappa_k := \revise{\frac{\kappa_1}{2}\lr{1+\frac{1}{1+\ln k}}>\frac{\kappa_1}{2}=:\kappa_*},\quad 2\delta\kappa_k:=\kappa_{k-1}-\kappa_k, \quad \mathbbm{d}_k = \norm{D_k}_{\kappa_k},\quad v_{k}:=\norm{V_{k}}_{\kappa_k}, \quad \tilde{v}_k = c_w (-\ln\delta\kappa_{k+1})^{d-1} v_k,
\end{equation}
and ultimately obtain a iteration relation: if  \begin{equation}\label{eq:v<k}
    \tilde{v}_{k-1}\le c_a\lr{k\ln^2 k}^{-\frac{3d+1}{\alpha}} \Delta,
\end{equation}
holds, then \begin{subequations}
\begin{align}
    \mathbbm{d}_k &\le \mathbbm{d}_{k-1} + \tilde{v}_k, \label{eq:dk} \\
    \tilde{v}_k &\le \frac{\tilde{c}_v}{\Delta} \lr{k\ln^2 k}^{\frac{2d-1}{\alpha}} \lr{\ln k}^{d-1} \tilde{v}_{k-1}(\mathbbm{d}_{k-1}+ \tilde{c}_w\tilde{v}_{k-1}) \nonumber\\
    &\le \frac{c_v}{\Delta} k^{\frac{2d-1}{2\alpha-1}} \tilde{v}_{k-1}(\mathbbm{d}_{k-1}+ \tilde{c}_w\tilde{v}_{k-1}).\label{eq:tvk}
\end{align}
\end{subequations}
Here all constants $c_a,\tilde{c}_v,\tilde{c}_w,c_v$ are determined by $\kappa_1,c_-,c_w$.
In \eqref{eq:tvk}, we have replaced the two power functions of $\ln k$ by power functions of $k$ that multiply to $k^{\frac{2d-1}{2\alpha-1}-\frac{2d-1}{\alpha}}$, with the price of adjusting the prefactor. Along the way, we fix a minus-sign typo in the exponent of the corresponding Eq.(S54) in \cite{our_preth}.

One can verify that if \eqref{eq:e<Delta} holds for a sufficiently small constant $c_V$ determined by $c_a,c_v$, then at the initial step $k=2$, \eqref{eq:v<k} is satisfied and \eqref{eq:tvk} becomes \begin{equation}
    \tilde{v}_2\le \tilde{v}_1/2.
\end{equation} 
By iteration, \begin{equation}\label{eq:vk<exp}
    \tilde{v}_k\le \tilde{v}_{k-1}/2\le \cdots \le 2^{1-k}\tilde{v}_1,
\end{equation}
as long as $k<k_*$ where the power of $k$ factor in \eqref{eq:tvk} has not grown too large comparing to the suppression factor \begin{equation}\label{eq:dk<v1}
    \mathbbm{d}_{k-1}\le \sum_{k'=1}^{k-1}\tilde{v}_k\le  2 \tilde{v}_1.
\end{equation}
In this regime, $\tilde{v}_k$ decays faster than the right hand side of \eqref{eq:v<k} so it is automatically satisfied for all $k$ (as long as the constant $c_V$ is chosen sufficiently small).
\eqref{eq:vk<exp} and \eqref{eq:dk<v1} at the last step prove \eqref{eq:H*<}.
Similarly, \eqref{eq:A<eps} comes from \eqref{eq:PO<} and \eqref{eq:vk<exp}: \begin{equation}
    a_*\le \max_{k=1,\cdots,k_*-1} 2^k \norm{A_k}_{\kappa_k-\delta\kappa_{k+1}} \le \max_{k=1,\cdots k_*-1}2^k \tilde{v}_k\le  2 \tilde{v}_1.
\end{equation}
Here we can take maximum over $k$, because the total unitary $U$ can be written as \begin{equation}\label{eq:As=Ak}
    U = \mathcal{T} \mathrm{exp}\mlr{ \int^1_0 \dd s A(s)}, \where A(s) = \left\{ \begin{aligned}
        &0, & &s< 2^{1-k_*} \\
        &2^k A_k, & 2^{-k} \le \, & s < 2^{1-k}, \quad (k=1,\cdots, k_*-1)
    \end{aligned} \right.
\end{equation}

It remains to define $V_*$ and show \eqref{eq:V<exp}. 
Since \begin{equation}
    D_{k_*} = \PP(V_1+\cdots + V_{k_*-1}),
\end{equation}
does not stabilize $\kzero$, we need to extract such contribution and put it into $V_*$.
We use the following Proposition, proven later, for each $\PP V_k$.
\begin{prop}\label{prop:stabilize}
Suppose $H_0$ satisfies \eqref{eq:H0S<h0} and has $(\mu,u)$-LRB \eqref{eq:LRB_meta}.
Suppose $(H_0,\kzero)$ is $(\Delta,R)$-metastable where $H_0$ stabilizes $\kzero$ as \eqref{eq:H0psi0=0}, and $R\ge c_{\rm R}$ is sufficiently large with $c_{\rm R}$ determined by $d,\mu,\Delta/h_0,\Delta/u$.
For any $\kappa$ satisfying \eqref{eq:kappa<u} and any Hermitian $V$, there exists Hermitian $V^{\rm P}$ that stabilizes $\kzero$, and \begin{equation}\label{eq:PV-VP<}
        \norm{\PP V-V^{\rm P}}_0 \le 3c_w (-\ln (\kappa/2))^{d-1}\ee^{-\frac{\kappa}{2}(R/2)^\alpha } \norm{V}_\kappa.
    \end{equation}
\end{prop}

As a result, we define \begin{equation}
    H_* = H_0 + \sum_{k=1}^{k_*-1}V_k^{\rm P},
\end{equation}
that satisfies \eqref{eq:H*=E*}, and \begin{equation}\label{eq:V*=}
    V_* := V_{k_*} + \sum_{k=1}^{k_*-1} \lr{\PP V_k - V_k^{\rm P}},
\end{equation}
with bound \begin{align}
    \norm{V_*}_0 &\le \norm{V_{k_*}}_{\kappa_{k_*}} + 3c_w  \sum_{k=1}^{k_*-1}(-\ln (\kappa_k/2))^{d-1} \ee^{-\frac{\kappa_k}{2}(R/2)^\alpha } \norm{V_k}_{\kappa_k} \nonumber\\
    &\le \frac{c'_V}{2} \lr{ \tilde{v}_{k_*} + \sum_{k=1}^{k_*-1} \ee^{-\frac{\kappa_k}{2}(R/2)^\alpha } \tilde{v}_k } \le c'_V \tilde{v}_1 \lr{ 2^{-k_*} + \sum_{k=1}^{k_*-1} 2^{-k} \ee^{-\frac{\kappa_k}{2}(R/2)^\alpha } },
\end{align}
for some constant $c'_V$,
where we have used \eqref{eq:dkvk=} and \eqref{eq:vk<exp}.
This proves \eqref{eq:V<exp}.

Finally, the bound \eqref{eq:V<exp} can be simplified by \begin{align}\label{eq:vstar=two}
    v_* &\le \tilde{c}_*  \epsilon\, 2^{-k_*} +\tilde{c}_*\epsilon k_*\max_{1\le k\le k_*} 2^{-k} \ee^{-\frac{\kappa_k}{2}(R/2)^\alpha } \nonumber\\
    &\le \exp\mlr{-\Omega\lr{\lr{\frac{\Delta}{\epsilon}}^{\frac{2\alpha-1}{2d-1}}}} + \epsilon^{\frac{2(d-\alpha)}{2d-1}}
        \exp\mlr{-\Omega\lr{\frac{R^\alpha}{\ln\mlr{ \min\lr{R,\Delta/\epsilon}}} }}, 
\end{align}
where we have used that function $ck+R^\alpha/\ln k$ with constant $c$ is minimized at $k/\ln^2 k \sim R^\alpha$ with value $\sim R^\alpha/\ln k\sim R^\alpha/\ln R$. Applying Proposition \ref{cor1} with \eqref{eq:vstar=two} yields \eqref{eq:tstarsec3} and completes the proof.
\end{proof}

\subsection{Local operators after superprojection almost stabilizes the metastable state}\label{app:one_local_O}
The goal is now to prove Proposition \ref{prop:stabilize}. Since $\PP V=\sum_S \PP V_S$, we first focus on one local operator $\OO=V_S$ in this subsection, showing an approximation $\OO^{\mathrm{P}}\approx \PP \OO$ that stabilizes the state, and combine different $V_S$ terms in the end of Appendix \ref{app:all_local_O}. The idea for a single local operator is as follows: We know $\PP \OO$ is quasi-local around the original $\OO$ from Proposition \ref{prop:POAO<}. Suppose $\PP \OO$ is exactly localized in a region of diameter $R$, then $\ket{\phi}:=\lr{\PP\OO-\alr{\PP\OO}_{\bm{0}}}\kzero$ has energy $\ge \Delta$ from the metastability condition (assuming $\kzero$ has energy $0$). However, $\PP\OO\kzero$ has no support on energies $>\Delta/2$ from \eqref{eq:w<w0}: The only resolution is then $\ket{\phi}=0$, i.e. $\PP \OO$ exactly stabilizes $\kzero$. In reality, although $\PP \OO$ is not exactly localized inside this size-$R$ region, the leakage outside is exponentially small in $R$; this leads to an approximate stabilizer $\OO^{\rm P}\approx \PP \OO$ instead, with exponentially small error. This last step requires a technical Lemma in terms of energy distributions of locally-excited states, which we prove in Appendix \ref{app:proof14}.

For a single local operator $\OO$ supported on $S_0$, $\PP \OO$ can be decomposed to two parts, a nearby part and a faraway one with cutoff distance $r$: \begin{equation}\label{eq:PO=Q}
    \PP \OO = \QQ_r' \PP \OO + (1-\QQ_r') \PP \OO = \QQ_r' \PP \OO+ \int^\infty_{-\infty}\dd t\, w(t) (1-\QQ_r')(\ee^{t\LL_0}\OO), 
\end{equation}
where the superoperator $\QQ_r'$ is defined by \begin{equation}\label{eq:QQprimedef}
\QQ_r^\prime \OO := \int\limits_{\text{Haar outside }B(S_0, r)} \mathrm{d}U \; U^\dagger \mathcal{O} U, 
\end{equation} 
selects the part of the operator that acts on $B(S_0,r)$, i.e. sites of distance $\le r$ to the original support $S_0$. The faraway part will be negligible for sufficiently large $r$ due to the locality of $\PP$ (see, e.g. \eqref{eq:PO<}). In the following, we show that the nearby part almost stabilizes $\kzero$ from the metastability condition.

\begin{prop}\label{prop:stabilize_local}
Suppose $H_0$ satisfies \eqref{eq:H0S<h0} and has $(\mu,u)$-LRB \eqref{eq:LRB_meta}.
Suppose $(H_0,\kzero)$ is $(\Delta,R)$-metastable where $H_0\kzero=0$\footnote{In this subsection, we set $E_0=0$ in \eqref{eq:H0psi0=0} without loss of generality. }. There exists a constant $c_{\rm R}$ determined by $d,\mu,\Delta/h_0,\Delta/u$, such that if \begin{equation}\label{eq:R>cR}
    R\ge c_{\rm R},
\end{equation}
the following holds:

For any local Hermitian operator $\OO$ supported in $S_0$ with diameter satisfying \begin{equation}\label{eq:diam<R}
    \max(1,\mathrm{diam}S_0)=:s\le R/2,
\end{equation}
let \begin{equation}\label{eq:r=R-s}
    r = \floor{(R-s)/2}\ge R/4-1,
\end{equation}
(where we have used \eqref{eq:diam<R}). 
There exists a Hermitian operator $\OO^{\rm P}$ supported in $B(S_0, r)$ that stabilizes $\kzero$, and
\begin{equation}\label{eq:PO-OP<}
        \norm{\QQ_r'\PP\OO - \OO^{\rm P}} \le 8\norm{\OO}\frac{h_0c_d R^d}{\Delta} \mlr{ \ee^{-\mu r/3} + c_\delta \ee^{-(2\Delta r/u)^{1-\delta}} },
    \end{equation}
for any arbitrarily small $\delta>0$, with $c_\delta$ determined solely by $\delta$.
\end{prop}

\begin{proof}
Lieb-Robinson bound \eqref{eq:LRB_meta} implies \begin{equation}\label{eq:Qr<}
    \norm{(1-\QQ_r') \lr{\ee^{t\LL_0}\OO}} \le 2\norm{\OO} \min\lr{1, c_d c_{\rm LR}s^{d-1} \ee^{\mu(u|t|-r)} },
\end{equation}
which bounds the second term in \eqref{eq:PO=Q},
\begin{align}\label{eq:QPO<}
    \norm{(1-\QQ_r') \PP \OO} &\le 2\norm{\OO} \int^\infty_{-\infty}\dd t\, w(t) \min\lr{1, c_d c_{\rm LR}s^{d-1} \ee^{\mu(u|t|-r)} } \nonumber\\
    &\le 2\norm{\OO} \mlr{ 2\int^{t_r}_0\dd t\, w(t)c_d c_{\rm LR}s^{d-1} \ee^{\mu(ut-r)} + 2\int_{t_r}^\infty\dd t\, 2(\ee\Delta/2)^2 t\ee^{-\frac{1}{7} \frac{\Delta t}{\ln^2 (\Delta t/2)}}} \nonumber\\
    &\le 2\norm{\OO} \mlr{ c_d c_{\rm LR}s^{d-1} \lr{\max_{|t|\le t_r} \ee^{\mu(ut-r)}} + \int_{t_r}^\infty\dd t\, (\ee\Delta)^2 t\ee^{-\frac{1}{7} \frac{\Delta t}{\ln^2 (\Delta t/2)}}} \nonumber\\
    &\le 2\norm{\OO}\mlr{ c_dc_{\rm LR} s^{d-1}\ee^{-2\mu r/3} + c_\delta \ee^{-(2\Delta r/u)^{1-\delta}} } \le 2\norm{\OO}\mlr{ \ee^{-\mu r/3} + c_\delta \ee^{-(2\Delta r/u)^{1-\delta}} }.
\end{align}
In the second line, we have used the even parity of $w(t)$ and defined \begin{equation}
    t_r:= r/(3u) \ge 2\ee^{1/\sqrt{2}}/\Delta.
\end{equation}
which holds because of \eqref{eq:r=R-s} and \eqref{eq:R>cR} for a sufficiently large $c_{\rm R}$. Therefore we can use bound \eqref{eq:w<} in the second line of \eqref{eq:QPO<}. In the third line of \eqref{eq:QPO<}, we have used $\int^\infty_0w(t)=1/2$. In the last line of \eqref{eq:QPO<}, we have used \eqref{eq:r=R-s} and introduced an arbitrarily small $\delta>0$ that determines the constant $c_\delta$.

According to decomposition \eqref{eq:PO=Q},
\begin{equation}\label{eq:POpsi=psi+psi}
    \lr{ \PP\OO - \expval{\PP\OO}_{\bm{0}} } \kzero = \ket{\psi_{\le R}} + \ket{\psi_{>R}},
\end{equation}
where \begin{equation}\label{eq:psi<R}
    \ket{\psi_{\le R}}= \ket{\psi_{\le R}}_{S}\otimes \kzero_{S^{\rm c}} := \lr{\QQ_r' \PP \OO - \alr{\QQ_r' \PP \OO}_{\bm{0}} } \kzero 
\end{equation}
with $S = B(S_0,r)$ and $\diam S\le R$ from \eqref{eq:r=R-s}, and \begin{equation}\label{eq:psi>R<}
    \norm{\ket{\psi_{>R}}} = \norm{ \mlr{(1-\QQ_r') \PP \OO - \alr{(1-\QQ_r') \PP \OO}_{\bm{0}}} \kzero } \le \norm{(1-\QQ_r') \PP \OO} ,
\end{equation}
due to \begin{equation}\label{eq:A-A<A}
    \norm{(A-\expval{A}_{\psi})\ket{\psi}}\le \norm{A},
\end{equation}
for any operator $A$ and state $\psi$.

$\ket{\psi_{\le R}}$ describes some local excitation above the all-zero state, which has the following nice property in terms of its energy distribution (proven in Appendix \ref{app:proof14}):
\begin{lem}\label{lem:PI}
Suppose $H_0$ satisfies \eqref{eq:H0S<h0}.
Suppose $(H_0,\kzero)$ is $(\Delta,R)$-metastable where $H_0\kzero=0$. Let \begin{equation}
    \Pi_{>E}:= \sum_{E'>E} \ket{E'}\bra{E'},
\end{equation}
be the projection operator onto eigenstates $\{\ket{E'}\}$ of $H_0$ with energy larger than $E$. Then
    \begin{equation}\label{eq:PI>Delta}
        \norm{\Pi_{>\Delta/2} \ket{\phi}_S\otimes \kzero_{S^{\rm c}} } \ge \frac{\Delta}{4h_0|S|} \norm{\ket{\phi}_S\otimes \kzero_{S^{\rm c}} },\quad \text{for any } \ket{\phi}_S\perp \kzero_S \text{ with } \diam S \le R.
    \end{equation}
\end{lem}

Therefore using \eqref{eq:cd} we have
\begin{equation}\label{eq:PI>psi}
    \norm{\Pi_{> \Delta/2}\ket{\psi_{\le R}}} \ge \frac{\Delta}{4h_0c_d R^d } \norm{\ket{\psi_{\le R}}}.
\end{equation}

To upperbound $\norm{\ket{\psi_{\le R}}}$, it suffices to bound $\norm{\Pi_{> \Delta/2}\ket{\psi_{\le R}}}$ due to \eqref{eq:PI>psi}. To this end, notice that $\PP$ projects onto the $\le \Delta/2$ energy subspace: Using \eqref{eq:w<w0}, \begin{align}\label{eq:PIPO<w0}
    \Pi_{> \Delta/2} \lr{ \PP\OO - \expval{\PP\OO}_{\bm{0}} } \kzero=\Pi_{> \Delta/2} \PP\OO \kzero  = \sum_{E> \Delta/2} \hat{w}(E) \OO_{E, 0} \ket{E}=0,
\end{align}
where $\OO_{E, 0}:= \bra{E} \OO\kzero$. Therefore, \eqref{eq:PI>psi} and \eqref{eq:POpsi=psi+psi} lead to \begin{align}\label{eq:psi<R<}
    \norm{\ket{\psi_{\le R}}} &\le \frac{4h_0c_d R^d}{\Delta} \norm{\Pi_{> \Delta/2}\ket{\psi_{\le R}}} = \frac{4h_0c_d R^d}{\Delta} \norm{\Pi_{> \Delta/2}\ket{\psi_{> R}}}
    \le \frac{4h_0c_d R^d}{\Delta} \norm{\ket{\psi_{> R}}} \nonumber\\
    &\le 8\norm{\OO}\frac{h_0 c_d R^d}{\Delta} \mlr{ \ee^{-\mu r/3} + c_\delta \ee^{-(2\Delta r/u)^{1-\delta}} }.
\end{align}
where the last step uses \eqref{eq:psi>R<} and \eqref{eq:QPO<}.

Define the Hermitian operator \begin{equation}\label{eq:OP=}
    \OO^{\rm P} = \QQ_r'\PP\OO - \lr{ \ket{\psi_{\le R}}_S\bra{\bm{0}} + \mathrm{H.c.} },
\end{equation}
supported in $S$. It stabilizes $\kzero$ because \begin{equation}
    \OO^{\rm P}\kzero = \QQ_r'\PP\OO\kzero - \ket{\psi_{\le R}} = \alr{\QQ_r' \PP \OO}_{\bm{0}}\kzero,
\end{equation}
from \eqref{eq:psi<R}; furthermore, \eqref{eq:psi<R<} leads to the bound \eqref{eq:PO-OP<}.

\end{proof}

\subsection{Proof of Lemma \ref{lem:PI}}\label{app:proof14}

\begin{proof}
Intuitively, we only need to show $\ket{\phi\otimes \bm{0}}:=\ket{\phi}_S\otimes \kzero_{S^{\rm c}}$ has small weight on extremely large $\gtrsim |S|$ energies, so that the large average energy $H_0[\ket{\phi\otimes \bm{0}}]$ implies sufficient weight in range $[\Delta/2, \mathrm{O}(|S|)]$. 

To upperbound the weight on large energies, observe that \begin{equation}
    H_0\ket{\phi\otimes \bm{0}} = H_0 \lr{\ket{\phi}_S\bra{\bm{0}}}\kzero = [H_0, \ket{\phi}_S\bra{\bm{0}}] \kzero + \lr{\ket{\phi}_S\bra{\bm{0}}} H_0\kzero = \mlr{\sum_{S':S'\cap S\neq \emptyset} H_{0,S'}, \ket{\phi}_S\bra{\bm{0}}} \kzero
\end{equation}
where we have used $H_0\kzero=0$, and only terms in $H_0$ that intersect $S$ contribute to the commutator.
The energy variance is then \begin{equation}\label{eq:H02<}
    \bra{\phi\otimes \bm{0}} H_0^2\ket{\phi\otimes \bm{0}}= \norm{H_0\ket{\phi\otimes \bm{0}}}^2 \le \lr{2\sum_{S':S'\cap S\neq \emptyset} \norm{ H_{0,S'}} } ^2\le 4\lr{\sum_{i\in S}\sum_{S'\ni i} \norm{ H_{0,S'}} } ^2\le 4 h_0^2|S|^2,
\end{equation}
using \eqref{eq:H0S<h0}.

For any random variable $Z$, the Paley–Zygmund inequality holds (as a lowerbound version of Markov inequality) \begin{equation}
    \mathrm{Pr}[Z>\theta \mathrm{E}[Z]] \ge (1-\theta)^2 \frac{\mathrm{E}[Z]^2}{ \mathrm{E}[Z^2]},
\end{equation}
for any $\theta\in (0,1)$.
Let $Z=E$ be the energy distribution of $\ket{\phi\otimes \bm{0}}$ and $\theta \mathrm{E}[Z]=\Delta/2$, we use energy expectation \eqref{eq:H0>Delta1} due to $\mathrm{diam}S\le R$ and variance \eqref{eq:H02<} to get \begin{equation}
    \norm{\Pi_{>\Delta/2} \ket{\phi\otimes \bm{0}}}^2 \ge \lr{1-\frac{\Delta}{2 H_0[\ket{\phi\otimes \bm{0}}]}}^2 \frac{H_0[\ket{\phi\otimes \bm{0}}]^2}{4h_0^2|S|^2} \ge \lr{\frac{\Delta}{4h_0|S|}}^2,
\end{equation}
for any normalized state $\ket{\phi\otimes \bm{0}}$. \eqref{eq:PI>Delta} then follows.

\end{proof}

\subsection{Proof of Proposition \ref{prop:stabilize}}\label{app:all_local_O}
\begin{proof}
The constant $c_{\rm R}$ is just the one determined by Proposition \ref{prop:stabilize_local}.
We decompose $V=V_{\le R/2}+V_{>R/2}$ where $V_{\le R/2}$ contains all local terms $V_{S_0}$ with diameter $\mathrm{diam}S_0\le R/2$, for which we can apply Proposition \ref{prop:stabilize_local} and obtain a stabilizer $V_{S_0}^{\rm P}$ of $\kzero$ that approximates $\QQ_r'\PP V_{S_0}$ (here $\QQ_r'$ depends on $S_0$ with $r$ given by \eqref{eq:r=R-s}). They sum up to the desired total stabilizer defined by \begin{equation}
    V^{\rm P}:= \sum_{V_{S_0}:\mathrm{diam}S_0 \le R/2} V_{S_0}^{\rm P}.
\end{equation}
Therefore we have \begin{equation}\label{eq:PV-VP=}
    \PP V - V^{\rm P} = \PP V_{>R/2} + \sum_{V_{S_0}:\mathrm{diam}S_0 \le R/2}(1-\QQ_r')\PP V_{S_0} + \sum_{V_{S_0}:\mathrm{diam}S_0 \le R/2} \lr{ \QQ_r'\PP V_{S_0} - V_{S_0}^{\rm P} }.
\end{equation}

The first term in \eqref{eq:PV-VP=} is bounded by Proposition \ref{prop:POAO<} \begin{align}\label{eq:PV>R<}
    \norm{\PP V_{>R/2}}_0 &\le  c_w (-\ln (\kappa/2))^{d-1}\norm{V_{>R/2}}_{\kappa/2} \nonumber\\
    &\le c_w (-\ln (\kappa/2))^{d-1} \ee^{-\frac{\kappa}{2}(R/2)^\alpha }\norm{V}_{\kappa}.
\end{align}

The second term in \eqref{eq:PV-VP=} is the diameter $\ge R$ part of $\PP V_{\le R/2}$ and bounded similarly \begin{align}\label{eq:1-QV<}
    \norm{\sum_{V_{S_0}:\mathrm{diam}S_0 \le R/2}(1-\QQ_r')\PP V_{S_0}}_0 &\le \ee^{-\frac{\kappa}{2}R^\alpha } \norm{\PP V_{\le R/2} }_{\kappa/2} \le \ee^{-\frac{\kappa}{2}R^\alpha } c_w (-\ln (\kappa/2))^{d-1} \norm{V_{\le R/2}}_{\kappa}\nonumber\\
    & \le c_w (-\ln (\kappa/2))^{d-1}\ee^{-\frac{\kappa}{2}R^\alpha }  \norm{V}_{\kappa}.
\end{align}

To bound the final term in \eqref{eq:PV-VP=}, we use the fact that for $\lr{ \QQ_r'\PP V_{S_0} - V_{S_0}^{\rm P} }$ to act on a site $i$, $V_{S_0}$ should act on a site $j$ with $\mathrm{d}(i,j)\le r\le R/2$. Therefore \begin{align}\label{eq:QPV-VP<}
    \norm{\sum_{V_{S_0}:\mathrm{diam}S_0 \le R/2} \lr{ \QQ_r'\PP V_{S_0} - V_{S_0}^{\rm P} }}_0 &\le \max_i \sum_{j:\mathrm{d}(i,j)\le R/2} \sum_{S_0\ni j} \norm{\QQ_r'\PP V_{S_0} - V_{S_0}^{\rm P}  } \nonumber\\
    &\le c_d R^d \max_j \sum_{S_0\ni j} 8\norm{V_{S_0}}\frac{h_0c_d R^d}{\Delta} \mlr{ \ee^{-\mu R/15} + c_\delta \ee^{-(2\Delta R/5u)^{1-\delta}} } \nonumber\\
    &\le 8 c_d^2 \norm{V}_\kappa \frac{h_0 }{\Delta}R^{2d} \mlr{ \ee^{-\mu R/15} + c_\delta \ee^{-(2\Delta R/5u)^{1-\delta}} }
\end{align}
Here in the second line, we have used \eqref{eq:cd} to bound the number of $j$s, and \eqref{eq:PO-OP<} with $r\ge R/5$ from \eqref{eq:r=R-s} and that $R$ is sufficiently large.

Using the triangle inequality to combine \eqref{eq:PV>R<}, \eqref{eq:1-QV<} (which is bounded by the right hand side of \eqref{eq:PV>R<}) and \eqref{eq:QPV-VP<}, we have \begin{align}\label{eq:PV-VP<2<3}
    \norm{\PP V - V^{\rm P}}_0 &\le \norm{V}_{\kappa}\glr{ 2c_w (-\ln (\kappa/2))^{d-1} \ee^{-\frac{\kappa}{2}(R/2)^\alpha } + 8 c_d^2 \frac{h_0 }{\Delta}R^{2d} \mlr{ \ee^{-\mu R/15} + c_\delta \ee^{-(2\Delta R/5u)^{1-\delta}} } } \nonumber\\
    &\le  \norm{V}_{\kappa}\cdot 3c_w (-\ln (\kappa/2))^{d-1} \ee^{-\frac{\kappa}{2}(R/2)^\alpha } ,
\end{align}
where we have set $\delta = 1-\alpha$ and used again the largeness of $R$, together with \eqref{eq:kappa<u}.
\end{proof}

\subsection{Beyond short-range entangled states}
Before stating the generalized result, we introduce one more definition: We say an operator $\OO$ is $k$-local, if there exists a decomposition $\OO=\sum_S\OO_S$ with all $|S|\le k$ in the sum, where $S$ need \emph{not} be connected. 

\begin{thm}\label{thm:LRE}
    Suppose the conditions of Theorem \ref{thm:main} hold with the metastable state $\kzero$ replaced by any state $\kpsi$ that may not be SRE, and $H_0$ satisfies one extra property that it is $k_0$-local with a constant $k_0$.
    In particular, $H_0$ stabilizes $\kpsi$. Then the conclusions of Theorem \ref{thm:main} still hold subject to a slight modification on the locally-diagonalizable property: For any $\alpha\in(0,1)$ and $\kappa_1$ satisfying \eqref{eq:kap<mu} there exist constants $c_{\rm R},c_*,c_V,c_A,c_D$ determined by $\alpha,d,\mu,\Delta/h_0,\Delta/u,\kappa_1$ and $k_0$ (where $c_{\rm R}$ does not depend on $\alpha,\kappa_1$), such that if $R\ge c_{\rm R}$, for any perturbation $V$ with small local norm $\epsilon$ obeying \eqref{eq:e<Delta}, $(H_0,V,\kpsi)$ is $(v_*, c_D  \epsilon, c_A  \epsilon,\revise{\kappa_1/2})$-locally-diagonalizable$^*$.
    Here locally-diagonalizable$^*$ means Definition \ref{def:localdiag}, but with \eqref{eq:V<exp0} replaced by \begin{equation}\label{eq:V<exp0_weaker}
        \max_i \sum_{S\ni i}\norm{V_{*,S}\kpsi} \le v_*,
    \end{equation} 
    and $v_*$ is small so that $\kpsi$ has life-time $t_*$ in \eqref{eq:tstarsec3} ($\delta>0$ being any positive constant) in the sense of \eqref{eq:rho-rho}: Local correlations functions are preserved for long time starting from a locally-dressed state $\ket{\widetilde{\psi}_0}$.
\end{thm}

Note that by our definition, $k_0$-local Hamiltonians include e.g. two-qubit interactions that decay exponentially with distance. Although \eqref{eq:V<exp0_weaker} is weaker than the original \eqref{eq:V<exp0} by bounding $\norm{V_{*,S}\kpsi}\le \norm{V_{*,S}}$, it is still sufficient to imply long lifetime for the particular initial state, as we will see in the proof. 

\begin{proof}
We first generalize Proposition \ref{cor1} and show that the weaker locally-diagonalizable$^*$ condition also implies long lifetime; the proof of this result will appear at the end.
\begin{prop}\label{prop:diagonal_star}
    If the triple $(H_0,V,\kpsi)$ is $(v_*,\mathbbm{d}_*,a_*, \kappa_*)$-locally-diagonalizable$^*$ with quasilocal unitary $U$, and $\norm{H}_{\kappa_1}\le h$ for \eqref{eq:H=H0+V}, then the conclusions 1 (locality of $U$) and 3 (lifetime of dressed $\kpsi$) in Proposition \ref{cor1} still hold.
\end{prop}

It remains to establish the locally-diagonalizable$^*$ condition with the extra $k_0$ property. Recall that the proof of Theorem \ref{thm:main} works by first performing SWTs that rotate the Hamiltonian to $H_0+D_{k_*}+V_{k_*}$, and then separating $D_{k_*}=D_{k_*}^{\rm P} + \lr{D_{k_*}-D_{k_*}^{\rm P}}$ to define $H_*=H_0+D_{k_*}^{\rm P}$ and $V_*=\lr{D_{k_*}-D_{k_*}^{\rm P}} + V_{k_*}$. The SWT part is always achievable because it only deals with operators and has nothing to do with any state being metastable. The second separation part is done by Proposition \ref{prop:stabilize_local} (Proposition \ref{prop:stabilize} is its immediate generalization, from one local operator to sums of them), which is the only place using the product structure of the metastable state.

The problem here for general $\kpsi$ is that we no longer have a way to separate $D_{k_*}$ such that the non-stabilizing part has small operator norms. So we do not separate the operator, and define $H_*=H_0+\alr{D_{k_*}}_{\psi_0}$, \begin{equation}
    V_*=D_{k_*}-\alr{D_{k_*}}_{\psi_0}+V_{k_*},
\end{equation}
and aim at the weaker \eqref{eq:V<exp0_weaker}. This reduces to bounding \begin{equation}\label{eq:D-Dpsi<}
   \norm{V_{*,S}|\psi_0\rangle} \le  \norm{\lr{D_{k_*,S}-\alr{D_{k_*,S}}_{\psi_0}} \kpsi} +  \norm{V_{k_*,S} \kpsi},
\end{equation}
since the local operator norms of $V_{k_*}$ are already bounded in the previous proof.

Actually, the proof of Proposition \ref{prop:stabilize_local} for $\kpsi=\kzero$ follows immediately from the exponential upper bound \eqref{eq:psi<R<} on \eqref{eq:D-Dpsi<}\footnote{Strictly speaking, \eqref{eq:psi<R<} only bounds \eqref{eq:D-Dpsi<} for sets $S$ that are not larger than the metastability range $R$. \eqref{eq:D-Dpsi<} for larger $S$ are bounded in the previous proof straightforwardly by Lieb-Robinson bounds.}, by constructing the stabilizer $\OO^{\rm P}$ in \eqref{eq:OP=}. It is this stabilizer construction step that lead to the stronger locally-diagonalizable condition, but it is unnecessary here for the weaker locally-diagonalizable$^*$ condition. As a result, here it suffices to generalize \eqref{eq:psi<R<} to \begin{equation}\label{eq:psi<R<_LRE}
    \norm{\ket{\psi_{\le R}}} \le 2c_{\rm LRE}\norm{\OO}\frac{h_0 c_d R^d}{\Delta} \mlr{ \ee^{-\mu r/3} + c_\delta \ee^{-(2\Delta r/u)^{1-\delta}} },
\end{equation}
where $\ket{\psi_{\le R}}=\lr{\QQ_r' \PP \OO - \alr{\QQ_r' \PP \OO}_{\psi_0} } \kpsi=:\widetilde{\OO}\kpsi$. Previously $c_{\rm LRE}=4$, while here we determine this constant shortly in \eqref{eq:cLRE=}; this change just modifies constants and does not influence the qualitative results of Theorem \ref{thm:LRE}. Operator $\widetilde{\OO}$ is supported in $S$ with $\diam S= R$; furthermore, \begin{equation}
    \norm{\widetilde{\OO}} \le 2\norm{\QQ_r' \PP \OO} \le 2\norm{\PP \OO} \le 2 \norm{\OO},
\end{equation} 
where the last inequality comes from $\int w(t)\dd t=1$ and the triangle inequality.

Following the manipulations in \eqref{eq:psi<R<}, we see that \eqref{eq:psi<R<_LRE} holds as long as Lemma \ref{lem:PI} can be generalized to show
\begin{equation}\label{eq:lemmaPI_LRE}
    \norm{\Pi_{>\Delta/2} \ket{\psi_{\le R}} } \ge \frac{\Delta}{c_{\rm LRE}h_0|S|} \norm{\ket{\psi_{\le R}} }.
\end{equation}
Note that we have set the energy zero to $H_0\kpsi=0$, without loss of generality. Tracing the proof of Lemma \ref{lem:PI}, \eqref{eq:lemmaPI_LRE} follows from bounding the energy variance \begin{equation}\label{eq:H2<cLRE}
    \bra{\psi_{\le R}} H_0^2 \ket{\psi_{\le R}} \le \lr{c_{\rm LRE}h_0|S|}^2 \norm{\ket{\psi_{\le R}}}^2.
\end{equation}

With the $k_0$-locality of $H_0$, we can invoke Theorem 2.1 in \cite{Arad_energy16} on the energy distribution of local operator $\widetilde{\OO}$ to get \begin{equation}\label{eq:PI_O_PI<}
    \norm{\Pi_{> E} \widetilde{\OO}\kpsi} \le \norm{\Pi_{> E} \widetilde{\OO} \Pi_{\le 0}}\le \norm{\widetilde{\OO}} \ee^{-\frac{1}{2h_0k_0}(E-2h_0|S|)} \le 2\norm{\OO} \ee^{-\frac{1}{2h_0k_0}(E-2h_0|S|)}.
\end{equation}
$\norm{\Pi_{<-E} \widetilde{\OO} \kpsi}$ obeys the same bound. For any $E$, \begin{align}\label{eq:H2<LRE_long}
    \bra{\psi_{\le R}} H_0^2 \ket{\psi_{\le R}} &\le E^2 \norm{\ket{\psi_{\le R}}}^2 + \int\limits^\infty_E \dd E' (E')^2 \lr{\norm{\Pi_{> E} \widetilde{\OO}\kpsi}^2 + \norm{\Pi_{<- E} \widetilde{\OO}\kpsi}^2} \nonumber\\
    &\le E^2 \norm{\ket{\psi_{\le R}}}^2 + 8\norm{\OO}^2 \int\limits^\infty_E \dd E' (E+E'-E)^2 \ee^{-\frac{1}{h_0k_0}(E'-2h_0|S|)} \nonumber\\
    &= E^2 \norm{\ket{\psi_{\le R}}}^2 + 8\norm{\OO}^2h_0k_0 \ee^{-\frac{1}{h_0k_0}(E-2h_0|S|)}\mlr{E^2+2Eh_0k_0+2(h_0k_0)^2} .
\end{align}
Assuming that $\norm{\ket{\psi_{\le R}}} \ge 8\norm{\OO}\frac{h_0 c_d R^d}{\Delta} \ee^{-\mu R/3}$, because otherwise \eqref{eq:psi<R<_LRE} holds trivially, and taking $E\ge 3h_0k_0$, \begin{align}
   \bra{\psi_{\le R}} H_0^2 \ket{\psi_{\le R}} &\le E^2\norm{\ket{\psi_{\le R}}}^2+\norm{\ket{\psi_{\le R}}}^2 \frac{1}{8}\lr{\frac{\Delta}{h_0c_dR^d}}^2\ee^{2\mu R/3} h_0k_0 \ee^{-\frac{1}{h_0k_0}(E-2h_0|S|)}E^2\lr{1+\frac{2}{3}+\frac{2}{9}}  \nonumber\\
   &\le E^2\norm{\ket{\psi_{\le R}}}^2\mlr{1+\frac{1}{4}\lr{\frac{\Delta}{h_0c_dR^d}}^2\ee^{2\mu R/3} h_0k_0 \ee^{-\frac{1}{h_0k_0}(E-2h_0|S|)}}.
\end{align}
\eqref{eq:H2<cLRE} then follows by choosing \begin{equation}\label{eq:cLRE=}
    E=c_{\rm LRE} h_0|S|/\sqrt{2}, \where c_{\rm LRE} = \sqrt{2}(2+k_0/\mu+3k_0),
\end{equation} 
(which satisfies the previous condition $E\ge 3h_0k_0$) because the second term in \eqref{eq:H2<LRE_long} scales as $\lesssim R^{-2d}\ee^{2\mu R/3-\mu|S|}\le \ee^{-\mu R/3}$ from $|S|\ge R$ and is thus upper bounded by $1$ if $R$ is sufficiently large given by the condition \eqref{eq:R>cR}. 

To conclude, if $R$ satisfies \eqref{eq:R>cR} where $c_{\rm R}$ now also depends on $k_0$, either $\ket{\psi_{\le R}}=\widetilde{\OO}\kpsi$ is already very small, or its relative weight at very large energy $\ge E$ is bounded by the locality of $\widetilde{\OO}$ in \eqref{eq:PI_O_PI<}, so that $\ket{\psi_{\le R}}$ is still small, as given by \eqref{eq:psi<R<_LRE} following the previous proof. This implies that the effective Hamiltonians $D_k=\PP V_k$ obtained from SWT still almost stabilize the state $\kpsi$, with \eqref{eq:D-Dpsi<} being exponentially small in $R$. Due to Proposition \ref{prop:diagonal_star}, this is sufficient to imply a long lifetime of $\kpsi$. The conclusions of Theorem \ref{thm:main} thus generalize, where the constants gain extra dependence on $k_0$.
\end{proof}

We expect the restriction of Theorem \ref{thm:LRE} to strictly local Hamiltonians is a technical issue. To overcome this problem potentially requires a generalization of imaginary-time operator growth bounds used in \cite{Arad_energy16} beyond strictly local Hamiltonians. We finally prove the generalization on implications of local diagonalizability:

\begin{proof}[Proof of Proposition \ref{prop:diagonal_star}]
The locality of $U$ is independent of the change from \eqref{eq:V<exp0} to \eqref{eq:V<exp0_weaker}, so follows immediately. 

Although conclusion 2 no longer holds, a weaker version still holds that replaces the operator norm in \eqref{eq:coper} with expectation in $\ket{\widetilde{\psi}_0}$: \begin{equation}\label{eq:coper_star}
    \abs{\bra{\widetilde{\psi}_0}\lr{\ee^{\ii t H}\mathcal{O} \ee^{-\ii t H} - U\ee^{\ii t H_*}U^\dagger \mathcal{O} U\ee^{-\ii t H_*}U^\dagger }\ket{\widetilde{\psi}_0}} \le c_{\mathrm{oper}} v_* |S|(|t|+1)^{d+1} \norm{\OO}.
\end{equation}
To show this, we follow the proof in \cite{our_preth} of its Corollary 4, starting from Eq.(S69):

Define superoperators \begin{equation}
    \LL=\ii[H,\cdot], \quad \tilde{\LL}=\ii[U^\dagger HU,\cdot],\quad \tilde{\LL}'=\ii[H_*,\cdot], \quad \mathcal{U} = U^\dagger \cdot U.
\end{equation}
The left hand side of \eqref{eq:coper_star} is then \begin{equation}\label{eq:L-LUO}
    \bra{\widetilde{\psi}_0}\lr{\ee^{t\LL}-\mathcal{U}^\dagger \ee^{t\tilde{\LL}'}\mathcal{U}}\OO \ket{\widetilde{\psi}_0} = \bra{\psi_0}\mathcal{U}^\dagger \lr{\ee^{t\tilde{\LL}}-\ee^{t\tilde{\LL}'}}\mathcal{U}\OO \kpsi = \bra{\psi_0}\lr{\ee^{t\tilde{\LL}}-\ee^{t\tilde{\LL}'}}\mathcal{U}\OO \kpsi,
\end{equation}
where we have used $\ket{\widetilde{\psi}_0}=U\kpsi$. Using the Duhamel identity \begin{equation}
    \ee^{t\tilde{\LL}}-\ee^{t\tilde{\LL}'} = \int^t_0 \dd t' \ee^{(t-t') \tilde{\LL}'}(\tilde{\LL}-\tilde{\LL}') \ee^{t' \tilde{\LL}},
\end{equation}
\eqref{eq:L-LUO} is further bounded by \begin{align}\label{eq:L-LUO<}
    \abs{\bra{\psi_0}\lr{\ee^{t\tilde{\LL}}-\ee^{t\tilde{\LL}'}}\mathcal{U}\OO \kpsi} &\le \int^t_0 \dd t' \abs{\bra{\psi_0}(\tilde{\LL}-\tilde{\LL}') \ee^{t' \tilde{\LL}} \mathcal{U}\OO\kpsi } \le |t| \max_{0\le t'\le t} \abs{\bra{\psi_0}[V_*, \OO(t')]\kpsi} \nonumber\\
    &\le |t| \max_{0\le t'\le t} \sum_r \abs{\bra{\psi_0}[\sum_{S\cap S_r\neq \emptyset} V_{*,S}, \OO_r(t')]\kpsi}.
\end{align}
Here in the first line, we have used $H_*$ stabilizes $\kpsi$ and \eqref{eq:UHU=}. We then define $\OO(t'):=U^\dagger \lr{\ee^{t'\LL}\OO} U=\sum_r \OO_r(t')$ and expand it into $\OO_r(t')$ supported in a radius-$r$ ball $S_r$ around the initial operator $\OO$. The commutator in \eqref{eq:L-LUO<} does not vanish only when the support of $V_*$ overlaps with $S_r$.

In the original proof, we have bound \eqref{eq:V<exp0} on the operator norms of $V_*$, so \eqref{eq:L-LUO<} is bounded by \begin{equation}\label{eq:V_Vpsi}
    \abs{\bra{\psi_0}[\sum_{S\cap S_r\neq \emptyset} V_{*,S}, \OO_r(t')]\kpsi}\le 2\sum_{S\cap S_r\neq \emptyset} \norm{V_{*,S}} \norm{\OO_r(t')} \le 2|S_r|v_*\norm{\OO_r(t')}.
\end{equation}
Here, although we only have the weaker \eqref{eq:V<exp0_weaker}, the result of \eqref{eq:V_Vpsi} still holds by just replacing the factor of $\norm{V_{*,S}}$ by  $\norm{V_{*,S}\kpsi}$, because $V_*$ acts directly on $\kpsi$. This is essentially the only proof step that needs modification.

The rest of the proof then follows verbatim from \cite{our_preth}: \eqref{eq:coper_star} holds by combining \eqref{eq:L-LUO<}, \eqref{eq:V_Vpsi} and constraining the locality of $\OO_r(t')$. Finally, the result for local density matrix \eqref{eq:rho-rho} holds as it is equivalent to \eqref{eq:coper_star}: The dynamics dropping off $V_*$ does not evolve the initial state, so by choosing all final local operators $\OO$, the final state $\rho_S(t)$ is close to the initial one.
\end{proof}

\section{Slow false vacuum decay}\label{app:fvd}
This appendix contains our proof of (a more generalized version of) Theorem \ref{thm:Ising_t_d}.  As the discussion is rather lengthy, we begin by first summarizing a more general setting for the problem, where we will prove our formal results.

\subsection{Setup of the solvable Hamiltonian}\label{sec:H0'}

We consider a commuting projector Hamiltonian \begin{equation}\label{eq:H0'=l0}
    H_0'=\sum_{f: \diam f\le \ell_0} H_{0,f}', \where H_{0,f}'=\lambda_f P^\perp_f ,\quad P^\perp_f=I-P_f,
\end{equation} 
that has finite range $\ell_0$, and local strength \begin{equation}
    h_0'=\max_i \sum_{f\ni i}\abs{\lambda_f}.
\end{equation}
Here each local projector is labeled by a ``factor'' $f=(S,\eta)$, where $S\subset \Lambda$ is the support of $P_f$, and $\eta=1,2,\cdots$ means multiple factors can act on a given subset $S$. Nevertheless, it is useful to still view $f=(S,\eta)$ as a subset, so we define $\diam f=\diam S$, and let $i\in f$ ($f\subset S'$) denote $i\in S$ ($S\subset S'$), etc.
The terminology ``commuting projector'' means \begin{equation}
    P_f^2 = P_f,\quad [P_f, P_{\tilde{f}}]=0,\quad \forall f,\tilde{f}.
\end{equation}
We assume \begin{equation}\label{eq:H2>DeltaH}
    \lambda_f\ge \Delta',\quad \forall f,
\end{equation}
and $H_0'$ is frustration-free: a projector onto the ground state subspace is $P=\prod_f P_f\neq 0$ that minimizes all local terms in $H$ simultaneously. Note that any gapped and frustration-free $H$ made out of commuting terms can always be put in the form of (\ref{eq:H0'=l0}), as any operator can be decomposed into a sum of projectors onto its eigenstates.

We will work with such a general $H_0'$ \eqref{eq:H0'=l0}, but it is useful to keep in mind a canonical example, the Ising Hamiltonian \begin{equation}\label{eq:H0=Ising}
    H_{\rm Ising}=\frac{1}{2}\sum_{i\sim j}\lambda_{i\sim j} \lr{ 1-Z_i Z_j },\where \lambda_{i\sim j}\ge \Delta',
\end{equation}
which sums over nearest-neighboring sites $i\sim j$ on a $d$-dimensional lattice. $Z_i$ is the Pauli-$Z$ on site $i$, and the coupling coefficients can be inhomogeneous. For this example, each factor $f$ has its own subset.

Given set $F$, let
\begin{equation}\label{eq:PF=Pf}
    P_F := \prod_{f\subset F} P_f, \quad P^\perp_F := I - P_F ,
\end{equation}
be the local ground state projector in that set. In particular $P=P_\Lambda$. Because the projectors commute, we have \begin{equation}\label{eq:PHP>Delta'}
    P^\perp_F H_{0,F}' P^\perp_F\ge \Delta',\where H_{0,F}':= \sum_{f\subset F}H_{0,f}',
\end{equation}
because any state in $P^\perp_F$ must violate at least one ``check'' $2P_f-I$, i.e. the check operator acting on the state returns $-1$.
In particular, setting $F=\Lambda$ in \eqref{eq:PHP>Delta'} implies $H_0'$ has a global gap at least $\Delta'$. 

\subsection{Operator decomposition based on syndromes}
Based on the commuting ``Pauli'' checks $\{2P_f-I\}$, it is useful to arrange operators based on their commutation relations with the checks.  Since $2P_f-I$ has eigenvalues $\pm 1$, any operator can be decomposed into a sum of terms that either commute or anticommute with each ``Pauli" individually (this can be shown by a simple calculation in the eigenbasis of the checks). For an operator $\OO$, if it either commutes or anti-commutes with any check $2P_f-I$ (see the precise statement \eqref{eq:comm_or_anti}), then we say it has syndrome $\bs=\{f:\{\OO, 2P_f-I\}=0\}$. In other words, the syndrome is the set of factors that anti-commute with the operator (the factors not included in the syndrome commute with $\OO$).
We then expand any operator $\OO$ that we will encounter (except $H_0'$ itself) by \begin{equation}\label{eq:decomp}
    \OO=\sum_S\OO_S,\quad \OO_S=\sum_{\bs: f\in \bs \Rightarrow f\subset S } \OO_{S,\bs},
\end{equation}
where $S$ is connected as before, and the second equation further expands each local term to terms with different syndrome $\bs$. Here $\OO_{S,\bs}$ is supported in $S$, and has syndrome $\bs$:
\begin{equation}\label{eq:comm_or_anti}
    \OO_{S,\bs} (2P_f-I) = \left\{ \begin{array}{cc}
      - (2P_f-I) \OO_{S,\bs},   & \forall f\in \bs \\
      (2P_f-I) \OO_{S,\bs},   & \forall f\notin \bs
    \end{array}\right.
\end{equation}
To gain more intuition, note that this is equivalent to $
    \text{either}\quad P_f \OO_{S,\bs} P_f=P_f^\perp \OO_{S,\bs} P_f^\perp =0, \quad(f\in \bs) \quad \text{or}\quad  P_f^\perp \OO_{S,\bs} P_f=P_f \OO_{S,\bs} P_f^\perp =0 \quad(f\notin \bs)$.
In the summation over syndrome \eqref{eq:decomp}, we further require that
$\OO_{S,\bs}$ is ``strongly supported'' \cite{roeck2019weak_driven} in $S$, in the sense that it only anti-commutes with factors that are completely included in $S$: \begin{equation}\label{eq:HcommV}
    [P_f, \OO_{S,\bs}] = 0,\quad \forall f\not\subset S.
\end{equation}
Note that $P_f$ commutes with $\OO_{S,\bs}$ trivially for $f$ that does not intersect with $S$; \eqref{eq:HcommV} also requires $P_f$ to commute if $f$ is not fully contained in $S$. A more explicit construction of  $\OO_{S',\bs}$ is given by 
\begin{equation}\label{eq:syndrome_add}
    \OO_{S}=\sum_{\bs_1, \bs_2} \widetilde{P}_{\bs_1} \OO_{S}\widetilde{P}_{\bs_2} = \sum_{\bs} \OO_{S,\bs},\where \OO_{S,\bs}=\sum_{\bs_1, \bs_2: \bs_1+\bs_2=\bs}\widetilde{P}_{\bs_1} \OO_{S}\widetilde{P}_{\bs_2}.
\end{equation} 
Here $\bs\in\{0,1\}^{\{f:f\subset S'\}}$ is viewed as a binary string over the set of $f$s contained in $S'$, and $\bs_1+\bs_2$ returns another binary string that is the elementwise summation of the two binary strings (modulo $2$). One can verify that \eqref{eq:comm_or_anti} is then satisfied.

Later on, we will use decomposition $\OO=\sum_{S,\bs}\OO_{S,\bs}$ \eqref{eq:decomp} with the strongly-supported condition $f\in \bs \Rightarrow f\subset S$ being implicit. We will see that this condition is useful to deal with operator growth by the finite-range $H_0'$ in SWTs. Based on this decomposition, we define the local $K$-norm of $\OO$ that cares about the syndromes \begin{equation}\label{eq:Knorm}
\norm{\OO}_{K}:= \max_{i \in \Lambda} \sum_{S \ni i, \bs} \ee^{K|S|} \norm{\OO_{S,\bs}}.
\end{equation}
\eqref{eq:Knorm} implies stronger locality than the previous $\kappa$-norm \eqref{eq:kappa_norm} that is based on the support \emph{diameter}, because the operator needs to decay exponentially with the support \emph{volume} in order for $\norm{\OO}_{K}$ to be finite.
Slightly abusing notation, we differentiate these two local norms simply by the subscript of $\norm{\cdot}_{\kappa/ K}$. We will exclusively use the volume local norm in this section, while only use the diameter local norm elsewhere.. 

One can always massage an operator into decomposition of the form \eqref{eq:decomp}, with a mild increase in local norm:
\begin{prop}\label{prop:decomp}
For any operator $\OO=\sum_{S} \OO_S$ with \begin{equation}
    \max_i \sum_{S\ni i} \norm{\OO_S}\ee^{K|S|}\le 1,
\end{equation}
one can always decompose it to the form \eqref{eq:decomp}, which induces local norm $\norm{\OO}_{K'}\le c_d (2\ell_0)^d$ with $K'=c_{\ell_0}K-c_{\rm factor}$, where constant $c_{\ell_0}$ is determined by $d,\ell_0$, and constant $c_{\rm factor}=\ln(2)\times \max_i |\{f:f\ni i\}|$ (the maximal number of factors that act on a given site). 
\end{prop}

Note that $c_{\ell}$ does not depend explicitly on the local Hilbert space dimension $q$. Before the proof, we present the idea of the syndrome decomposition using the Ising example: Each local term $\OO_S$ can be expanded into Pauli strings (products of $Z_i$s and $X_j$s) contained in $S$, and each Pauli string has syndrome $\bs$ containing all edges $i\sim j$ with exactly one Pauli-$X$ on either $i$ or $j$. $\OO_{S,\bs}$ is then the sum over Pauli strings of the given syndrome $\bs$.

\begin{proof}
To satisfy \eqref{eq:HcommV}, we first fatten the support of each $\OO_{S}$ by a buffer region of finite thickness (see Fig.~\ref{fig:syndrome}(b)), viewing it as a term that is strongly supported in $S'=B(S,\ell_0)$ instead, so that $\OO=\sum_{S'}\OO_{S'}$. Since $|B(S,\ell_0)|\le |S|/c_{\ell_0}$ for some constant $c_{\ell_0}$ determined by $d,\ell_0$, we have \begin{align}\label{eq:fattend<}
    \max_i \sum_{S'\ni i} \norm{\OO_{S'}}\ee^{c_{\ell_0}K|S'|}\le \max_i\sum_{j:\mathsf{d}(i,j)\le \ell_0} \sum_{S\ni j} \norm{\OO_{S'}}\ee^{c_{\ell_0}K|S'|} \le c_d (2\ell_0)^d   \max_j \sum_{S\ni j} \norm{\OO_{S}}\ee^{K|S|} \le c_d (2\ell_0)^d,
\end{align}
Here we have used that if $i\in B(S,\ell_0)$, then there exists $j\in S$ within distance $\ell_0$ to $i$. The number of such $j$s is bounded by $c_d (2\ell_0)^d$ because they form a subset of diameter $\le 2\ell_0$.  Note that a simple bound is $c_{\ell_0} \ge 1/c_d(2\ell_0)^d$.

Then we decompose each $\OO_{S'}$ by syndromes as follows. The projectors $\{P_f: f\subset S'\}$ can be simultaneously diagonalized by a set of orthogonal subspaces $\{\widetilde{P}_{\bs}\}$ labeled by syndrome $\bs$ that span the Hilbert space of the subsystem $S'$, where $P_f\widetilde{P}_{\bs}=\widetilde{P}_{\bs}$ (otherwise $=0$) if and only if $f\in \bs$. Expand the operator $\OO_{S'}$ onto this basis \eqref{eq:syndrome_add}. 
Since \begin{equation}
\norm{\OO_{S',\bs}}=\norm{\sum_{\bs_1}\widetilde{P}_{\bs_1} \OO_{S'}\widetilde{P}_{\bs_1+\bs}}=\max_{\bs_1}\norm{\widetilde{P}_{\bs_1} \OO_{S'}\widetilde{P}_{\bs_1+\bs}} \le \norm{\OO_{S'}},
\end{equation}
we have \begin{equation}
\norm{\OO}_{K'}=\max_i \sum_{S'\ni i}\ee^{K'|S'|}\sum_{\bs} \norm{\OO_{S',\bs}}\le
    \max_i \sum_{S'\ni i} \ee^{K'|S'|} \norm{\OO_{S'}}\ee^{c_{\rm factor}|S'|}=\max_i \sum_{S'\ni i} \norm{\OO_{S'}}\ee^{c_{\ell_0}K|S'|}\le c_d (2\ell_0)^d,
\end{equation}
where we have used that the number of $\bs$ for $S'$ is $2^{\abs{\{f:f\subset S'\}}}\le 2^{|S'|\times \max_i \abs{\{f:f\ni i\}} } =\ee^{c_{\rm factor}|S'|}$, and invoked \eqref{eq:fattend<} in the end.
\end{proof}

Let $\kpsi$ be one ground state of $H_0'$ (i.e. $P\kpsi=\kpsi$).  We assume that this stat has the following property.
\begin{defn}\label{def:nondegen}
We say the ground state $\kpsi$ of $H_0'$ is $R'$-locally nondegenerate, if the following holds for any operator $\OO_{S,\bs}$ supported in $S\subset \Lambda$ with $|S|\le c_d (R')^d$ that has syndrome $\bs$: Suppose all $f\in \bs$ overlaps (has a common site) with a smaller subset $F\subset S$ (including the empty set meaning $\OO_{S,\bs}$ commutes with all projectors) such that all connected components of subset $S\setminus F$ are connected with $S^{\rm c}$. 
Then there exists an operator ${\widetilde{\OO}}_{F,\bs}$ supported in the smaller region $F$ with the same syndrome $\bs$, such that $\norm{{\widetilde{\OO}}_{F,\bs}}\le \norm{\OO_{S,\bs}}$ and 
    \begin{equation}
        \OO_{S,\bs} \kpsi = {\widetilde{\OO}}_{F,\bs} \kpsi.
    \end{equation}
\end{defn}

See Fig.~\ref{fig:syndrome}(a) for illustration.
In particular, let $\OO_{S,\bs}$ be a local operator with $|S|\le c_d (R')^d$ that commutes with all $P_f$ so that $\bs=\emptyset$, the above condition implies that \begin{equation}\label{eq:empty_s_stabilize}
    \OO_{S,\emptyset}\kpsi\propto \kpsi.
\end{equation}
This criterion forbids trivial ground state degeneracies (which explains the name ``locally nondegenerate'' of this condition) like a dangling qubit that no Hamiltonian term acts on, where the two local states can be flipped to each other easily by perturbation. The parameter $R'$ will end up rather unimportant, as it will generally be much larger than the metastability length scale $R$. Indeed, the Ising Hamiltonian \eqref{eq:H0=Ising} satisfies Definition \ref{def:nondegen} with $R'$ scaling with the system size: $\OO_{S,\bs}$ acts on $\kzero$ as if it is all-identity in region $S\setminus F$, because the qubits in $S\setminus F$ are forced to be all-zero in $\OO_{S,\bs} \kzero$. To see this, consider a qubit $i\in\partial S$, which is connected by a $ZZ$ check to an outside $\ket{0}$ state that is untouched by the operator. Since this check is not violated by $\OO_{S,\bs}$, $i$ is still in the $\ket{0}$ state after the operator action. Then one can turn to qubit $j\in S$ connected to $i$ and repeat the same argument to prove $j$ remains to be $\ket{0}$ if $\{i,j\}\notin \bs$, so on and so forth. An alternative way is to thinking about $\OO_{S,\bs}$ as a superposition of Pauli strings, where all $X$s are contained in $F$
. Outside $F$, there can be $Z$s but they act trivially on $\kzero$.

Definition \ref{def:nondegen} can be viewed as a formal condition for (spontaneous) discrete symmetry breaking in the ground state sector.  It is interesting to compare this with the topological order condition in \cite{topo_Hastings}, which enables a stability proof of commuting topological order. Indeed, the topological order condition also implies the nontrivial ground state degeneracy \eqref{eq:empty_s_stabilize} with $R'$ of the order of the system size, and our later proof shares the similar SWT techniques with their stability proof.

\begin{figure}
    \centering
    \includegraphics[width=1.\textwidth]{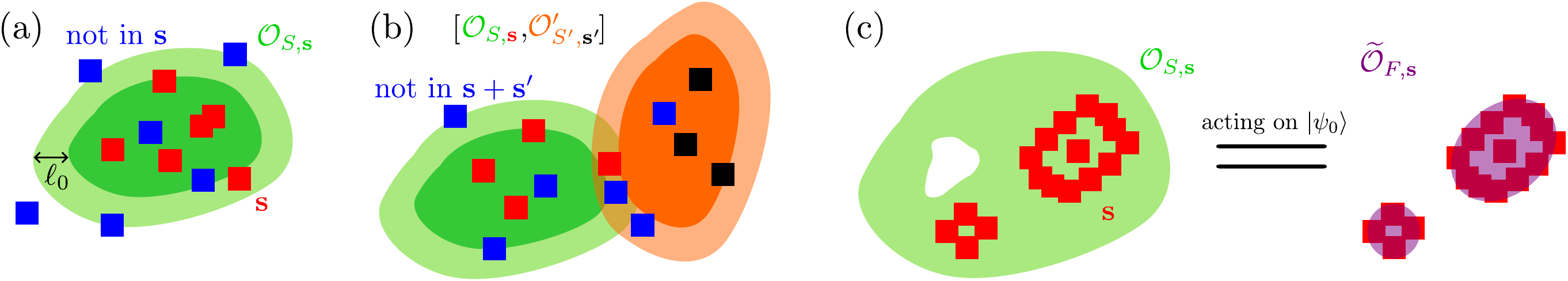}
    \caption{(a) Illustration of a local term $\OO_{S,\bs}$ in decomposition \eqref{eq:decomp}. Blue (red) boxes show examples of checks $2P_f-I$ that commute (anti-commute) with $\OO_{S,\bs}$. $\OO_{S,\bs}$ is strongly supported in the larger (lighter) green region $S$, so that the red syndrome can only be inside $S$, i.e. they must intersect with the smaller darker green region. Any local operator can be ``fattened'' in its support to satisfy this condition. (b) Illustration that any commutator $[\OO_{S,\bs}, \OO'_{S',\bs'}]$ is strongly supported in $S\cup S'$. This commutator also has syndrome $\bs+\bs'$. For example, the blue check overlapping with the two darker regions can be in both syndromes $\bs$ and $\bs'$ so that it is not included in $\bs+\bs'$. (c) Illustration of the locally nondegenerate condition Definition \ref{def:nondegen}. Here $S$ is always connected in decomposition \eqref{eq:decomp} (which can have ``holes'' as shown by the white region), while $F$ can be disconnected. The condition guarantees that when acting on $\kpsi$, the local operator can shrink its support to $F$ (purple regions) such that its boundary $\partial F$ are covered by syndrome. }
    \label{fig:syndrome}
\end{figure}

\subsection{Preprocessing with SWTs}

Consider perturbing $H_0'$ by
\begin{equation}\label{eq:V1'=}
    V'_1=\sum_{S:\diam(S)\le \ell_V}\sum_{\bs} V'_{1,S,\bs},
\end{equation}
with finite range $\ell_V$. Here we have expanded $V'_1$ into syndromes without loss of generality, as explained above. In order to prove a long-time stability for this system, it turns out useful to first preprocess the problem by a \emph{finite} round of SWTs, which is summarized in the following.

\begin{prop}\label{prop:finiteSWT}
Let $\ell_V$ be a finite number. There exist constants $c_V,c_A',c_D',c_*',\ell_D'$ determined by $d,\ell_V$ such that the following holds:

For any perturbation \eqref{eq:V1'=}, if \begin{equation}\label{eq:V1'<Delta'}
    \norm{V'_1}_0 =:\epsilon'\le c_V \Delta',
\end{equation}
($\norm{\cdot}_0$ means $\norm{\cdot}_K$ with $K=0$; similar below)
then there exists a quasi-local unitary $U_{\rm finite}=\ee^{\int^1_0\dd s A_{\rm finite}(s)}$ that makes \begin{equation}
    U_{\rm finite}^\dagger (H_0'+V_1')U_{\rm finite} = H_0'+D'+V_1,
\end{equation}
where \begin{equation}\label{eq:D'=D'Ss}
    D'=\sum_{S:\diam S\le \ell_D'} \sum_{\bs} D'_{S,\bs}, 
\end{equation}
has finite range $\ell'_D$, and \begin{equation}\label{eq:finiteSWT_bound}
    \norm{D'}_0\le c_D' \epsilon',\quad \norm{A_{\rm finite}(s)}_1\le \frac{c_A' \epsilon'}{\Delta'} ,\quad \norm{V_1}_{1}\le c_*'\epsilon' \left(\frac{\epsilon'}{\Delta'}\right)^{d-2/3}.
\end{equation}
Furthermore, $D'$ is locally-block-diagonal with respect to $H_0'$.
\end{prop}

Here we have defined the following:
\begin{defn}\label{def:D'}
We say $D'$ is locally-block-diagonal with respect to $H_0'$, if its decomposition \eqref{eq:D'=D'Ss} satisfies \begin{equation}\label{eq:DcommutP}
    [D'_{S,\bs}, P_S ]=0,\quad \forall S,\bs.
\end{equation}
\end{defn}

Here $P_S$ is defined in \eqref{eq:PF=Pf} and includes all $f\subset S$.
Note that $D'_{S,\bs}$ also commutes with $P_f$ for any $f\not\subset S$ because it is strongly-supported in $S$.
As an example, for 1d Ising $H_0'$, $D'_{S,\bs}=X_2(1-Z_1Z_3)$ with $\bs=\{\alr{1,2},\alr{2,3}\}$ satisfies \eqref{eq:DcommutP} because it annihilates local patterns $\ket{000}_{123},\ket{111}_{123}$ with no domain walls (DWs).
In words, $D'$ does not create local excitations out of the ``vacuum" (however, the action on excited states which have $\ge 1$ syndrome flip is unconstrained). As a result, Definition \ref{def:D'} implies $D'$ stabilizes $\kpsi$ locally: \begin{equation}\label{eq:D'kpsi=kpsi}
    D'_{S,\bs} \kpsi = D'_{S,\bs}P_S \kpsi = P_S D'_{S,\bs} \kpsi = P_S D'_{S,\bs}P_S \kpsi \propto \kpsi,
\end{equation}
if $\kpsi$ is a $R'$-locally nondegenerate ground state of $H_0'$ with $R'\ge \ell'_D$. The last step of \eqref{eq:D'kpsi=kpsi} utilizes $[P_S D'_{S,\bs}P_S, P_f]=0,\forall f$ and \eqref{eq:empty_s_stabilize}.
In the proof of Proposition \ref{prop:finiteSWT}, we will see that this locally-block-diagonal property holds because $D'$ is the effective Hamiltonian (up to some SWT order) generated by the perturbation $V_1'$. 

Proposition \ref{prop:finiteSWT} suppresses the local strength of the perturbation to high order, from the original $V_1'\sim \epsilon'$ to $V_1\ll (\epsilon')^d$. Later, we will treat the new perturbation problem $H_0+V_1$ where \begin{equation}\label{eq:H0=H0'+D'}
    H_0 = H_0'+D',
\end{equation}
using the idea\footnote{We will show $H_0$ satisfies the metastability condition for completeness in Appendix \ref{app:relabound}; however, this is not necessary for the main result Theorem \ref{thm:solvable} of bounding the life-time under $H_0+V_1$.} that $H_0$ is metastable up to operator \emph{volume} $R^d$ with $R\sim \Delta'/\epsilon'$. Here, we first show how to obtain $D',V_1$:

\begin{proof}[Proof of Proposition \ref{prop:finiteSWT}]
We perform a finite $k_d$ (determined later in \eqref{eq:kd=3d2}) number of SWT steps, which are formally the same as \eqref{eq:UkHUk=}, \eqref{eq:Hk+V=Hk+1}, \eqref{eq:H0A+V=PV} and \eqref{eq:Vk+1}, except that we add primes over $H_0,V_k$ and $\PP$. For example, \eqref{eq:H0A+V=PV} now becomes
\begin{equation}\label{eq:H01A+V'=PV}
    [H_0', A_k] + V_k' = \PP' V_k'=: D_{k+1}-D_k.
\end{equation}

Instead of using the filter function, here $A_k, \PP'V_k'$ (with their syndrome decompositions \eqref{eq:decomp}) are constructed in a way that explicitly exploits the solvable structure of $H_0'$, such that $\PP' V_k'$ is locally-block-diagonal, and the resulting $V_{k+1}'$ is at higher order of $\epsilon'$ comparing to $V_k'$. More precisely, we prove in Appendix \ref{app:DA<} that
\begin{lem}\label{lem:DA<}
For any $V$, there exist $A, \PP' V$ such that \begin{equation}
    [H_0', A] + V = \PP' V,
\end{equation}
where $\PP'V$ is locally-block-diagonal with respect to $H_0'$, and \begin{equation}\label{eq:AK<VK}
    \max\lr{\Delta'\norm{A}_K, \frac{1}{2} \norm{\PP' V}_K} \le \norm{V}_K,
\end{equation}
for any $K$. In particular, if $V$ has finite range $\ell_V$, then $A,\PP'V$ has the same finite range $\ell_V$.

Furthermore, suppose $D'$ is of finite range $\ell_D'$, locally-block-diagonal, and \begin{equation}\label{eq:bbmd'=}
    \mathbbm{d}'=\norm{D'}_0,
\end{equation}
satisfies \eqref{eq:d'e<Delta} and determines $R$ by \eqref{eq:R=1d} with sufficiently small constants $c_V,c_{\rm R}$ determined by $d,\ell_0,\ell_D'$. Then if $V$ only contains $V_S$ with $|S|\le c_d R^d$, \begin{equation}
    \label{eq:DA<}
    \norm{[D',A]}_{K} \le \frac{1}{3}\norm{V}_K.
\end{equation}
\end{lem}
We describe how to find the local terms $A_{S,\bs}, (\PP' V)_{S,\bs}, ([D',A])_{S,\bs}$ in the proof.
\eqref{eq:DA<} is not needed for the current proof, but will play an important role later.

For the first SWT step, the above Lemma implies $D_1=\PP' V'_1,A_1$ has finite range $\ell_V$ and local norm \begin{equation}
    \norm{D_1}_0\le 2\epsilon',\quad \norm{A_1}_0\le \frac{\epsilon'}{\Delta'}.
\end{equation}
Technicalities arise for $V_2'$: it is not finite-range anymore. To tackle this, we define \begin{equation}
    K_1':= \frac{1}{3c_d \ell_V^d}\log\frac{\Delta'}{\epsilon'},\quad K_k':= \frac{K_1'}{1+\ln k},\quad \delta K_{k+1}'=K_k'-K_{k+1}',
\end{equation}
so that \begin{equation}
    \norm{V'_1}_{K_1'}=\max_i \sum_{S\ni i,\bs} \norm{V'_{1,S,\bs}} \ee^{K_1'|S|} \le \max_i \sum_{S\ni i,\bs} \norm{V'_{1,S,\bs}} \ee^{K_1'c_d \ell_V^2 } \le \epsilon' \ee^{K_1'c_d \ell_V^d } = \epsilon'\left(\frac{\epsilon'}{\Delta'}\right)^{-1/3}.
\end{equation}
$D_1/2, \Delta' A_1$ have exactly the same bound.

We then invoke Lemma 4.1 in \cite{abanin2017rigorous}: 
\begin{lem}[adapted from \cite{abanin2017rigorous}]\label{lem:clustexp}
Consider two operators $A,\OO$ with local decompositions \eqref{eq:decomp}. If \begin{equation}\label{eq:AK<dK}
    \norm{A}_K \le \frac{\delta K}{3}:= \frac{K-K'}{3},
\end{equation}
then \begin{subequations}\label{eq:clustexp}
    \begin{align}\label{eq:clustexp1}
        \norm{\lr{\ee^{-\cA}-1}\OO}_{K'}, 2\norm{\int_0^1 \dd s \lr{ \ee^{-s\cA_k}-1}\OO}_{K'} &\le \frac{18}{K' \delta K} \norm{A}_K \norm{\OO}_K, \\
        \norm{\ee^{-\cA}\OO}_{K'} &\le \lr{1+\frac{18}{K' \delta K} \norm{A}_K} \norm{\OO}_K. \label{eq:clustexp2}
    \end{align}
\end{subequations} 
\end{lem}

\begin{proof}

We define the local term of e.g. \eqref{eq:clustexp2} by
\begin{equation}\label{eq:AO=AAOAA}
    \lr{\ee^{-\cA}\OO}_{S,\bs} = \sum_{m=0}^\infty \sum_{\substack{S_0,S_1,\cdots,S_m: \\ S_0\cup S_1\cup \cdots\cup S_m=S }} \sum_{\substack{\bs_0,\bs_1,\cdots,\bs_m: \\ \bs_0+ \bs_1+ \cdots+ \bs_m=\bs }} a_{S_0,S_1,\cdots,S_m}[A_{S_m,\bs_m}, \cdots [A_{S_2,\bs_2},[A_{S_1,\bs_1}, \OO_{S_0,\bs_0}]] \cdots] ,
\end{equation}
where the coefficient $a_{S_0,S_1,\cdots,S_m}$ is explicitly zero if e.g. $S_1\cap S_0=\emptyset$ such that the multicommutator vanishes. The syndromes are added as in \eqref{eq:syndrome_add}. One can verify that (\emph{1}) Summing \eqref{eq:AO=AAOAA} over $S,\bs$ returns the whole $\ee^{-\cA}\OO$. (\emph{2}) \eqref{eq:AO=AAOAA} is strongly supported in $S$ as inherited from $A,\OO$ (see Fig.~\ref{fig:syndrome}(c) for illustration). (\emph{3}) It also has syndrome $\bs$, which is the elementwise summation of the binary strings $\bs_1,\cdots,\bs_m\in\{0,1\}^{\{f\}}$ (modulo $2$). 
This is because one can commute $2P_f-I$ through $\OO_{S_0,\bs_0}, A_{S_1,\bs_1}$ etc one-by-one by adding $\pm$ phases; the final phase is $-$ if and only if $f\in \bs$.

Taking the operator norm of \eqref{eq:AO=AAOAA} and summing over $\bs$, we have
\begin{align}
    \sum_\bs \norm{\lr{(\ee^{-\cA}-1)\OO}_{S,\bs}} &\le \sum_{m=1}^\infty \sum_{\substack{S_0,S_1,\cdots,S_m: \\ S_0\cup S_1\cup \cdots\cup S_m=S }} \sum_{\bs_0,\bs_1,\cdots,\bs_m} \abs{a_{S_0,S_1,\cdots,S_m}}\lr{2\norm{A_{S_m,\bs_m}}} \cdots \lr{2\norm{A_{S_1,\bs_1}}} \norm{\OO_{S_0,\bs_0}} \nonumber\\
    &= \sum_{m=1}^\infty \sum_{\substack{S_0,S_1,\cdots,S_m: \\ S_0\cup S_1\cup \cdots\cup S_m=S }} \abs{a_{S_0,S_1,\cdots,S_m}}\lr{2\sum_{\bs_m}\norm{A_{S_m,\bs_m}}} \cdots \lr{2\sum_{\bs_1}\norm{A_{S_1,\bs_1}}} \lr{\sum_{\bs_0}\norm{\OO_{S_0,\bs_0}}}.
\end{align}
Viewing e.g. $\lr{\sum_{\bs_0}\norm{\OO_{S_0,\bs_0}}}$ as $\norm{Z_{S_0}}$ in Eq.~(5.1) in \cite{abanin2017rigorous},
the proof there then applies straightforwardly to yield \eqref{eq:clustexp2} and \eqref{eq:clustexp1} for the first argument. The key point here is that our extra summation over syndromes can be absorbed into our syndrome-dependent local norms $\norm{A}_K, \norm{\OO}_K$. 

The second argument (the integral) in \eqref{eq:clustexp1} also holds by the following trick: \begin{equation}\label{eq:ints=}
    2\int_0^1 \dd s \lr{ \ee^{-s\cA_k}-1}\OO = \sum_{m=1}^\infty \lr{2\int^1_0 (-s)^m \dd s} \sum_{S_0,S_1,S_2,\cdots,S_m} \sum_{\bs_0,\cdots,\bs_m} a_{S,S_1,\cdots,S_m}[A_{S_m,\bs_m}, \cdots [A_{S_2,\bs_2},[A_{S_1,\bs_1}, \OO_{S_0,\bs_0}]] \cdots],
\end{equation}
which is again decomposed to local terms labeled by $S,\bs$ using the rules in \eqref{eq:AO=AAOAA}. Since each local term in \eqref{eq:ints=} is proportional to the corresponding term in $(\ee^{-\cA}-1)\OO$ with coefficient $\abs{\lr{2\int^1_0 (-s)^m \dd s}}\le 1$, they satisfy the same bound \eqref{eq:clustexp1}.
\end{proof}

Using this Lemma, $V_2'$ (expressed in \eqref{eq:Vk+1} with primes added) is then bounded by \begin{equation}
    \norm{V'_2}_{K_2'} \le \frac{18}{K_2'\delta K_2'} \lr{\frac{1}{2}\times 3 + 1}\frac{\norm{V'_1}_{K_1'}^2}{\Delta'}\le c_1'\epsilon'\left(\frac{\epsilon'}{\Delta'}\right)^{1/3},
\end{equation}
for some constant $c_1'$ determined by $d,\ell_V$ and $c_V$ in \eqref{eq:V1'<Delta'}. Here $V'_2$ does not scale as $(\epsilon')^2$ due to the way we invoke the $K'$ norm. Nevertheless, $V'_2$ is still suppressed by some power of $\epsilon'$ comparing to $V'_1$, and this is sufficient. Note that the condition \eqref{eq:AK<dK} for Lemma \ref{lem:clustexp} is always satisfied in this proof, because $K'_k$s are large $\sim \log(\Delta'/\epsilon')$.

At higher but finite order $2\le k\le k_d$, the $D_k$ term in \eqref{eq:Vk+1} dominates: Assuming \begin{equation}\label{eq:DK'<}
    \norm{D_k}_{K_{k}'} \le 2 \norm{D_1}_{K_1'} \le 4 \epsilon'\left(\frac{\epsilon'}{\Delta'}\right)^{-1/3},
\end{equation}
we have
\begin{equation}\label{eq:V'k+1<V'k}
    \norm{V'_{k+1}}_{K_{k+1}'} \le \frac{18}{K_{k+1}'\delta K_{k+1}'\Delta'} 4 \epsilon'\left(\frac{\epsilon'}{\Delta'}\right)^{-1/3} \norm{V'_{k}}_{K_{k}'} \le c'_k \left(\frac{\epsilon'}{\Delta'}\right)^{2/3} \norm{V'_{k}}_{K_{k}'},  
\end{equation}
which implies that \begin{equation}
    \norm{V'_{k+1}}_{K_{k+1}'}\le \tilde{c}'_k \epsilon' \left(\frac{\epsilon'}{\Delta'}\right)^{\frac{2k-1}{3}},
\end{equation}
for constants $c'_k, \tilde{c}'_k$ determined by $d,\ell_V$. Consistently, \eqref{eq:V'k+1<V'k} guarantees \eqref{eq:DK'<} because each $D_{k+1}$ is suppressed by a factor $(\epsilon'/\Delta')^{2/3}$ comparing to $D_k$, e.g. \begin{equation}
    \norm{D_2}_{K'_2} \le 2 \norm{V'_2}_{K_2'} \le 2c_1'\epsilon'\left(\frac{\epsilon'}{\Delta'}\right)^{1/3},
\end{equation} 
so that $D_1$ dominates $D_k$ for sufficiently small $c_V$.

According to \eqref{eq:V'k+1<V'k}, we define \begin{equation}\label{eq:kd=3d2}
    k_d := \left\lfloor \frac{3d}{2}\right\rfloor,  
\end{equation}
which implies $2k_d \ge 3d-1$, so that after order $k_d$, \begin{equation}\label{eq:V'K'<}
    \norm{V'_{k_d+1}}_{K_{k_d+1}'}\le \tilde{c}'_{k_d} \epsilon' \left(\frac{\epsilon'}{\Delta'}\right)^{d-2/3}.
\end{equation}
At this order, the Hamiltonian is $H_0'+D_{k_d+1}+V'_{k_d+1}$, where $D_{k_d+1}$ is not quite the desired $D'$ because it is not finite-range. Therefore, we truncate the operator support to define \begin{equation}
    D':=\sum_{S:\diam S\le \ell_D'} \sum_{\bs}D_{k_d+1, S,\bs}, 
\end{equation}
which is locally-block-diagonal because the orginal $\PP'V_k'$s are, 
where $\ell_D'$ is chosen such that \begin{equation}
    V_1:=V'_{k_d+1} + D_{k_d+1}-D',
\end{equation}
is sufficiently small. More precisely, let $K_1=1$, and \begin{equation}
    \ell_D':= \ell_V \mlr{(6d-2)(1+\ln(k_d+1))}^{1/d},
\end{equation} we have \begin{align}
    \norm{D_{k_d+1}-D'}_{K_1} &\le \norm{D_{k_d+1}}_{K_{k_d+1}'} \ee^{-(K_{k_d+1}'-1) c_d (\ell_D')^d } \le 4 \epsilon'\left(\frac{\epsilon'}{\Delta'}\right)^{-1/3} \ee^{-K_{k_d+1}' c_d (\ell_D')^d/2 } \nonumber\\
    &= 4\epsilon'\left(\frac{\epsilon'}{\Delta'}\right)^{-1/3+\frac{(\ell'_D/\ell_V)^d}{6(1+\ln (k_d+1))}} \le 4\epsilon'\left(\frac{\epsilon'}{\Delta'}\right)^{d-2/3}.
\end{align}
Here the first inequality comes from reducing the decay exponent from $K_{k_d+1}'$ to $K_1$ so that the contribution to the local norm from each local term in $D_{k_d+1}-D'$ is suppressed exponentially, because they all have large support (see the similar later \eqref{eq:V>RK<}).  We then used \eqref{eq:DK'<} and $K_{k_d+1}'\ge 2$ that holds for sufficiently small $c_V\ge \epsilon'/\Delta'$. 
Combining with \eqref{eq:V'K'<}, we have \begin{equation}
    \epsilon=\norm{V_1}_{K_1}\le c_*'\epsilon'\left(\frac{\epsilon'}{\Delta'}\right)^{d-2/3},
\end{equation}
with $c_*'$ determined by $d,\ell_V$.

On the other hand, \begin{equation}\label{eq:d'<e'}
   \norm{D'}_0\le \norm{D_1}_0 + \norm{D_{k_d+1}-D_1}_{K_{k_d+1}'} \le 2\epsilon' + 2\norm{D_2}_{K'_2} \le 2\epsilon'+ 2c_1'\epsilon'\left(\frac{\epsilon'}{\Delta'}\right)^{1/3} \le c_D' \epsilon',
\end{equation}
for some constant $c_D'$ determined by $c_1',c_V$. Here we have separated $D_1$ out of $D_{k_d+1}$ to have a stronger bound than \eqref{eq:DK'<}. Similarly, \begin{equation}
    \norm{A_{\rm finite}(s)}_{K_1} \le \norm{A_1}_{K_1}+ 2\norm{A_2}_{K_1} \le \norm{A_1}_0 \ee^{c_d \ell_V^d} + \frac{2\norm{V_2'}_{K_1}}{\Delta'}\le \frac{c_A' \epsilon'}{\Delta'},
\end{equation}
with constant $c_A'$ determined by $d,\ell_V$, and we have again separated $A_{\rm finite}$ to its $A_1$ contribution and the rest that is subdominant $\le 2\norm{A_2}_{K_1}$ like \eqref{eq:d'<e'}. 
\end{proof}

\subsection{Proof of Theorem \ref{thm:Ising_t_d}}
We are now ready to prove a generalized version of Theorem \ref{thm:Ising_t_d}.  First, we state the following important result:

\begin{thm}\label{thm:solvable}
Let $\kpsi$ be a ground state of $H_0'$ that is $R'$-locally nondegenerate. Suppose $D'$ is locally-block-diagonal with respect to $H_0'$ (and thus stabilizes $\kpsi$ due to \eqref{eq:D'kpsi=kpsi}), has finite range $\ell'_D$ and local norm \eqref{eq:bbmd'=}.  Define $H_0$ by \eqref{eq:H0=H0'+D'} and $H=H_0+V_1$, where $V_1$ has local norm $\epsilon:=\norm{V_{1}}_{K_1}$ ($K_1\le 1$). There exist constants $c_{\rm R},c_V,c_*$ determined by $d,\ell_0,\ell'_D,K_1$, such that the following holds:
Suppose 
\begin{equation}\label{eq:d'e<Delta}
     \max(\mathbbm{d}',\epsilon) \le c_V \Delta',
\end{equation}
and define $R$ by \begin{equation}\label{eq:R=1d}
    R= c_{\rm R}\Delta'/\mathbbm{d}'.
\end{equation}
If $R<R^\prime$, then $(H_0,V_1,\kpsi)$ is $(v_*, \mathbbm{d}_*,a_*,\revise{K_1/6})$-locally-diagonalizable (see Definition \ref{def:localdiag}), where \begin{equation}\label{eq:V<exp2}
    v_*=\epsilon\, 2^{\revise{2}-k_*}, \quad
    \mathbbm{d}_* = \revise{4}\epsilon, \quad
    a_* =2\epsilon/\Delta',
\end{equation}
and \revise{\begin{equation}\label{eq:k*=_}
    k_*:=\left\lfloor c_* \min\lr{\frac{\Delta'/\epsilon}{\ln^2( \Delta'/\epsilon)} , R^d} \right\rfloor.
\end{equation}
}
\end{thm}

Here we assume $K_1\le 1$ (same for any $K$ below) in the same spirit of \eqref{eq:kappa=mu}, as one can always decrease $K_1$ for free. 
$D',V_1$ in the above Theorem \ref{thm:solvable} does not need to come from Proposition \ref{prop:finiteSWT}, so we treat the parameters $\mathbbm{d}',\epsilon$ as independent (instead of $\epsilon\sim (\mathbbm{d}')^{d+1/3}$ from \eqref{eq:finiteSWT_bound}). Nevertheless, combining Theorem \ref{thm:solvable} with the preprocessing Proposition \ref{prop:finiteSWT} proves the parametrically tighter bound \eqref{eq:tFVD>d} for false vacuum decay near the solvable point: 

\begin{cor}[general version of Theorem \ref{thm:Ising_t_d}]\label{cor:FVDRd}
Let $\kpsi$ be a ground state of $H_0'$ that is $R'$-locally nondegenerate. Let $\ell_V$ be a finite number. There exist constants $c_V,c_A,c_D,c_*,\tilde{c}_*$ determined by $d,\ell_0,\ell_V$ such that the following holds: For any perturbation $V_1'$ \eqref{eq:V1'=}
with finite range $\ell_V$ and sufficiently small local norm \eqref{eq:V1'<Delta'},
$(H_0',V'_1, \kpsi)$ is $(\tilde{c}_* \epsilon' 2^{-k_*}, c_D\epsilon',c_A\epsilon',\revise{1/6})$-locally-diagonalizable (see Definition \ref{def:localdiag}), where \revise{ \begin{equation}\label{eq:k*cor}
    k_* =\left\lfloor c_* (\Delta'/\epsilon')^d \right\rfloor.
\end{equation}
}
\end{cor}

\begin{proof}

For sufficiently small $c_V$, Proposition \ref{prop:finiteSWT} rotates the Hamiltonian to $H_0'+D'+V_1$ that satisfies all conditions of Theorem \ref{thm:solvable} with $K_1=1$. We can then invoke Theorem \ref{thm:solvable} and get \begin{equation}\label{eq:UUfinite}
    U^\dagger U_{\rm finite}^\dagger (H_0'+V'_1) U_{\rm finite} U = U^\dagger(H_0'+D'+V_1) U = H_*+V_*,
\end{equation}
where $U$ is the SWT unitary when invoking Theorem \ref{thm:solvable}, and $(H_0'+D', V_1,\kzero)$ is $(\epsilon 2^{1-k_*},5\epsilon,2\epsilon/\Delta',K_*)$-locally-diagonalizable, where \revise{$K_*=1/6$} and $k_*$ is given by \eqref{eq:k*cor}. Here we have used the first entry in \eqref{eq:k*=_} $\sim \Delta'/\epsilon\sim (\Delta'/\epsilon')^{d+1/3}$ is larger than the second one $\sim R^d\sim (\Delta'/\epsilon')^{d}$. Finally, we combine the two unitaries in \eqref{eq:UUfinite}) to get that the original problem $(H_0', V_1',\kzero)$ is $(v_*,\mathbbm{d}_*,a_*,K_*)$-locally-diagonalizable. Here $v_*=\epsilon 2^{\revise{2}-k_*}\le \tilde{c}_*\epsilon' 2^{-k_*}$ is the same one in \eqref{eq:V<exp2} and further bounded by replacing $\epsilon\rightarrow \epsilon'$ with an constant $\tilde{c}_*$. Like \eqref{eq:As=Ak}, the total generator $A_{\rm total}(s)$ for $U_{\rm finite} U$ can be chosen as $2A(2s)$ (with $A(s)$ the generator for $U$) for $s\le 1/2$ and $2 A_{\rm finite}(2s-1)$ for $s>1/2$, so that \begin{equation}
    a_*=\max_s \norm{A_{\rm total}(s)}_{K_*} \le 2 \max\lr{\norm{A_{\rm finite}(s)}_{K_*} ,2\epsilon/\Delta'} \le c_A \epsilon'/\Delta',
\end{equation}
for some constant $c_A$,
where we have used $K_*\le 1$, \eqref{eq:finiteSWT_bound} and \eqref{eq:V<exp2}.
Similarly, 
\begin{equation}
    \mathbbm{d}_*=\norm{H_*+V_*-H_0'}_{K_*}\le \norm{D'}_{K_*}+ \norm{H_*+V_*-H_0}_{K_*}\le \norm{D'}_{K_*} + \revise{4}\epsilon \le c_D \epsilon'.
\end{equation}
This finishes the proof.
\end{proof}

\begin{proof}[Proof of Theorem \ref{thm:solvable}]
We still use the SWT framework, where at the $k$-th order, suppose we have rotated the Hamiltonian to (note the difference between $D'$ and $D$) \revise{ \begin{equation}\label{eq:UkHUk=Ek}
    U^\dagger_{k-1} H U_{k-1}=H_0+D_k+V_k+E_k,
\end{equation}
and we want to rotate one step further by \begin{equation}\label{eq:Hk+V=Hk+1_Ek}
    \ee^{-\mathcal{A}_k}(H_0 + D_k+ V_k+E_k) = H_0 + D_{k+1}+ V_{k+1}+E_{k+1},
\end{equation}
aiming at suppressing the local norm of $V_k$. Note that we have added a ``garbage'' term $E_k$ comparing to \eqref{eq:UkHUk=} in previous proofs, whose meaning will become clear shortly.} As in the finite SWT rounds, here we keep track of the syndrome decompositions \eqref{eq:decomp} of operators.

At each step, we use Lemma \ref{lem:DA<} to find $A_k, \PP' V_{k,\le}$ such that \begin{equation}\label{eq:H01A+V=PV}
    [H_0', A_k] + V_{k,\le} = \PP' V_{k,\le},
\end{equation}
where $V_{k,\le}=\sum_{S:|S|\le c_dR^d} V_{k,S}$ neglects very non-local terms in $V_k$, and \begin{equation}\label{eq:P'V=D-D}
    \PP' V_{k,\le}=: D_{k+1}-D_k,
\end{equation} 
stabilizes $\kpsi$ so that $D_{k+1}$ is also a stabilizer ($D_k$ already stabilizes $\kpsi$ from previous iterations). The only difference to the finite SWT steps \eqref{eq:H01A+V'=PV} is that we do not rotate very non-local terms in $V_k$ because their support exceeds the maximal size $c_d R^d$ for metastability; see the similar size cutoff \eqref{eq:diam<R} in the previous approach. 

Since $H_0'$ is only part of $H_0$, now we have extra contributions comparing to \eqref{eq:Vk+1} to
the perturbation at next order. Writing $V_{k,>}=V_{k}-V_{k,\le}$, we have \revise{
\begin{equation}\label{eq:Ek+1=}
    E_{k+1} = \ee^{-\cA_k}E_k + \ee^{-\cA_k}V_{k,>},
\end{equation}
} and
\begin{align}\label{eq:Vk+1_}
    V_{k+1} &= \ee^{-\cA_k}(H_0' + D_k + V_{k,\le}) - H_0' - D_{k+1}+ \lr{\ee^{-\cA_k}-1}D'\nonumber\\ 
    = \int_0^1 &\dd s \lr{ \ee^{-s\cA_k}-1}(\PP'-1)V_{k,\le} +(\ee^{-\cA_k}-1)(D_k+V_{k,\le})+ \int_0^1 \dd s \lr{ \ee^{-s\cA_k}-1}[D', A_k]+ [D', A_k].
\end{align}
Here the first two terms come from \eqref{eq:Vk+1} by substituting $V_k\rightarrow V_{k,\le}$, and the last two terms are the higher-order and first-order contributions to the last term in the first line.
We want to iterate this procedure up to order $k_*\sim \min(R^d,\Delta'/\epsilon)$, so that for $k\le k_*$, the local norm of \eqref{eq:Vk+1_} is bounded by e.g. $1/2\times v_k$, where $v_k$ is the local norm of $V_k$. If so, then $v_k\sim 2^{-k}$. \revise{At the same time, we will show that the garbage term remains small $\norm{E_{k_*}}_0= \exp[-\Omega(R^d)]=\exp[-\Omega(k_*)]$ because it basically only contains operators of large support $|S|\gtrsim R^d$, and we are now counting the number of sites on which operators act, not simply the diameter of the set on which they act. As a result, we can get the desired $v_*=v_{k_*}+\norm{E_{k_*}}_0$ at an order $k_*$ limited by $k_*\lesssim R^d$. To show $v_{k+1}\le 1/2\times v_k$,} let us estimate/preview the local norms in \eqref{eq:Vk+1_}. 
Assuming $A_k\sim v_k$, the first and the \revise{third} term scale as $v_k^2\ll v_k$, so are also subdominant. The second is dominated by the $[D_k,A_k]\sim [D_1,A_k]$ term, and scales as $\epsilon v_k \times \text{(support volume of $A_k$)} \sim \epsilon v_k k$, where we have used $D_1\sim \epsilon$, and volume of $A_k$ roughly grows linearly with $k$. This second term is thus subdominant up to $k\sim 1/\epsilon$, corresponding to the first argument in \eqref{eq:k*=_}. These terms above are easy to bound by \eqref{eq:AK<VK} and Lemma \ref{lem:clustexp}. The main problem is the last term in \eqref{eq:Vk+1_}, which naively scales as \begin{equation}\label{eq:finalterm_naive}
    \mathbbm{d}' v_k\times \text{(support volume of $A_k$)}\sim \mathbbm{d}'v_k k.
\end{equation}
This would lead to a final $k_* \lesssim 1/\mathbbm{d}'\sim R\ll R^d$, which violates the second argument in \eqref{eq:k*=_}.

To tackle the last term in \eqref{eq:Vk+1_}, we use \eqref{eq:DA<} of Lemma \ref{lem:DA<}, which is stronger than the naive scaling \eqref{eq:finalterm_naive}. This is our key for promoting $R\rightarrow R^d$ dependence in $k_*$. Here $k_*\sim R^d$ is the order at which \revise{$E_k$} becomes dominant, so that we have to stop the iteration.
In a nutshell, to show \eqref{eq:DA<} we exploit the fact that large support operators in $A_k$ can be much smaller in norm than the corresponding terms of $V_k$, because $V_k$ violates $\gg 1$ ground state projectors $P_f$. Note that in order for Lemma \ref{lem:DA<} to hold, the constants $c_V,c_{\rm R}$ do not need to depend on $K_1$. We will use the same $c_{\rm R}$ independent of $K_1$ below, but will choose a $c_V$ determined by all the constants $d,\ell_0,\ell'_D,K_1$; for this new $c_V$, Lemma \ref{lem:DA<} will still hold because it is still sufficiently small.

After sketching the idea, we now start to bound the local norms. \revise{Define $v_1=\epsilon,\; \mathbbm{d}_k := \norm{D_k}_{K_k},\; v_{k}:=\norm{V_{k}}_{K_k},$ where for $k\ge 2$, \cite{yin_ldpc} 
\begin{equation}
    K_k:=\frac{K_1}{2}\lr{1+\frac{1}{1+\ln k}},\quad \delta K_{k}:=K_{k-1}-K_{k}= \frac{K_1\ln(k/(k-1)}{2(1+\ln k)[1+\ln(k-1)]}\ge \frac{K_1}{6k \ln^2 k}, 
\end{equation}
where the $6$ in the denominator was deduced by numerically plotting the function. Note that our choice of $K_k$ improves upon $K_k=K_1/(1+\ln k)$ used in the previous work \cite{abanin2017rigorous}, where our $K_k\ge K_1/2$ and does not vanish at large $k$. This ensures that the decay properties in terms of the locality of the Hamiltonian remains almost unchanged when going to high orders in perturbation theory. Furthermore, this improvement enables us to get a pure $R^d$ scaling without logarithmic factors in \eqref{eq:k*=_}. The only price to pay is making each $\delta K_k$ smaller by a constant factor of $2$, which does not change the qualitative results, as we will see.}

From the definition of the $K$-norm in \eqref{eq:Knorm}, it is manifest that $\norm{V_{k,\le}}_{K_k} \le \norm{V_k}_{K_k}=v_k$, because  the sum over support $S$ is simply truncated to contain fewer terms.  
Applying Lemma \ref{lem:DA<} to \eqref{eq:H01A+V=PV} then yields \begin{equation}
    \norm{\PP' V_{k,\le}}_K\le 2v_k,\quad \norm{A_k}_K \le v_k/\Delta',\quad \norm{[D',A_k]}_{K_k}\le v_k/3,\quad \mathbbm{d}_{k+1}\le \mathbbm{d}_k + 2v_k, \label{eq:dk+1<dk}
\end{equation}
using \eqref{eq:P'V=D-D}.
Combining this with \eqref{eq:Vk+1_} and Lemma \ref{lem:clustexp}, if at step $k$, \begin{equation}\label{eq:v_k<dK}
    v_k \le \Delta' \delta K_{k+1}/6,
\end{equation}
then
\begin{align}
    \frac{v_{k+1}}{v_k}&\le \frac{18}{\Delta' K_{k+1}\delta K_{k+1}} \mlr{\frac{1}{2}(2+1)v_k + \mathbbm{d}_k+v_k +\frac{1}{2}\cdot\frac{1}{3} v_k} + \frac{1}{3} \nonumber\\
    &\le \frac{18}{\Delta' K_{k+1}\delta K_{k+1}}(\mathbbm{d}_k + 3v_k) + \frac{1}{3}. \label{eq:vk_vk<}
\end{align}

It then remains to iterate up to some order $k_*$.
This part is similar to previous proofs. 
At small $k$, \eqref{eq:vk_vk<} is bounded by $1/2$, because 
the first term $\propto \mathbbm{d}_k+3v_k\sim v_1\ll 1$. Indeed, suppose \begin{equation}\label{eq:vk+1<vk_}
    v_{k+1}\le v_k/2,\quad k\le k_*-1,
\end{equation}
with $k_*$ to be determined shortly, we have \begin{equation}\label{eq:vk<vk-1_}
    v_k\le v_{k-1}/2\le \cdots \le v_1 2^{1-k}=2^{1-k}\epsilon, \quad \mathbbm{d}_k+3v_k \le 2(v_1+v_2+\cdots+v_{k-1}+v_k + v_k) \le 4v_1=4\epsilon,
\end{equation} 
for $k\le k_*$ from \eqref{eq:dk+1<dk}. For sufficiently small $\epsilon/\Delta'$, the condition \eqref{eq:v_k<dK} is also satisfied consistently because it holds at $k=1$, and the left hand side decays exponentially faster $\sim 2^{-k}$ than the right hand side $\sim 1/(k\ln^2 k)$.

To determine the final step $k_*$, we require \begin{equation}\label{eq:1_3<1_2}
    \frac{18}{\Delta' K_{k}\delta K_{k}}4\epsilon + \frac{1}{3} \le \frac{1}{2},\quad \forall k\le k_*.
\end{equation}
Since the left hand side increases with $k$ monotonically, it is sufficient if \eqref{eq:1_3<1_2} holds at $k=k_*$.
If we choose $k_*$ by \eqref{eq:k*=_} 
for some sufficiently small constant $c_*$ determined by $K_1, d$, \eqref{eq:1_3<1_2} is satisfied, because 
the first term \revise{\begin{equation}
    \frac{18}{\Delta' K_{k_*}\delta K_{k_*}}4\epsilon \le \frac{72\times 2\times 6\epsilon}{K_1^2\Delta'}k_*\ln^2 k_* =\mathrm{O}\lr{\frac{\ln^2 k_*}{\ln^2( \Delta'/\epsilon)}}\ll 1,
\end{equation}
in terms of scaling with $\epsilon/\Delta'\rightarrow 0$ at sufficiently small $c_V$. }

\revise{
To summarize the iteration, using the key Lemma \ref{lem:DA<} to bound the last term in \eqref{eq:Vk+1_}, for sufficiently small constant
$c_{V}$ determined by $d,\ell_0,\ell_D',K_1$, the condition \eqref{eq:v_k<dK} for Lemma \ref{lem:clustexp} and the resulting inductive relation \eqref{eq:vk+1<vk_} hold together up to $k_*$ given by \eqref{eq:R=1d}. 
}

At the last step, the Hamiltonian is rotated to $H_0+D_{k_*}+\revise{V_{*}}$, where only \revise{$V_*=V_{k_*}+E_{k_*}$} does not stabilize $\kpsi$. \revise{Let $K_*:=K_1/6< K_{k_*}$,} \eqref{eq:vk<vk-1_} leads to \begin{equation}\label{eq:astar<2}
    \norm{V_{k_*}}_{K_*}\le v_{k_*}\le \epsilon\, 2^{1-k_*}, \quad a_*\le (v_1+v_2+\cdots)/\Delta' \le 2\epsilon/\Delta',
\end{equation}
where we have used \eqref{eq:dk+1<dk} at last. 
\revise{
From iteration relation \eqref{eq:Ek+1=}, the garbage term is \begin{equation}\label{eq:E=sum}
    E_{k_*} = \sum_{k=1}^{k_*-1} \ee^{-\widetilde{\cA}_k} V_{k,>},
\end{equation}
where $\ee^{-\widetilde{\cA}_k}=\ee^{-\cA_{k_*-1}}\cdots \ee^{-\cA_{k}}$ (strictly speaking $\widetilde{\cA}_k$ is time-dependent like \eqref{eq:As=Ak}, while we ignore this for notational simplicity) satisfies \begin{equation}\label{eq:tcA<}
    \norm{\widetilde{\cA}_k}_{2K_*}\le 2\epsilon/\Delta',
\end{equation}
similarly as \eqref{eq:astar<2}. }
For truncating the operator size of $V_k$, we have
\begin{align}\label{eq:V>RK<}
    \norm{V_{k,>}}_{\revise{2K_*}} &= \max_i \sum_{S\ni i, |S|>c_d R^d,\bs} \norm{V_{k,S,\bs}} \ee^{\revise{2K_*}|S|} \notag \\ &\le \max_i \sum_{S\ni i, |S|>c_d R^d,\bs} \norm{V_{k,S,\bs}}\ee^{-\revise{K_*} c_d R^d} \ee^{K_k|S|} \le v_k \ee^{-\revise{K_*} c_d R^d}.
\end{align}
Here in the second step, we used the fact that $|S|\ge  c_dR^d$ \revise{and $K_k-2K_*>K_1/2-2K_*=K_*$}. \revise{Because \eqref{eq:AK<dK} is satisfied by \eqref{eq:tcA<} and $K=2K_*,K'=K_*$ for sufficiently small constant $c_V$, we can apply \eqref{eq:clustexp2} and \eqref{eq:V>RK<} to bound \eqref{eq:E=sum}: \begin{equation}
    \norm{E_{k_*}}_{K_*} \le \lr{1+\frac{18}{K_*^2}} \frac{2\epsilon}{\Delta'} \sum_{k=1}^{k_*-1} v_k\ee^{-K_* c_d R^d} \le \frac{4\epsilon^2}{\Delta'} \lr{1+\frac{3\times 6^3}{K_1^2}} \ee^{-K_* c_d R^d} \le \epsilon \, 2^{1-k_*},
\end{equation}
where the last inequality holds for suitable small constants $c_V$ and $c_*$ determined by $K_1$ and $d$. Combining this with \eqref{eq:astar<2}, we have \begin{equation}
    v_*=\norm{V_{k_*}+E_{k_*}}_0\le \norm{V_{k_*}}_{K_*} + \norm{E_{k_*}}_{K_*} \le \epsilon \, 2^{2-k_*},\quad \mathbbm{d}_*\le \norm{D_{k_*}}_{K_*} + \norm{V_{k_*}+E_{k_*}}_{K_*} \le \mathbbm{d}_{k_*}+2v_{k_*}\le 4\epsilon,
\end{equation}
where we have used \eqref{eq:vk<vk-1_} at last. This finishes the proof.
}
\end{proof}

\subsection{Proof of Lemma \ref{lem:DA<}}\label{app:DA<}
\begin{proof}

Since \eqref{eq:H01A+V=PV} is linear, it is sufficient to first focus on one local term $V_{S,\bs}$ 
and find the corresponding $A,\PP' V$; the whole $A$ and $\PP' V$ are then defined by just summing over $S,\bs$: e.g. $A=\sum_{S,\mathbf{s}} A_{S,\mathbf{s}}$.

As we will see, we only need to cancel the part \begin{equation}
    P_{S}^\perp V_{S,\bs} \kpsi \bra{\psi_0}+\mathrm{H.c.} = P_{S}^\perp {\widetilde{V}}_{F,\bs} \kpsi \bra{\psi_0}+\mathrm{H.c.},
\end{equation} 
of $V_{S,\bs}$, where we have replaced $V_{S,\bs}$ by ${\widetilde{V}}_{F,\bs}$ on a potentially smaller support $F\subset S$ with $\norm{{\widetilde{V}}_{F,\bs}} \le \norm{V_{S,\bs}}$ using the locally nondegenerate condition (Definition \ref{def:nondegen}), such that \begin{equation}\label{eq:PcommuteVF}
    [P_f, \widetilde{V}_{F,\bs}] = 0,\quad \forall f: f\cap F=\emptyset.
\end{equation}
Furthermore, we can assume that \begin{equation}\label{eq:dF_anticomm}
    \forall i\in \partial F, \quad \exists f\ni i \quad \text{such that}\quad \{2P_f-I, {\widetilde{V}}_{F,\bs} \}=0,
\end{equation} 
(i.e. any boundary site has a violated check for ${\widetilde{V}}_{F,\bs}$)
because otherwise $[P_f, {\widetilde{V}}_{F,\bs}]=0,\forall f\ni i$, so that we can further shrink $F$ to $F\setminus\{i\}$ as a new smaller $F$.

After shrinking, we follow the standard SWT and define an anti-Hermitian operator \begin{equation}\label{eq:A=Hinverse}
    A_{S,\bs}(\bar{F}) = P_{\bar{F}} {\widetilde{V}}_{F,\bs} P^\perp_{\bar{F}}\lr{H_{0, \bar{F}}'}^{-1}- \lr{H_{0, \bar{F}}'}^{-1} P^\perp_{\bar{F}} {\widetilde{V}}_{F,\bs} P_{\bar{F}},
\end{equation}
as the corresponding $A$ term for $V_{S,\bs}$, which is viewed as supported in $S$, but actually only acts nontrivially on \begin{equation}\label{eq:Fbar=F}
    \bar{F} := F \cup \{i\in \Lambda: \exists f\ni i \text{ such that } \{2P_f-I, \widetilde{V}_{F,\bs}\}= 0\}\subset S.
\end{equation}
This is emphasized by the notation $A_{S,\bs}(\bar{F})$, which will be useful later for bounding $[D',A]$.
Here $\bar{F}\subset S$ comes from the fact that $V_{S,\bs}$ is strongly-supported in $S$, and that $\widetilde{V}_{F,\bs}$ has the same syndrome as $V_{S,\bs}$.
Note that the (pseudo)inverse in \eqref{eq:A=Hinverse} is well-defined because of \eqref{eq:PHP>Delta'}: it is always multiplied by $P^\perp_{\bar{F}}$ that projects out the kernel. We verify that the resulting \begin{equation}\label{eq:P'V=}
    \PP' V_{S,\bs} = [H_{0,\bar{F}}', A_{S,\bs}(\bar{F})] + V_{S,\bs} = -P^\perp_{\bar{F}} {\widetilde{V}}_{F,\bs} P_{\bar{F}}-P_{\bar{F}}{\widetilde{V}}_{F,\bs} P^\perp_{\bar{F}} + V_{S,\bs},
\end{equation} 
stabilizes $\kpsi$: \begin{align}\label{eq:P'Vpsi=}
    \PP' V_{S,\bs} \kpsi
    &= -P^\perp_{\bar{F}} {\widetilde{V}}_{F,\bs} P_{\bar{F}} \kpsi +{\widetilde{V}}_{F,\bs}\kpsi = P_{\bar{F}} {\widetilde{V}}_{F,\bs}\kpsi= P{\widetilde{V}}_{F,\bs} \kpsi \propto \kpsi.
\end{align}
Here in \eqref{eq:P'V=}, we have used \begin{equation}\label{eq:HfA=0}
    [H_{0,f}', A_{S,\bs}(\bar{F})] = 0,\quad \forall f: f\not \subset \bar{F},
\end{equation}
from the definition of $\bar{F}$ \eqref{eq:Fbar=F} and the fact that
$H_{0,f}'$ commutes with all projectors and $H_0^\prime$ in \eqref{eq:A=Hinverse}, and only possibly fails to commute with ${\widetilde{V}}_{F,\bs}$ if $f\cap F \ne \emptyset$. In \eqref{eq:P'Vpsi=}, the second term of \eqref{eq:P'V=} annihilates $\kpsi$, and we have used the locally-nondegenerate condition for the last term. In \eqref{eq:P'Vpsi=}, we then absorbed $P_{\bar{F}}$ into $\kpsi$ and again used the locally-nondegenerate condition that the ground state component of ${\widetilde{V}}_{F,\bs} \kpsi$ is just $\propto \kpsi$.

The operators constructed are all good decompositions: $\bar{F}\subset S$ and \eqref{eq:HfA=0} guarantees $A_{S,\bs}(\bar{F})$ is strongly supported in $S$ (we do not want to use $\bar{F}$ as the support because it can be disconnected). $A_{S,\bs}(\bar{F})$ also has syndrome $\bs$ because for each of the two terms in \eqref{eq:A=Hinverse}, one can commute $2P_f-I$ through by adding the same $\pm$ sign when going through $\widetilde{V}_{F,\bs}$. Finally, the properties also hold for $\PP'V_{S,\bs}$ due to \eqref{eq:H01A+V=PV}, where $\PP' V_{S,\bs}=(\PP' V)_{S,\bs}$ is strongly supported in $S$ because $[H_0', A_{S,\bs}(\bar{F})]$ is also.

Then we prove the bounds in \eqref{eq:AK<VK}. Here we work with the original extensive operators \begin{equation}
    V=\sum_{S,\bs} V_{S,\bs},\quad A=\sum_{S,\bs} A_{S,\bs}(\bar{F}),\quad \PP' V=\sum_{S,\bs} (\PP'V)_{S,\bs}.
\end{equation}
Since \begin{equation}\label{eq:Anaivebound}
    \norm{A_{S,\bs}(\bar{F})}\le \max\mlr{\norm{P_{\bar{F}} {\widetilde{V}}_{F,\bs} P^\perp_{\bar{F}}\lr{H_{0, \bar{F}}'}^{-1}}, \norm{\lr{H_{0, \bar{F}}'}^{-1} P^\perp_{\bar{F}} {\widetilde{V}}_{F,\bs} P_{\bar{F}}}} \le \norm{\widetilde{V}_{F,\bs}}/\Delta'\le \norm{V_{S,\bs}}/\Delta',
\end{equation}
from \eqref{eq:PHP>Delta'}, we have \begin{equation}\label{eq:D'AK<VK}
    \norm{A}_{K}= \max_i \sum_{S: S\ni i, \bs} \norm{A_{S,\bs}(\bar{F})} \ee^{K|S|} \le \frac{1}{\Delta'}\max_i \sum_{S: S\ni i, \bs} \norm{V_{S,\bs}} \ee^{K|S|} = \norm{V}_K/\Delta'.
\end{equation}
Note that the strongly-supported condition of $V_{S,\bs}$ is crucial for this equation to hold: otherwise $A$ can be supported in a region $S'$ with $|S'|-|S|\gg 1$. Using \eqref{eq:P'V=}, \begin{equation}\label{eq:P'V<2V}
    \norm{\PP' V_{S,\bs}} \le \norm{P^\perp_{\bar{F}} {\widetilde{V}}_{F,\bs} P_{\bar{F}}+P_{\bar{F}}{\widetilde{V}}_{F,\bs} P^\perp_{\bar{F}}} + \norm{V_{S,\bs}} \le \norm{\widetilde{V}_{F,\bs}} + \norm{V_{S,\bs}} \le 2\norm{V_{S,\bs}}.
\end{equation}
The second inequality of \eqref{eq:P'V<2V} follows from the similar \eqref{eq:Anaivebound}.
Similar to \eqref{eq:D'AK<VK}, this leads to $\norm{\PP' V}_K\le 2\norm{V}_K$ by summing $S,\bs$, which proves \eqref{eq:AK<VK}. 

\comment{
\begin{equation}
    \bra{\psi}(P\OO P^\perp + P^\perp \OO P)\ket{\psi} \le 2\norm{\OO} \norm{P\ket{\psi}} \norm{P^\perp \ket{\psi}} \le \norm{\OO}\lr{\norm{P\ket{\psi}}^2+ \norm{P^\perp \ket{\psi}}^2} = \norm{\OO},
\end{equation}
for any orthogonal projectors $P,P^\perp$ and any normalized state $\ket{\psi}$.
}

The proof until now applies for any $V$. From now on, we turn to operator $[D',A]$ and aim to show \eqref{eq:DA<} by the constraint that all local terms $V_{S,\bs}$ have support smaller than some metastability size: \begin{equation}\label{eq:S<Rd}
    |S|\le c_d R^d.
\end{equation}

We assign the local terms as \begin{equation}\label{eq:D'ASs}
    \lr{[D', A]}_{S'',\bs''} = \sum_{S,\bs}\sum_{\substack{S':S'\cap \bar{F}\neq \emptyset,\bs': \\
    S\cup S'=S'',\bs+\bs'=\bs''}} [D'_{S',\bs'}, A_{S,\bs}(\bar{F})],
\end{equation}
where each term is strongly supported in $S'\cup S$ (instead of the smaller set $S'\cup \bar{F}$, which can be disconnected), and has syndrome $\bs+\bs'$. Crucially, only $S'$ that overlaps with $\bar{F}$ contributes because $A_{S,\bs}(\bar{F})$ acts as a non-identity operator only in $\bar{F}$.
Hence, \eqref{eq:D'ASs} is a valid decomposition of the form \eqref{eq:decomp}. To get \eqref{eq:DA<}, we need a stronger bound than the naive \eqref{eq:Anaivebound}.  Due to \eqref{eq:dF_anticomm}, $\widetilde{V}_{F,\bs} P_{\bar{F}}\ket{\psi}$ for any $\ket{\psi}$ violates at least $|\partial F|/(c_d\ell_0^d)$ number of terms $H_{0,f}'$, because a single term covers at most $c_d\ell_0^d$ sites. As a result, $P^\perp_{\bar{F}}$ in \eqref{eq:A=Hinverse} can actually be replaced by $P_{\bar{F}}^{\gtrsim|\partial F|}$ that projects onto the subspace of energy (measured by $H'_{0,\bar{F}}$) larger than (or equal to) $\Delta'|\partial F|/(c_d\ell_0^d)$. Using  \begin{equation}\label{eq:PHP>boundary}
    P_{\bar{F}}^{\gtrsim|\partial F|} H_{0,\bar{F}}' P_{\bar{F}}^{\gtrsim|\partial F|} \ge \frac{\Delta'|\partial F|}{c_d\ell_0^d},
\end{equation}
instead of \eqref{eq:PHP>Delta'}, \eqref{eq:Anaivebound} becomes \begin{equation}
    \norm{A_{S,\bs}(\bar{F})} \le \frac{c_d \ell_0^d\norm{V_{S,\bs}}}{\Delta'|\partial F|}.
\end{equation}
This leads to \begin{align}\label{eq:sumS'DA<}
    \sum_{S':S'\cap \bar{F}\neq \emptyset,\bs'} \norm{[D'_{S',\bs'}, A_{S,\bs}(\bar{F})]} \ee^{K|S'\cup S|} &\le \sum_{i\in \bar{F}} \sum_{S'\ni i,\bs'} 2\norm{D'_{S',\bs'}}\norm{A_{S,\bs}(\bar{F})} \ee^{K|S|+c_d(\ell'_D)^d} \nonumber\\
    &\le 2c_d\ell_0^d \ee^{c_d(\ell'_D)^d} \frac{\mathbbm{d}'}{\Delta'} \norm{V_{S,\bs}} \frac{|F|}{|\partial F|} \ee^{K|S|} \le c_{\rm DA} \frac{\mathbbm{d}'}{\Delta'} \norm{V_{S,\bs}} R\,\ee^{K|S|}\nonumber\\
    &\le \frac{1}{4} \norm{V_{S,\bs}} \ee^{K|S|},
\end{align}
for some constant $c_{\rm DA}$ determined by $d,\ell_0,\ell'_D$. In the second line of \eqref{eq:sumS'DA<}, we have used $|\bar{F}|\le |F|+c_{\rm bar}|\partial F|$ for some constant $c_{\rm bar}$, $|F|\le c_d R^d$ from \eqref{eq:S<Rd}, and the geometrical fact that boundary should grow with volume as \begin{equation}\label{eq:boundary>volume}
    |\partial F| \ge c'_d |F|^{1-1/d},
\end{equation}
for some constant $c'_d$ determined by $d$. We have chosen $c_{\rm R}=1/(4 c_{\rm DA})$ in \eqref{eq:R=1d} so that the last line of \eqref{eq:sumS'DA<} holds. It remains to further sum over $S$: \begin{align}
    \norm{[D',A]}_K &\le \max_i \sum_{S,\bs} \sum_{S':S'\cap \bar{F}\neq \emptyset, i\in (S'\cup S),\bs'}\norm{[D'_{S',\bs'}, A_{S,\bs}(\bar{F})]} \ee^{K|S'\cup S|} \nonumber\\
    &\le \max_i \sum_{S\ni i,\bs} \sum_{S':S'\cap \bar{F}\neq \emptyset,\bs'} \norm{[D'_{S',\bs'}, A_{S,\bs}(\bar{F})]} \ee^{K|S'\cup S|} + \max_i \sum_{S'\ni i,\bs'} \sum_{S:S'\cap \bar{F}\neq \emptyset,\bs} \norm{[D'_{S',\bs'}, A_{S,\bs}(\bar{F})]} \ee^{K|S'\cup S|} \nonumber\\
    &\le \frac{1}{4} \max_i \sum_{S\ni i,\bs} \norm{V_{S,\bs}}\ee^{K|S|} + \max_i \sum_{S'\ni i,\bs'} \sum_{j\in S'} \sum_{S: j\in \bar{F},\bs} 2\norm{D'_{S',\bs'}}\norm{A_{S,\bs}(\bar{F})} \ee^{K|S|+c_d(\ell'_D)^d} \nonumber\\
    &\le \frac{1}{4} \norm{V}_K + 2c_d(\ell_D')^d \ee^{c_d(\ell'_D)^d} \norm{D'}_0 \norm{A}_K \le \lr{\frac{1}{4}+ 2c_d(\ell_D')^d \ee^{c_d(\ell'_D)^d} \frac{\mathbbm{d}'}{\Delta'}} \norm{V}_K \nonumber\\
    &\le \frac{1}{3}\norm{V}_K.
\end{align}
Here in the second line, the two terms correspond to the two situations that either $S$ or $S'$ contains $i$. We have used \eqref{eq:sumS'DA<} for the first term in the third line, and invoked \eqref{eq:D'AK<VK} in the fourth line. The last line follows from choosing a sufficiently small constant $c_V$ in \eqref{eq:d'e<Delta}.
\end{proof}

\subsection{A stronger notion of metastability}\label{app:relabound}

We have proven a stability timescale $t_*\gtrsim \ee^{(\epsilon')^{-d}}$ for perturbing $H_0'$, as defined in previous sections, using the idea that $H_0'+D'$ is metastable in volume $\lesssim R^d$ for any $D'$ that is locally block-diagonal with respect to $H_0'$. Here we make this intuition rigorous, and thus connect our proof of slow false vacuum decay to the metastability theory developed in this paper.

\begin{thm}\label{thm:H+D_metastable}
Let $\kpsi$ be a ground state of $H_0'$ that is $R'$-locally nondegenerate. Suppose $D'$ is locally-block-diagonal with respect to $H_0'$, annihilates $\kpsi$ locally: \begin{equation}\label{eq:D_annihilate}
    D'_{S,\bs} \kpsi = 0,
\end{equation}
and has finite range $\ell'_D$ and local norm \eqref{eq:bbmd'=}. There exists a constant $c_{\rm R}$ determined by $d,\ell_0,\ell'_D$, so that if $R<R'$ with $R$ defined by \eqref{eq:R=1d}, then $(H_0,\kpsi)$ with $H_0:=H_0'+D'$ obeys\begin{equation}\label{eq:H0>Delta2}
    H_0[(\OO-\alr{\OO}_{\psi_0})\kpsi]\ge \Delta'/2, \quad \text{for any } \OO \text{ strongly supported in } S \text{ with } |S|  \le c_d R^d.
\end{equation}
\end{thm}

We present several remarks before the proof:
\eqref{eq:D_annihilate} is assumed without loss of generality, because $D'$ automatically stabilizes $\kpsi$ due to its locally-block-diagonality and \eqref{eq:D'kpsi=kpsi}, so one can subtract each local term with a $c$-number to achieve \eqref{eq:D_annihilate}. 
According to the fattening procedure in Fig.~\ref{fig:syndrome}(b), the strongly-supported conditions in \eqref{eq:H0>Delta2} can be relaxed to just $\OO$ is supported in $S$, by reducing $R$ by a finite amount. 

\eqref{eq:H0>Delta2} immediately implies $(H_0,\kpsi)$ is metastable for our original definition for metastability \eqref{eq:H0>D} with range $R$, because any $S$ with $\diam S\le R$ satisfies $|S|\le c_d R^d$. However, by counting the volume instead of the diameter, \eqref{eq:H0>Delta2} is much stronger and constrains e.g. stripe-like regions with $\diam S>R$.  Indeed, the stronger criterion \eqref{eq:H0>Delta2} more carefully captures our intuition about metastability for the simplest models, such as the false vacuum decay problem in an Ising ferromagnet.  Using this much sharper notion of metastability is the key idea for the previous proof that obtains the tight $d$ exponent in $t_*$.

\begin{proof}
Let $Q$ be the projector onto the subspace \begin{equation}\label{eq:Q=spanO}
    \mathcal{Q} = \mathrm{span}((\OO-\alr{\OO}_{\psi_0})\kpsi: \OO \text{ is strongly supported in } S \text{ with } |S| \le c_d R^d ),
\end{equation}
i.e. the subspace of interest that does not include any ground state of $H_0'$ due to $R<R'$.  We aim to show that for some constant $c_{\rm R}$,
\begin{equation}\label{eq:QH0Q>Delta}
    QH_0Q\ge \Delta'/2,
\end{equation}
which is sufficient to obtain \eqref{eq:H0>Delta2}.

As discussed around \eqref{eq:syndrome_add}, the total Hilbert space can be separated to subspaces $\widetilde{P}_\bs$ labeled by syndrome $\bs$. For example, the ground state satisfies $P_f \kpsi = \kpsi$ for any $f$ so it has syndrome $\bs=\emptyset$. For any $\ket{\psi_\bs}\in \mathcal{Q}$
with syndrome $\bs$, it is a superposition of states $\OO_{S,\bs}\kpsi$ where $S$ contains $\bs$ and $|S|\le c_d R^d$.
Moreover, since $D'$ is locally-block-diagonal with $R'>R$, each state in the superposition can be rewritten as $\widetilde{\OO}_{F,\bs}\kpsi$ where $F=F(\bs)$ is the (possibly disconnected) subset determined by the shrinking procedure described below \eqref{eq:dF_anticomm}, such that every boundary site in $\partial F$ intersects with the syndrome $\bs$, and \begin{equation}\label{eq:Fs<Rd}
   1 \le |F|\le c_d R^d.
\end{equation} 
Then the different terms in the superposition can be combined, so that $\ket{\psi_\bs}=\OO_{F,\bs}\kpsi$ for a single operator $\OO_{F,\bs}$ supported in $F=F(\bs)$ determined by $\bs$.

We first investigate what happens when acting $H_0$ on $\ket{\psi_\bs}$ with fixed syndrome, and deal with the general case of superposition of syndromes later. The $H_0'$ part just stablizes the state \begin{equation}\label{eq:H0'=E'}
    H_0'\ket{\psi_\bs} = E_\bs' \ket{\psi_\bs}, 
\end{equation}
where \begin{equation}
    E_\bs' \ge \frac{\Delta'|\partial F|}{c_d\ell_0^d}.
\end{equation}
The bound on $E_\bs'$ is because a constant portion of the boundary checks of $F$ are violated, similar to \eqref{eq:PHP>boundary}.
On the other hand, \eqref{eq:D_annihilate} implies \begin{equation}
    D'\ket{\psi_\bs}= \sum_{S:S\cap F\neq \emptyset}D'_S \ket{\psi_\bs},
\end{equation}
so we immediately have \begin{equation}\label{eq:D_syndrome<}
    \norm{D'\ket{\psi_\bs}}\le \mathbbm{d}'|F| \le  \frac{\mathbbm{d}'|F|^{1/d} |\partial F|}{c'_d} \le c_d^{1+1/d}\ell_0^d \frac{\mathbbm{d}'R}{c'_d \Delta'}  E_\bs' ,
\end{equation}
where we have used the boundary-volume relation \eqref{eq:boundary>volume}, and the energy bound \eqref{eq:H0'=E'}. This proves \eqref{eq:QH0Q>Delta} when acting on a state with fixed syndrome. 

Expand \begin{equation}
D'\ket{\psi_\bs}=\sum_{S:S\cap F\neq 0,\bs''}D'_{S,\bs''}\ket{\psi_\bs}=\sum_{\bs'}\sum_{S:S\cap F\neq 0}D'_{S,\bs-\bs'}\ket{\psi_\bs}
\end{equation}
into syndromes, where the number of syndromes $\bs'$ that appear is bounded by $c_{\rm syn}|F|$ with constant $c_{\rm syn}$ determined by $d,\ell'_D$, because due to the locality of $D^\prime$, $\mathbf{s}^{\prime\prime}$ and $|\mathbf{s}|$ can only differ by O(1) syndromes, and since $D^\prime$ stabilizes $|\psi_0\rangle$, $D^\prime$ must act non-trivially on an already included syndrome in $\mathbf{s}$.  Furthermore, each $\bs'\neq \bs$ corresponds to a particular $\bs''\neq \emptyset$, so is contributed by a finite amount $c_{\rm syn}'$ of local terms $D'_{S,\bs''}$. Therefore we have for $\bs\neq \bs'$, \begin{equation}\label{eq:D_off_diagonal}
    \norm{\widetilde{P}_{\bs'}Q H_0 Q\widetilde{P}_\bs} \le \norm{\widetilde{P}_{\bs'} H_0 \widetilde{P}_\bs} = \norm{\widetilde{P}_{\bs'}D'\widetilde{P}_\bs}\left\{\begin{array}{cc}
        \le c_{\rm syn}' \mathbbm{d}', & \text{if } \bs,\bs' \text{ are connected},\\
        =0, & \text{otherwise}.
    \end{array}\right.
\end{equation}
Here we say $\bs,\bs'$ are connected if $\bs\neq \bs'$ and there exists a local term $D'_{S,\bs-\bs'}$. We have shown that given $\bs$, there are at most $c_{\rm syn}|F|$ many $\bs'$ that are connected to it. We have used $H_0'$ preserves syndrome, and inserted the projector $Q$ freely, because $Q$ does not care about nonempty syndromes and thus commutes with $\widetilde{P}_\bs$.

\eqref{eq:D_off_diagonal} bounds the block-off-diagonal elements of matrix $H_0$, while \eqref{eq:H0'=E'} and \eqref{eq:D_syndrome<} already bound the block-diagonal elements \begin{equation}\label{eq:diagonal>}
    \widetilde{P}_\bs QH_0 Q\widetilde{P}_\bs \ge E_\bs' - \max_{\psi_\bs}\norm{D'\ket{\psi_\bs}}\ge \Delta'|\partial F|/(c_d\ell_0^d)-\mathbbm{d}'|F|.
\end{equation}
We then invoke Gershgorin's circle theorem for block matrices (see, e.g. Theorem 1.13.1 in \cite{gershgorin_book}): For any eigenvalue $\lambda$ of $QH_0Q$, there always exists a syndrome $\bs$ such that either (\emph{1}) $\lambda$ is an eigenvalue of $\widetilde{P}_\bs QH_0Q\widetilde{P}_\bs$, or (\emph{2}) \begin{equation}
    \norm{(\widetilde{P}_\bs QH_0Q\widetilde{P}_\bs - \lambda)^{-1}} \ge \lr{\sum_{\bs'\neq \bs}\norm{\widetilde{P}_{\bs'}Q H_0 Q\widetilde{P}_\bs}}^{-1} \ge \lr{c_{\rm syn} |F|\cdot c_{\rm syn}' \mathbbm{d}'}^{-1},
\end{equation}
which implies \begin{align}\label{eq:eig_H0>}
    \lambda &\ge \lambda_{\rm min}\lr{\widetilde{P}_\bs QH_0 Q\widetilde{P}_\bs}-c_{\rm syn}c_{\rm syn}' \mathbbm{d}' |F|   \ge \frac{\Delta'|\partial F|}{c_d\ell_0^d}- \lr{c_{\rm syn}c_{\rm syn}'+1} \mathbbm{d}' |F| \nonumber\\
    &\ge \mlr{\frac{\Delta'c_d'|F|^{-1/d}}{c_d\ell_0^d}- \lr{c_{\rm syn}c_{\rm syn}'+1} \mathbbm{d}'} |F| \ge \frac{\Delta'c_d'c_d^{-1-1/d}}{\ell_0^d R}- \lr{c_{\rm syn}c_{\rm syn}'+1} \mathbbm{d}',
\end{align}
where $\lambda_{\rm min}(\cdot)$ denotes the minimal eigenvalue, and we have used \eqref{eq:diagonal>} and \eqref{eq:boundary>volume} together with the geometrical facts \eqref{eq:Fs<Rd}. The last expression \eqref{eq:eig_H0>} is larger than $\Delta'/2$ for $R$ defined by \eqref{eq:R=1d} with a sufficiently small constant $c_{\rm R}$. This proves \eqref{eq:QH0Q>Delta} and furthermore \eqref{eq:H0>Delta2}.
\end{proof}

As a direct implication of this metastability, the following Corollary shows that any locally-block-diagonal $D'$ is ``relatively bounded'' by $H_0'$ for states $\ket{\psi}$ obtained by locally exciting the ground state $\kpsi$. Intriguingly, this relative boundedness also plays a key role in proving stability of topological order \cite{topo_Hastings}, where $\ket{\psi}$ can be any state. Therefore the following \eqref{eq:rela_bound} serves as a natural generalization of this concept, which is equivalently to \begin{equation}\label{eq:QDQ<H}
    Q(D')^2Q \le\frac{1}{4} Q(H_{0}')^2Q.
\end{equation}
The prefactors $1/2,1/4$ are chosen for convenience, and can be replaced by any positive constant (which will determine how small $c_{\rm R}$ should be).

\begin{cor}
    Given the same conditions as Theorem \ref{thm:H+D_metastable}, $D'$ is relatively bounded by $H_0'$ for sufficiently local excitations above $\kpsi$: for all $\ket{\psi} \in \mathcal{Q}$, \begin{equation}\label{eq:rela_bound}
    \norm{D'\ket{\psi}} \le \frac{1}{2} \norm{ H_0'\ket{\psi}}.
\end{equation}
\end{cor}

\begin{proof}
We have proven \eqref{eq:QH0Q>Delta} with $Q$ defined in \eqref{eq:Q=spanO}. By reducing $c_{\rm R}$ by a factor of $2$, \eqref{eq:QH0Q>Delta} also holds for $D$ replaced by $\pm 2D$, so we have
\begin{equation}
    QH_0'Q \pm 2QD'Q\ge 0.
\end{equation}
Multiplying the two equations \begin{subequations}
    \begin{align}
        (QH_0'Q + 2QD'Q)(QH_0'Q - 2QD'Q) \ge 0, \\
        (QH_0'Q - 2QD'Q)(QH_0'Q + 2QD'Q) \ge 0,
    \end{align}
\end{subequations}
and adding them, we have \begin{equation}
    4 \lr{QD'Q}^2 \le \lr{QH_0'Q}^2.
\end{equation}
I.e. \begin{equation}\label{eq:QD'<H0'}
    \norm{QD'\ket{\psi}} \le \frac{1}{2} \norm{ QH_0'\ket{\psi}},
\end{equation}
for any $\ket{\psi} \in \mathcal{Q}$. Since $D',H_0'$ are finite-range, if $\ket{\psi} \in \widetilde{\mathcal{Q}}$ where $\widetilde{\mathcal{Q}}$ is $\mathcal{Q}$ but with $R$ reduced by a finite constant, the final $Q$ projectors in \eqref{eq:QD'<H0'} act trivially and can be dropped. Relabeling $\widetilde{\mathcal{Q}}\rightarrow \mathcal{Q}$ leads to \eqref{eq:rela_bound}.
\end{proof}

\end{appendix}

\bibliography{preth}

\end{document}

%% file: preamble.tex
\usepackage{amsmath,amssymb,amsfonts,mathtools,dsfont,relsize}
\usepackage{graphicx}

\usepackage[colorlinks=true, urlcolor=violet, linkcolor=blue, citecolor=red, hyperindex=true, linktocpage=true]{hyperref}

\usepackage{amsthm}
\usepackage{qtree}
\newtheorem{thm}{Theorem}
\newtheorem{cor}[thm]{Corollary}
\newtheorem{lem}[thm]{Lemma}
\newtheorem{prop}[thm]{Proposition}

\newtheorem{defn}[thm]{Definition}

\usepackage{soul} 

\usepackage{enumerate}
\usepackage{stmaryrd} 

\renewcommand{\thesection}{\arabic{section}}
\renewcommand{\thesubsection}{\thesection.\arabic{subsection}}

\makeatletter
\renewcommand{\p@subsection}{}
\renewcommand{\p@subsubsection}{}
\makeatother

\DeclarePairedDelimiter{\floor}{\lfloor}{\rfloor}
\DeclarePairedDelimiter{\expval}{\langle}{\rangle}

\DeclarePairedDelimiter{\ket}{\lvert}{\rangle}
\DeclarePairedDelimiter{\bra}{\langle}{\rvert}

\newcommand{\ii}[0]{\mathrm{i}}

\usepackage{qcircuit}

\newcommand{\lr}[1]{\left( #1\right)}
\newcommand{\mlr}[1]{\left[ #1\right]}
\newcommand{\glr}[1]{\left\{ #1\right\}}
\newcommand{\alr}[1]{\left\langle #1\right\rangle}
\newcommand{\norm}[1]{\left\lVert#1\right\rVert}
\newcommand{\abs}[1]{\left\lvert#1\right\rvert}

\newcommand{\ee}{\mathrm{e}}
\newcommand{\dd}{\mathrm{d}}

\newcommand{\OO}{\mathcal{O}}
\newcommand{\LL}{\mathcal{L}}
\newcommand{\cA}{\mathcal{A}}

\newcommand{\bA}{\mathbb{A}}
\newcommand{\PP}{\mathbb{P}}
\newcommand{\QQ}{\mathbb{Q}}

\makeatletter 
    
\renewcommand\onecolumngrid{
\do@columngrid{one}{\@ne}%
\def\set@footnotewidth{\onecolumngrid}
\def\footnoterule{\kern-6pt\hrule width 1.5in\kern6pt}%
}

\renewcommand\twocolumngrid{
        \def\footnoterule{
        \dimen@\skip\footins\divide\dimen@\thr@@
        \kern-\dimen@\hrule width.5in\kern\dimen@}
        \do@columngrid{mlt}{\tw@}
}%

\makeatother

%% file: main_revise_arxiv.bbl
\providecommand{\noopsort}[1]{}\providecommand{\singleletter}[1]{#1}%
\begin{thebibliography}{129}%
\makeatletter
\providecommand \@ifxundefined [1]{%
 \@ifx{#1\undefined}
}%
\providecommand \@ifnum [1]{%
 \ifnum #1\expandafter \@firstoftwo
 \else \expandafter \@secondoftwo
 \fi
}%
\providecommand \@ifx [1]{%
 \ifx #1\expandafter \@firstoftwo
 \else \expandafter \@secondoftwo
 \fi
}%
\providecommand \natexlab [1]{#1}%
\providecommand \enquote  [1]{``#1''}%
\providecommand \bibnamefont  [1]{#1}%
\providecommand \bibfnamefont [1]{#1}%
\providecommand \citenamefont [1]{#1}%
\providecommand \href@noop [0]{\@secondoftwo}%
\providecommand \href [0]{\begingroup \@sanitize@url \@href}%
\providecommand \@href[1]{\@@startlink{#1}\@@href}%
\providecommand \@@href[1]{\endgroup#1\@@endlink}%
\providecommand \@sanitize@url [0]{\catcode `\\12\catcode `\$12\catcode
  `\&12\catcode `\#12\catcode `\^12\catcode `\_12\catcode `\%12\relax}%
\providecommand \@@startlink[1]{}%
\providecommand \@@endlink[0]{}%
\providecommand \url  [0]{\begingroup\@sanitize@url \@url }%
\providecommand \@url [1]{\endgroup\@href {#1}{\urlprefix }}%
\providecommand \urlprefix  [0]{URL }%
\providecommand \Eprint [0]{\href }%
\providecommand \doibase [0]{http://dx.doi.org/}%
\providecommand \selectlanguage [0]{\@gobble}%
\providecommand \bibinfo  [0]{\@secondoftwo}%
\providecommand \bibfield  [0]{\@secondoftwo}%
\providecommand \translation [1]{[#1]}%
\providecommand \BibitemOpen [0]{}%
\providecommand \bibitemStop [0]{}%
\providecommand \bibitemNoStop [0]{.\EOS\space}%
\providecommand \EOS [0]{\spacefactor3000\relax}%
\providecommand \BibitemShut  [1]{\csname bibitem#1\endcsname}%
\let\auto@bib@innerbib\@empty
\bibitem [{\citenamefont {Langer}(1969)}]{langer}%
  \BibitemOpen
  \bibfield  {author} {\bibinfo {author} {\bibfnamefont {J.~S.}\ \bibnamefont
  {Langer}},\ }\bibfield  {title} {\enquote {\bibinfo {title} {Statistical
  theory of the decay of metastable states},}\ }\href@noop {} {\bibfield
  {journal} {\bibinfo  {journal} {Ann. Phys.}\ }\textbf {\bibinfo {volume}
  {54}},\ \bibinfo {pages} {258} (\bibinfo {year} {1969})}\BibitemShut
  {NoStop}%
\bibitem [{\citenamefont {Cassandro}\ \emph {et~al.}(1984)\citenamefont
  {Cassandro}, \citenamefont {Galves}, \citenamefont {Olivieri},\ and\
  \citenamefont {Vares}}]{cassandro}%
  \BibitemOpen
  \bibfield  {author} {\bibinfo {author} {\bibfnamefont {M.}~\bibnamefont
  {Cassandro}}, \bibinfo {author} {\bibfnamefont {A.}~\bibnamefont {Galves}},
  \bibinfo {author} {\bibfnamefont {E.}~\bibnamefont {Olivieri}}, \ and\
  \bibinfo {author} {\bibfnamefont {M.~E.}\ \bibnamefont {Vares}},\ }\bibfield
  {title} {\enquote {\bibinfo {title} {Metastable behavior of stochastic
  dynamics: A pathwise approach},}\ }\href@noop {} {\bibfield  {journal}
  {\bibinfo  {journal} {J. Stat. Phys.}\ }\textbf {\bibinfo {volume} {35}},\
  \bibinfo {pages} {603} (\bibinfo {year} {1984})}\BibitemShut {NoStop}%
\bibitem [{\citenamefont {Gamow}(1928)}]{tunnel_wavef28}%
  \BibitemOpen
  \bibfield  {author} {\bibinfo {author} {\bibfnamefont {George}\ \bibnamefont
  {Gamow}},\ }\bibfield  {title} {\enquote {\bibinfo {title} {Zur
  quantentheorie des atomkernes},}\ }\href@noop {} {\bibfield  {journal}
  {\bibinfo  {journal} {Zeitschrift f{\"u}r Physik}\ }\textbf {\bibinfo
  {volume} {51}},\ \bibinfo {pages} {204--212} (\bibinfo {year}
  {1928})}\BibitemShut {NoStop}%
\bibitem [{\citenamefont {Nieva}\ \emph {et~al.}(1992)\citenamefont {Nieva},
  \citenamefont {Osquiguil}, \citenamefont {Guimpel}, \citenamefont
  {Maenhoudt}, \citenamefont {Wuyts}, \citenamefont {Bruynseraede},
  \citenamefont {Maple},\ and\ \citenamefont {Schuller}}]{Nieva1992}%
  \BibitemOpen
  \bibfield  {author} {\bibinfo {author} {\bibfnamefont {G.}~\bibnamefont
  {Nieva}}, \bibinfo {author} {\bibfnamefont {E.}~\bibnamefont {Osquiguil}},
  \bibinfo {author} {\bibfnamefont {J.}~\bibnamefont {Guimpel}}, \bibinfo
  {author} {\bibfnamefont {M.}~\bibnamefont {Maenhoudt}}, \bibinfo {author}
  {\bibfnamefont {B.}~\bibnamefont {Wuyts}}, \bibinfo {author} {\bibfnamefont
  {Y.}~\bibnamefont {Bruynseraede}}, \bibinfo {author} {\bibfnamefont {M.~B.}\
  \bibnamefont {Maple}}, \ and\ \bibinfo {author} {\bibfnamefont {Ivan~K.}\
  \bibnamefont {Schuller}},\ }\bibfield  {title} {\enquote {\bibinfo {title}
  {{Photoinduced enhancement of superconductivity}},}\ }\href {\doibase
  10.1063/1.107069} {\bibfield  {journal} {\bibinfo  {journal} {Applied Physics
  Letters}\ }\textbf {\bibinfo {volume} {60}},\ \bibinfo {pages} {2159--2161}
  (\bibinfo {year} {1992})}\BibitemShut {NoStop}%
\bibitem [{\citenamefont {Stojchevska}\ \emph {et~al.}(2014)\citenamefont
  {Stojchevska}, \citenamefont {Vaskivskyi}, \citenamefont {Mertelj},
  \citenamefont {Kusar}, \citenamefont {Svetin}, \citenamefont {Brazovskii},\
  and\ \citenamefont {Mihailovic}}]{Stojchevska}%
  \BibitemOpen
  \bibfield  {author} {\bibinfo {author} {\bibfnamefont {L.}~\bibnamefont
  {Stojchevska}}, \bibinfo {author} {\bibfnamefont {I.}~\bibnamefont
  {Vaskivskyi}}, \bibinfo {author} {\bibfnamefont {T.}~\bibnamefont {Mertelj}},
  \bibinfo {author} {\bibfnamefont {P.}~\bibnamefont {Kusar}}, \bibinfo
  {author} {\bibfnamefont {D.}~\bibnamefont {Svetin}}, \bibinfo {author}
  {\bibfnamefont {S.}~\bibnamefont {Brazovskii}}, \ and\ \bibinfo {author}
  {\bibfnamefont {D.}~\bibnamefont {Mihailovic}},\ }\bibfield  {title}
  {\enquote {\bibinfo {title} {Ultrafast switching to a stable hidden quantum
  state in an electronic crystal},}\ }\href {\doibase 10.1126/science.1241591}
  {\bibfield  {journal} {\bibinfo  {journal} {Science}\ }\textbf {\bibinfo
  {volume} {344}},\ \bibinfo {pages} {177--180} (\bibinfo {year}
  {2014})}\BibitemShut {NoStop}%
\bibitem [{\citenamefont {Budden}\ \emph {et~al.}(2021)\citenamefont {Budden},
  \citenamefont {Gebert}, \citenamefont {Buzzi}, \citenamefont {Jotzu},
  \citenamefont {Wang}, \citenamefont {Matsuyama}, \citenamefont {Meier},
  \citenamefont {Laplace}, \citenamefont {Pontiroli}, \citenamefont {Ricc{\`o}}
  \emph {et~al.}}]{budden2021evidence}%
  \BibitemOpen
  \bibfield  {author} {\bibinfo {author} {\bibfnamefont {M}~\bibnamefont
  {Budden}}, \bibinfo {author} {\bibfnamefont {T}~\bibnamefont {Gebert}},
  \bibinfo {author} {\bibfnamefont {M}~\bibnamefont {Buzzi}}, \bibinfo {author}
  {\bibfnamefont {G}~\bibnamefont {Jotzu}}, \bibinfo {author} {\bibfnamefont
  {E}~\bibnamefont {Wang}}, \bibinfo {author} {\bibfnamefont {T}~\bibnamefont
  {Matsuyama}}, \bibinfo {author} {\bibfnamefont {G}~\bibnamefont {Meier}},
  \bibinfo {author} {\bibfnamefont {Y}~\bibnamefont {Laplace}}, \bibinfo
  {author} {\bibfnamefont {D}~\bibnamefont {Pontiroli}}, \bibinfo {author}
  {\bibfnamefont {M}~\bibnamefont {Ricc{\`o}}},  \emph {et~al.},\ }\bibfield
  {title} {\enquote {\bibinfo {title} {Evidence for metastable photo-induced
  superconductivity in k3c60},}\ }\href {\doibase
  https://doi.org/10.1038/s41567-020-01148-1} {\bibfield  {journal} {\bibinfo
  {journal} {Nature Physics}\ }\textbf {\bibinfo {volume} {17}},\ \bibinfo
  {pages} {611--618} (\bibinfo {year} {2021})}\BibitemShut {NoStop}%
\bibitem [{\citenamefont {Coleman}(1977)}]{coleman}%
  \BibitemOpen
  \bibfield  {author} {\bibinfo {author} {\bibfnamefont {Sidney}\ \bibnamefont
  {Coleman}},\ }\bibfield  {title} {\enquote {\bibinfo {title} {Fate of the
  false vacuum: Semiclassical theory},}\ }\href {\doibase
  10.1103/PhysRevD.15.2929} {\bibfield  {journal} {\bibinfo  {journal} {Phys.
  Rev. D}\ }\textbf {\bibinfo {volume} {15}},\ \bibinfo {pages} {2929--2936}
  (\bibinfo {year} {1977})}\BibitemShut {NoStop}%
\bibitem [{\citenamefont {Devoto}\ \emph {et~al.}(2022)\citenamefont {Devoto},
  \citenamefont {Devoto}, \citenamefont {Luzio},\ and\ \citenamefont
  {Ridolfi}}]{FVD_rev22}%
  \BibitemOpen
  \bibfield  {author} {\bibinfo {author} {\bibfnamefont {Federica}\
  \bibnamefont {Devoto}}, \bibinfo {author} {\bibfnamefont {Simone}\
  \bibnamefont {Devoto}}, \bibinfo {author} {\bibfnamefont {Luca~Di}\
  \bibnamefont {Luzio}}, \ and\ \bibinfo {author} {\bibfnamefont {Giovanni}\
  \bibnamefont {Ridolfi}},\ }\bibfield  {title} {\enquote {\bibinfo {title}
  {False vacuum decay: an introductory review},}\ }\href {\doibase
  10.1088/1361-6471/ac7f24} {\bibfield  {journal} {\bibinfo  {journal} {Journal
  of Physics G: Nuclear and Particle Physics}\ }\textbf {\bibinfo {volume}
  {49}},\ \bibinfo {pages} {103001} (\bibinfo {year} {2022})}\BibitemShut
  {NoStop}%
\bibitem [{\citenamefont {Sher}(1989)}]{SHER1989273}%
  \BibitemOpen
  \bibfield  {author} {\bibinfo {author} {\bibfnamefont {Marc}\ \bibnamefont
  {Sher}},\ }\bibfield  {title} {\enquote {\bibinfo {title} {Electroweak higgs
  potential and vacuum stability},}\ }\href {\doibase
  https://doi.org/10.1016/0370-1573(89)90061-6} {\bibfield  {journal} {\bibinfo
   {journal} {Physics Reports}\ }\textbf {\bibinfo {volume} {179}},\ \bibinfo
  {pages} {273--418} (\bibinfo {year} {1989})}\BibitemShut {NoStop}%
\bibitem [{\citenamefont {Turner}\ and\ \citenamefont
  {Wilczek}(1982)}]{wilczek}%
  \BibitemOpen
  \bibfield  {author} {\bibinfo {author} {\bibfnamefont {M.~S.}\ \bibnamefont
  {Turner}}\ and\ \bibinfo {author} {\bibfnamefont {F.}~\bibnamefont
  {Wilczek}},\ }\bibfield  {title} {\enquote {\bibinfo {title} {Is our vacuum
  metastable?}}\ }\href@noop {} {\bibfield  {journal} {\bibinfo  {journal}
  {Nature}\ }\textbf {\bibinfo {volume} {298}},\ \bibinfo {pages} {633}
  (\bibinfo {year} {1982})}\BibitemShut {NoStop}%
\bibitem [{\citenamefont {Elias-Miró}\ \emph {et~al.}(2012)\citenamefont
  {Elias-Miró}, \citenamefont {Espinosa}, \citenamefont {Giudice},
  \citenamefont {Isidori}, \citenamefont {Riotto},\ and\ \citenamefont
  {Strumia}}]{ELIASMIRO2012222}%
  \BibitemOpen
  \bibfield  {author} {\bibinfo {author} {\bibfnamefont {Joan}\ \bibnamefont
  {Elias-Miró}}, \bibinfo {author} {\bibfnamefont {José~R.}\ \bibnamefont
  {Espinosa}}, \bibinfo {author} {\bibfnamefont {Gian~F.}\ \bibnamefont
  {Giudice}}, \bibinfo {author} {\bibfnamefont {Gino}\ \bibnamefont {Isidori}},
  \bibinfo {author} {\bibfnamefont {Antonio}\ \bibnamefont {Riotto}}, \ and\
  \bibinfo {author} {\bibfnamefont {Alessandro}\ \bibnamefont {Strumia}},\
  }\bibfield  {title} {\enquote {\bibinfo {title} {Higgs mass implications on
  the stability of the electroweak vacuum},}\ }\href {\doibase
  https://doi.org/10.1016/j.physletb.2012.02.013} {\bibfield  {journal}
  {\bibinfo  {journal} {Physics Letters B}\ }\textbf {\bibinfo {volume}
  {709}},\ \bibinfo {pages} {222--228} (\bibinfo {year} {2012})}\BibitemShut
  {NoStop}%
\bibitem [{\citenamefont {Rutkevich}(1999)}]{falseVac_Ising}%
  \BibitemOpen
  \bibfield  {author} {\bibinfo {author} {\bibfnamefont {S.~B.}\ \bibnamefont
  {Rutkevich}},\ }\bibfield  {title} {\enquote {\bibinfo {title} {Decay of the
  metastable phase in $d=1$ and $d=2$ ising models},}\ }\href {\doibase
  10.1103/PhysRevB.60.14525} {\bibfield  {journal} {\bibinfo  {journal} {Phys.
  Rev. B}\ }\textbf {\bibinfo {volume} {60}},\ \bibinfo {pages} {14525--14528}
  (\bibinfo {year} {1999})}\BibitemShut {NoStop}%
\bibitem [{\citenamefont {Lagnese}\ \emph {et~al.}(2021)\citenamefont
  {Lagnese}, \citenamefont {Surace}, \citenamefont {Kormos},\ and\
  \citenamefont {Calabrese}}]{falseVac_spinchain}%
  \BibitemOpen
  \bibfield  {author} {\bibinfo {author} {\bibfnamefont {Gianluca}\
  \bibnamefont {Lagnese}}, \bibinfo {author} {\bibfnamefont {Federica~Maria}\
  \bibnamefont {Surace}}, \bibinfo {author} {\bibfnamefont {M\'arton}\
  \bibnamefont {Kormos}}, \ and\ \bibinfo {author} {\bibfnamefont {Pasquale}\
  \bibnamefont {Calabrese}},\ }\bibfield  {title} {\enquote {\bibinfo {title}
  {False vacuum decay in quantum spin chains},}\ }\href {\doibase
  10.1103/PhysRevB.104.L201106} {\bibfield  {journal} {\bibinfo  {journal}
  {Phys. Rev. B}\ }\textbf {\bibinfo {volume} {104}},\ \bibinfo {pages}
  {L201106} (\bibinfo {year} {2021})}\BibitemShut {NoStop}%
\bibitem [{\citenamefont {Andreassen}\ \emph {et~al.}(2016)\citenamefont
  {Andreassen}, \citenamefont {Farhi}, \citenamefont {Frost},\ and\
  \citenamefont {Schwartz}}]{Andreassen:2016cff}%
  \BibitemOpen
  \bibfield  {author} {\bibinfo {author} {\bibfnamefont {Anders}\ \bibnamefont
  {Andreassen}}, \bibinfo {author} {\bibfnamefont {David}\ \bibnamefont
  {Farhi}}, \bibinfo {author} {\bibfnamefont {William}\ \bibnamefont {Frost}},
  \ and\ \bibinfo {author} {\bibfnamefont {Matthew~D.}\ \bibnamefont
  {Schwartz}},\ }\bibfield  {title} {\enquote {\bibinfo {title} {{Direct
  Approach to Quantum Tunneling}},}\ }\href {\doibase
  10.1103/PhysRevLett.117.231601} {\bibfield  {journal} {\bibinfo  {journal}
  {Phys. Rev. Lett.}\ }\textbf {\bibinfo {volume} {117}},\ \bibinfo {pages}
  {231601} (\bibinfo {year} {2016})},\ \Eprint
  {http://arxiv.org/abs/1602.01102} {arXiv:1602.01102 [hep-th]} \BibitemShut
  {NoStop}%
\bibitem [{\citenamefont {Chen}\ \emph {et~al.}(2024)\citenamefont {Chen},
  \citenamefont {Huang}, \citenamefont {Preskill},\ and\ \citenamefont
  {Zhou}}]{Chen:2023idc}%
  \BibitemOpen
  \bibfield  {author} {\bibinfo {author} {\bibfnamefont {Chi-Fang}\
  \bibnamefont {Chen}}, \bibinfo {author} {\bibfnamefont {Hsin-Yuan}\
  \bibnamefont {Huang}}, \bibinfo {author} {\bibfnamefont {John}\ \bibnamefont
  {Preskill}}, \ and\ \bibinfo {author} {\bibfnamefont {Leo}\ \bibnamefont
  {Zhou}},\ }\bibfield  {title} {\enquote {\bibinfo {title} {Local minima in
  quantum systems},}\ }in\ \href {\doibase 10.1145/3618260.3649675} {\emph
  {\bibinfo {booktitle} {Proceedings of the 56th Annual ACM Symposium on Theory
  of Computing}}},\ \bibinfo {series and number} {STOC 2024}\ (\bibinfo
  {publisher} {Association for Computing Machinery},\ \bibinfo {address} {New
  York, NY, USA},\ \bibinfo {year} {2024})\ p.\ \bibinfo {pages}
  {1323–1330}\BibitemShut {NoStop}%
\bibitem [{\citenamefont {Yin}\ and\ \citenamefont {Lucas}(2023)}]{our_preth}%
  \BibitemOpen
  \bibfield  {author} {\bibinfo {author} {\bibfnamefont {Chao}\ \bibnamefont
  {Yin}}\ and\ \bibinfo {author} {\bibfnamefont {Andrew}\ \bibnamefont
  {Lucas}},\ }\bibfield  {title} {\enquote {\bibinfo {title}
  {{Prethermalization and the Local Robustness of Gapped Systems}},}\ }\href
  {\doibase 10.1103/PhysRevLett.131.050402} {\bibfield  {journal} {\bibinfo
  {journal} {Phys. Rev. Lett.}\ }\textbf {\bibinfo {volume} {131}},\ \bibinfo
  {pages} {050402} (\bibinfo {year} {2023})},\ \Eprint
  {http://arxiv.org/abs/2209.11242} {arXiv:2209.11242 [cond-mat.str-el]}
  \BibitemShut {NoStop}%
\bibitem [{\citenamefont {Fokker}(1914)}]{Fokker}%
  \BibitemOpen
  \bibfield  {author} {\bibinfo {author} {\bibfnamefont {A.~D.}\ \bibnamefont
  {Fokker}},\ }\bibfield  {title} {\enquote {\bibinfo {title} {Die mittlere
  {E}nergie rotierender elektrischer {D}ipole im {S}trahlungsfeld},}\ }\href
  {\doibase https://doi.org/10.1002/andp.19143480507} {\bibfield  {journal}
  {\bibinfo  {journal} {Annalen der Physik}\ }\textbf {\bibinfo {volume}
  {348}},\ \bibinfo {pages} {810--820} (\bibinfo {year} {1914})}\BibitemShut
  {NoStop}%
\bibitem [{\citenamefont {Planck}(1917)}]{Planck}%
  \BibitemOpen
  \bibfield  {author} {\bibinfo {author} {\bibfnamefont {M.}~\bibnamefont
  {Planck}},\ }\href {https://books.google.com/books?id=Sf4wGwAACAAJ} {\emph
  {\bibinfo {title} {{\"U}ber einen {S}atz der statistischen {D}ynamik und
  seine {E}rweiterung in der {Q}uantentheorie}}},\ Sitzungsberichte der
  K{\"o}niglich-Preussischen Akademie der Wissenschaften zu Berlin\ (\bibinfo
  {publisher} {Reimer},\ \bibinfo {year} {1917})\BibitemShut {NoStop}%
\bibitem [{\citenamefont {Gardiner}(2009)}]{gardiner}%
  \BibitemOpen
  \bibfield  {author} {\bibinfo {author} {\bibfnamefont {C.}~\bibnamefont
  {Gardiner}},\ }\href@noop {} {\emph {\bibinfo {title} {Stochastic
  Methods}}},\ \bibinfo {edition} {4th}\ ed.\ (\bibinfo  {publisher}
  {Springer},\ \bibinfo {year} {2009})\BibitemShut {NoStop}%
\bibitem [{\citenamefont {Banks}\ \emph {et~al.}(1973)\citenamefont {Banks},
  \citenamefont {Bender},\ and\ \citenamefont {Wu}}]{WKB_highd73}%
  \BibitemOpen
  \bibfield  {author} {\bibinfo {author} {\bibfnamefont {Thomas}\ \bibnamefont
  {Banks}}, \bibinfo {author} {\bibfnamefont {Carl~M.}\ \bibnamefont {Bender}},
  \ and\ \bibinfo {author} {\bibfnamefont {Tai~Tsun}\ \bibnamefont {Wu}},\
  }\bibfield  {title} {\enquote {\bibinfo {title} {Coupled anharmonic
  oscillators. i. equal-mass case},}\ }\href {\doibase 10.1103/PhysRevD.8.3346}
  {\bibfield  {journal} {\bibinfo  {journal} {Phys. Rev. D}\ }\textbf {\bibinfo
  {volume} {8}},\ \bibinfo {pages} {3346--3366} (\bibinfo {year}
  {1973})}\BibitemShut {NoStop}%
\bibitem [{\citenamefont {Siegert}(1939)}]{tunnel_wavef39}%
  \BibitemOpen
  \bibfield  {author} {\bibinfo {author} {\bibfnamefont {A.~J.~F.}\
  \bibnamefont {Siegert}},\ }\bibfield  {title} {\enquote {\bibinfo {title} {On
  the derivation of the dispersion formula for nuclear reactions},}\ }\href
  {\doibase 10.1103/PhysRev.56.750} {\bibfield  {journal} {\bibinfo  {journal}
  {Phys. Rev.}\ }\textbf {\bibinfo {volume} {56}},\ \bibinfo {pages} {750--752}
  (\bibinfo {year} {1939})}\BibitemShut {NoStop}%
\bibitem [{\citenamefont {Callan}\ and\ \citenamefont
  {Coleman}(1977)}]{coleman2}%
  \BibitemOpen
  \bibfield  {author} {\bibinfo {author} {\bibfnamefont {Curtis~G.}\
  \bibnamefont {Callan}}\ and\ \bibinfo {author} {\bibfnamefont {Sidney}\
  \bibnamefont {Coleman}},\ }\bibfield  {title} {\enquote {\bibinfo {title}
  {Fate of the false vacuum. ii. first quantum corrections},}\ }\href {\doibase
  10.1103/PhysRevD.16.1762} {\bibfield  {journal} {\bibinfo  {journal} {Phys.
  Rev. D}\ }\textbf {\bibinfo {volume} {16}},\ \bibinfo {pages} {1762--1768}
  (\bibinfo {year} {1977})}\BibitemShut {NoStop}%
\bibitem [{\citenamefont {Metropolis}\ \emph {et~al.}(1953)\citenamefont
  {Metropolis}, \citenamefont {Rosenbluth}, \citenamefont {Rosenbluth},
  \citenamefont {Teller},\ and\ \citenamefont {Teller}}]{Metropolis_1953}%
  \BibitemOpen
  \bibfield  {author} {\bibinfo {author} {\bibfnamefont {Nicholas}\
  \bibnamefont {Metropolis}}, \bibinfo {author} {\bibfnamefont {Arianna~W.}\
  \bibnamefont {Rosenbluth}}, \bibinfo {author} {\bibfnamefont {Marshall~N.}\
  \bibnamefont {Rosenbluth}}, \bibinfo {author} {\bibfnamefont {Augusta~H.}\
  \bibnamefont {Teller}}, \ and\ \bibinfo {author} {\bibfnamefont {Edward}\
  \bibnamefont {Teller}},\ }\bibfield  {title} {\enquote {\bibinfo {title}
  {{Equation of State Calculations by Fast Computing Machines}},}\ }\href
  {\doibase 10.1063/1.1699114} {\bibfield  {journal} {\bibinfo  {journal} {The
  Journal of Chemical Physics}\ }\textbf {\bibinfo {volume} {21}},\ \bibinfo
  {pages} {1087--1092} (\bibinfo {year} {1953})}\BibitemShut {NoStop}%
\bibitem [{\citenamefont {Andreassen}\ \emph {et~al.}(2017)\citenamefont
  {Andreassen}, \citenamefont {Farhi}, \citenamefont {Frost},\ and\
  \citenamefont {Schwartz}}]{anders}%
  \BibitemOpen
  \bibfield  {author} {\bibinfo {author} {\bibfnamefont {Anders}\ \bibnamefont
  {Andreassen}}, \bibinfo {author} {\bibfnamefont {David}\ \bibnamefont
  {Farhi}}, \bibinfo {author} {\bibfnamefont {William}\ \bibnamefont {Frost}},
  \ and\ \bibinfo {author} {\bibfnamefont {Matthew~D.}\ \bibnamefont
  {Schwartz}},\ }\bibfield  {title} {\enquote {\bibinfo {title} {Precision
  decay rate calculations in quantum field theory},}\ }\href {\doibase
  10.1103/PhysRevD.95.085011} {\bibfield  {journal} {\bibinfo  {journal} {Phys.
  Rev. D}\ }\textbf {\bibinfo {volume} {95}},\ \bibinfo {pages} {085011}
  (\bibinfo {year} {2017})}\BibitemShut {NoStop}%
\bibitem [{\citenamefont {Anderson}(1967)}]{anderson}%
  \BibitemOpen
  \bibfield  {author} {\bibinfo {author} {\bibfnamefont {P.~W.}\ \bibnamefont
  {Anderson}},\ }\bibfield  {title} {\enquote {\bibinfo {title} {Infrared
  catastrophe in fermi gases with local scattering potentials},}\ }\href
  {\doibase 10.1103/PhysRevLett.18.1049} {\bibfield  {journal} {\bibinfo
  {journal} {Phys. Rev. Lett.}\ }\textbf {\bibinfo {volume} {18}},\ \bibinfo
  {pages} {1049--1051} (\bibinfo {year} {1967})}\BibitemShut {NoStop}%
\bibitem [{\citenamefont {Lerose}\ \emph {et~al.}(2020)\citenamefont {Lerose},
  \citenamefont {Surace}, \citenamefont {Mazza}, \citenamefont {Perfetto},
  \citenamefont {Collura},\ and\ \citenamefont {Gambassi}}]{falseVac_confine}%
  \BibitemOpen
  \bibfield  {author} {\bibinfo {author} {\bibfnamefont {Alessio}\ \bibnamefont
  {Lerose}}, \bibinfo {author} {\bibfnamefont {Federica~M.}\ \bibnamefont
  {Surace}}, \bibinfo {author} {\bibfnamefont {Paolo~P.}\ \bibnamefont
  {Mazza}}, \bibinfo {author} {\bibfnamefont {Gabriele}\ \bibnamefont
  {Perfetto}}, \bibinfo {author} {\bibfnamefont {Mario}\ \bibnamefont
  {Collura}}, \ and\ \bibinfo {author} {\bibfnamefont {Andrea}\ \bibnamefont
  {Gambassi}},\ }\bibfield  {title} {\enquote {\bibinfo {title} {Quasilocalized
  dynamics from confinement of quantum excitations},}\ }\href {\doibase
  10.1103/PhysRevB.102.041118} {\bibfield  {journal} {\bibinfo  {journal}
  {Phys. Rev. B}\ }\textbf {\bibinfo {volume} {102}},\ \bibinfo {pages}
  {041118} (\bibinfo {year} {2020})}\BibitemShut {NoStop}%
\bibitem [{\citenamefont {Sinha}\ \emph {et~al.}(2021)\citenamefont {Sinha},
  \citenamefont {Chanda},\ and\ \citenamefont {Dziarmaga}}]{Bubble2021}%
  \BibitemOpen
  \bibfield  {author} {\bibinfo {author} {\bibfnamefont {Aritra}\ \bibnamefont
  {Sinha}}, \bibinfo {author} {\bibfnamefont {Titas}\ \bibnamefont {Chanda}}, \
  and\ \bibinfo {author} {\bibfnamefont {Jacek}\ \bibnamefont {Dziarmaga}},\
  }\bibfield  {title} {\enquote {\bibinfo {title} {Nonadiabatic dynamics across
  a first-order quantum phase transition: Quantized bubble nucleation},}\
  }\href {\doibase 10.1103/PhysRevB.103.L220302} {\bibfield  {journal}
  {\bibinfo  {journal} {Phys. Rev. B}\ }\textbf {\bibinfo {volume} {103}},\
  \bibinfo {pages} {L220302} (\bibinfo {year} {2021})}\BibitemShut {NoStop}%
\bibitem [{\citenamefont {Pomponio}\ \emph {et~al.}(2022)\citenamefont
  {Pomponio}, \citenamefont {Werner}, \citenamefont {Zarand},\ and\
  \citenamefont {Takacs}}]{Pomponio22}%
  \BibitemOpen
  \bibfield  {author} {\bibinfo {author} {\bibfnamefont {O.}~\bibnamefont
  {Pomponio}}, \bibinfo {author} {\bibfnamefont {M.~A.}\ \bibnamefont
  {Werner}}, \bibinfo {author} {\bibfnamefont {G.}~\bibnamefont {Zarand}}, \
  and\ \bibinfo {author} {\bibfnamefont {G.}~\bibnamefont {Takacs}},\
  }\bibfield  {title} {\enquote {\bibinfo {title} {{Bloch oscillations and the
  lack of the decay of the false vacuum in a one-dimensional quantum spin
  chain}},}\ }\href {\doibase 10.21468/SciPostPhys.12.2.061} {\bibfield
  {journal} {\bibinfo  {journal} {SciPost Phys.}\ }\textbf {\bibinfo {volume}
  {12}},\ \bibinfo {pages} {061} (\bibinfo {year} {2022})}\BibitemShut
  {NoStop}%
\bibitem [{\citenamefont {Milsted}\ \emph {et~al.}(2022)\citenamefont
  {Milsted}, \citenamefont {Liu}, \citenamefont {Preskill},\ and\ \citenamefont
  {Vidal}}]{Milsted22}%
  \BibitemOpen
  \bibfield  {author} {\bibinfo {author} {\bibfnamefont {Ashley}\ \bibnamefont
  {Milsted}}, \bibinfo {author} {\bibfnamefont {Junyu}\ \bibnamefont {Liu}},
  \bibinfo {author} {\bibfnamefont {John}\ \bibnamefont {Preskill}}, \ and\
  \bibinfo {author} {\bibfnamefont {Guifre}\ \bibnamefont {Vidal}},\ }\bibfield
   {title} {\enquote {\bibinfo {title} {Collisions of false-vacuum bubble walls
  in a quantum spin chain},}\ }\href {\doibase 10.1103/PRXQuantum.3.020316}
  {\bibfield  {journal} {\bibinfo  {journal} {PRX Quantum}\ }\textbf {\bibinfo
  {volume} {3}},\ \bibinfo {pages} {020316} (\bibinfo {year}
  {2022})}\BibitemShut {NoStop}%
\bibitem [{\citenamefont {Lagnese}\ \emph {et~al.}(2024)\citenamefont
  {Lagnese}, \citenamefont {Surace}, \citenamefont {Morampudi},\ and\
  \citenamefont {Wilczek}}]{lagnese2023detecting}%
  \BibitemOpen
  \bibfield  {author} {\bibinfo {author} {\bibfnamefont {Gianluca}\
  \bibnamefont {Lagnese}}, \bibinfo {author} {\bibfnamefont {Federica~Maria}\
  \bibnamefont {Surace}}, \bibinfo {author} {\bibfnamefont {Sid}\ \bibnamefont
  {Morampudi}}, \ and\ \bibinfo {author} {\bibfnamefont {Frank}\ \bibnamefont
  {Wilczek}},\ }\bibfield  {title} {\enquote {\bibinfo {title} {Detecting a
  long-lived false vacuum with quantum quenches},}\ }\href {\doibase
  10.1103/PhysRevLett.133.240402} {\bibfield  {journal} {\bibinfo  {journal}
  {Phys. Rev. Lett.}\ }\textbf {\bibinfo {volume} {133}},\ \bibinfo {pages}
  {240402} (\bibinfo {year} {2024})}\BibitemShut {NoStop}%
\bibitem [{\citenamefont {Lencs\'es}\ \emph {et~al.}(2022)\citenamefont
  {Lencs\'es}, \citenamefont {Mussardo},\ and\ \citenamefont
  {Tak\'acs}}]{FVDTakacs}%
  \BibitemOpen
  \bibfield  {author} {\bibinfo {author} {\bibfnamefont {M.}~\bibnamefont
  {Lencs\'es}}, \bibinfo {author} {\bibfnamefont {G.}~\bibnamefont {Mussardo}},
  \ and\ \bibinfo {author} {\bibfnamefont {G.}~\bibnamefont {Tak\'acs}},\
  }\bibfield  {title} {\enquote {\bibinfo {title} {Variations on vacuum decay:
  The scaling ising and tricritical ising field theories},}\ }\href {\doibase
  10.1103/PhysRevD.106.105003} {\bibfield  {journal} {\bibinfo  {journal}
  {Phys. Rev. D}\ }\textbf {\bibinfo {volume} {106}},\ \bibinfo {pages}
  {105003} (\bibinfo {year} {2022})}\BibitemShut {NoStop}%
\bibitem [{\citenamefont {Sz\'asz-Schagrin}\ and\ \citenamefont
  {Tak\'acs}(2022)}]{FVDphi4}%
  \BibitemOpen
  \bibfield  {author} {\bibinfo {author} {\bibfnamefont {D.}~\bibnamefont
  {Sz\'asz-Schagrin}}\ and\ \bibinfo {author} {\bibfnamefont {G.}~\bibnamefont
  {Tak\'acs}},\ }\bibfield  {title} {\enquote {\bibinfo {title} {False vacuum
  decay in the ($1+1$)-dimensional ${\ensuremath{\varphi}}^{4}$ theory},}\
  }\href {\doibase 10.1103/PhysRevD.106.025008} {\bibfield  {journal} {\bibinfo
   {journal} {Phys. Rev. D}\ }\textbf {\bibinfo {volume} {106}},\ \bibinfo
  {pages} {025008} (\bibinfo {year} {2022})}\BibitemShut {NoStop}%
\bibitem [{\citenamefont {Maki}\ \emph {et~al.}(2023)\citenamefont {Maki},
  \citenamefont {Berti}, \citenamefont {Carusotto},\ and\ \citenamefont
  {Biella}}]{Maki23}%
  \BibitemOpen
  \bibfield  {author} {\bibinfo {author} {\bibfnamefont {Jeffrey~Allan}\
  \bibnamefont {Maki}}, \bibinfo {author} {\bibfnamefont {Anna}\ \bibnamefont
  {Berti}}, \bibinfo {author} {\bibfnamefont {Iacopo}\ \bibnamefont
  {Carusotto}}, \ and\ \bibinfo {author} {\bibfnamefont {Alberto}\ \bibnamefont
  {Biella}},\ }\bibfield  {title} {\enquote {\bibinfo {title} {{Monte Carlo
  matrix-product-state approach to the false vacuum decay in the monitored
  quantum Ising chain}},}\ }\href {\doibase 10.21468/SciPostPhys.15.4.152}
  {\bibfield  {journal} {\bibinfo  {journal} {SciPost Phys.}\ }\textbf
  {\bibinfo {volume} {15}},\ \bibinfo {pages} {152} (\bibinfo {year}
  {2023})}\BibitemShut {NoStop}%
\bibitem [{\citenamefont {Batini}\ \emph {et~al.}(2024)\citenamefont {Batini},
  \citenamefont {Chatrchyan},\ and\ \citenamefont {Berges}}]{Batini}%
  \BibitemOpen
  \bibfield  {author} {\bibinfo {author} {\bibfnamefont {Laura}\ \bibnamefont
  {Batini}}, \bibinfo {author} {\bibfnamefont {Aleksandr}\ \bibnamefont
  {Chatrchyan}}, \ and\ \bibinfo {author} {\bibfnamefont {J\"urgen}\
  \bibnamefont {Berges}},\ }\bibfield  {title} {\enquote {\bibinfo {title}
  {Real-time dynamics of false vacuum decay},}\ }\href {\doibase
  10.1103/PhysRevD.109.023502} {\bibfield  {journal} {\bibinfo  {journal}
  {Phys. Rev. D}\ }\textbf {\bibinfo {volume} {109}},\ \bibinfo {pages}
  {023502} (\bibinfo {year} {2024})}\BibitemShut {NoStop}%
\bibitem [{\citenamefont {Zenesini}\ \emph {et~al.}(2024)\citenamefont
  {Zenesini}, \citenamefont {Berti}, \citenamefont {Cominotti}, \citenamefont
  {Rogora}, \citenamefont {Moss}, \citenamefont {Billam}, \citenamefont
  {Carusotto}, \citenamefont {Lamporesi}, \citenamefont {Recati},\ and\
  \citenamefont {Ferrari}}]{ferrari}%
  \BibitemOpen
  \bibfield  {author} {\bibinfo {author} {\bibfnamefont {A.}~\bibnamefont
  {Zenesini}}, \bibinfo {author} {\bibfnamefont {A.}~\bibnamefont {Berti}},
  \bibinfo {author} {\bibfnamefont {R.}~\bibnamefont {Cominotti}}, \bibinfo
  {author} {\bibfnamefont {C.}~\bibnamefont {Rogora}}, \bibinfo {author}
  {\bibfnamefont {I.~G.}\ \bibnamefont {Moss}}, \bibinfo {author}
  {\bibfnamefont {T.~P.}\ \bibnamefont {Billam}}, \bibinfo {author}
  {\bibfnamefont {I.}~\bibnamefont {Carusotto}}, \bibinfo {author}
  {\bibfnamefont {G.}~\bibnamefont {Lamporesi}}, \bibinfo {author}
  {\bibfnamefont {A.}~\bibnamefont {Recati}}, \ and\ \bibinfo {author}
  {\bibfnamefont {G.}~\bibnamefont {Ferrari}},\ }\bibfield  {title} {\enquote
  {\bibinfo {title} {False vacuum decay via bubble formation in ferromagnetic
  superfluids},}\ }\href@noop {} {\bibfield  {journal} {\bibinfo  {journal}
  {Nature Phys.}\ }\textbf {\bibinfo {volume} {20}},\ \bibinfo {pages} {558}
  (\bibinfo {year} {2024})}\BibitemShut {NoStop}%
\bibitem [{\citenamefont {Darbha}\ \emph {et~al.}(2024)\citenamefont {Darbha},
  \citenamefont {Kornja\ifmmode~\check{c}\else \v{c}\fi{}a}, \citenamefont
  {Liu}, \citenamefont {Balewski}, \citenamefont {Hirsbrunner}, \citenamefont
  {Lopes}, \citenamefont {Wang}, \citenamefont {Van~Beeumen}, \citenamefont
  {Camps},\ and\ \citenamefont {Klymko}}]{darbha2024false}%
  \BibitemOpen
  \bibfield  {author} {\bibinfo {author} {\bibfnamefont {Siva}\ \bibnamefont
  {Darbha}}, \bibinfo {author} {\bibfnamefont {Milan}\ \bibnamefont
  {Kornja\ifmmode~\check{c}\else \v{c}\fi{}a}}, \bibinfo {author}
  {\bibfnamefont {Fangli}\ \bibnamefont {Liu}}, \bibinfo {author}
  {\bibfnamefont {Jan}\ \bibnamefont {Balewski}}, \bibinfo {author}
  {\bibfnamefont {Mark~R.}\ \bibnamefont {Hirsbrunner}}, \bibinfo {author}
  {\bibfnamefont {Pedro L.~S.}\ \bibnamefont {Lopes}}, \bibinfo {author}
  {\bibfnamefont {Sheng-Tao}\ \bibnamefont {Wang}}, \bibinfo {author}
  {\bibfnamefont {Roel}\ \bibnamefont {Van~Beeumen}}, \bibinfo {author}
  {\bibfnamefont {Daan}\ \bibnamefont {Camps}}, \ and\ \bibinfo {author}
  {\bibfnamefont {Katherine}\ \bibnamefont {Klymko}},\ }\bibfield  {title}
  {\enquote {\bibinfo {title} {False vacuum decay and nucleation dynamics in
  neutral atom systems},}\ }\href {\doibase 10.1103/PhysRevB.110.155103}
  {\bibfield  {journal} {\bibinfo  {journal} {Phys. Rev. B}\ }\textbf {\bibinfo
  {volume} {110}},\ \bibinfo {pages} {155103} (\bibinfo {year}
  {2024})}\BibitemShut {NoStop}%
\bibitem [{\citenamefont {Vodeb}\ \emph {et~al.}(2024)\citenamefont {Vodeb},
  \citenamefont {Desaules}, \citenamefont {Hallam}, \citenamefont {Rava},
  \citenamefont {Humar}, \citenamefont {Willsch}, \citenamefont {Jin},
  \citenamefont {Willsch}, \citenamefont {Michielsen},\ and\ \citenamefont
  {Papi\'c}}]{vodeb2024stirring}%
  \BibitemOpen
  \bibfield  {author} {\bibinfo {author} {\bibfnamefont {Jaka}\ \bibnamefont
  {Vodeb}}, \bibinfo {author} {\bibfnamefont {Jean-Yves}\ \bibnamefont
  {Desaules}}, \bibinfo {author} {\bibfnamefont {Andrew}\ \bibnamefont
  {Hallam}}, \bibinfo {author} {\bibfnamefont {Andrea}\ \bibnamefont {Rava}},
  \bibinfo {author} {\bibfnamefont {Gregor}\ \bibnamefont {Humar}}, \bibinfo
  {author} {\bibfnamefont {Dennis}\ \bibnamefont {Willsch}}, \bibinfo {author}
  {\bibfnamefont {Fengping}\ \bibnamefont {Jin}}, \bibinfo {author}
  {\bibfnamefont {Madita}\ \bibnamefont {Willsch}}, \bibinfo {author}
  {\bibfnamefont {Kristel}\ \bibnamefont {Michielsen}}, \ and\ \bibinfo
  {author} {\bibfnamefont {Zlatko}\ \bibnamefont {Papi\'c}},\ }\bibfield
  {title} {\enquote {\bibinfo {title} {{Stirring the false vacuum via
  interacting quantized bubbles on a 5564-qubit quantum annealer}},}\
  }\href@noop {} {\  (\bibinfo {year} {2024})},\ \Eprint
  {http://arxiv.org/abs/2406.14718} {arXiv:2406.14718 [quant-ph]} \BibitemShut
  {NoStop}%
\bibitem [{\citenamefont {Deutsch}(1991)}]{ETH_91}%
  \BibitemOpen
  \bibfield  {author} {\bibinfo {author} {\bibfnamefont {J.~M.}\ \bibnamefont
  {Deutsch}},\ }\bibfield  {title} {\enquote {\bibinfo {title} {Quantum
  statistical mechanics in a closed system},}\ }\href {\doibase
  10.1103/PhysRevA.43.2046} {\bibfield  {journal} {\bibinfo  {journal} {Phys.
  Rev. A}\ }\textbf {\bibinfo {volume} {43}},\ \bibinfo {pages} {2046--2049}
  (\bibinfo {year} {1991})}\BibitemShut {NoStop}%
\bibitem [{\citenamefont {Srednicki}(1994)}]{ETH_94}%
  \BibitemOpen
  \bibfield  {author} {\bibinfo {author} {\bibfnamefont {Mark}\ \bibnamefont
  {Srednicki}},\ }\bibfield  {title} {\enquote {\bibinfo {title} {Chaos and
  quantum thermalization},}\ }\href {\doibase 10.1103/PhysRevE.50.888}
  {\bibfield  {journal} {\bibinfo  {journal} {Phys. Rev. E}\ }\textbf {\bibinfo
  {volume} {50}},\ \bibinfo {pages} {888--901} (\bibinfo {year}
  {1994})}\BibitemShut {NoStop}%
\bibitem [{\citenamefont {D'Alessio}\ \emph {et~al.}(2016)\citenamefont
  {D'Alessio}, \citenamefont {Kafri}, \citenamefont {Polkovnikov},\ and\
  \citenamefont {Rigol}}]{ETH_rev16}%
  \BibitemOpen
  \bibfield  {author} {\bibinfo {author} {\bibfnamefont {Luca}\ \bibnamefont
  {D'Alessio}}, \bibinfo {author} {\bibfnamefont {Yariv}\ \bibnamefont
  {Kafri}}, \bibinfo {author} {\bibfnamefont {Anatoli}\ \bibnamefont
  {Polkovnikov}}, \ and\ \bibinfo {author} {\bibfnamefont {Marcos}\
  \bibnamefont {Rigol}},\ }\bibfield  {title} {\enquote {\bibinfo {title} {From
  quantum chaos and eigenstate thermalization to statistical mechanics and
  thermodynamics},}\ }\href {\doibase 10.1080/00018732.2016.1198134} {\bibfield
   {journal} {\bibinfo  {journal} {Advances in Physics}\ }\textbf {\bibinfo
  {volume} {65}},\ \bibinfo {pages} {239--362} (\bibinfo {year}
  {2016})}\BibitemShut {NoStop}%
\bibitem [{\citenamefont {Mori}\ \emph {et~al.}(2018)\citenamefont {Mori},
  \citenamefont {Ikeda}, \citenamefont {Kaminishi},\ and\ \citenamefont
  {Ueda}}]{ETH_rev18}%
  \BibitemOpen
  \bibfield  {author} {\bibinfo {author} {\bibfnamefont {Takashi}\ \bibnamefont
  {Mori}}, \bibinfo {author} {\bibfnamefont {Tatsuhiko~N}\ \bibnamefont
  {Ikeda}}, \bibinfo {author} {\bibfnamefont {Eriko}\ \bibnamefont
  {Kaminishi}}, \ and\ \bibinfo {author} {\bibfnamefont {Masahito}\
  \bibnamefont {Ueda}},\ }\bibfield  {title} {\enquote {\bibinfo {title}
  {Thermalization and prethermalization in isolated quantum systems: a
  theoretical overview},}\ }\href {\doibase 10.1088/1361-6455/aabcdf}
  {\bibfield  {journal} {\bibinfo  {journal} {Journal of Physics B: Atomic,
  Molecular and Optical Physics}\ }\textbf {\bibinfo {volume} {51}},\ \bibinfo
  {pages} {112001} (\bibinfo {year} {2018})}\BibitemShut {NoStop}%
\bibitem [{\citenamefont {Abanin}\ \emph {et~al.}(2015)\citenamefont {Abanin},
  \citenamefont {De~Roeck},\ and\ \citenamefont {Huveneers}}]{deroeckprl}%
  \BibitemOpen
  \bibfield  {author} {\bibinfo {author} {\bibfnamefont {Dmitry~A.}\
  \bibnamefont {Abanin}}, \bibinfo {author} {\bibfnamefont {Wojciech}\
  \bibnamefont {De~Roeck}}, \ and\ \bibinfo {author} {\bibfnamefont
  {Fran\ifmmode \mbox{\c{c}}\else~\c{c}\fi{}ois}\ \bibnamefont {Huveneers}},\
  }\bibfield  {title} {\enquote {\bibinfo {title} {Exponentially slow heating
  in periodically driven many-body systems},}\ }\href {\doibase
  10.1103/PhysRevLett.115.256803} {\bibfield  {journal} {\bibinfo  {journal}
  {Phys. Rev. Lett.}\ }\textbf {\bibinfo {volume} {115}},\ \bibinfo {pages}
  {256803} (\bibinfo {year} {2015})}\BibitemShut {NoStop}%
\bibitem [{\citenamefont {Abanin}\ \emph
  {et~al.}(2017{\natexlab{a}})\citenamefont {Abanin}, \citenamefont {De~Roeck},
  \citenamefont {Ho},\ and\ \citenamefont {Huveneers}}]{abanin2017rigorous}%
  \BibitemOpen
  \bibfield  {author} {\bibinfo {author} {\bibfnamefont {Dmitry}\ \bibnamefont
  {Abanin}}, \bibinfo {author} {\bibfnamefont {Wojciech}\ \bibnamefont
  {De~Roeck}}, \bibinfo {author} {\bibfnamefont {Wen~Wei}\ \bibnamefont {Ho}},
  \ and\ \bibinfo {author} {\bibfnamefont {Fran{\c{c}}ois}\ \bibnamefont
  {Huveneers}},\ }\bibfield  {title} {\enquote {\bibinfo {title} {A rigorous
  theory of many-body prethermalization for periodically driven and closed
  quantum systems},}\ }\href {\doibase
  https://doi.org/10.1007/s00220-017-2930-x} {\bibfield  {journal} {\bibinfo
  {journal} {Communications in Mathematical Physics}\ }\textbf {\bibinfo
  {volume} {354}},\ \bibinfo {pages} {809--827} (\bibinfo {year}
  {2017}{\natexlab{a}})}\BibitemShut {NoStop}%
\bibitem [{\citenamefont {Kuwahara}\ \emph {et~al.}(2016)\citenamefont
  {Kuwahara}, \citenamefont {Mori},\ and\ \citenamefont {Saito}}]{Floq_KS_16}%
  \BibitemOpen
  \bibfield  {author} {\bibinfo {author} {\bibfnamefont {Tomotaka}\
  \bibnamefont {Kuwahara}}, \bibinfo {author} {\bibfnamefont {Takashi}\
  \bibnamefont {Mori}}, \ and\ \bibinfo {author} {\bibfnamefont {Keiji}\
  \bibnamefont {Saito}},\ }\bibfield  {title} {\enquote {\bibinfo {title}
  {Floquet–magnus theory and generic transient dynamics in periodically
  driven many-body quantum systems},}\ }\href {\doibase
  https://doi.org/10.1016/j.aop.2016.01.012} {\bibfield  {journal} {\bibinfo
  {journal} {Annals of Physics}\ }\textbf {\bibinfo {volume} {367}},\ \bibinfo
  {pages} {96--124} (\bibinfo {year} {2016})}\BibitemShut {NoStop}%
\bibitem [{\citenamefont {Mori}\ \emph {et~al.}(2016)\citenamefont {Mori},
  \citenamefont {Kuwahara},\ and\ \citenamefont {Saito}}]{Floq_KS_PRL}%
  \BibitemOpen
  \bibfield  {author} {\bibinfo {author} {\bibfnamefont {Takashi}\ \bibnamefont
  {Mori}}, \bibinfo {author} {\bibfnamefont {Tomotaka}\ \bibnamefont
  {Kuwahara}}, \ and\ \bibinfo {author} {\bibfnamefont {Keiji}\ \bibnamefont
  {Saito}},\ }\bibfield  {title} {\enquote {\bibinfo {title} {Rigorous bound on
  energy absorption and generic relaxation in periodically driven quantum
  systems},}\ }\href {\doibase 10.1103/PhysRevLett.116.120401} {\bibfield
  {journal} {\bibinfo  {journal} {Phys. Rev. Lett.}\ }\textbf {\bibinfo
  {volume} {116}},\ \bibinfo {pages} {120401} (\bibinfo {year}
  {2016})}\BibitemShut {NoStop}%
\bibitem [{\citenamefont {Abanin}\ \emph
  {et~al.}(2017{\natexlab{b}})\citenamefont {Abanin}, \citenamefont {De~Roeck},
  \citenamefont {Ho},\ and\ \citenamefont {Huveneers}}]{Floq_PRB}%
  \BibitemOpen
  \bibfield  {author} {\bibinfo {author} {\bibfnamefont {Dmitry~A.}\
  \bibnamefont {Abanin}}, \bibinfo {author} {\bibfnamefont {Wojciech}\
  \bibnamefont {De~Roeck}}, \bibinfo {author} {\bibfnamefont {Wen~Wei}\
  \bibnamefont {Ho}}, \ and\ \bibinfo {author} {\bibfnamefont {Fran\ifmmode
  \mbox{\c{c}}\else~\c{c}\fi{}ois}\ \bibnamefont {Huveneers}},\ }\bibfield
  {title} {\enquote {\bibinfo {title} {Effective hamiltonians,
  prethermalization, and slow energy absorption in periodically driven
  many-body systems},}\ }\href {\doibase 10.1103/PhysRevB.95.014112} {\bibfield
   {journal} {\bibinfo  {journal} {Phys. Rev. B}\ }\textbf {\bibinfo {volume}
  {95}},\ \bibinfo {pages} {014112} (\bibinfo {year}
  {2017}{\natexlab{b}})}\BibitemShut {NoStop}%
\bibitem [{\citenamefont {Mallayya}\ \emph {et~al.}(2019)\citenamefont
  {Mallayya}, \citenamefont {Rigol},\ and\ \citenamefont
  {De~Roeck}}]{preth_PRX19}%
  \BibitemOpen
  \bibfield  {author} {\bibinfo {author} {\bibfnamefont {Krishnanand}\
  \bibnamefont {Mallayya}}, \bibinfo {author} {\bibfnamefont {Marcos}\
  \bibnamefont {Rigol}}, \ and\ \bibinfo {author} {\bibfnamefont {Wojciech}\
  \bibnamefont {De~Roeck}},\ }\bibfield  {title} {\enquote {\bibinfo {title}
  {Prethermalization and thermalization in isolated quantum systems},}\ }\href
  {\doibase 10.1103/PhysRevX.9.021027} {\bibfield  {journal} {\bibinfo
  {journal} {Phys. Rev. X}\ }\textbf {\bibinfo {volume} {9}},\ \bibinfo {pages}
  {021027} (\bibinfo {year} {2019})}\BibitemShut {NoStop}%
\bibitem [{\citenamefont {De~Roeck}\ and\ \citenamefont
  {Verreet}(2019)}]{roeck2019weak_driven}%
  \BibitemOpen
  \bibfield  {author} {\bibinfo {author} {\bibfnamefont {Wojciech}\
  \bibnamefont {De~Roeck}}\ and\ \bibinfo {author} {\bibfnamefont {Victor}\
  \bibnamefont {Verreet}},\ }\bibfield  {title} {\enquote {\bibinfo {title}
  {Very slow heating for weakly driven quantum many-body systems},}\
  }\href@noop {} {\bibfield  {journal} {\bibinfo  {journal} {arXiv preprint
  arXiv:1911.01998}\ } (\bibinfo {year} {2019})}\BibitemShut {NoStop}%
\bibitem [{\citenamefont {Luitz}\ \emph {et~al.}(2020)\citenamefont {Luitz},
  \citenamefont {Moessner}, \citenamefont {Sondhi},\ and\ \citenamefont
  {Khemani}}]{preth_Kehmani20}%
  \BibitemOpen
  \bibfield  {author} {\bibinfo {author} {\bibfnamefont {David~J.}\
  \bibnamefont {Luitz}}, \bibinfo {author} {\bibfnamefont {Roderich}\
  \bibnamefont {Moessner}}, \bibinfo {author} {\bibfnamefont {S.~L.}\
  \bibnamefont {Sondhi}}, \ and\ \bibinfo {author} {\bibfnamefont {Vedika}\
  \bibnamefont {Khemani}},\ }\bibfield  {title} {\enquote {\bibinfo {title}
  {Prethermalization without temperature},}\ }\href {\doibase
  10.1103/PhysRevX.10.021046} {\bibfield  {journal} {\bibinfo  {journal} {Phys.
  Rev. X}\ }\textbf {\bibinfo {volume} {10}},\ \bibinfo {pages} {021046}
  (\bibinfo {year} {2020})}\BibitemShut {NoStop}%
\bibitem [{\citenamefont {Else}\ \emph {et~al.}(2020)\citenamefont {Else},
  \citenamefont {Ho},\ and\ \citenamefont {Dumitrescu}}]{quasiperiodic}%
  \BibitemOpen
  \bibfield  {author} {\bibinfo {author} {\bibfnamefont {Dominic~V.}\
  \bibnamefont {Else}}, \bibinfo {author} {\bibfnamefont {Wen~Wei}\
  \bibnamefont {Ho}}, \ and\ \bibinfo {author} {\bibfnamefont {Philipp~T.}\
  \bibnamefont {Dumitrescu}},\ }\bibfield  {title} {\enquote {\bibinfo {title}
  {Long-lived interacting phases of matter protected by multiple
  time-translation symmetries in quasiperiodically driven systems},}\ }\href
  {\doibase 10.1103/PhysRevX.10.021032} {\bibfield  {journal} {\bibinfo
  {journal} {Phys. Rev. X}\ }\textbf {\bibinfo {volume} {10}},\ \bibinfo
  {pages} {021032} (\bibinfo {year} {2020})}\BibitemShut {NoStop}%
\bibitem [{\citenamefont {Surace}\ and\ \citenamefont
  {Motrunich}(2023)}]{SuraceMotrunich}%
  \BibitemOpen
  \bibfield  {author} {\bibinfo {author} {\bibfnamefont {Federica~Maria}\
  \bibnamefont {Surace}}\ and\ \bibinfo {author} {\bibfnamefont {Olexei}\
  \bibnamefont {Motrunich}},\ }\bibfield  {title} {\enquote {\bibinfo {title}
  {Weak integrability breaking perturbations of integrable models},}\ }\href
  {\doibase 10.1103/PhysRevResearch.5.043019} {\bibfield  {journal} {\bibinfo
  {journal} {Phys. Rev. Res.}\ }\textbf {\bibinfo {volume} {5}},\ \bibinfo
  {pages} {043019} (\bibinfo {year} {2023})}\BibitemShut {NoStop}%
\bibitem [{\citenamefont {Bernien}\ \emph {et~al.}(2017)\citenamefont {Bernien}
  \emph {et~al.}}]{PXP_exp17}%
  \BibitemOpen
  \bibfield  {author} {\bibinfo {author} {\bibfnamefont {Hannes}\ \bibnamefont
  {Bernien}} \emph {et~al.},\ }\bibfield  {title} {\enquote {\bibinfo {title}
  {{Probing many-body dynamics on a 51-atom quantum simulator}},}\ }\href
  {\doibase 10.1038/nature24622} {\bibfield  {journal} {\bibinfo  {journal}
  {Nature}\ }\textbf {\bibinfo {volume} {551}},\ \bibinfo {pages} {579--584}
  (\bibinfo {year} {2017})},\ \Eprint {http://arxiv.org/abs/1707.04344}
  {arXiv:1707.04344 [quant-ph]} \BibitemShut {NoStop}%
\bibitem [{\citenamefont {Shiraishi}\ and\ \citenamefont
  {Mori}(2017)}]{scar_Mori17}%
  \BibitemOpen
  \bibfield  {author} {\bibinfo {author} {\bibfnamefont {Naoto}\ \bibnamefont
  {Shiraishi}}\ and\ \bibinfo {author} {\bibfnamefont {Takashi}\ \bibnamefont
  {Mori}},\ }\bibfield  {title} {\enquote {\bibinfo {title} {Systematic
  construction of counterexamples to the eigenstate thermalization
  hypothesis},}\ }\href {\doibase 10.1103/PhysRevLett.119.030601} {\bibfield
  {journal} {\bibinfo  {journal} {Phys. Rev. Lett.}\ }\textbf {\bibinfo
  {volume} {119}},\ \bibinfo {pages} {030601} (\bibinfo {year}
  {2017})}\BibitemShut {NoStop}%
\bibitem [{\citenamefont {Turner}\ \emph
  {et~al.}(2018{\natexlab{a}})\citenamefont {Turner}, \citenamefont
  {Michailidis}, \citenamefont {Abanin}, \citenamefont {Serbyn},\ and\
  \citenamefont {Papi{\'c}}}]{PXP_NP18}%
  \BibitemOpen
  \bibfield  {author} {\bibinfo {author} {\bibfnamefont {Christopher~J}\
  \bibnamefont {Turner}}, \bibinfo {author} {\bibfnamefont {Alexios~A}\
  \bibnamefont {Michailidis}}, \bibinfo {author} {\bibfnamefont {Dmitry~A}\
  \bibnamefont {Abanin}}, \bibinfo {author} {\bibfnamefont {Maksym}\
  \bibnamefont {Serbyn}}, \ and\ \bibinfo {author} {\bibfnamefont {Zlatko}\
  \bibnamefont {Papi{\'c}}},\ }\bibfield  {title} {\enquote {\bibinfo {title}
  {Weak ergodicity breaking from quantum many-body scars},}\ }\href {\doibase
  https://doi.org/10.1038/s41567-018-0137-5} {\bibfield  {journal} {\bibinfo
  {journal} {Nature Physics}\ }\textbf {\bibinfo {volume} {14}},\ \bibinfo
  {pages} {745--749} (\bibinfo {year} {2018}{\natexlab{a}})}\BibitemShut
  {NoStop}%
\bibitem [{\citenamefont {Moudgalya}\ \emph
  {et~al.}(2018{\natexlab{a}})\citenamefont {Moudgalya}, \citenamefont
  {Rachel}, \citenamefont {Bernevig},\ and\ \citenamefont
  {Regnault}}]{scar_exact18}%
  \BibitemOpen
  \bibfield  {author} {\bibinfo {author} {\bibfnamefont {Sanjay}\ \bibnamefont
  {Moudgalya}}, \bibinfo {author} {\bibfnamefont {Stephan}\ \bibnamefont
  {Rachel}}, \bibinfo {author} {\bibfnamefont {B.~Andrei}\ \bibnamefont
  {Bernevig}}, \ and\ \bibinfo {author} {\bibfnamefont {Nicolas}\ \bibnamefont
  {Regnault}},\ }\bibfield  {title} {\enquote {\bibinfo {title} {Exact excited
  states of nonintegrable models},}\ }\href {\doibase
  10.1103/PhysRevB.98.235155} {\bibfield  {journal} {\bibinfo  {journal} {Phys.
  Rev. B}\ }\textbf {\bibinfo {volume} {98}},\ \bibinfo {pages} {235155}
  (\bibinfo {year} {2018}{\natexlab{a}})}\BibitemShut {NoStop}%
\bibitem [{\citenamefont {Moudgalya}\ \emph
  {et~al.}(2018{\natexlab{b}})\citenamefont {Moudgalya}, \citenamefont
  {Regnault},\ and\ \citenamefont {Bernevig}}]{scar_exact18_1}%
  \BibitemOpen
  \bibfield  {author} {\bibinfo {author} {\bibfnamefont {Sanjay}\ \bibnamefont
  {Moudgalya}}, \bibinfo {author} {\bibfnamefont {Nicolas}\ \bibnamefont
  {Regnault}}, \ and\ \bibinfo {author} {\bibfnamefont {B.~Andrei}\
  \bibnamefont {Bernevig}},\ }\bibfield  {title} {\enquote {\bibinfo {title}
  {Entanglement of exact excited states of affleck-kennedy-lieb-tasaki models:
  Exact results, many-body scars, and violation of the strong eigenstate
  thermalization hypothesis},}\ }\href {\doibase 10.1103/PhysRevB.98.235156}
  {\bibfield  {journal} {\bibinfo  {journal} {Phys. Rev. B}\ }\textbf {\bibinfo
  {volume} {98}},\ \bibinfo {pages} {235156} (\bibinfo {year}
  {2018}{\natexlab{b}})}\BibitemShut {NoStop}%
\bibitem [{\citenamefont {Lin}\ and\ \citenamefont
  {Motrunich}(2019)}]{PXP_exact19}%
  \BibitemOpen
  \bibfield  {author} {\bibinfo {author} {\bibfnamefont {Cheng-Ju}\
  \bibnamefont {Lin}}\ and\ \bibinfo {author} {\bibfnamefont {Olexei~I.}\
  \bibnamefont {Motrunich}},\ }\bibfield  {title} {\enquote {\bibinfo {title}
  {Exact quantum many-body scar states in the rydberg-blockaded atom chain},}\
  }\href {\doibase 10.1103/PhysRevLett.122.173401} {\bibfield  {journal}
  {\bibinfo  {journal} {Phys. Rev. Lett.}\ }\textbf {\bibinfo {volume} {122}},\
  \bibinfo {pages} {173401} (\bibinfo {year} {2019})}\BibitemShut {NoStop}%
\bibitem [{\citenamefont {Schecter}\ and\ \citenamefont
  {Iadecola}(2019)}]{scar_XY19}%
  \BibitemOpen
  \bibfield  {author} {\bibinfo {author} {\bibfnamefont {Michael}\ \bibnamefont
  {Schecter}}\ and\ \bibinfo {author} {\bibfnamefont {Thomas}\ \bibnamefont
  {Iadecola}},\ }\bibfield  {title} {\enquote {\bibinfo {title} {Weak
  ergodicity breaking and quantum many-body scars in spin-1 $xy$ magnets},}\
  }\href {\doibase 10.1103/PhysRevLett.123.147201} {\bibfield  {journal}
  {\bibinfo  {journal} {Phys. Rev. Lett.}\ }\textbf {\bibinfo {volume} {123}},\
  \bibinfo {pages} {147201} (\bibinfo {year} {2019})}\BibitemShut {NoStop}%
\bibitem [{\citenamefont {Wang}\ \emph {et~al.}(2024)\citenamefont {Wang},
  \citenamefont {Yuan}, \citenamefont {Zhang}, \citenamefont {Wang},
  \citenamefont {Deng},\ and\ \citenamefont {Duan}}]{scar_decofree24}%
  \BibitemOpen
  \bibfield  {author} {\bibinfo {author} {\bibfnamefont {He-Ran}\ \bibnamefont
  {Wang}}, \bibinfo {author} {\bibfnamefont {Dong}\ \bibnamefont {Yuan}},
  \bibinfo {author} {\bibfnamefont {Shun-Yao}\ \bibnamefont {Zhang}}, \bibinfo
  {author} {\bibfnamefont {Zhong}\ \bibnamefont {Wang}}, \bibinfo {author}
  {\bibfnamefont {Dong-Ling}\ \bibnamefont {Deng}}, \ and\ \bibinfo {author}
  {\bibfnamefont {L.-M.}\ \bibnamefont {Duan}},\ }\bibfield  {title} {\enquote
  {\bibinfo {title} {Embedding quantum many-body scars into decoherence-free
  subspaces},}\ }\href {\doibase 10.1103/PhysRevLett.132.150401} {\bibfield
  {journal} {\bibinfo  {journal} {Phys. Rev. Lett.}\ }\textbf {\bibinfo
  {volume} {132}},\ \bibinfo {pages} {150401} (\bibinfo {year}
  {2024})}\BibitemShut {NoStop}%
\bibitem [{\citenamefont {Moudgalya}\ \emph {et~al.}(2022)\citenamefont
  {Moudgalya}, \citenamefont {Bernevig},\ and\ \citenamefont
  {Regnault}}]{scar_HSF_rev22}%
  \BibitemOpen
  \bibfield  {author} {\bibinfo {author} {\bibfnamefont {Sanjay}\ \bibnamefont
  {Moudgalya}}, \bibinfo {author} {\bibfnamefont {B~Andrei}\ \bibnamefont
  {Bernevig}}, \ and\ \bibinfo {author} {\bibfnamefont {Nicolas}\ \bibnamefont
  {Regnault}},\ }\bibfield  {title} {\enquote {\bibinfo {title} {Quantum
  many-body scars and hilbert space fragmentation: a review of exact
  results},}\ }\href {\doibase 10.1088/1361-6633/ac73a0} {\bibfield  {journal}
  {\bibinfo  {journal} {Reports on Progress in Physics}\ }\textbf {\bibinfo
  {volume} {85}},\ \bibinfo {pages} {086501} (\bibinfo {year}
  {2022})}\BibitemShut {NoStop}%
\bibitem [{\citenamefont {Serbyn}\ \emph {et~al.}(2021)\citenamefont {Serbyn},
  \citenamefont {Abanin},\ and\ \citenamefont {Papi\'c}}]{PXP_rev_Papic}%
  \BibitemOpen
  \bibfield  {author} {\bibinfo {author} {\bibfnamefont {Maksym}\ \bibnamefont
  {Serbyn}}, \bibinfo {author} {\bibfnamefont {Dmitry~A.}\ \bibnamefont
  {Abanin}}, \ and\ \bibinfo {author} {\bibfnamefont {Zlatko}\ \bibnamefont
  {Papi\'c}},\ }\bibfield  {title} {\enquote {\bibinfo {title} {{Quantum
  many-body scars and weak breaking of ergodicity}},}\ }\href {\doibase
  10.1038/s41567-021-01230-2} {\bibfield  {journal} {\bibinfo  {journal}
  {Nature Phys.}\ }\textbf {\bibinfo {volume} {17}},\ \bibinfo {pages}
  {675--685} (\bibinfo {year} {2021})},\ \Eprint
  {http://arxiv.org/abs/2011.09486} {arXiv:2011.09486 [quant-ph]} \BibitemShut
  {NoStop}%
\bibitem [{\citenamefont {Chandran}\ \emph {et~al.}(2023)\citenamefont
  {Chandran}, \citenamefont {Iadecola}, \citenamefont {Khemani},\ and\
  \citenamefont {Moessner}}]{PXP_rev_Kehmani}%
  \BibitemOpen
  \bibfield  {author} {\bibinfo {author} {\bibfnamefont {Anushya}\ \bibnamefont
  {Chandran}}, \bibinfo {author} {\bibfnamefont {Thomas}\ \bibnamefont
  {Iadecola}}, \bibinfo {author} {\bibfnamefont {Vedika}\ \bibnamefont
  {Khemani}}, \ and\ \bibinfo {author} {\bibfnamefont {Roderich}\ \bibnamefont
  {Moessner}},\ }\bibfield  {title} {\enquote {\bibinfo {title} {Quantum
  many-body scars: A quasiparticle perspective},}\ }\href {\doibase
  https://doi.org/10.1146/annurev-conmatphys-031620-101617} {\bibfield
  {journal} {\bibinfo  {journal} {Annual Review of Condensed Matter Physics}\
  }\textbf {\bibinfo {volume} {14}},\ \bibinfo {pages} {443--469} (\bibinfo
  {year} {2023})}\BibitemShut {NoStop}%
\bibitem [{\citenamefont {Lin}\ \emph {et~al.}(2020{\natexlab{a}})\citenamefont
  {Lin}, \citenamefont {Chandran},\ and\ \citenamefont
  {Motrunich}}]{LinChandranMotrunich}%
  \BibitemOpen
  \bibfield  {author} {\bibinfo {author} {\bibfnamefont {Cheng-Ju}\
  \bibnamefont {Lin}}, \bibinfo {author} {\bibfnamefont {Anushya}\ \bibnamefont
  {Chandran}}, \ and\ \bibinfo {author} {\bibfnamefont {Olexei~I.}\
  \bibnamefont {Motrunich}},\ }\bibfield  {title} {\enquote {\bibinfo {title}
  {Slow thermalization of exact quantum many-body scar states under
  perturbations},}\ }\href {\doibase 10.1103/PhysRevResearch.2.033044}
  {\bibfield  {journal} {\bibinfo  {journal} {Phys. Rev. Res.}\ }\textbf
  {\bibinfo {volume} {2}},\ \bibinfo {pages} {033044} (\bibinfo {year}
  {2020}{\natexlab{a}})}\BibitemShut {NoStop}%
\bibitem [{\citenamefont {Surace}\ \emph
  {et~al.}(2021{\natexlab{a}})\citenamefont {Surace}, \citenamefont {Votto},
  \citenamefont {Lazo}, \citenamefont {Silva}, \citenamefont {Dalmonte},\ and\
  \citenamefont {Giudici}}]{Surace2021}%
  \BibitemOpen
  \bibfield  {author} {\bibinfo {author} {\bibfnamefont {Federica~Maria}\
  \bibnamefont {Surace}}, \bibinfo {author} {\bibfnamefont {Matteo}\
  \bibnamefont {Votto}}, \bibinfo {author} {\bibfnamefont {Eduardo~Gonzalez}\
  \bibnamefont {Lazo}}, \bibinfo {author} {\bibfnamefont {Alessandro}\
  \bibnamefont {Silva}}, \bibinfo {author} {\bibfnamefont {Marcello}\
  \bibnamefont {Dalmonte}}, \ and\ \bibinfo {author} {\bibfnamefont {Giuliano}\
  \bibnamefont {Giudici}},\ }\bibfield  {title} {\enquote {\bibinfo {title}
  {Exact many-body scars and their stability in constrained quantum chains},}\
  }\href {\doibase 10.1103/PhysRevB.103.104302} {\bibfield  {journal} {\bibinfo
   {journal} {Phys. Rev. B}\ }\textbf {\bibinfo {volume} {103}},\ \bibinfo
  {pages} {104302} (\bibinfo {year} {2021}{\natexlab{a}})}\BibitemShut
  {NoStop}%
\bibitem [{\citenamefont {Dooley}(2021)}]{scar_metro21}%
  \BibitemOpen
  \bibfield  {author} {\bibinfo {author} {\bibfnamefont {Shane}\ \bibnamefont
  {Dooley}},\ }\bibfield  {title} {\enquote {\bibinfo {title} {Robust quantum
  sensing in strongly interacting systems with many-body scars},}\ }\href
  {\doibase 10.1103/PRXQuantum.2.020330} {\bibfield  {journal} {\bibinfo
  {journal} {PRX Quantum}\ }\textbf {\bibinfo {volume} {2}},\ \bibinfo {pages}
  {020330} (\bibinfo {year} {2021})}\BibitemShut {NoStop}%
\bibitem [{\citenamefont {Sala}\ \emph {et~al.}(2020)\citenamefont {Sala},
  \citenamefont {Rakovszky}, \citenamefont {Verresen}, \citenamefont {Knap},\
  and\ \citenamefont {Pollmann}}]{shatter1}%
  \BibitemOpen
  \bibfield  {author} {\bibinfo {author} {\bibfnamefont {Pablo}\ \bibnamefont
  {Sala}}, \bibinfo {author} {\bibfnamefont {Tibor}\ \bibnamefont {Rakovszky}},
  \bibinfo {author} {\bibfnamefont {Ruben}\ \bibnamefont {Verresen}}, \bibinfo
  {author} {\bibfnamefont {Michael}\ \bibnamefont {Knap}}, \ and\ \bibinfo
  {author} {\bibfnamefont {Frank}\ \bibnamefont {Pollmann}},\ }\bibfield
  {title} {\enquote {\bibinfo {title} {Ergodicity breaking arising from hilbert
  space fragmentation in dipole-conserving hamiltonians},}\ }\href {\doibase
  10.1103/PhysRevX.10.011047} {\bibfield  {journal} {\bibinfo  {journal} {Phys.
  Rev. X}\ }\textbf {\bibinfo {volume} {10}},\ \bibinfo {pages} {011047}
  (\bibinfo {year} {2020})}\BibitemShut {NoStop}%
\bibitem [{\citenamefont {Khemani}\ \emph {et~al.}(2020)\citenamefont
  {Khemani}, \citenamefont {Hermele},\ and\ \citenamefont
  {Nandkishore}}]{shatter2}%
  \BibitemOpen
  \bibfield  {author} {\bibinfo {author} {\bibfnamefont {Vedika}\ \bibnamefont
  {Khemani}}, \bibinfo {author} {\bibfnamefont {Michael}\ \bibnamefont
  {Hermele}}, \ and\ \bibinfo {author} {\bibfnamefont {Rahul}\ \bibnamefont
  {Nandkishore}},\ }\bibfield  {title} {\enquote {\bibinfo {title}
  {Localization from hilbert space shattering: From theory to physical
  realizations},}\ }\href {\doibase 10.1103/PhysRevB.101.174204} {\bibfield
  {journal} {\bibinfo  {journal} {Phys. Rev. B}\ }\textbf {\bibinfo {volume}
  {101}},\ \bibinfo {pages} {174204} (\bibinfo {year} {2020})}\BibitemShut
  {NoStop}%
\bibitem [{\citenamefont {Moudgalya}\ \emph {et~al.}(2021)\citenamefont
  {Moudgalya}, \citenamefont {Prem}, \citenamefont {Nandkishore}, \citenamefont
  {Regnault},\ and\ \citenamefont {Bernevig}}]{frag_Rahul_Bern}%
  \BibitemOpen
  \bibfield  {author} {\bibinfo {author} {\bibfnamefont {Sanjay}\ \bibnamefont
  {Moudgalya}}, \bibinfo {author} {\bibfnamefont {Abhinav}\ \bibnamefont
  {Prem}}, \bibinfo {author} {\bibfnamefont {Rahul}\ \bibnamefont
  {Nandkishore}}, \bibinfo {author} {\bibfnamefont {Nicolas}\ \bibnamefont
  {Regnault}}, \ and\ \bibinfo {author} {\bibfnamefont {B.~Andrei}\
  \bibnamefont {Bernevig}},\ }\enquote {\bibinfo {title} {Thermalization and
  its absence within krylov subspaces of a constrained hamiltonian},}\ in\
  \href {\doibase 10.1142/9789811231711_0009} {\emph {\bibinfo {booktitle}
  {Memorial Volume for Shoucheng Zhang}}}\ (\bibinfo  {publisher} {WORLD
  SCIENTIFIC},\ \bibinfo {year} {2021})\ p.\ \bibinfo {pages}
  {147–209}\BibitemShut {NoStop}%
\bibitem [{\citenamefont {Yang}\ \emph {et~al.}(2020)\citenamefont {Yang},
  \citenamefont {Liu}, \citenamefont {Gorshkov},\ and\ \citenamefont
  {Iadecola}}]{iadecola}%
  \BibitemOpen
  \bibfield  {author} {\bibinfo {author} {\bibfnamefont {Zhi-Cheng}\
  \bibnamefont {Yang}}, \bibinfo {author} {\bibfnamefont {Fangli}\ \bibnamefont
  {Liu}}, \bibinfo {author} {\bibfnamefont {Alexey~V.}\ \bibnamefont
  {Gorshkov}}, \ and\ \bibinfo {author} {\bibfnamefont {Thomas}\ \bibnamefont
  {Iadecola}},\ }\bibfield  {title} {\enquote {\bibinfo {title} {Hilbert-space
  fragmentation from strict confinement},}\ }\href {\doibase
  10.1103/PhysRevLett.124.207602} {\bibfield  {journal} {\bibinfo  {journal}
  {Phys. Rev. Lett.}\ }\textbf {\bibinfo {volume} {124}},\ \bibinfo {pages}
  {207602} (\bibinfo {year} {2020})}\BibitemShut {NoStop}%
\bibitem [{\citenamefont {Stephen}\ \emph {et~al.}(2024)\citenamefont
  {Stephen}, \citenamefont {Hart},\ and\ \citenamefont
  {Nandkishore}}]{loop_David24}%
  \BibitemOpen
  \bibfield  {author} {\bibinfo {author} {\bibfnamefont {David~T.}\
  \bibnamefont {Stephen}}, \bibinfo {author} {\bibfnamefont {Oliver}\
  \bibnamefont {Hart}}, \ and\ \bibinfo {author} {\bibfnamefont {Rahul~M.}\
  \bibnamefont {Nandkishore}},\ }\bibfield  {title} {\enquote {\bibinfo {title}
  {Ergodicity breaking provably robust to arbitrary perturbations},}\ }\href
  {\doibase 10.1103/PhysRevLett.132.040401} {\bibfield  {journal} {\bibinfo
  {journal} {Phys. Rev. Lett.}\ }\textbf {\bibinfo {volume} {132}},\ \bibinfo
  {pages} {040401} (\bibinfo {year} {2024})}\BibitemShut {NoStop}%
\bibitem [{\citenamefont {Stahl}\ \emph {et~al.}(2024)\citenamefont {Stahl},
  \citenamefont {Nandkishore},\ and\ \citenamefont {Hart}}]{loop_Charlie24}%
  \BibitemOpen
  \bibfield  {author} {\bibinfo {author} {\bibfnamefont {Charles}\ \bibnamefont
  {Stahl}}, \bibinfo {author} {\bibfnamefont {Rahul}\ \bibnamefont
  {Nandkishore}}, \ and\ \bibinfo {author} {\bibfnamefont {Oliver}\
  \bibnamefont {Hart}},\ }\bibfield  {title} {\enquote {\bibinfo {title}
  {{Topologically stable ergodicity breaking from emergent higher-form
  symmetries in generalized quantum loop models}},}\ }\href {\doibase
  10.21468/SciPostPhys.16.3.068} {\bibfield  {journal} {\bibinfo  {journal}
  {SciPost Phys.}\ }\textbf {\bibinfo {volume} {16}},\ \bibinfo {pages} {068}
  (\bibinfo {year} {2024})}\BibitemShut {NoStop}%
\bibitem [{\citenamefont {Yoshinaga}\ \emph {et~al.}(2022)\citenamefont
  {Yoshinaga}, \citenamefont {Hakoshima}, \citenamefont {Imoto}, \citenamefont
  {Matsuzaki},\ and\ \citenamefont {Hamazaki}}]{2dIsing_PRL22}%
  \BibitemOpen
  \bibfield  {author} {\bibinfo {author} {\bibfnamefont {Atsuki}\ \bibnamefont
  {Yoshinaga}}, \bibinfo {author} {\bibfnamefont {Hideaki}\ \bibnamefont
  {Hakoshima}}, \bibinfo {author} {\bibfnamefont {Takashi}\ \bibnamefont
  {Imoto}}, \bibinfo {author} {\bibfnamefont {Yuichiro}\ \bibnamefont
  {Matsuzaki}}, \ and\ \bibinfo {author} {\bibfnamefont {Ryusuke}\ \bibnamefont
  {Hamazaki}},\ }\bibfield  {title} {\enquote {\bibinfo {title} {Emergence of
  hilbert space fragmentation in ising models with a weak transverse field},}\
  }\href {\doibase 10.1103/PhysRevLett.129.090602} {\bibfield  {journal}
  {\bibinfo  {journal} {Phys. Rev. Lett.}\ }\textbf {\bibinfo {volume} {129}},\
  \bibinfo {pages} {090602} (\bibinfo {year} {2022})}\BibitemShut {NoStop}%
\bibitem [{\citenamefont {Hart}\ and\ \citenamefont
  {Nandkishore}(2022)}]{2dIsing_Hart22}%
  \BibitemOpen
  \bibfield  {author} {\bibinfo {author} {\bibfnamefont {Oliver}\ \bibnamefont
  {Hart}}\ and\ \bibinfo {author} {\bibfnamefont {Rahul}\ \bibnamefont
  {Nandkishore}},\ }\bibfield  {title} {\enquote {\bibinfo {title} {Hilbert
  space shattering and dynamical freezing in the quantum ising model},}\ }\href
  {\doibase 10.1103/PhysRevB.106.214426} {\bibfield  {journal} {\bibinfo
  {journal} {Phys. Rev. B}\ }\textbf {\bibinfo {volume} {106}},\ \bibinfo
  {pages} {214426} (\bibinfo {year} {2022})}\BibitemShut {NoStop}%
\bibitem [{\citenamefont {Simon}(2014)}]{barry_simon_book}%
  \BibitemOpen
  \bibfield  {author} {\bibinfo {author} {\bibfnamefont {Barry}\ \bibnamefont
  {Simon}},\ }\href@noop {} {\emph {\bibinfo {title} {The Statistical Mechanics
  of Lattice Gases, Volume I}}},\ Vol.\ \bibinfo {volume} {260}\ (\bibinfo
  {publisher} {Princeton University Press},\ \bibinfo {year}
  {2014})\BibitemShut {NoStop}%
\bibitem [{\citenamefont {Lieb}\ and\ \citenamefont
  {Robinson}(1972)}]{Lieb1972}%
  \BibitemOpen
  \bibfield  {author} {\bibinfo {author} {\bibfnamefont {Elliott~H.}\
  \bibnamefont {Lieb}}\ and\ \bibinfo {author} {\bibfnamefont {Derek~W.}\
  \bibnamefont {Robinson}},\ }\bibfield  {title} {\enquote {\bibinfo {title}
  {The finite group velocity of quantum spin systems},}\ }\href {\doibase
  10.1007/BF01645779} {\bibfield  {journal} {\bibinfo  {journal} {Commun. Math.
  Phys.}\ }\textbf {\bibinfo {volume} {28}},\ \bibinfo {pages} {251--257}
  (\bibinfo {year} {1972})}\BibitemShut {NoStop}%
\bibitem [{\citenamefont {Chen}\ \emph {et~al.}(2023)\citenamefont {Chen},
  \citenamefont {Lucas},\ and\ \citenamefont {Yin}}]{ourreview}%
  \BibitemOpen
  \bibfield  {author} {\bibinfo {author} {\bibfnamefont {Chi-Fang~(Anthony)}\
  \bibnamefont {Chen}}, \bibinfo {author} {\bibfnamefont {Andrew}\ \bibnamefont
  {Lucas}}, \ and\ \bibinfo {author} {\bibfnamefont {Chao}\ \bibnamefont
  {Yin}},\ }\bibfield  {title} {\enquote {\bibinfo {title} {Speed limits and
  locality in many-body quantum dynamics},}\ }\href {\doibase
  10.1088/1361-6633/acfaae} {\bibfield  {journal} {\bibinfo  {journal} {Reports
  on Progress in Physics}\ }\textbf {\bibinfo {volume} {86}},\ \bibinfo {pages}
  {116001} (\bibinfo {year} {2023})}\BibitemShut {NoStop}%
\bibitem [{\citenamefont {Bachmann}\ \emph {et~al.}(2012)\citenamefont
  {Bachmann}, \citenamefont {Michalakis}, \citenamefont {Nachtergaele},\ and\
  \citenamefont {Sims}}]{bachmann2012automorphic}%
  \BibitemOpen
  \bibfield  {author} {\bibinfo {author} {\bibfnamefont {Sven}\ \bibnamefont
  {Bachmann}}, \bibinfo {author} {\bibfnamefont {Spyridon}\ \bibnamefont
  {Michalakis}}, \bibinfo {author} {\bibfnamefont {Bruno}\ \bibnamefont
  {Nachtergaele}}, \ and\ \bibinfo {author} {\bibfnamefont {Robert}\
  \bibnamefont {Sims}},\ }\bibfield  {title} {\enquote {\bibinfo {title}
  {Automorphic equivalence within gapped phases of quantum lattice systems},}\
  }\href {\doibase https://doi.org/10.1007/s00220-011-1380-0} {\bibfield
  {journal} {\bibinfo  {journal} {Communications in Mathematical Physics}\
  }\textbf {\bibinfo {volume} {309}},\ \bibinfo {pages} {835--871} (\bibinfo
  {year} {2012})}\BibitemShut {NoStop}%
\bibitem [{\citenamefont {Bravyi}\ \emph {et~al.}(2010)\citenamefont {Bravyi},
  \citenamefont {Hastings},\ and\ \citenamefont {Michalakis}}]{topo_Hastings}%
  \BibitemOpen
  \bibfield  {author} {\bibinfo {author} {\bibfnamefont {Sergey}\ \bibnamefont
  {Bravyi}}, \bibinfo {author} {\bibfnamefont {Matthew~B.}\ \bibnamefont
  {Hastings}}, \ and\ \bibinfo {author} {\bibfnamefont {Spyridon}\ \bibnamefont
  {Michalakis}},\ }\bibfield  {title} {\enquote {\bibinfo {title} {Topological
  quantum order: Stability under local perturbations},}\ }\href {\doibase
  10.1063/1.3490195} {\bibfield  {journal} {\bibinfo  {journal} {Journal of
  Mathematical Physics}\ }\textbf {\bibinfo {volume} {51}},\ \bibinfo {pages}
  {093512} (\bibinfo {year} {2010})}\BibitemShut {NoStop}%
\bibitem [{\citenamefont {Haegeman}\ \emph {et~al.}(2011)\citenamefont
  {Haegeman}, \citenamefont {Cirac}, \citenamefont {Osborne}, \citenamefont
  {Pi\ifmmode~\check{z}\else \v{z}\fi{}orn}, \citenamefont {Verschelde},\ and\
  \citenamefont {Verstraete}}]{TDVP1}%
  \BibitemOpen
  \bibfield  {author} {\bibinfo {author} {\bibfnamefont {Jutho}\ \bibnamefont
  {Haegeman}}, \bibinfo {author} {\bibfnamefont {J.~Ignacio}\ \bibnamefont
  {Cirac}}, \bibinfo {author} {\bibfnamefont {Tobias~J.}\ \bibnamefont
  {Osborne}}, \bibinfo {author} {\bibfnamefont {Iztok}\ \bibnamefont
  {Pi\ifmmode~\check{z}\else \v{z}\fi{}orn}}, \bibinfo {author} {\bibfnamefont
  {Henri}\ \bibnamefont {Verschelde}}, \ and\ \bibinfo {author} {\bibfnamefont
  {Frank}\ \bibnamefont {Verstraete}},\ }\bibfield  {title} {\enquote {\bibinfo
  {title} {Time-dependent variational principle for quantum lattices},}\ }\href
  {\doibase 10.1103/PhysRevLett.107.070601} {\bibfield  {journal} {\bibinfo
  {journal} {Phys. Rev. Lett.}\ }\textbf {\bibinfo {volume} {107}},\ \bibinfo
  {pages} {070601} (\bibinfo {year} {2011})}\BibitemShut {NoStop}%
\bibitem [{\citenamefont {Haegeman}\ \emph {et~al.}(2016)\citenamefont
  {Haegeman}, \citenamefont {Lubich}, \citenamefont {Oseledets}, \citenamefont
  {Vandereycken},\ and\ \citenamefont {Verstraete}}]{TDVP2}%
  \BibitemOpen
  \bibfield  {author} {\bibinfo {author} {\bibfnamefont {Jutho}\ \bibnamefont
  {Haegeman}}, \bibinfo {author} {\bibfnamefont {Christian}\ \bibnamefont
  {Lubich}}, \bibinfo {author} {\bibfnamefont {Ivan}\ \bibnamefont
  {Oseledets}}, \bibinfo {author} {\bibfnamefont {Bart}\ \bibnamefont
  {Vandereycken}}, \ and\ \bibinfo {author} {\bibfnamefont {Frank}\
  \bibnamefont {Verstraete}},\ }\bibfield  {title} {\enquote {\bibinfo {title}
  {Unifying time evolution and optimization with matrix product states},}\
  }\href {\doibase 10.1103/PhysRevB.94.165116} {\bibfield  {journal} {\bibinfo
  {journal} {Phys. Rev. B}\ }\textbf {\bibinfo {volume} {94}},\ \bibinfo
  {pages} {165116} (\bibinfo {year} {2016})}\BibitemShut {NoStop}%
\bibitem [{\citenamefont {Hauschild}\ and\ \citenamefont
  {Pollmann}(2018)}]{tenpy}%
  \BibitemOpen
  \bibfield  {author} {\bibinfo {author} {\bibfnamefont {Johannes}\
  \bibnamefont {Hauschild}}\ and\ \bibinfo {author} {\bibfnamefont {Frank}\
  \bibnamefont {Pollmann}},\ }\bibfield  {title} {\enquote {\bibinfo {title}
  {{Efficient numerical simulations with Tensor Networks: Tensor Network Python
  (TeNPy)}},}\ }\href {\doibase 10.21468/SciPostPhysLectNotes.5} {\bibfield
  {journal} {\bibinfo  {journal} {SciPost Phys. Lect. Notes}\ ,\ \bibinfo
  {pages} {5}} (\bibinfo {year} {2018})},\ \bibinfo {note} {code available from
  \url{https://github.com/tenpy/tenpy}},\ \Eprint
  {http://arxiv.org/abs/1805.00055} {arXiv:1805.00055} \BibitemShut {NoStop}%
\bibitem [{\citenamefont {Turner}\ \emph
  {et~al.}(2018{\natexlab{b}})\citenamefont {Turner}, \citenamefont
  {Michailidis}, \citenamefont {Abanin}, \citenamefont {Serbyn},\ and\
  \citenamefont {Papi\ifmmode~\acute{c}\else \'{c}\fi{}}}]{PXP_Papic18}%
  \BibitemOpen
  \bibfield  {author} {\bibinfo {author} {\bibfnamefont {C.~J.}\ \bibnamefont
  {Turner}}, \bibinfo {author} {\bibfnamefont {A.~A.}\ \bibnamefont
  {Michailidis}}, \bibinfo {author} {\bibfnamefont {D.~A.}\ \bibnamefont
  {Abanin}}, \bibinfo {author} {\bibfnamefont {M.}~\bibnamefont {Serbyn}}, \
  and\ \bibinfo {author} {\bibfnamefont {Z.}~\bibnamefont
  {Papi\ifmmode~\acute{c}\else \'{c}\fi{}}},\ }\bibfield  {title} {\enquote
  {\bibinfo {title} {Quantum scarred eigenstates in a rydberg atom chain:
  Entanglement, breakdown of thermalization, and stability to perturbations},}\
  }\href {\doibase 10.1103/PhysRevB.98.155134} {\bibfield  {journal} {\bibinfo
  {journal} {Phys. Rev. B}\ }\textbf {\bibinfo {volume} {98}},\ \bibinfo
  {pages} {155134} (\bibinfo {year} {2018}{\natexlab{b}})}\BibitemShut
  {NoStop}%
\bibitem [{\citenamefont {Khemani}\ \emph {et~al.}(2019)\citenamefont
  {Khemani}, \citenamefont {Laumann},\ and\ \citenamefont
  {Chandran}}]{PXP_Khemani19}%
  \BibitemOpen
  \bibfield  {author} {\bibinfo {author} {\bibfnamefont {Vedika}\ \bibnamefont
  {Khemani}}, \bibinfo {author} {\bibfnamefont {Chris~R.}\ \bibnamefont
  {Laumann}}, \ and\ \bibinfo {author} {\bibfnamefont {Anushya}\ \bibnamefont
  {Chandran}},\ }\bibfield  {title} {\enquote {\bibinfo {title} {Signatures of
  integrability in the dynamics of rydberg-blockaded chains},}\ }\href
  {\doibase 10.1103/PhysRevB.99.161101} {\bibfield  {journal} {\bibinfo
  {journal} {Phys. Rev. B}\ }\textbf {\bibinfo {volume} {99}},\ \bibinfo
  {pages} {161101} (\bibinfo {year} {2019})}\BibitemShut {NoStop}%
\bibitem [{\citenamefont {Choi}\ \emph {et~al.}(2019)\citenamefont {Choi},
  \citenamefont {Turner}, \citenamefont {Pichler}, \citenamefont {Ho},
  \citenamefont {Michailidis}, \citenamefont {Papi\ifmmode~\acute{c}\else
  \'{c}\fi{}}, \citenamefont {Serbyn}, \citenamefont {Lukin},\ and\
  \citenamefont {Abanin}}]{PXP_SU2_19}%
  \BibitemOpen
  \bibfield  {author} {\bibinfo {author} {\bibfnamefont {Soonwon}\ \bibnamefont
  {Choi}}, \bibinfo {author} {\bibfnamefont {Christopher~J.}\ \bibnamefont
  {Turner}}, \bibinfo {author} {\bibfnamefont {Hannes}\ \bibnamefont
  {Pichler}}, \bibinfo {author} {\bibfnamefont {Wen~Wei}\ \bibnamefont {Ho}},
  \bibinfo {author} {\bibfnamefont {Alexios~A.}\ \bibnamefont {Michailidis}},
  \bibinfo {author} {\bibfnamefont {Zlatko}\ \bibnamefont
  {Papi\ifmmode~\acute{c}\else \'{c}\fi{}}}, \bibinfo {author} {\bibfnamefont
  {Maksym}\ \bibnamefont {Serbyn}}, \bibinfo {author} {\bibfnamefont
  {Mikhail~D.}\ \bibnamefont {Lukin}}, \ and\ \bibinfo {author} {\bibfnamefont
  {Dmitry~A.}\ \bibnamefont {Abanin}},\ }\bibfield  {title} {\enquote {\bibinfo
  {title} {Emergent su(2) dynamics and perfect quantum many-body scars},}\
  }\href {\doibase 10.1103/PhysRevLett.122.220603} {\bibfield  {journal}
  {\bibinfo  {journal} {Phys. Rev. Lett.}\ }\textbf {\bibinfo {volume} {122}},\
  \bibinfo {pages} {220603} (\bibinfo {year} {2019})}\BibitemShut {NoStop}%
\bibitem [{\citenamefont {Iadecola}\ \emph {et~al.}(2019)\citenamefont
  {Iadecola}, \citenamefont {Schecter},\ and\ \citenamefont {Xu}}]{PXP_GS19}%
  \BibitemOpen
  \bibfield  {author} {\bibinfo {author} {\bibfnamefont {Thomas}\ \bibnamefont
  {Iadecola}}, \bibinfo {author} {\bibfnamefont {Michael}\ \bibnamefont
  {Schecter}}, \ and\ \bibinfo {author} {\bibfnamefont {Shenglong}\
  \bibnamefont {Xu}},\ }\bibfield  {title} {\enquote {\bibinfo {title} {Quantum
  many-body scars from magnon condensation},}\ }\href {\doibase
  10.1103/PhysRevB.100.184312} {\bibfield  {journal} {\bibinfo  {journal}
  {Phys. Rev. B}\ }\textbf {\bibinfo {volume} {100}},\ \bibinfo {pages}
  {184312} (\bibinfo {year} {2019})}\BibitemShut {NoStop}%
\bibitem [{\citenamefont {Ho}\ \emph {et~al.}(2019)\citenamefont {Ho},
  \citenamefont {Choi}, \citenamefont {Pichler},\ and\ \citenamefont
  {Lukin}}]{PXP_TDVP19}%
  \BibitemOpen
  \bibfield  {author} {\bibinfo {author} {\bibfnamefont {Wen~Wei}\ \bibnamefont
  {Ho}}, \bibinfo {author} {\bibfnamefont {Soonwon}\ \bibnamefont {Choi}},
  \bibinfo {author} {\bibfnamefont {Hannes}\ \bibnamefont {Pichler}}, \ and\
  \bibinfo {author} {\bibfnamefont {Mikhail~D.}\ \bibnamefont {Lukin}},\
  }\bibfield  {title} {\enquote {\bibinfo {title} {Periodic orbits,
  entanglement, and quantum many-body scars in constrained models: Matrix
  product state approach},}\ }\href {\doibase 10.1103/PhysRevLett.122.040603}
  {\bibfield  {journal} {\bibinfo  {journal} {Phys. Rev. Lett.}\ }\textbf
  {\bibinfo {volume} {122}},\ \bibinfo {pages} {040603} (\bibinfo {year}
  {2019})}\BibitemShut {NoStop}%
\bibitem [{\citenamefont {Bull}\ \emph {et~al.}(2020)\citenamefont {Bull},
  \citenamefont {Desaules},\ and\ \citenamefont {Papi\ifmmode~\acute{c}\else
  \'{c}\fi{}}}]{PXP_SU2_20}%
  \BibitemOpen
  \bibfield  {author} {\bibinfo {author} {\bibfnamefont {Kieran}\ \bibnamefont
  {Bull}}, \bibinfo {author} {\bibfnamefont {Jean-Yves}\ \bibnamefont
  {Desaules}}, \ and\ \bibinfo {author} {\bibfnamefont {Zlatko}\ \bibnamefont
  {Papi\ifmmode~\acute{c}\else \'{c}\fi{}}},\ }\bibfield  {title} {\enquote
  {\bibinfo {title} {Quantum scars as embeddings of weakly broken lie algebra
  representations},}\ }\href {\doibase 10.1103/PhysRevB.101.165139} {\bibfield
  {journal} {\bibinfo  {journal} {Phys. Rev. B}\ }\textbf {\bibinfo {volume}
  {101}},\ \bibinfo {pages} {165139} (\bibinfo {year} {2020})}\BibitemShut
  {NoStop}%
\bibitem [{\citenamefont {Turner}\ \emph {et~al.}(2021)\citenamefont {Turner},
  \citenamefont {Desaules}, \citenamefont {Bull},\ and\ \citenamefont
  {Papi\ifmmode~\acute{c}\else \'{c}\fi{}}}]{PXP_TDVP21}%
  \BibitemOpen
  \bibfield  {author} {\bibinfo {author} {\bibfnamefont {C.~J.}\ \bibnamefont
  {Turner}}, \bibinfo {author} {\bibfnamefont {J.-Y.}\ \bibnamefont
  {Desaules}}, \bibinfo {author} {\bibfnamefont {K.}~\bibnamefont {Bull}}, \
  and\ \bibinfo {author} {\bibfnamefont {Z.}~\bibnamefont
  {Papi\ifmmode~\acute{c}\else \'{c}\fi{}}},\ }\bibfield  {title} {\enquote
  {\bibinfo {title} {Correspondence principle for many-body scars in ultracold
  rydberg atoms},}\ }\href {\doibase 10.1103/PhysRevX.11.021021} {\bibfield
  {journal} {\bibinfo  {journal} {Phys. Rev. X}\ }\textbf {\bibinfo {volume}
  {11}},\ \bibinfo {pages} {021021} (\bibinfo {year} {2021})}\BibitemShut
  {NoStop}%
\bibitem [{\citenamefont {Windt}\ and\ \citenamefont
  {Pichler}(2022)}]{PXP_all2all22}%
  \BibitemOpen
  \bibfield  {author} {\bibinfo {author} {\bibfnamefont {Bennet}\ \bibnamefont
  {Windt}}\ and\ \bibinfo {author} {\bibfnamefont {Hannes}\ \bibnamefont
  {Pichler}},\ }\bibfield  {title} {\enquote {\bibinfo {title} {Squeezing
  quantum many-body scars},}\ }\href {\doibase 10.1103/PhysRevLett.128.090606}
  {\bibfield  {journal} {\bibinfo  {journal} {Phys. Rev. Lett.}\ }\textbf
  {\bibinfo {volume} {128}},\ \bibinfo {pages} {090606} (\bibinfo {year}
  {2022})}\BibitemShut {NoStop}%
\bibitem [{\citenamefont {Pan}\ and\ \citenamefont
  {Zhai}(2022)}]{PXP_composite22}%
  \BibitemOpen
  \bibfield  {author} {\bibinfo {author} {\bibfnamefont {Lei}\ \bibnamefont
  {Pan}}\ and\ \bibinfo {author} {\bibfnamefont {Hui}\ \bibnamefont {Zhai}},\
  }\bibfield  {title} {\enquote {\bibinfo {title} {Composite spin approach to
  the blockade effect in rydberg atom arrays},}\ }\href {\doibase
  10.1103/PhysRevResearch.4.L032037} {\bibfield  {journal} {\bibinfo  {journal}
  {Phys. Rev. Res.}\ }\textbf {\bibinfo {volume} {4}},\ \bibinfo {pages}
  {L032037} (\bibinfo {year} {2022})}\BibitemShut {NoStop}%
\bibitem [{\citenamefont {Omiya}\ and\ \citenamefont
  {M\"uller}(2023)}]{PXP_spin1_23}%
  \BibitemOpen
  \bibfield  {author} {\bibinfo {author} {\bibfnamefont {Keita}\ \bibnamefont
  {Omiya}}\ and\ \bibinfo {author} {\bibfnamefont {Markus}\ \bibnamefont
  {M\"uller}},\ }\bibfield  {title} {\enquote {\bibinfo {title} {Quantum
  many-body scars in bipartite rydberg arrays originating from hidden projector
  embedding},}\ }\href {\doibase 10.1103/PhysRevA.107.023318} {\bibfield
  {journal} {\bibinfo  {journal} {Phys. Rev. A}\ }\textbf {\bibinfo {volume}
  {107}},\ \bibinfo {pages} {023318} (\bibinfo {year} {2023})}\BibitemShut
  {NoStop}%
\bibitem [{\citenamefont {Ovchinnikov}\ \emph {et~al.}(2003)\citenamefont
  {Ovchinnikov}, \citenamefont {Dmitriev}, \citenamefont {Krivnov},\ and\
  \citenamefont {Cheranovskii}}]{PXPspec}%
  \BibitemOpen
  \bibfield  {author} {\bibinfo {author} {\bibfnamefont {A.~A.}\ \bibnamefont
  {Ovchinnikov}}, \bibinfo {author} {\bibfnamefont {D.~V.}\ \bibnamefont
  {Dmitriev}}, \bibinfo {author} {\bibfnamefont {V.~Ya.}\ \bibnamefont
  {Krivnov}}, \ and\ \bibinfo {author} {\bibfnamefont {V.~O.}\ \bibnamefont
  {Cheranovskii}},\ }\bibfield  {title} {\enquote {\bibinfo {title}
  {Antiferromagnetic ising chain in a mixed transverse and longitudinal
  magnetic field},}\ }\href {\doibase 10.1103/PhysRevB.68.214406} {\bibfield
  {journal} {\bibinfo  {journal} {Phys. Rev. B}\ }\textbf {\bibinfo {volume}
  {68}},\ \bibinfo {pages} {214406} (\bibinfo {year} {2003})}\BibitemShut
  {NoStop}%
\bibitem [{\citenamefont {Lin}\ and\ \citenamefont
  {Motrunich}(2017)}]{falseVac_quasip}%
  \BibitemOpen
  \bibfield  {author} {\bibinfo {author} {\bibfnamefont {Cheng-Ju}\
  \bibnamefont {Lin}}\ and\ \bibinfo {author} {\bibfnamefont {Olexei~I.}\
  \bibnamefont {Motrunich}},\ }\bibfield  {title} {\enquote {\bibinfo {title}
  {Quasiparticle explanation of the weak-thermalization regime under quench in
  a nonintegrable quantum spin chain},}\ }\href {\doibase
  10.1103/PhysRevA.95.023621} {\bibfield  {journal} {\bibinfo  {journal} {Phys.
  Rev. A}\ }\textbf {\bibinfo {volume} {95}},\ \bibinfo {pages} {023621}
  (\bibinfo {year} {2017})}\BibitemShut {NoStop}%
\bibitem [{\citenamefont {Robertson}\ \emph {et~al.}(2024)\citenamefont
  {Robertson}, \citenamefont {Senese},\ and\ \citenamefont
  {Essler}}]{Robertson2024}%
  \BibitemOpen
  \bibfield  {author} {\bibinfo {author} {\bibfnamefont {Jacob~H.}\
  \bibnamefont {Robertson}}, \bibinfo {author} {\bibfnamefont {Riccardo}\
  \bibnamefont {Senese}}, \ and\ \bibinfo {author} {\bibfnamefont {Fabian
  H.~L.}\ \bibnamefont {Essler}},\ }\bibfield  {title} {\enquote {\bibinfo
  {title} {Decay of long-lived oscillations after quantum quenches in gapped
  interacting quantum systems},}\ }\href {\doibase 10.1103/PhysRevA.109.032208}
  {\bibfield  {journal} {\bibinfo  {journal} {Phys. Rev. A}\ }\textbf {\bibinfo
  {volume} {109}},\ \bibinfo {pages} {032208} (\bibinfo {year}
  {2024})}\BibitemShut {NoStop}%
\bibitem [{\citenamefont {Lin}\ \emph {et~al.}(2020{\natexlab{b}})\citenamefont
  {Lin}, \citenamefont {Chandran},\ and\ \citenamefont
  {Motrunich}}]{PXP_pert20}%
  \BibitemOpen
  \bibfield  {author} {\bibinfo {author} {\bibfnamefont {Cheng-Ju}\
  \bibnamefont {Lin}}, \bibinfo {author} {\bibfnamefont {Anushya}\ \bibnamefont
  {Chandran}}, \ and\ \bibinfo {author} {\bibfnamefont {Olexei~I.}\
  \bibnamefont {Motrunich}},\ }\bibfield  {title} {\enquote {\bibinfo {title}
  {Slow thermalization of exact quantum many-body scar states under
  perturbations},}\ }\href {\doibase 10.1103/PhysRevResearch.2.033044}
  {\bibfield  {journal} {\bibinfo  {journal} {Phys. Rev. Res.}\ }\textbf
  {\bibinfo {volume} {2}},\ \bibinfo {pages} {033044} (\bibinfo {year}
  {2020}{\natexlab{b}})}\BibitemShut {NoStop}%
\bibitem [{\citenamefont {Surace}\ \emph
  {et~al.}(2021{\natexlab{b}})\citenamefont {Surace}, \citenamefont {Votto},
  \citenamefont {Lazo}, \citenamefont {Silva}, \citenamefont {Dalmonte},\ and\
  \citenamefont {Giudici}}]{PXP_pert21}%
  \BibitemOpen
  \bibfield  {author} {\bibinfo {author} {\bibfnamefont {Federica~Maria}\
  \bibnamefont {Surace}}, \bibinfo {author} {\bibfnamefont {Matteo}\
  \bibnamefont {Votto}}, \bibinfo {author} {\bibfnamefont {Eduardo~Gonzalez}\
  \bibnamefont {Lazo}}, \bibinfo {author} {\bibfnamefont {Alessandro}\
  \bibnamefont {Silva}}, \bibinfo {author} {\bibfnamefont {Marcello}\
  \bibnamefont {Dalmonte}}, \ and\ \bibinfo {author} {\bibfnamefont {Giuliano}\
  \bibnamefont {Giudici}},\ }\bibfield  {title} {\enquote {\bibinfo {title}
  {Exact many-body scars and their stability in constrained quantum chains},}\
  }\href {\doibase 10.1103/PhysRevB.103.104302} {\bibfield  {journal} {\bibinfo
   {journal} {Phys. Rev. B}\ }\textbf {\bibinfo {volume} {103}},\ \bibinfo
  {pages} {104302} (\bibinfo {year} {2021}{\natexlab{b}})}\BibitemShut
  {NoStop}%
\bibitem [{\citenamefont {Krinitsin}\ \emph {et~al.}(2024)\citenamefont
  {Krinitsin}, \citenamefont {Tausendpfund}, \citenamefont {Rizzi},\ and\
  \citenamefont {Schmitt}}]{2dising_DW24}%
  \BibitemOpen
  \bibfield  {author} {\bibinfo {author} {\bibfnamefont {Wladislaw}\
  \bibnamefont {Krinitsin}}, \bibinfo {author} {\bibfnamefont {Niklas}\
  \bibnamefont {Tausendpfund}}, \bibinfo {author} {\bibfnamefont {Matteo}\
  \bibnamefont {Rizzi}}, \ and\ \bibinfo {author} {\bibfnamefont {Markus}\
  \bibnamefont {Schmitt}},\ }\bibfield  {title} {\enquote {\bibinfo {title}
  {Domain wall dynamics of a two dimensional quantum ising model using tree
  tensor networks},}\ }\href@noop {} {\bibfield  {journal} {\bibinfo  {journal}
  {Bulletin of the American Physical Society}\ } (\bibinfo {year}
  {2024})}\BibitemShut {NoStop}%
\bibitem [{\citenamefont {Balducci}\ \emph {et~al.}(2023)\citenamefont
  {Balducci}, \citenamefont {Gambassi}, \citenamefont {Lerose}, \citenamefont
  {Scardicchio},\ and\ \citenamefont {Vanoni}}]{Ising2D_1}%
  \BibitemOpen
  \bibfield  {author} {\bibinfo {author} {\bibfnamefont {Federico}\
  \bibnamefont {Balducci}}, \bibinfo {author} {\bibfnamefont {Andrea}\
  \bibnamefont {Gambassi}}, \bibinfo {author} {\bibfnamefont {Alessio}\
  \bibnamefont {Lerose}}, \bibinfo {author} {\bibfnamefont {Antonello}\
  \bibnamefont {Scardicchio}}, \ and\ \bibinfo {author} {\bibfnamefont {Carlo}\
  \bibnamefont {Vanoni}},\ }\bibfield  {title} {\enquote {\bibinfo {title}
  {Interface dynamics in the two-dimensional quantum ising model},}\ }\href
  {\doibase 10.1103/PhysRevB.107.024306} {\bibfield  {journal} {\bibinfo
  {journal} {Phys. Rev. B}\ }\textbf {\bibinfo {volume} {107}},\ \bibinfo
  {pages} {024306} (\bibinfo {year} {2023})}\BibitemShut {NoStop}%
\bibitem [{\citenamefont {Balducci}\ \emph {et~al.}(2022)\citenamefont
  {Balducci}, \citenamefont {Gambassi}, \citenamefont {Lerose}, \citenamefont
  {Scardicchio},\ and\ \citenamefont {Vanoni}}]{Ising2D_2}%
  \BibitemOpen
  \bibfield  {author} {\bibinfo {author} {\bibfnamefont {Federico}\
  \bibnamefont {Balducci}}, \bibinfo {author} {\bibfnamefont {Andrea}\
  \bibnamefont {Gambassi}}, \bibinfo {author} {\bibfnamefont {Alessio}\
  \bibnamefont {Lerose}}, \bibinfo {author} {\bibfnamefont {Antonello}\
  \bibnamefont {Scardicchio}}, \ and\ \bibinfo {author} {\bibfnamefont {Carlo}\
  \bibnamefont {Vanoni}},\ }\bibfield  {title} {\enquote {\bibinfo {title}
  {Localization and melting of interfaces in the two-dimensional quantum ising
  model},}\ }\href {\doibase 10.1103/PhysRevLett.129.120601} {\bibfield
  {journal} {\bibinfo  {journal} {Phys. Rev. Lett.}\ }\textbf {\bibinfo
  {volume} {129}},\ \bibinfo {pages} {120601} (\bibinfo {year}
  {2022})}\BibitemShut {NoStop}%
\bibitem [{\citenamefont {Pave\v{s}i\'c}\ \emph {et~al.}(2024)\citenamefont
  {Pave\v{s}i\'c}, \citenamefont {Jaschke},\ and\ \citenamefont
  {Montangero}}]{pavevsic2024constrained}%
  \BibitemOpen
  \bibfield  {author} {\bibinfo {author} {\bibfnamefont {Luka}\ \bibnamefont
  {Pave\v{s}i\'c}}, \bibinfo {author} {\bibfnamefont {Daniel}\ \bibnamefont
  {Jaschke}}, \ and\ \bibinfo {author} {\bibfnamefont {Simone}\ \bibnamefont
  {Montangero}},\ }\bibfield  {title} {\enquote {\bibinfo {title} {{Constrained
  dynamics and confinement in the two-dimensional quantum Ising model}},}\
  }\href@noop {} {\  (\bibinfo {year} {2024})},\ \Eprint
  {http://arxiv.org/abs/2406.11979} {arXiv:2406.11979 [quant-ph]} \BibitemShut
  {NoStop}%
\bibitem [{\citenamefont {Fukui}\ and\ \citenamefont
  {Kawakami}(1998)}]{kawakami}%
  \BibitemOpen
  \bibfield  {author} {\bibinfo {author} {\bibfnamefont {Takahiro}\
  \bibnamefont {Fukui}}\ and\ \bibinfo {author} {\bibfnamefont {Norio}\
  \bibnamefont {Kawakami}},\ }\bibfield  {title} {\enquote {\bibinfo {title}
  {Breakdown of the mott insulator: Exact solution of an asymmetric hubbard
  model},}\ }\href {\doibase 10.1103/PhysRevB.58.16051} {\bibfield  {journal}
  {\bibinfo  {journal} {Phys. Rev. B}\ }\textbf {\bibinfo {volume} {58}},\
  \bibinfo {pages} {16051--16056} (\bibinfo {year} {1998})}\BibitemShut
  {NoStop}%
\bibitem [{\citenamefont {Lian}(2023)}]{biaolian}%
  \BibitemOpen
  \bibfield  {author} {\bibinfo {author} {\bibfnamefont {Biao}\ \bibnamefont
  {Lian}},\ }\bibfield  {title} {\enquote {\bibinfo {title} {Quantum breakdown
  model: From many-body localization to chaos with scars},}\ }\href {\doibase
  10.1103/PhysRevB.107.115171} {\bibfield  {journal} {\bibinfo  {journal}
  {Phys. Rev. B}\ }\textbf {\bibinfo {volume} {107}},\ \bibinfo {pages}
  {115171} (\bibinfo {year} {2023})}\BibitemShut {NoStop}%
\bibitem [{\citenamefont {Yin}\ \emph {et~al.}(2024)\citenamefont {Yin},
  \citenamefont {Nandkishore},\ and\ \citenamefont {Lucas}}]{MBL_LDPC24}%
  \BibitemOpen
  \bibfield  {author} {\bibinfo {author} {\bibfnamefont {Chao}\ \bibnamefont
  {Yin}}, \bibinfo {author} {\bibfnamefont {Rahul}\ \bibnamefont
  {Nandkishore}}, \ and\ \bibinfo {author} {\bibfnamefont {Andrew}\
  \bibnamefont {Lucas}},\ }\bibfield  {title} {\enquote {\bibinfo {title}
  {Eigenstate localization in a many-body quantum system},}\ }\href {\doibase
  10.1103/PhysRevLett.133.137101} {\bibfield  {journal} {\bibinfo  {journal}
  {Phys. Rev. Lett.}\ }\textbf {\bibinfo {volume} {133}},\ \bibinfo {pages}
  {137101} (\bibinfo {year} {2024})}\BibitemShut {NoStop}%
\bibitem [{\citenamefont {Altshuler}\ \emph {et~al.}(2010)\citenamefont
  {Altshuler}, \citenamefont {Krovi},\ and\ \citenamefont
  {Roland}}]{Altshuler_2010}%
  \BibitemOpen
  \bibfield  {author} {\bibinfo {author} {\bibfnamefont {Boris}\ \bibnamefont
  {Altshuler}}, \bibinfo {author} {\bibfnamefont {Hari}\ \bibnamefont {Krovi}},
  \ and\ \bibinfo {author} {\bibfnamefont {Jérémie}\ \bibnamefont {Roland}},\
  }\bibfield  {title} {\enquote {\bibinfo {title} {Anderson localization makes
  adiabatic quantum optimization fail},}\ }\href {\doibase
  10.1073/pnas.1002116107} {\bibfield  {journal} {\bibinfo  {journal}
  {Proceedings of the National Academy of Sciences}\ }\textbf {\bibinfo
  {volume} {107}},\ \bibinfo {pages} {12446–12450} (\bibinfo {year}
  {2010})}\BibitemShut {NoStop}%
\bibitem [{\citenamefont {Baldwin}\ \emph {et~al.}(2017)\citenamefont
  {Baldwin}, \citenamefont {Laumann}, \citenamefont {Pal},\ and\ \citenamefont
  {Scardicchio}}]{MBL_glass17}%
  \BibitemOpen
  \bibfield  {author} {\bibinfo {author} {\bibfnamefont {C.~L.}\ \bibnamefont
  {Baldwin}}, \bibinfo {author} {\bibfnamefont {C.~R.}\ \bibnamefont
  {Laumann}}, \bibinfo {author} {\bibfnamefont {A.}~\bibnamefont {Pal}}, \ and\
  \bibinfo {author} {\bibfnamefont {A.}~\bibnamefont {Scardicchio}},\
  }\bibfield  {title} {\enquote {\bibinfo {title} {Clustering of nonergodic
  eigenstates in quantum spin glasses},}\ }\href {\doibase
  10.1103/PhysRevLett.118.127201} {\bibfield  {journal} {\bibinfo  {journal}
  {Phys. Rev. Lett.}\ }\textbf {\bibinfo {volume} {118}},\ \bibinfo {pages}
  {127201} (\bibinfo {year} {2017})}\BibitemShut {NoStop}%
\bibitem [{\citenamefont {Ikeda}\ \emph {et~al.}(2024)\citenamefont {Ikeda},
  \citenamefont {Sugiura},\ and\ \citenamefont {Polkovnikov}}]{Floq_GS24}%
  \BibitemOpen
  \bibfield  {author} {\bibinfo {author} {\bibfnamefont {Tatsuhiko~N.}\
  \bibnamefont {Ikeda}}, \bibinfo {author} {\bibfnamefont {Sho}\ \bibnamefont
  {Sugiura}}, \ and\ \bibinfo {author} {\bibfnamefont {Anatoli}\ \bibnamefont
  {Polkovnikov}},\ }\bibfield  {title} {\enquote {\bibinfo {title} {Robust
  effective ground state in a nonintegrable floquet quantum circuit},}\ }\href
  {\doibase 10.1103/PhysRevLett.133.030401} {\bibfield  {journal} {\bibinfo
  {journal} {Phys. Rev. Lett.}\ }\textbf {\bibinfo {volume} {133}},\ \bibinfo
  {pages} {030401} (\bibinfo {year} {2024})}\BibitemShut {NoStop}%
\bibitem [{\citenamefont {Dyson}(1969)}]{dyson}%
  \BibitemOpen
  \bibfield  {author} {\bibinfo {author} {\bibfnamefont {F.~J.}\ \bibnamefont
  {Dyson}},\ }\bibfield  {title} {\enquote {\bibinfo {title} {Existence of a
  phase-transition in a one-dimensional ising ferromagnet},}\ }\href@noop {}
  {\bibfield  {journal} {\bibinfo  {journal} {Comm. Math. Phys.}\ }\textbf
  {\bibinfo {volume} {12}},\ \bibinfo {pages} {91} (\bibinfo {year}
  {1969})}\BibitemShut {NoStop}%
\bibitem [{\citenamefont {Israel}(1975)}]{israel}%
  \BibitemOpen
  \bibfield  {author} {\bibinfo {author} {\bibfnamefont {R.~B.}\ \bibnamefont
  {Israel}},\ }\bibfield  {title} {\enquote {\bibinfo {title} {Existence of
  phase transitions for long-range interactions},}\ }\href@noop {} {\bibfield
  {journal} {\bibinfo  {journal} {Comm. Math. Phys.}\ }\textbf {\bibinfo
  {volume} {43}},\ \bibinfo {pages} {59} (\bibinfo {year} {1975})}\BibitemShut
  {NoStop}%
\bibitem [{\citenamefont {McCraw}(1980)}]{mccraw}%
  \BibitemOpen
  \bibfield  {author} {\bibinfo {author} {\bibfnamefont {R.~J.}\ \bibnamefont
  {McCraw}},\ }\bibfield  {title} {\enquote {\bibinfo {title} {Metastability in
  a long range one-dimensional ising model},}\ }\href@noop {} {\bibfield
  {journal} {\bibinfo  {journal} {Phys. Lett.}\ }\textbf {\bibinfo {volume}
  {A75}},\ \bibinfo {pages} {379} (\bibinfo {year} {1980})}\BibitemShut
  {NoStop}%
\bibitem [{\citenamefont {van Enter}\ \emph {et~al.}(2019)\citenamefont {van
  Enter}, \citenamefont {Kimura}, \citenamefont {Ruszel},\ and\ \citenamefont
  {Spitoni}}]{van_Enter_2019}%
  \BibitemOpen
  \bibfield  {author} {\bibinfo {author} {\bibfnamefont {Aernout C.~D.}\
  \bibnamefont {van Enter}}, \bibinfo {author} {\bibfnamefont {Bruno}\
  \bibnamefont {Kimura}}, \bibinfo {author} {\bibfnamefont {Wioletta}\
  \bibnamefont {Ruszel}}, \ and\ \bibinfo {author} {\bibfnamefont {Cristian}\
  \bibnamefont {Spitoni}},\ }\bibfield  {title} {\enquote {\bibinfo {title}
  {Nucleation for one-dimensional long-range ising models},}\ }\href {\doibase
  10.1007/s10955-019-02238-y} {\bibfield  {journal} {\bibinfo  {journal}
  {Journal of Statistical Physics}\ }\textbf {\bibinfo {volume} {174}},\
  \bibinfo {pages} {1327–1345} (\bibinfo {year} {2019})}\BibitemShut
  {NoStop}%
\bibitem [{\citenamefont {Collura}\ \emph {et~al.}(2022)\citenamefont
  {Collura}, \citenamefont {De~Luca}, \citenamefont {Rossini},\ and\
  \citenamefont {Lerose}}]{lerose}%
  \BibitemOpen
  \bibfield  {author} {\bibinfo {author} {\bibfnamefont {Mario}\ \bibnamefont
  {Collura}}, \bibinfo {author} {\bibfnamefont {Andrea}\ \bibnamefont
  {De~Luca}}, \bibinfo {author} {\bibfnamefont {Davide}\ \bibnamefont
  {Rossini}}, \ and\ \bibinfo {author} {\bibfnamefont {Alessio}\ \bibnamefont
  {Lerose}},\ }\bibfield  {title} {\enquote {\bibinfo {title} {Discrete
  time-crystalline response stabilized by domain-wall confinement},}\ }\href
  {\doibase 10.1103/PhysRevX.12.031037} {\bibfield  {journal} {\bibinfo
  {journal} {Phys. Rev. X}\ }\textbf {\bibinfo {volume} {12}},\ \bibinfo
  {pages} {031037} (\bibinfo {year} {2022})}\BibitemShut {NoStop}%
\bibitem [{\citenamefont {Kastner}(2011)}]{longrange_metastab_Ising}%
  \BibitemOpen
  \bibfield  {author} {\bibinfo {author} {\bibfnamefont {Michael}\ \bibnamefont
  {Kastner}},\ }\bibfield  {title} {\enquote {\bibinfo {title} {Diverging
  equilibration times in long-range quantum spin models},}\ }\href {\doibase
  10.1103/PhysRevLett.106.130601} {\bibfield  {journal} {\bibinfo  {journal}
  {Phys. Rev. Lett.}\ }\textbf {\bibinfo {volume} {106}},\ \bibinfo {pages}
  {130601} (\bibinfo {year} {2011})}\BibitemShut {NoStop}%
\bibitem [{\citenamefont {Lerose}\ \emph {et~al.}(2019)\citenamefont {Lerose},
  \citenamefont {\ifmmode \check{Z}\else
  \v{Z}\fi{}unkovi\ifmmode~\check{c}\else \v{c}\fi{}}, \citenamefont {Marino},
  \citenamefont {Gambassi},\ and\ \citenamefont
  {Silva}}]{longrange_metastab_fluc}%
  \BibitemOpen
  \bibfield  {author} {\bibinfo {author} {\bibfnamefont {Alessio}\ \bibnamefont
  {Lerose}}, \bibinfo {author} {\bibfnamefont {Bojan}\ \bibnamefont {\ifmmode
  \check{Z}\else \v{Z}\fi{}unkovi\ifmmode~\check{c}\else \v{c}\fi{}}}, \bibinfo
  {author} {\bibfnamefont {Jamir}\ \bibnamefont {Marino}}, \bibinfo {author}
  {\bibfnamefont {Andrea}\ \bibnamefont {Gambassi}}, \ and\ \bibinfo {author}
  {\bibfnamefont {Alessandro}\ \bibnamefont {Silva}},\ }\bibfield  {title}
  {\enquote {\bibinfo {title} {Impact of nonequilibrium fluctuations on
  prethermal dynamical phase transitions in long-range interacting spin
  chains},}\ }\href {\doibase 10.1103/PhysRevB.99.045128} {\bibfield  {journal}
  {\bibinfo  {journal} {Phys. Rev. B}\ }\textbf {\bibinfo {volume} {99}},\
  \bibinfo {pages} {045128} (\bibinfo {year} {2019})}\BibitemShut {NoStop}%
\bibitem [{\citenamefont {Liu}\ \emph {et~al.}(2019)\citenamefont {Liu},
  \citenamefont {Lundgren}, \citenamefont {Titum}, \citenamefont {Pagano},
  \citenamefont {Zhang}, \citenamefont {Monroe},\ and\ \citenamefont
  {Gorshkov}}]{longrange_metastab_confine}%
  \BibitemOpen
  \bibfield  {author} {\bibinfo {author} {\bibfnamefont {Fangli}\ \bibnamefont
  {Liu}}, \bibinfo {author} {\bibfnamefont {Rex}\ \bibnamefont {Lundgren}},
  \bibinfo {author} {\bibfnamefont {Paraj}\ \bibnamefont {Titum}}, \bibinfo
  {author} {\bibfnamefont {Guido}\ \bibnamefont {Pagano}}, \bibinfo {author}
  {\bibfnamefont {Jiehang}\ \bibnamefont {Zhang}}, \bibinfo {author}
  {\bibfnamefont {Christopher}\ \bibnamefont {Monroe}}, \ and\ \bibinfo
  {author} {\bibfnamefont {Alexey~V.}\ \bibnamefont {Gorshkov}},\ }\bibfield
  {title} {\enquote {\bibinfo {title} {Confined quasiparticle dynamics in
  long-range interacting quantum spin chains},}\ }\href {\doibase
  10.1103/PhysRevLett.122.150601} {\bibfield  {journal} {\bibinfo  {journal}
  {Phys. Rev. Lett.}\ }\textbf {\bibinfo {volume} {122}},\ \bibinfo {pages}
  {150601} (\bibinfo {year} {2019})}\BibitemShut {NoStop}%
\bibitem [{\citenamefont {Defenu}(2021)}]{longrange_metastab_spec}%
  \BibitemOpen
  \bibfield  {author} {\bibinfo {author} {\bibfnamefont {Nicolò}\ \bibnamefont
  {Defenu}},\ }\bibfield  {title} {\enquote {\bibinfo {title} {Metastability
  and discrete spectrum of long-range systems},}\ }\href {\doibase
  10.1073/pnas.2101785118} {\bibfield  {journal} {\bibinfo  {journal}
  {Proceedings of the National Academy of Sciences}\ }\textbf {\bibinfo
  {volume} {118}},\ \bibinfo {pages} {e2101785118} (\bibinfo {year}
  {2021})}\BibitemShut {NoStop}%
\bibitem [{\citenamefont {Giachetti}\ and\ \citenamefont
  {Defenu}(2023)}]{longrange_metastab_violent}%
  \BibitemOpen
  \bibfield  {author} {\bibinfo {author} {\bibfnamefont {Guido}\ \bibnamefont
  {Giachetti}}\ and\ \bibinfo {author} {\bibfnamefont {Nicol\`o}\ \bibnamefont
  {Defenu}},\ }\bibfield  {title} {\enquote {\bibinfo {title} {{On the
  Conditions for a Quantum Violent Relaxation}},}\ }\href@noop {} {\  (\bibinfo
  {year} {2023})},\ \Eprint {http://arxiv.org/abs/2312.14768} {arXiv:2312.14768
  [quant-ph]} \BibitemShut {NoStop}%
\bibitem [{\citenamefont {Chen}\ and\ \citenamefont
  {Lucas}(2019)}]{chen2019finite}%
  \BibitemOpen
  \bibfield  {author} {\bibinfo {author} {\bibfnamefont {Chi-Fang}\
  \bibnamefont {Chen}}\ and\ \bibinfo {author} {\bibfnamefont {Andrew}\
  \bibnamefont {Lucas}},\ }\bibfield  {title} {\enquote {\bibinfo {title}
  {{Finite speed of quantum scrambling with long range interactions}},}\ }\href
  {\doibase 10.1103/PhysRevLett.123.250605} {\bibfield  {journal} {\bibinfo
  {journal} {Phys. Rev. Lett.}\ }\textbf {\bibinfo {volume} {123}},\ \bibinfo
  {pages} {250605} (\bibinfo {year} {2019})},\ \Eprint
  {http://arxiv.org/abs/1907.07637} {arXiv:1907.07637 [quant-ph]} \BibitemShut
  {NoStop}%
\bibitem [{\citenamefont {Kuwahara}\ and\ \citenamefont
  {Saito}(2020)}]{Kuwahara:2019rlw}%
  \BibitemOpen
  \bibfield  {author} {\bibinfo {author} {\bibfnamefont {Tomotaka}\
  \bibnamefont {Kuwahara}}\ and\ \bibinfo {author} {\bibfnamefont {Keiji}\
  \bibnamefont {Saito}},\ }\bibfield  {title} {\enquote {\bibinfo {title}
  {{Strictly linear light cones in long-range interacting systems of arbitrary
  dimensions}},}\ }\href {\doibase 10.1103/PhysRevX.10.031010} {\bibfield
  {journal} {\bibinfo  {journal} {Phys. Rev. X}\ }\textbf {\bibinfo {volume}
  {10}},\ \bibinfo {pages} {031010} (\bibinfo {year} {2020})},\ \Eprint
  {http://arxiv.org/abs/1910.14477} {arXiv:1910.14477 [quant-ph]} \BibitemShut
  {NoStop}%
\bibitem [{\citenamefont {Tran}\ \emph {et~al.}(2021)\citenamefont {Tran},
  \citenamefont {Guo}, \citenamefont {Baldwin}, \citenamefont {Ehrenberg},
  \citenamefont {Gorshkov},\ and\ \citenamefont {Lucas}}]{Tran:2021ogo}%
  \BibitemOpen
  \bibfield  {author} {\bibinfo {author} {\bibfnamefont {Minh~C.}\ \bibnamefont
  {Tran}}, \bibinfo {author} {\bibfnamefont {Andrew~Y.}\ \bibnamefont {Guo}},
  \bibinfo {author} {\bibfnamefont {Christopher~L.}\ \bibnamefont {Baldwin}},
  \bibinfo {author} {\bibfnamefont {Adam}\ \bibnamefont {Ehrenberg}}, \bibinfo
  {author} {\bibfnamefont {Alexey~V.}\ \bibnamefont {Gorshkov}}, \ and\
  \bibinfo {author} {\bibfnamefont {Andrew}\ \bibnamefont {Lucas}},\ }\bibfield
   {title} {\enquote {\bibinfo {title} {{Lieb-Robinson Light Cone for Power-Law
  Interactions}},}\ }\href {\doibase 10.1103/PhysRevLett.127.160401} {\bibfield
   {journal} {\bibinfo  {journal} {Phys. Rev. Lett.}\ }\textbf {\bibinfo
  {volume} {127}},\ \bibinfo {pages} {160401} (\bibinfo {year} {2021})},\
  \Eprint {http://arxiv.org/abs/2103.15828} {arXiv:2103.15828 [quant-ph]}
  \BibitemShut {NoStop}%
\bibitem [{\citenamefont {Machado}\ \emph {et~al.}(2020)\citenamefont
  {Machado}, \citenamefont {Else}, \citenamefont {Kahanamoku-Meyer},
  \citenamefont {Nayak},\ and\ \citenamefont {Yao}}]{Floq_power}%
  \BibitemOpen
  \bibfield  {author} {\bibinfo {author} {\bibfnamefont {Francisco}\
  \bibnamefont {Machado}}, \bibinfo {author} {\bibfnamefont {Dominic~V.}\
  \bibnamefont {Else}}, \bibinfo {author} {\bibfnamefont {Gregory~D.}\
  \bibnamefont {Kahanamoku-Meyer}}, \bibinfo {author} {\bibfnamefont {Chetan}\
  \bibnamefont {Nayak}}, \ and\ \bibinfo {author} {\bibfnamefont {Norman~Y.}\
  \bibnamefont {Yao}},\ }\bibfield  {title} {\enquote {\bibinfo {title}
  {Long-range prethermal phases of nonequilibrium matter},}\ }\href {\doibase
  10.1103/PhysRevX.10.011043} {\bibfield  {journal} {\bibinfo  {journal} {Phys.
  Rev. X}\ }\textbf {\bibinfo {volume} {10}},\ \bibinfo {pages} {011043}
  (\bibinfo {year} {2020})}\BibitemShut {NoStop}%
\bibitem [{\citenamefont {Schuch}\ \emph {et~al.}(2011)\citenamefont {Schuch},
  \citenamefont {Harrison}, \citenamefont {Osborne},\ and\ \citenamefont
  {Eisert}}]{schuch}%
  \BibitemOpen
  \bibfield  {author} {\bibinfo {author} {\bibfnamefont {Norbert}\ \bibnamefont
  {Schuch}}, \bibinfo {author} {\bibfnamefont {Sarah~K.}\ \bibnamefont
  {Harrison}}, \bibinfo {author} {\bibfnamefont {Tobias~J.}\ \bibnamefont
  {Osborne}}, \ and\ \bibinfo {author} {\bibfnamefont {Jens}\ \bibnamefont
  {Eisert}},\ }\bibfield  {title} {\enquote {\bibinfo {title} {Information
  propagation for interacting-particle systems},}\ }\href {\doibase
  10.1103/PhysRevA.84.032309} {\bibfield  {journal} {\bibinfo  {journal} {Phys.
  Rev. A}\ }\textbf {\bibinfo {volume} {84}},\ \bibinfo {pages} {032309}
  (\bibinfo {year} {2011})}\BibitemShut {NoStop}%
\bibitem [{\citenamefont {Kuwahara}\ and\ \citenamefont
  {Saito}(2021)}]{boson_kuwahara}%
  \BibitemOpen
  \bibfield  {author} {\bibinfo {author} {\bibfnamefont {Tomotaka}\
  \bibnamefont {Kuwahara}}\ and\ \bibinfo {author} {\bibfnamefont {Keiji}\
  \bibnamefont {Saito}},\ }\bibfield  {title} {\enquote {\bibinfo {title}
  {Lieb-robinson bound and almost-linear light cone in interacting boson
  systems},}\ }\href {\doibase 10.1103/PhysRevLett.127.070403} {\bibfield
  {journal} {\bibinfo  {journal} {Phys. Rev. Lett.}\ }\textbf {\bibinfo
  {volume} {127}},\ \bibinfo {pages} {070403} (\bibinfo {year}
  {2021})}\BibitemShut {NoStop}%
\bibitem [{\citenamefont {Yin}\ and\ \citenamefont {Lucas}(2022)}]{boson_yin}%
  \BibitemOpen
  \bibfield  {author} {\bibinfo {author} {\bibfnamefont {Chao}\ \bibnamefont
  {Yin}}\ and\ \bibinfo {author} {\bibfnamefont {Andrew}\ \bibnamefont
  {Lucas}},\ }\bibfield  {title} {\enquote {\bibinfo {title} {Finite speed of
  quantum information in models of interacting bosons at finite density},}\
  }\href {\doibase 10.1103/PhysRevX.12.021039} {\bibfield  {journal} {\bibinfo
  {journal} {Phys. Rev. X}\ }\textbf {\bibinfo {volume} {12}},\ \bibinfo
  {pages} {021039} (\bibinfo {year} {2022})}\BibitemShut {NoStop}%
\bibitem [{\citenamefont {Faupin}\ \emph {et~al.}(2022)\citenamefont {Faupin},
  \citenamefont {Lemm},\ and\ \citenamefont {Sigal}}]{boson_lemm}%
  \BibitemOpen
  \bibfield  {author} {\bibinfo {author} {\bibfnamefont {J\'er\'emy}\
  \bibnamefont {Faupin}}, \bibinfo {author} {\bibfnamefont {Marius}\
  \bibnamefont {Lemm}}, \ and\ \bibinfo {author} {\bibfnamefont
  {Israel~Michael}\ \bibnamefont {Sigal}},\ }\bibfield  {title} {\enquote
  {\bibinfo {title} {Maximal speed for macroscopic particle transport in the
  bose-hubbard model},}\ }\href {\doibase 10.1103/PhysRevLett.128.150602}
  {\bibfield  {journal} {\bibinfo  {journal} {Phys. Rev. Lett.}\ }\textbf
  {\bibinfo {volume} {128}},\ \bibinfo {pages} {150602} (\bibinfo {year}
  {2022})}\BibitemShut {NoStop}%
\bibitem [{\citenamefont {Eisert}\ and\ \citenamefont {Gross}(2009)}]{gross}%
  \BibitemOpen
  \bibfield  {author} {\bibinfo {author} {\bibfnamefont {J.}~\bibnamefont
  {Eisert}}\ and\ \bibinfo {author} {\bibfnamefont {D.}~\bibnamefont {Gross}},\
  }\bibfield  {title} {\enquote {\bibinfo {title} {Supersonic quantum
  communication},}\ }\href {\doibase 10.1103/PhysRevLett.102.240501} {\bibfield
   {journal} {\bibinfo  {journal} {Phys. Rev. Lett.}\ }\textbf {\bibinfo
  {volume} {102}},\ \bibinfo {pages} {240501} (\bibinfo {year}
  {2009})}\BibitemShut {NoStop}%
\bibitem [{\citenamefont {Bravyi}\ and\ \citenamefont
  {Gosset}(2015)}]{FFspinchain_classify15}%
  \BibitemOpen
  \bibfield  {author} {\bibinfo {author} {\bibfnamefont {Sergey}\ \bibnamefont
  {Bravyi}}\ and\ \bibinfo {author} {\bibfnamefont {David}\ \bibnamefont
  {Gosset}},\ }\bibfield  {title} {\enquote {\bibinfo {title} {{Gapped and
  gapless phases of frustration-free spin- 1 2 chains}},}\ }\href {\doibase
  10.1063/1.4922508} {\bibfield  {journal} {\bibinfo  {journal} {Journal of
  Mathematical Physics}\ }\textbf {\bibinfo {volume} {56}},\ \bibinfo {pages}
  {061902} (\bibinfo {year} {2015})}\BibitemShut {NoStop}%
\bibitem [{\citenamefont {Arad}\ \emph {et~al.}(2016)\citenamefont {Arad},
  \citenamefont {Kuwahara},\ and\ \citenamefont {Landau}}]{Arad_energy16}%
  \BibitemOpen
  \bibfield  {author} {\bibinfo {author} {\bibfnamefont {Itai}\ \bibnamefont
  {Arad}}, \bibinfo {author} {\bibfnamefont {Tomotaka}\ \bibnamefont
  {Kuwahara}}, \ and\ \bibinfo {author} {\bibfnamefont {Zeph}\ \bibnamefont
  {Landau}},\ }\bibfield  {title} {\enquote {\bibinfo {title} {Connecting
  global and local energy distributions in quantum spin models on a lattice},}\
  }\href {\doibase 10.1088/1742-5468/2016/03/033301} {\bibfield  {journal}
  {\bibinfo  {journal} {Journal of Statistical Mechanics: Theory and
  Experiment}\ }\textbf {\bibinfo {volume} {2016}},\ \bibinfo {pages} {033301}
  (\bibinfo {year} {2016})}\BibitemShut {NoStop}%
\bibitem [{\citenamefont {Yin}\ and\ \citenamefont {Lucas}(2024)}]{yin_ldpc}%
  \BibitemOpen
  \bibfield  {author} {\bibinfo {author} {\bibfnamefont {Chao}\ \bibnamefont
  {Yin}}\ and\ \bibinfo {author} {\bibfnamefont {Andrew}\ \bibnamefont
  {Lucas}},\ }\bibfield  {title} {\enquote {\bibinfo {title} {{Low-density
  parity-check codes as stable phases of quantum matter}},}\ }\href@noop {} {\
  (\bibinfo {year} {2024})},\ \Eprint {http://arxiv.org/abs/2411.01002}
  {arXiv:2411.01002 [quant-ph]} \BibitemShut {NoStop}%
\bibitem [{\citenamefont {Tretter}(2008)}]{gershgorin_book}%
  \BibitemOpen
  \bibfield  {author} {\bibinfo {author} {\bibfnamefont {Christiane}\
  \bibnamefont {Tretter}},\ }\href@noop {} {\emph {\bibinfo {title} {Spectral
  theory of block operator matrices and applications}}}\ (\bibinfo  {publisher}
  {World Scientific},\ \bibinfo {year} {2008})\BibitemShut {NoStop}%
\end{thebibliography}%
